\documentclass[10pt,twoside]{book}

\usepackage[english]{babel}

\usepackage[latin1]{inputenc}
\hoffset - 0cm 
\voffset - 2cm 
\usepackage{pdfpages}

\usepackage{amsthm}
\usepackage{graphicx}  
\usepackage{amsmath,amssymb}
\usepackage[titletoc]{appendix}
\usepackage{dsfont}
\usepackage{xcolor}
\usepackage{afterpage}
\usepackage{amsfonts}
\usepackage{color}
\usepackage[all]{xy}
\usepackage{mathtools}
\usepackage{tcolorbox}
\usepackage[nottoc]{tocbibind}
\usepackage{wasysym}
\usepackage{bbm}
\usepackage[capitalize]{cleveref}
\usepackage{cancel}

\newtheorem{theorem}{Theorem}[chapter]

\newtheorem{lemma}[theorem]{Lemma}
\newtheorem{proposition}[theorem]{Proposition}
\newtheorem{corollary}[theorem]{Corollary}

\newtheorem{definition}{Definition}
\numberwithin{definition}{chapter}
\newtheorem{question}{Question}\newtheorem{remark}{Remark}
\newtheorem{observation}{Observation}
\numberwithin{observation}{chapter}
\newtheorem{conj}{Conjecture}

\newcommand{\tr}{\operatorname{Tr}}
\newcommand{\id}{\mathds{1}}
\newcommand{\myinv}[1]{#1^{\scalebox{0.9}[1.0]{-}1}}
\newcommand{\bed}{\[}
\newcommand{\eed}{\]}
\newcommand{\beq}{\begin{equation}}
\newcommand{\eeq}{\end{equation}}
\newcommand{\ket} [1] {\vert #1 \rangle}
\newcommand{\bra} [1] {\langle #1 \vert}
\newcommand{\N}{\mathbb{N}}
\newcommand{\C}{\mathbb{C}}

\usepackage{tikz}
\usetikzlibrary{3d}
\usetikzlibrary{shapes,snakes}
\usetikzlibrary{decorations.pathreplacing,matrix,arrows.meta,decorations.markings,calc}
\usetikzlibrary{shapes.geometric}

\tikzset{
  every picture/.style = {
    baseline={([yshift=-.5ex]current bounding box.center)}, 
    scale=1.2,
    transform shape,
    font=\scriptsize
  }
}

\tikzset{
  tensor/.pic={
    \begin{scope}[canvas is zx plane at y=0]
      \filldraw (0,0) circle (0.07);
    \end{scope}
    \draw (0,0,0)--(0,0.2,0);
    
  }
}

\tikzset{
  tensord/.pic={
    \begin{scope}[canvas is zx plane at y=0]
      \filldraw (0,0) circle (0.07);
    \end{scope}
    \draw (0,0,0)--(0,-0.2,0);
    
  }
}
\tikzset{
  tensorud/.pic={
    \begin{scope}[canvas is zx plane at y=0]
      \filldraw (0,0) circle (0.07);
    \end{scope}
    \draw (0,0.2,0)--(0,-0.2,0);
    
  }
}

\tikzset{
  3dpeps/.pic={
    \begin{scope}[canvas is zx plane at y=0]
      \draw (-0.5,0)--(0.5,0);
      \draw (0,-0.4)--(0,0.4);
      \filldraw (0,0) circle (0.07);
    \end{scope}
    \draw (0,0,0)--(0,0.2,0);
    
  }
}
\tikzset{
  3dpepsshort/.pic={
    \begin{scope}[canvas is zx plane at y=0]
      \draw[thick] (-0.5,0)--(0.5,0);
      \draw[thick] (0,-0.3)--(0,0.3);
      \filldraw (0,0) circle (0.07);
    \end{scope}
    \draw[thick] (0,0,0)--(0,0.2,0);
    
  }
}

\tikzset{
  3dpepsdownshort/.pic={
    \draw[thick] (0,0,0)--(0,-0.2,0);
    \begin{scope}[canvas is zx plane at y=0]
      \draw[thick] (-0.5,0)--(0.5,0);
      \draw[thick] (0,-0.3)--(0,0.3);
      \filldraw (0,0) circle (0.07);
    \end{scope} 
  }
}
 \tikzset{
  3dpepspshort/.pic={
    \begin{scope}[canvas is zx plane at y=0]
      \draw[thick] (-0.3,0)--(0.3,0);
      \draw[thick] (0,-0.1)--(0,0.1);
      \filldraw[draw=black,fill=blue] (0,0) circle (0.07);
    \end{scope}
  }
}

\tikzset{
  3dpepsdownpshort/.pic={
    \begin{scope}[canvas is zx plane at y=0]
      \draw[thick] (-0.3,0)--(0.3,0);
      \draw[thick] (0,-0.1)--(0,0.1);
      \filldraw[draw=black,fill=blue] (0,0) circle (0.07);
    \end{scope} 
  }
}

\tikzset{
  3dpepsres/.pic={
    \begin{scope}[canvas is zx plane at y=0]
      \draw (-0.5,0)--(0.5,0);
      \draw (0,-0.4)--(0,0.4);
      \filldraw (-0.07,-0.07) rectangle (0.07,0.07); 
    \end{scope}
    \draw (0,0,0)--(0,0.16,0);
  }
}

 \tikzset{
  3dpepspb/.pic={
    \begin{scope}[canvas is zx plane at y=0]
      \draw[thick] (-0.3,0)--(0.3,0);
      \draw[thick] (0,-0.25)--(0,0.25);
      \filldraw[draw=black,fill=black] (0,0) circle (0.07);
    \end{scope}
  }
}

\tikzset{
  3dpepsdownpb/.pic={
    \begin{scope}[canvas is zx plane at y=0]
      \draw[thick] (-0.3,0)--(0.3,0);
      \draw[thick] (0,-0.25)--(0,0.25);
      \filldraw[draw=black,fill=black] (0,0) circle (0.07);
    \end{scope} 
  }
}

\tikzset{
  3dpepsdown/.pic={
    \draw (0,0,0)--(0,-0.2,0);
    \begin{scope}[canvas is zx plane at y=0]
      \draw (-0.5,0)--(0.5,0);
      \draw (0,-0.4)--(0,0.4);
      \filldraw (0,0) circle (0.07);
    \end{scope} 
  }
}
\tikzset{
  hopf/.style={
    line cap=round,
    line width=1.5mm,
  },
  epsilon/.style={
   draw = red,
   thin, 
   rounded corners,
   fill opacity=0.2,
   fill=red,
  },
  every picture/.style = {
    baseline={([yshift=-.5ex]current bounding box.center)}, 
    font=\scriptsize
  }  
}
\tikzset{
  pepsGisopen/.pic={
 \draw (0.07,0.1,0)--(0.07,0,0) --(0.4,0,0);
 \draw (-0.07,0.1,0) --(-0.07,0,0) --(-0.4,0,0);
 \draw (0,0.1,-0.1)--(0,0,-0.1) --(0,0,-0.5);
  \draw (0,0.1,0.1) --(0,0,0.1) --(0,0,0.5);
}}

\tikzset{
  tensor/.style={
    inner sep = 0.055cm,
    shape = circle,
    draw,
    canvas is zx plane at y=0,
    fill
  },
    tensorr/.style={
    inner sep = 0.045cm,
    shape = circle,
    draw,
    fill
  },
    tensorB/.style={
    inner sep = 0.065cm,
    shape = rectangle,
    draw,
    canvas is zx plane at y=0,
    fill
  },
  every picture/.style = {
    baseline={([yshift=-.5ex]current bounding box.center)}, 
    transform shape,
    font=\scriptsize
  }
}

\tikzset{
  every picture/.style = {
    baseline={([yshift=-.5ex]current bounding box.center)}, 
    scale=1.2,
    transform shape,
    font=\scriptsize
  }
}

\tikzset{
  O2/.pic={
\filldraw[rounded corners=.05cm] (-0.4,-0.05) rectangle (0.4, 0.05);
  \draw (-0.3,0.2)--(-0.3,-0.2);
    \draw (0.3,0.2)--(0.3,-0.2);
     }
}
\tikzset{
  O2p/.pic={
\filldraw[rounded corners=.02cm] (-0.2,-0.025) rectangle (0.2, 0.025);
  \draw (-0.15,0.1)--(-0.15,-0.1);
    \draw (0.15,0.1)--(0.15,-0.1);
     }
}
\tikzset{
  O2SVD/.pic={
\draw (-0.35,0)--(0.15,0);     
  \draw (-0.3,0.12)--(-0.3,-0.12);
    \draw (0.1,0.12)--(0.1,-0.12);
      \filldraw[draw=black, fill=blue] (-0.35,-0.05) rectangle (-0.25,0.05);
        \filldraw[draw=black, fill=purple] (0.05,-0.05) rectangle (0.15,0.05);
     }
}

\tikzset{
  3dpeps/.pic={
    \begin{scope}[canvas is zx plane at y=0]
      \draw (-0.5,0)--(0.5,0);
      \draw (0,-0.4)--(0,0.4);
      \filldraw (0,0) circle (0.07);
    \end{scope}
    \draw (0,0,0)--(0,0.2,0);
  }
}
 \tikzset{
  3dpepsp/.pic={
    \begin{scope}[canvas is zx plane at y=0]
      \draw (-0.3,0)--(0.3,0);
      \draw (0,-0.25)--(0,0.25);
      \filldraw[draw=black,fill=blue] (0,0) circle (0.07);
    \end{scope}
  }
}

\tikzset{
  pepsplus/.pic={
      \draw (0,0,-2)--(0,0,2);
      \draw (-2,0,0)--(2,0,0);
      \foreach \x/\z in {0/-1.5,0/0,-1.5/0,1.5/0,0/1.5}{
        \pic at (\x,0,\z) {3dpeps};
      }
  }
}
\tikzset{
  pepsplusres/.pic={
      \draw (0,0,-2)--(0,0,2);
      \draw (-2,0,0)--(2,0,0);
      \foreach \x/\z in {0/-1.5,0/0,-1.5/0,1.5/0,0/1.5}{
        \pic at (\x,0,\z) {3dpepsres};
      }
  }
}
\tikzset{
  pepspluscontrac/.pic={
      \draw (0,0,-2)--(0,0,2);
      \draw (-2,0,0)--(2,0,0);
      \foreach \x/\z in {0/-1.5,0/0,-1.5/0,1.5/0,0/1.5}{
        \pic at (\x,0,\z) {3dpeps};
      }
       \foreach \c/\x in {zy plane at x/-1.5, zy plane at x/1.5,xy plane at z/-1.5,xy plane at z/1.5}{
        \begin{scope}[canvas is \c=\x]
          \draw (-0.5,0) rectangle (0.5,-0.15);
        \end{scope}
      }
  }
}

\tikzset{
  pepspluscontracdown/.pic={
      \draw (0,0,-2)--(0,0,2);
      \draw (-2,0,0)--(2,0,0);
      \foreach \x/\z in {0/-1.5,0/0,-1.5/0,1.5/0,0/1.5}{
        \pic at (\x,0,\z) {3dpepsdown};
      }
       \foreach \c/\x in {zy plane at x/-1.5, zy plane at x/1.5,xy plane at z/-1.5,xy plane at z/1.5}{
        \begin{scope}[canvas is \c=\x]
          \draw (-0.5,0.15) rectangle (0.5,0);
        \end{scope}
      }
  }
}

\tikzset{
  pepsplusrescontrac/.pic={
      \draw (0,0,-2)--(0,0,2);
      \draw (-2,0,0)--(2,0,0);
      \foreach \x/\z in {0/-1.5,0/0,-1.5/0,1.5/0,0/1.5}{
        \pic at (\x,0,\z) {3dpepsres};
      }
       \foreach \c/\x in {zy plane at x/-1.5, zy plane at x/1.5,xy plane at z/-1.5,xy plane at z/1.5}{
        \begin{scope}[canvas is \c=\x]
          \draw (-0.5,0) rectangle (0.5,-0.15);
        \end{scope}
      }
  }
}

\tikzset{
  openrectcirc/.pic={
 \foreach \x in {-0.9,0.9}{
  \draw (-0.2+\x,0.2,0) -- (-0.3+\x,0.2,0) -- (-0.3+\x,0,0) -- (0.3+\x,0,0) -- (0.3+\x,0.2,0) --(0.2+\x,0.2,0);
}
 \foreach \z in {-0.9,0.9}{
          \draw (0,0.2,-0.2+\z) -- (0,0.2,-0.3+\z) -- (0,0,-0.3+\z) -- (0,0,0.3+\z) -- (0,0.2,0.3+\z) --(0,0.2,0.2+\z);
}
        \filldraw[draw=black, fill=black] (-0.9,0,0) circle (0.05);
         \filldraw[draw=black, fill=black] (0.9,0,0) circle (0.05);
          \filldraw[draw=black, fill=black] (0,0,-0.9) circle (0.05);
            \filldraw[draw=black, fill=black] (0,0,0.9) circle (0.05);
  }
}
 
\tikzset{
  openrectcirc1/.pic={

  \draw  (-0.3-0.9,0,0) -- (0.3-0.9,0,0) -- (0.3-0.9,0.2,0) --(0.2-0.9,0.2,0);
    \draw (-0.2+0.9,0.2,0) -- (-0.3+0.9,0.2,0) -- (-0.3+0.9,0,0) -- (0.3+0.9,0,0);

  \draw (0,0.2,-0.2+0.9) -- (0,0.2,-0.3+0.9) -- (0,0,-0.3+0.9) -- (0,0,0.3+0.9);
 
          \draw (0,0,-0.3-0.9) -- (0,0,0.3-0.9) -- (0,0.2,0.3-0.9) --(0,0.2,0.2-0.9);

        \filldraw[draw=black, fill=black] (-0.9,0,0) circle (0.05);
         \filldraw[draw=black, fill=black] (0.9,0,0) circle (0.05);
          \filldraw[draw=black, fill=black] (0,0,-0.9) circle (0.05);
            \filldraw[draw=black, fill=black] (0,0,0.9) circle (0.05);
  }
}

\tikzset{
  pepsplusdown/.pic={
      \draw (0,0,-2)--(0,0,2);
      \draw (-2,0,0)--(2,0,0);
      \foreach \x/\z in {0/-1.5,0/0,-1.5/0,1.5/0,0/1.5}{
        \pic at (\x,0,\z) {3dpepsdown};
      }
  }
}

\tikzset{
PAplus/.pic={
 \foreach \x/\z in {-2/0,-1.5/-0.5,-1.5/0.5, 1.5/0.5,1.5/-0.5,2/0}{
          \draw (\x,-0.15,\z)--(\x,0,\z);
        }
        \pic at (0,-0.15,0) {pepsplusdown};
        \pic at (0,0,0) {pepsplus};
        \foreach \x/\z in {-0.4/1.5, -0.4/-1.5,0/-2, 0/2, 0.4/-1.5, 0.4/1.5}{
          \draw (\x,-0.15,\z)--(\x,0,\z);
        }
 }
}

\tikzset{
  3dpepsdownp/.pic={
    \begin{scope}[canvas is zx plane at y=0]
      \draw (-0.3,0)--(0.3,0);
      \draw (0,-0.25)--(0,0.25);
      \filldraw[draw=black,fill=blue] (0,0) circle (0.07);
    \end{scope} 
  }
}

\tikzset{
  3dGisopeps/.pic={
\draw (-0.1,0,0)--(-0.1,0.4,0);
       \begin{scope}[canvas is zx plane at y=0]
          \draw (0.1,0)--(0.5,0);
        \draw (-0.5,0)--(-0.1,0);
       \filldraw (0.3,0) circle (0.04);
        \filldraw (0,-0.25) circle (0.04); 
    \end{scope} 
    
     \begin{scope}[canvas is zx plane at y=0.4]
     \filldraw (-0.3,0) circle (0.04);
     \filldraw (0,0.25) circle (0.04);
       \draw (-0.5,0)--(-0.1,0);
      \draw[preaction={draw, line width=1pt, white}] (0.1,0)--(0.5,0);
    \end{scope} 
     \draw (0,0,0.1)--(0,0.4,0.1);
\draw (0,0,-0.1)--(0,0.4,-0.1);
       \draw[preaction={draw, line width=1pt, white}] (0.1,0,0)--(0.1,0.4,0);
         \begin{scope}[canvas is zx plane at y=0]
          \draw (0,0.1)--(0,0.4);
       \draw (0,-0.4)--(0,-0.1);
          \end{scope} 
                  \begin{scope}[canvas is zx plane at y=0.4]
          \draw (0,0.1)--(0,0.4);
       \draw (0,-0.4)--(0,-0.1);
          \end{scope} 

  }
}

\tikzset{
  3dGisopepsproj/.pic={
         \filldraw (-0.3,0,0) circle (0.04);
     \filldraw (0,0,0.3) circle (0.04);
       \filldraw (0.3,0,0) circle (0.04);
        \filldraw (0,0,-0.3) circle (0.04); 
\draw (-0.1,0,0)--(-0.1,0.4,0);
       \begin{scope}[canvas is zx plane at y=0]
          \draw (0.1,0)--(0.5,0);
        \draw (-0.5,0)--(-0.1,0);
    \end{scope} 
     \begin{scope}[canvas is zx plane at y=0.4]
       \draw (-0.5,0)--(-0.1,0);
      \draw[preaction={draw, line width=1pt, white}] (0.1,0)--(0.5,0);
    \end{scope} 
     \draw (0,0,0.1)--(0,0.4,0.1);
\draw (0,0,-0.1)--(0,0.4,-0.1);
       \draw[preaction={draw, line width=1pt, white}] (0.1,0,0)--(0.1,0.4,0);
         \begin{scope}[canvas is zx plane at y=0]
          \draw (0,0.1)--(0,0.4);
       \draw (0,-0.4)--(0,-0.1);
          \end{scope} 
                  \begin{scope}[canvas is zx plane at y=0.4]
          \draw (0,0.1)--(0,0.4);
       \draw (0,-0.4)--(0,-0.1);
          \end{scope} 
  }
}

\tikzset{
  3disopeps/.pic={
\draw (-0.1,0,0)--(-0.1,0.4,0);
       \begin{scope}[canvas is zx plane at y=0]
          \draw (0.1,0)--(0.5,0);
        \draw (-0.5,0)--(-0.1,0);
    \end{scope} 
    
     \begin{scope}[canvas is zx plane at y=0.4]
       \draw (-0.5,0)--(-0.1,0);
      \draw[preaction={draw, line width=1pt, white}] (0.1,0)--(0.5,0);
    \end{scope} 
     \draw (0,0,0.1)--(0,0.4,0.1);
\draw (0,0,-0.1)--(0,0.4,-0.1);
       \draw[preaction={draw, line width=1pt, white}] (0.1,0,0)--(0.1,0.4,0);
         \begin{scope}[canvas is zx plane at y=0]
          \draw (0,0.1)--(0,0.4);
       \draw (0,-0.4)--(0,-0.1);
          \end{scope} 
                  \begin{scope}[canvas is zx plane at y=0.4]
          \draw (0,0.1)--(0,0.4);
       \draw (0,-0.4)--(0,-0.1);
          \end{scope} 

  }
}
\tikzset{
  3dpepsshort/.pic={
    \begin{scope}[canvas is zx plane at y=0]
      \draw (-0.5,0)--(0.5,0);
      \draw (0,-0.3)--(0,0.3);
      \filldraw (0,0) circle (0.07);
    \end{scope}
    \draw (0,0,0)--(0,0.2,0);
    
  }
}

\tikzset{
  3dpepsdownshort/.pic={
    \draw (0,0,0)--(0,-0.2,0);
    \begin{scope}[canvas is zx plane at y=0]
      \draw (-0.5,0)--(0.5,0);
      \draw (0,-0.3)--(0,0.3);
      \filldraw (0,0) circle (0.07);
    \end{scope} 
  }
}

\begin{document}

\thispagestyle{empty}

\begin{center}

{\huge \textbf{Symmetries in topological tensor network states: classification, construction and detection}}\\[1cm]

{\LARGE \textbf{ Simetr\'ias en estados de redes tensoriales topol\'ogicas: clasificaci\'on, construcci\'on y detecci\'on}}\\[1cm]

\begin{figure}[ht!]
\begin{center}
\includegraphics[scale=0.15]{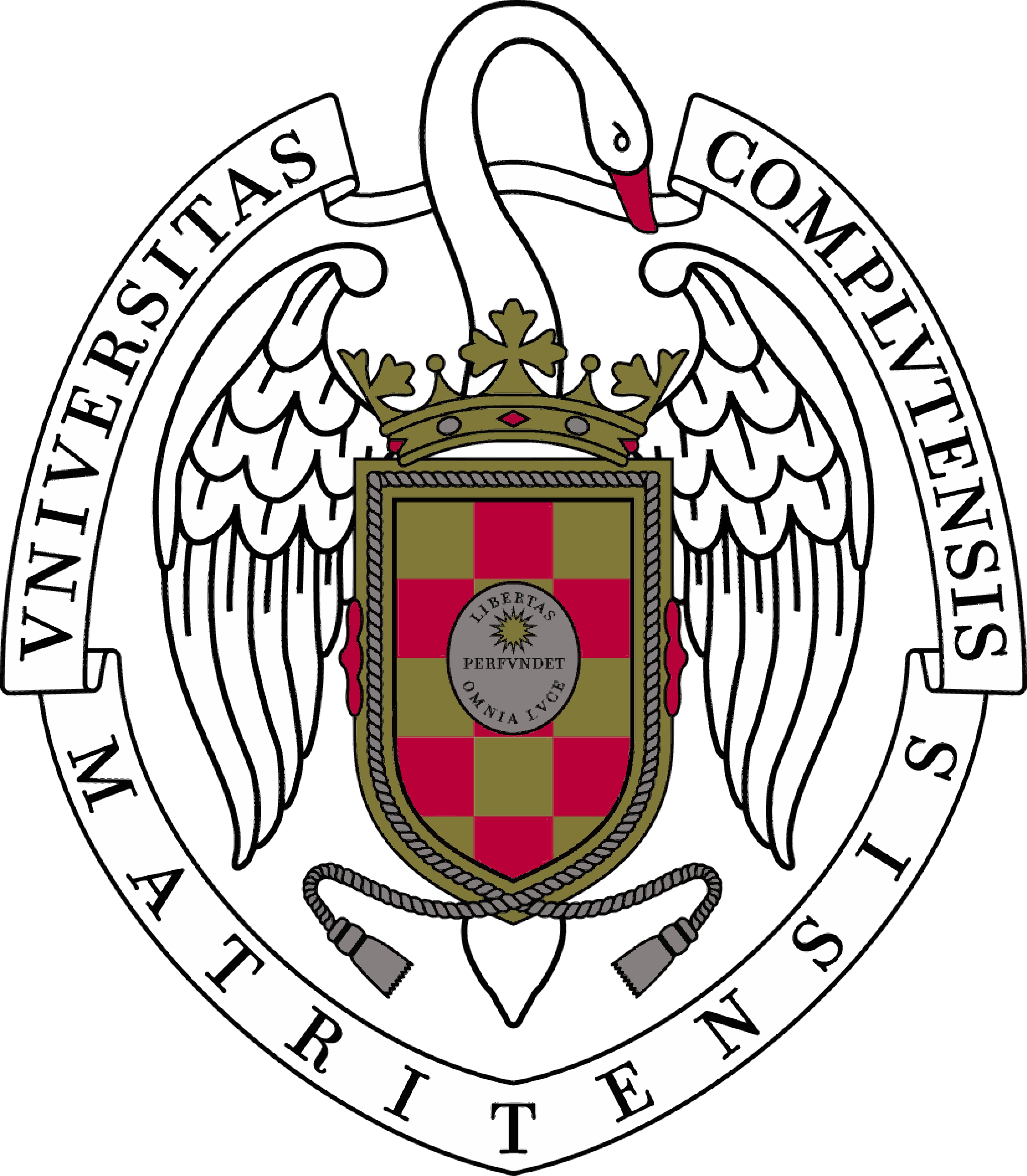}
\end{center}
\end{figure}

\textbf{{\LARGE UNIVERSIDAD COMPLUTENSE DE MADRID}}\\[0.3cm]
\textbf{{\LARGE Facultad de Ciencias Matem\'aticas}\\[0.3cm]
{\Large Departamento de An\'alisis y Matem\'atica Aplicada}}\\[2cm]

\hrule 
\vspace{25px}
 {\large A thesis submitted in fulfillment of the requirements \\ for the degree of Doctor of Philosophy by}\\[1cm]
{\Large \textbf{Jos\'e Garre Rubio}}\\[1cm]
{\large Under the supervision of} \\[0.4cm]
{\large David P\'erez Garc\'ia}\\[0.1cm]
{\large $\&$}\\[0.1cm]
{\large  Sofyan Iblisdir}\\[0.8cm]
\end{center}

\afterpage{\null\newpage}

\thispagestyle{empty}

\newpage\phantom{blabla}

\begin{flushright}
 \small {\it Estando el tío José 'el nano' segando con su Miguel, le dice José \\
 a su hermano: "Esta es mucha siega pa' los tres \\
  y aquí se nos va el verano".}\\
 -Refranero cartagenero; adaptación de José Garre Nieto.
\end{flushright}

%
%


\newpage\phantom{blabla}


\newpage

\section*{Abstract}

The exotic properties which appear in the quantum setting, mainly manifested in strongly-correlated systems, offer potential applications in future technologies. For instance, high-precision measurements or the new paradigm of a quantum computer.

One of the most prominent features of quantum physics is \emph{entanglement}: the correlations between the \emph{parties} of a system that cannot be described classically. This property is believed to be the one endowing quantum mechanics its complexity. Therefore, characterizing the entanglement properties of strongly-correlated systems plays a fundamental role for condensed-matter physics. However, this complexity comes hand in hand with a challenge: the number of parameters needed to describe a system grows exponentially with the number of parties in the system. This challenge lies at the heart of the mathematical description of quantum mechanics. Such a situation happens naturally in many-body systems and in particular, in condensed-matter physics where the relevant physics appears when considering large systems. Then, how can we deal with this difficulty? The key observation here is that realistic physical systems are the ones whose parties interact locally and this restricts the entanglement pattern in the low-energy sector (zero temperature). So the question is shifted to: Is there a framework that captures states with such entanglement pattern? The answer is yes: \emph{tensor network states}. 

Tensor network states describe many-body quantum states locally. This local description reflects the ubiquitous effort in physics to describe global (mesoscopic) properties in terms of local constituents. 
Despite their apparent simplicity, tensor network states are capable of approximating the entanglement pattern present in the low-energy sector of locally-interacting systems, as well as of precisely describing particularly relevant many-body quantum states (those which are fixed points of suitable renormalization processes). 

The fact that they are defined by a small number of parameters, all the while describing the relevant set of states, makes them highly suitable as a variational \emph{ansatz} for numerical calculations. 

Besides their initial numerical motivation, tensor network states also constitute a very valuable tool to study strongly correlated systems analytically. In this thesis, entitled 'Symmetries in topological tensor network states: classification, construction and detection' we focus on this line.

In particular, we analyze tensor networks that describe quantum phases which have been identified to possess a new type of non-local ordering: \emph{topological} order. One of the most important features of a system with topological order is the existence of quasiparticle excitations with unusual exchange properties which differ from regular particles such as bosons (like photons) or fermions (like electrons). These excitations, called \emph{anyons}, emerge as the collective behaviour of the constituents of the system and can transform non-trivially under global symmetries present in the system. As a paradigmatic example, the anyons of the fractional quantum Hall effect can have a fractional electric charge of an electron. Another example are spin liquid systems where the quasiparticles are chargeons and spinons, which can be seen as particles coming from the splitting of the electron (the electric charge plus the spin). That is, the symmetry charge has been \emph{fractionalized}.

It turns out that the low-energy sector of some topological models, including anyons, can be described in terms of tensor network states. Moreover, the topological order, a global property, is encoded locally: a symmetry of the local tensors. Such encodings are crucial since once the local structure that characterizes a global property is identified mathematically, one can focus on that family of tensor network states to systematically study different properties. In this sense, tensor network states arise as the formal framework to work in many-body physics. In particular, they have succeeded in rigorously classifying (gapped) quantum phases in one-dimensional systems under symmetries (operations that leave a system invariant). The goal of this thesis is to contribute to the classification of 2D topological tensor network states under symmetries.

An important problem, related to the local characterization of global symmetries, is how to characterize the relationship between tensors that describe the same tensor network state (and hence should be considered as 'equivalent'). Results clarifying this issue are usually entitled 'Fundamental Theorems' due to their profound implications. 

In this thesis, we prove such theorems for some previously studied classes of tensor network states  --injective and normal-- improving the existing results by relaxing some hypotheses. We also state a fundamental theorem for tensor network states describing topological phases, where no previous results were known as such. Specifically, we do it for the family of quantum double models of a finite group, the so-called $G$-injective PEPS. Once we have this theorem, global symmetries that act on-site (acting as a tensor product in each site of the network) can be characterized and classified. 

We classify global on-site symmetries, coming from a finite group $Q$, in $G$-injective PEPS obtaining a finite number of classes for each pair $(Q,G)$. This classification is related to the possible group extensions, $E$, of $G$ by $Q$. Moreover, we provide a method to construct a representative of each class, concluding that our classification is complete. The representatives are also constructed using the theory of group extensions, namely starting with a representation of the extension $E$ and then restricting the tensors to some (related) representation of $G$ locally. This representative is given in a renormalization group fixed point form, which facilitates the extraction of all properties of the phase. Since a general symmetric topological phase could be away from this representative point, a method to detect the class is required. We solve this problem by proposing an order parameter to detect the fractionalization of the symmetry (charge) in a given state which has been elusive in previous studies. The order parameter captures the symmetry class via the detection of an invariant quantity of the extension group. We conclude our thesis by touching on  some mathematical open problems in PEPS.

\newpage
\section*{Resumen}

Las inusuales propiedades que surgen en la mec\'anica cu\'antica, manifestadas especialmente en sistemas fuertemente correlacionados, ofrecen aplicaciones potenciales a nuevas tecnolog\'ias. Por ejemplo, mediciones de alta precisi\'on o la creaci\'on un nuevo modelo de computaci\'on (el ordenador cu\'antico).

Una de las propiedades m\'as destacadas de la f\'isica cu\'antica es el \emph{entrelazamiento}: las correlaciones entre las \emph{partes} de un sistema que no se pueden describir cl\'asicamente. Se cree que esta propiedad es la que dota a la mec\'anica cu\'antica de su complejidad. Por lo tanto, caracterizar el entrelazamiento de sistemas fuertemente correlacionados desempeña un papel fundamental en la f\'isica de la materia condensada. Pero esta complejidad viene unida a una dificultad: el n\'umero de par\'ametros necesario para describir un sistema crece exponencialmente con el n\'umero de partes en el sistema. Este reto est\'a en el coraz\'on de la descripci\'on matem\'atica de la mec\'anica cu\'antica. Esto ocurre de forma natural en sistemas de muchos cuerpos y en particular en f\'isica de la materia condensada, donde la f\'isica relevante aparece cuando el tamaño del sistema es grande. Entonces, ¿c\'omo podemos lidiar con esta dificultad? La observaci\'on clave es que los sistemas f\'isicos reales son aquellos en los que las partes interaccionan  localmente y esto restringe el patr\'on de entrelazamiento en el sector de baja energ\'ia de los sistemas. Por lo que ahora la cuesti\'on es: ¿hay un formalismo que capture los estados con ese tipo de patr\'on de entrelazamiento? La respuesta es s\'i: \emph{los estados de redes tensoriales}.

Los estados de redes tensoriales describen estados cu\'anticos de muchos cuerpos de forma local. Esta descripci\'on local refleja el esfuerzo 
ubicuo en la f\'isica de describir propiedades globales (mesosc\'opicas) en t\'erminos de constituyentes locales.
A pesar de su aparente simplicidad, los estados de redes tensoriales son capaces de aproximar el patr\'on de entrelazamiento en el sector de baja energ\'ia de sistemas que interaccionan localmente, as\'i como describir exactamente estados cu\'anticos de muchos cuerpos de enorme importancia (aquellos que son puntos fijos de un proceso de renormalizaci\'on).
El hecho de que est\'en definidos con un n\'umero pequeño de par\'ametros, describiendo al mismo tiempo el conjunto de estados relevantes, los hace adecuados como una soluci\'on variacional para c\'alculos num\'ericos.

A pesar de su motivaci\'on inicialmente num\'erica, los estados de redes tensoriales constituyen una herramienta anal\'itica muy valiosa para el estudio de sistemas fuertemente correlacionados. En esta tesis, con t\'itulo 'Simetr\'ias en estados de redes tensoriales topol\'ogicas: clasificaci\'on, construcci\'on y detecci\'on', nos centramos en esta l\'inea.

En particular, analizamos redes tensoriales que describen fases cu\'anticas que albergan un nuevo tipo de order no local: \emph{el orden topol\'ogico}. Uno de los aspectos m\'as importantes de un sistema con orden topol\'ogico es la existencia de excitaciones de cuasipart\'iculas con propiedades inusuales bajo intercambio que difieren tanto de los bosones (como fotones) como de los fermiones (como electrones). Estas excitaciones, denominadas anyones, emergen como un comportamiento colectivo de los constituyentes del sistema y pueden transformase de forma no trivial bajo las simetr\'ias globales del sistema. Como ejemplo paradigm\'atico, los anyones del efecto Hall cu\'antico fraccionario tienen una carga  que es una fracci\'on de la del electr\'on. Otro ejemplo son los sistemas de l\'iquidos de esp\'in cuyas cuasipart\'iculas son \emph{chargeons} y \emph{spinons} que pueden interpretarse como resultado de la separaci\'on de un electr\'on (carga el\'ectrica m\'as esp\'in). En este caso, la simetr\'ia se ha \emph{fraccionalizado}.

Resulta que el sector de baja energ\'ia de algunos modelos topol\'ogicos, incluidos los anyones, tienen una descripci\'on en t\'erminos de estados de redes de tensores. Adem\'as, el orden topol\'ogico, una propiedad global, se codifica localmente como una simetr\'ia en los tensores individuales. Esta codificaci\'on es crucial debido a que cuando una estructura local que caracteriza una propiedad local es identificada matem\'aticamente, uno puede centrarse en esa familia de estados de redes tensoriales para estudiar sistem\'aticamente sus propiedades. En este sentido, los estados de redes tensoriales surgen como el marco formal para trabajar en la f\'isica de muchos cuerpos. En particular, han conseguido clasificar rigurosamente las fases cu\'anticas (con \emph{gap}) en sistemas unidimensional con simetr\'ias (operaciones que dejan invariante el estado). El objetivo de esta tesis es contribuir a la clasificaci\'on de redes de tensores topol\'ogicas en dos dimensiones con simetr\'ias.

Una cuesti\'on importante, relacionada con la caracterizaci\'on local de simetr\'ias globales, es c\'omo caracterizar la relaci\'on que existe entre tensores que describen el mismo estado de red tensorial y que, por ende, deber\'ian considerarse equivalentes. Los resultados que aclaran esta cuesti\'on han sido denominados 'Teoremas Fundamentales' debido a sus profundas implicaciones.

En esta tesis probamos dichos teoremas para algunas clases de estados de redes de tensores previamente estudiadas, inyectivos y normales, mejorando los resultados existentes relajando algunas hip\'otesis. Adem\'as, probamos un teorema fundamental para estados de redes tensoriales que describen fases tipol\'ogicas, donde ning\'un resultado previo exist\'ia. Concretamente, lo probamos para la familia que describe los modelos dobles cu\'anticos de un grupo finito, los denominados PEPS $G$-inyectivos. Una vez tenemos este teorema, se pueden caracterizar y clasificar simetr\'ias globales que act\'uan localmente en cada sitio de la red.

Con la caracterización anterior, clasificamos simetr\'ias globales que act\'uan de forma local, provenientes de un grupo finito $Q$, en PEPS $G$-inyectivos. Obtenemos un n\'umero finito de clases para cada par $(Q,G)$. Esta clasificaci\'on est\'a relacionada con las posibles extensiones de grupo, $E$, de $G$ por $Q$. Adem\'as, proporcionamos un m\'etodo para construir un representante de cada clase, concluyendo que nuestra clasificaci\'on es completa. Los representantes tambi\'en se construyen usando la teor\'ia de extensiones de grupo, empezando por una representaci\'on de la extensi\'on $E$ y restringiendo localmente el tensor a los elementos de $G$. Este representante est\'a dado en una forma de punto fijo de renormalizaci\'on, lo que facilita la obtenci\'on de todas las propiedades de la fase cu\'antica. Debido a que un sistema con orden topol\'ogico y simetr\'ia estar\'a lejos de ese representante, es necesario un m\'etodo para detectar la fase. Nosotros resolvemos ese problema proponiendo unos invariantes de la fase y sus correspondientes par\'ametros de orden para detectar el patr\'on de fraccionalizaci\'on de la simetr\'ia en un estado determinado. Los par\'ametros de orden propuestos capturan la fase mediante la detecci\'on de cantidades invariantes en la extensi\'on de grupo. Concluimos  la tesis formulando algunas preguntas matemáticas abiertas en el campo de los PEPS.

\newpage

{\Large \textbf{Structure of the thesis}}\\

This thesis has been structured as follows: the first introductory chapter sets the basic formalism used in this work (the mathematics of quantum physics and tensor networks states). This introduction is based on the article \cite{Cirac19A}. The others chapters are meant to present the results of this thesis. We have included a summary of the results at the beginning of each chapter to facilitate the readability of the manuscript.  In addition, at the end of each chapter we discuss the results and the conclusions.

\

The second chapter states and proves the (Fundamental) theorems that characterize the tensors of two networks that describe the same state. The chapter includes the fundamental theorem for injective and normal PEPS in any dimension and geometry, based on the article \cite{Molnar18A}, and the fundamental theorem for $G$-injective PEPS based on the unpublished article \cite{Molnarinprep}. The third chapter exposes the classification of symmetries in $G$-injective PEPS by applying the fundamental theorem. It also includes a study of gauging and domain walls in these models, these results are part of the unpublished article \cite{Molnarinprep}. The fourth chapter shows how to construct a representative of each phase in the developed classification of symmetric $G$-injective PEPS. It also includes an application to MPS invariant under symmetries; these results are based on the article \cite{Garre17}. In the fifth chapter, we present a local method to identify the quantum phase of a given model, based on \cite{Garre19}. Finally, Chapter 6 outlines a review on mathematical open problems in PEPS, which is part of \cite{Cirac19A}.\\

The last chapter covers the conclusions, open questions and future work obtained from our results.

\

The author acknowledge financial support from MINECO (grant MTM2014- 54240-P), from Comunidad de Madrid (grant QUITEMAD+- CM, ref. S2013/ICE-2801), and the European Research Council (ERC) under the European Union's Horizon 2020 research and innovation programme (grant agreement No 648913). This work has been also partially supported by ICMAT Severo Ochoa project SEV-2015-0554 (MINECO).

\

{\Large \textbf{Publications}}

\begin{itemize}
\item \cite{Garre17} J. Garre-Rubio, D. P\'erez-Garc\'ia, S. Iblisdir: {\it Symmetry reduction induced by anyon condensation: a tensor network approach, Phys. Rev. B 96, 155123, 16 October 2017}

\item  \cite{Molnar18A} A. Molnar, J. Garre-Rubio, D. P\'erez-Garc\'ia, N. Schuch, J. I. Cirac: {\it Normal projected entangled pair states generating the same state, New Journal of Physics, 20, November 2018}

\item  \cite{Cirac19A} J. I. Cirac, J. Garre-Rubio, D. P\'erez-Garc\'ia : {\it Mathematical open problems in Projected Entangled Pair States}, Rev Mat Complut (2019) 32: 579. https://doi.org/10.1007/s13163-019-00318-x

\item \cite{Garre19} J. Garre-Rubio, S. Iblisdir: {\it Local order parameters for symmetry fractionalization}, New Journal of Physics, Volume 21, November 2019.  https://doi.org/10.1088/1367-2630/ab4fff
\end{itemize}
\

{\Large \textbf{In preparation}}

\begin{itemize}
 \item \cite{Molnarinprep} A. Molnar, J. Garre-Rubio, D. P\'erez-Garc\'ia, N. Schuch, J. I. Cirac: {\it In preparation}
\end{itemize}

\newpage

{\Large \textbf{Contributions}}\\
\begin{enumerate}

\item Talk 'Symmetries through anyon condensation with tensor networks' on 16/02/17 in the workshop 'Entanglement in Strongly Correlated Systems' at Centro de Ciencias de Benasque Pedro Pascual in Huesca, Spain.

\item Talk 'Symmetry reduction induced by anyon condensation: a tensor network approach' on 26/04/17 in the MathQI Seminar at UCM, Madrid, Spain.

\item Talk 'Symmetry reduction induced by anyon condensation: a tensor network approach' on 18/04/18 in the Special Seminar at the Max Planck Institute of Quantum Optics in Munich, Germany.

\item Talk 'Normal projected entangled pair states generating the same state' on 23/05/18 in the MathQI Seminar at UCM, Madrid, Spain.

\item Talk 'Symmetry reduction induced by anyon condensation: a TN approach' on 12/06/18 in the workshop 'Symposium on Quantum Matter' at University of Zurich, Switzerland. 

\item Poster 'An order parameter for symmetry fractionalization with projected entangled pair states' on 21 - 25/01/19 in the workshop 'Anyons in Quantum Many-Body Systems' at Max Planck Institute for the Physics of Complex Systems in Dresden, Germany.

\item Talk Benasque 2019, 'Symmetry fractionalization detection without dimensional compactification using PEPS' on 1/03/19 in the workshop 'Entanglement in Strongly Correlated Systems' at Centro de Ciencias de Benasque Pedro Pascual in Huesca, Spain.

\item Talk 'Symmetry obstructions to anyon condensation' on 22/05/19 in the workshop 'Bringing Young Mathematicians Together' at ICMAT, Madrid/Spain.

\end{enumerate}

{\bf Research visit:} Max Planck Institute of Quantum Optics in Munich, Germany hosted by Andras Molnar and Norbert Schuch.  3 weeks visit (from 9 to 27 April 2018).\\

\newpage
{\Large \textbf{List of acronyms}}\\

\begin{itemize}
\item [] AKLT: Affleck, Lieb, Kennedy and Tasaki
\item [] $\mathcal{D}(G)$: quantum Double model of the group $G$
\item [] GHZ: Greenberger-Horne-Zeilinger
\item [] irrep(s): irreducible representation(s)
\item [] LHS: Left Hand Side
\item [] MPS: Matrix Product State
\item [] PEPS: Projected Entangled Pair States
\item [] RG: Renormalization Group
\item [] RHS: Right Hand Side
\item [] SD: Schmidt decomposition
\item [] SET: Symmetry Enriched Topological
\item [] SF: Symmetry Frationalization
\item [] SPT: Symmetry Protected Topological
\item [] SVD: Singular Value Decomposition
\item [] TC: Toric Code
\item [] TN: Tensor Network
\end{itemize}

\tableofcontents 

\newpage

\pagenumbering{arabic}

\chapter{Introduction}\label{chapter:Intro}

\section{Basics notions on Quantum Physics}
Quantum mechanics were developed in the mid-20s and, since then, it has been the underlying description of any physical model, perhaps with the exception of gravitation. 
The manifestation of quantum mechanics as such universal framework came from its formulation as four axioms or postulates. These postulates provide a mathematical framework that has to be followed by any physical setup. The postulates, for bosonic systems and non-relativistic theories are the following-see \cite{NielsenChuang10} for a more complete description-:\\

\paragraph{Systems and states.} This postulate provides the frame where the physical objects are placed; a system is described by a Hilbert space $\mathcal{H}$, which we suppose here always finite-dimensional, so $\mathcal{H}=\mathbb{C}^d$, and a state is represented as a unit vector in that space.
We will use the standard notation in quantum mechanics, introduced by Dirac \cite{Dirac39} and named  {\it bra-ket notation}, where (column) vectors $v \in \mathcal{H}$ are denoted as $\ket{v}$. The scalar product between $\ket{u}$ and $\ket{v}$ is written as $\langle u |v \rangle$ which is justified by Riesz's theorem that establishes that any linear functional $F:\mathcal{H}\to \mathbb{C}$ is given by the scalar product $\langle \cdot ,\cdot \rangle$ with a vector $f\in \mathcal{H}$; $F [u]= \langle f,u\rangle$. Then, the \emph{bra} $\bra{u}$ represents the dual vector of the \emph{ket} $\ket{u}$ and a rank-one operator adopts the form of the product $\ket{u}\bra{v}$.

The simplest case of a system, but extremely relevant, is $\mathcal{H}=\mathbb{C}^{2}$. It is known as a qubit system, which is the quantum analog of a bit, and the canonical basis is usually denoted as $\{ \ket{0},\ket{1}\}$. Instances of qubit systems are the spin of an electron or the polarization of a photon. Quantum mechanics also allows for the description of not completely known states; it is a probabilistic theory. These states are represented as \emph{density} matrices $\rho= \sum_i p_i \ket{\psi_i} \bra{\psi_i}$, called mixed states, where $p_i \ge 0$ represents the probability for the system to be in the state $\ket{\psi_i}$ (if the $\ket{\psi_i}$ are mutually orthogonal) so they fullfil $\sum_i p_i=1$. The (pure) state $\ket{\psi}$ is just represented as $\ket{\psi}\bra{\psi}$ in the density matrix formalism. 

\paragraph{Measurements.} This postulate describes the way in which quantum measurements are implemented and how they affect the measured system. The magnitude to measure is represented by a hermitian operator $\mathcal{O}$, called observable, in the case of projective measurements (see Ref. \cite{NielsenChuang10} for the general case). The average value or the expectation value of $\mathcal{O}$ in the system described by the mixed state $\rho$ is $\langle \mathcal{O} \rangle_\rho=\tr[\mathcal{O}\rho]$ (which coincides with  $\bra{\psi} \mathcal{O} \ket{\psi}$ for a pure state $\ket{\psi}$).

\paragraph{Multiple systems.} The space associated to a composite system is mathematically represented by the tensor product of the components; $\mathcal{H}=\mathcal{H}_1\otimes\mathcal{H}_2 \otimes \cdots \otimes \mathcal{H}_N$. If we have $N$ systems, each of them in the state $\ket{\psi_i}\in \mathcal{H}_i$, the global state is $\ket{\psi_1}\otimes\ket{\psi_2}\otimes \cdots \otimes \ket{\psi_N}\in \mathcal{H}$. However, there are also states that cannot be written in a tensor product form, these are called entangled states. Let us consider an example: a two qubit system $\mathcal{H}=\mathbb{C}^{2}\otimes\mathbb{C}^{2}$  in the state $\ket{\phi}=(\ket{0}\otimes\ket{1}+\ket{1}\otimes \ket{0})/\sqrt{2}$ cannot be written as $\ket{\phi}= \ket{a}\otimes\ket{b}$. This property is known as entanglement and it is believed to be the one endowing quantum mechanics its complexity. We will simplify the notation in the tensor products writing $\ket{a}\ket{b}$ or $\ket{ab}$ instead of $\ket{a}\otimes\ket{b}$ so $\ket{\phi}\equiv (\ket{01}+\ket{10})/\sqrt{2}$.  \\

\paragraph{Evolution.} A quantum system changes with time according to a unitary transformation: $\ket{\psi(t_1)}=U(t_1,t_0)\ket{\psi(t_0)}$. The infinitesimal form of such evolutions is described by the Schr\"{o}dinger equation:
\begin{equation}
i\hbar \frac{d\ket{\psi}}{dt }=H\ket{\psi},
\end{equation}
where $\hbar$ is Planck's constant and $H$ is the self-adjoint operator known as the Hamiltonian of the system.
 The Hamiltonian is the observable that measures the energy of the system. Since we are dealing with a finite-dimensional Hilbert space, $H$ admits a discrete spectral decomposition, $H=\sum_i E_i \ket{e_i}\bra{e_i}$ where the eigenvectors $\ket{e_i}$ are called energy eigenstates and $E_i$ is the energy of $\ket{e_i}$.
 
 The eigenstate corresponding to the smallest eigenvalue (energy), $E_0$, is known as the ground state (GS) and the other eigenstates are called excited states. The difference between the two smallest eigenvalues (energy levels) is known as spectral gap, or just gap, of the Hamiltonian and plays a fundamental role in many problems.\\

The Hamiltonian plays then a fundamental role in the description of a system. But its study encounters two main difficulties. On the one hand, this operator has to be deduced from the physics of the problem --the interaction between the parties among other considerations-- which is not a simple task. On top of that, the Hamiltonian obtained in this way would be in general very complex. To simplify the task, effective Hamiltonians are defined that aim to capture the relevant features of the system. Then, effective models are proposed to describe the low energy sector of the problem,  where the relevant quantum behaviors are expected to appear. On the other hand, when one has the effective model, it has to be solved. That is, the ground state and the low-energy excitations of the Hamiltonian have to be found  together with their energies. Since we are interested in many body physics, that is, when a large number of parties is considered, we have to solve the problem in a \emph{huge} Hilbert space. Specifically, the dimension of the total Hilbert space grows exponentially with the number of parties in it because of the inherent tensor product structure. This means that to describe any (entangled) state an exponential number of parameters is needed, which makes the task intractable. For example $N$ qubit systems are described by the Hilbert space ${\mathbb{C}^{2}}^{ \otimes N}$, so the dimension of the full space is $2^N$. But the naive fact that the full Hilbert space of \emph{any} quantum system grows exponentially with the number of parties is not an unavoidable obstruction as we shall see in the next subsection.

\section{Setup}

The systems we will consider are placed on one-dimensional or two-dimensional finite size lattices $\Lambda$ where each vertex $v\in\Lambda$ represents a subsystem. The Hilbert space of each subsystem $\mathcal{H}_v$ is finite dimensional and isomorphic to $\mathbb{C}^{d_v}$ for some $d_v \in \mathbb{N}$. The total Hilbert space is thus:
$$\mathcal{H}_\Lambda=\bigotimes_{v\in \Lambda} \mathcal{H}_v. $$
We will focus on square lattices, so the one-dimensional cases are just segments or rings of length $L$ for open boundary conditions or periodic boundary conditions respectively. In 2D, we will consider an $L\times L$ square lattice. Therefore, for periodic boundary conditions  the system is placed on a torus ($L$ will correspond to the lattice size). 

The main assumption we will impose is that the interactions of the Hamiltonian are \emph{local}. This is motivated by the physical nature of the interactions:
\begin{definition} \label{def:localhamil}
An operator $H$ is a locally interacting Hamiltonian if it can be written as
$$H=\sum_i \mathfrak{h}_i\otimes \id_{\rm rest},$$ 
where $\mathfrak{h}_i$ acts only in $\bigotimes_{v\in \Omega_i } \mathcal{H}_v$ and $\Omega_i$ is a connected sublattice of $\Lambda$ with $|\Omega_i|\le C$( $C$ a constant independent of $i$ and $|\Lambda|$).
\end{definition}

We will further assume that the system is translationally invariant, meaning that $\mathcal{H}_v=\mathcal{H}_{v'}$ for all $v,v' \in \Lambda$, $\mathcal{H}_\Lambda= \mathcal{H}^{\otimes |\Lambda|}_v$, where $|\Lambda|$ is the total number of vertices, and the local terms $h\equiv h_i$ of $H$ are the same operators acting on translated sublattices. This implies that a given interaction $h$ defines the Hamiltonian $H$ for any lattice size.  

This allows us to define the limit when the system size grows to infinity (usually called thermodynamic limit). In particular, we can define the key notion of {\it gapped Hamiltonians}.

\begin{definition} \label{gapH} 
A family of Hamiltonians $H^{[L]}$ is gapped, where $L$ denotes the system size, if 
$$\mathcal{G}:=\liminf_{L\rightarrow \infty} \left(E_1^{[L]}-E_0^{[L]}\right)>0.$$
If this is the case $\mathcal{G}$ is called the gap of the system. 
\end{definition}
Note that for a finite Hilbert space the spectrum is discrete and then gapped, so the relevant information is how the gap behaves when the system size (and hence the Hilbert space dimension) grows to infinity.\\

The key observation here, proven rigorously in the 1D case, is that generally the ground states of locally interacting gapped Hamiltonians have a very restrictive pattern of entanglement -see \cite{Eisert10} for a review. The states satisfying this pattern correspond to the subset of the full Hilbert space we are interested in. To describe and characterize this pattern, let us introduce a measure of entanglement called entanglement entropy. Given a state $\ket{\psi}\in \mathcal{H}_{\mathcal{R}}\otimes \mathcal{H}_{{\mathcal{R}}^c}$ the reduced density matrix of the subsystem ${\mathcal{R}}\subset \Lambda$ is defined as the partial trace on the complementary of ${\mathcal{R}}$ in $\ket{\psi}$:
$$\rho_{\mathcal{R}}=\tr_{{\mathcal{R}}^c}\left[ \ket{\psi}\bra{\psi}\right],$$
where the partial trace is defined as the unique linear map fulfilling 
$$\tr_{{\mathcal{R}}^c}\left[ \ket{r_i}\bra{r_j} \otimes \ket{u}\bra{v} \right]= \ket{r_i}\bra{r_j} \langle v\ket{u}$$
for $\ket{r_i},\ket{r_j}\in \mathcal{H}_{\mathcal{R}}$ and $\ket{u},\ket{v}\in \mathcal{H}_{\mathcal{R}^{c}}$. The entanglement entropy of the subsystem ${\mathcal{R}}$ is defined as follows
$$S_{\mathcal{R}}( \ket{\psi})\equiv S_{VN}(\rho_{\mathcal{R}})= -\tr[\rho_{\mathcal{R}}\log(\rho_{\mathcal{R}})],$$
where $S_{VN}$ is the von Neumann entropy. We will now recall some basic properties of the entanglement entropy that can be found in e.g. \cite{NielsenChuang10}. The entanglement entropy for ${\mathcal{R}}$ is equal to that of ${\mathcal{R}}^c$ and  for product states, i.e states that can be written as $ \ket{\psi}= \ket{\phi_{\mathcal{R}}}\otimes  \ket{\sigma_{{\mathcal{R}}^c}}$,  it is zero. One important point is that the entanglement entropy is bounded by the logarithm of the dimension of the Hilbert space where ${\mathcal{R}}$ lives, $S_{\mathcal{R}}( \ket{\psi})\le \log|\mathcal{H}_{\mathcal{R}}|\propto |{\mathcal{R}}|$. In fact this maximum rate, a scaling with the volume, is the typical behaviour of a random state \cite{Hayden06}. But for ground states of locally interacting gapped Hamiltonians the entanglement entropy of a subsystem is expected to scale as the boundary of the region:
$$S_{\mathcal{R}}( \ket{\psi})\propto \log|\mathcal{H}_{\partial {\mathcal{R}}}|\propto |\partial {\mathcal{R}}|.$$
This is known as the Area Law Conjecture. It has been proven for one dimensional systems \cite{Hastings07A}, \cite{Arad13} and for some higher dimensional cases \cite{Hamza09,Masanes09}.

The area law seems to be the characteristic property of ground states of locally interacting gapped Hamiltonian so the following question arises naturally: does there exist a tractable parametrisation of the set of states fulfilling an area law? An answer is given by the so-called tensor network states, which by construction follow such entanglement patterns.

\section{Tensor Network States}\label{sec:TNS}

A tensor is a multilinear map $A: \mathbb{C}^{d_1}\otimes \cdots \otimes\mathbb{C}^{d_r}\mapsto \mathbb{C}$ and the rank of this tensor is the number of factors in the tensor product, rank$(A)=r$. Because of the linearity, one can work directly with a basis in each factor of the tensor product so one associates an index label to each element of the basis. The index in each factor runs from $1$ to the dimension of the space of that factor $d_j$. 
With the previous relation we will associate each factor of the tensor product to an index of the tensor.
Tensors can be composed. The tensor product of two tensors $A$ and $B: \mathbb{C}^{d'_1}\otimes \cdots \otimes\mathbb{C}^{d'_{r'}}\mapsto \mathbb{C}$ is the tensor $A\otimes B:\mathbb{C}^{d_1}\otimes \cdots \otimes\mathbb{C}^{d_r}\otimes \mathbb{C}^{d'_1}\otimes \cdots \otimes\mathbb{C}^{d'_{r'}}\mapsto \mathbb{C}$ with rank $r+r'$. We define the contraction of two indices $i$ and $j$ with $d_i=d_j$ as the map (defined in the basis): 
\begin{align*}
\mathcal{C}: \; & \mathbb{C}^{d_i}\otimes \mathbb{C}^{d_j} \to \mathbb{C}  \notag \\
 & \;\; \ket{\alpha \beta} \;\; \longmapsto \;\delta_{\alpha,\beta}  \notag,
\end{align*}
and extended to the whole space by linearity. Then, the contraction of two indices of $A$ is carried out by acting with $\mathcal{C}$ on those indices and with the identity on the rest of them. 
Tensors naturally describe states, or in general linear operators (like matrices), in a tensor product of spaces. If we consider that the tensor $A$ describes a multi-particle $ket$, we can write explicitly:
$$A=\sum_{l_1,\cdots,l_i,\cdots,l_j,\cdots, l_r=1}^{d_1,\cdots,d_i,\cdots,d_j,\cdots, d_r}A_{l_1,\cdots,l_i,\cdots,l_j,\cdots, l_r}\ket{l_1,\cdots,l_i,\cdots,l_j,\cdots, l_r}$$
and the contraction of the two indices $i,j$ is as follows:
$$ \mathcal{C}_{i,j}\otimes \id_{\rm rest} (A)= \sum_{l_1,\cdots,l_i,\cdots,\xcancel{l_j},\cdots, l_r=1}^{d_1,\cdots,d_i,\cdots,\xcancel{d_j},\cdots, d_r} A_{l_1,\cdots,{l_i},\cdots, {l_i},\cdots, l_r}\ket{l_1,\cdots,\xcancel{l_i},\cdots, \xcancel{l_j},\cdots, l_r},$$
where the resulting tensor has rank $r-2$. We will use the standard graphical notation of tensors, where they are shapes with legs attached, each of them representing an index. For example if $r=4$: 
\begin{equation*}
\begin{tikzpicture}
\filldraw (0,0) circle (0.1);
  \draw (0,0)--(-0.1,-0.3);
         \node at (0.05,-0.35) {$a_3$};
  \draw (0,0)--(0,0.3);
       \node at (0.2,0.25) {$a_1$};
     \draw (0,0)--(-0.3,0.2);
            \node at (-0.2,0.25) {$a_4$};
     \draw (0,0)--(0.3,0.1);
     \node at (0.3,-0.1) {$a_2$};
     \node at (-0.3,-0.2) {$A$};
 \end{tikzpicture},
 \end{equation*}
where have labeled the indices as $a_1,a_2, a_3, a_4$. The contraction of two indices is represented as a line connecting the legs of the corresponding indices, thus: 
\begin{equation*}
\mathcal{C}_{a_1,a_2}\otimes \id_{\rm rest} (A) \equiv
\begin{tikzpicture}
\filldraw (0,0) circle (0.1);
  \draw (0,0)--(-0.1,-0.3);
  \draw (0,0)--(0,0.3);
     \draw (0,0)--(-0.3,0.2);
     \draw (0,0)--(0.3,0);
     \node at (-0.3,-0.2) {$A$};
      \draw  (0.3,0) to [out=0, in=90] (0,0.3);
 \end{tikzpicture}.
 \end{equation*}
The graphical notation of tensors simplifies the overload of indices when working with these objects: just compare the previous expressions with their equivalent diagrams. This benefit grows when considering a large number of tensors and operations between them, as it typically occurs in many-body systems. 

The contraction between indices $a_1$ and $b_1$ of different tensors $A$ and $B$ is  represented graphically as:
\begin{equation*}
\mathcal{C}_{a_1,b_1} \otimes \id_{\rm rest} (A\otimes B) \equiv
 \begin{tikzpicture}
\filldraw (0,0) circle (0.1);
  \draw (0,0)--(-0.1,-0.3);
  \draw (0,0)--(0,0.3);
     \draw (0,0)--(-0.3,0.2);
     \draw (0,0)--(0.3,0);
     \node at (-0.3,-0.2) {$A$};
     \filldraw (1,0) circle (0.1);
  \draw (1,0)--(0.9,-0.3);
  \draw (1,0)--(1.1,0.3);
     \draw (1,0)--(0.7,0.2);
     \draw (1,0)--(0.7,-0.3);
     \node at (1.3,-0.3) {$B$};
     \draw  (0.3,0) to [out=0, in=-135] (0.7,-0.3);
 \end{tikzpicture}.
 \end{equation*}
Each index of a tensor can be associated with a ket or a bra (the dual space) so that the tensor itself can be seen as a multilinear operator. Simple examples of tensors are vectors and matrices:
\begin{equation*}
\ket{\phi}\equiv 
\begin{tikzpicture} [baseline=-1mm]
\filldraw (0,0) circle (0.08);
 \node at (0,0.25) {$\phi$};
  \draw (0,0)--(-0.4,0);
 \end{tikzpicture}
\; ,\;\; A\equiv
 \begin{tikzpicture}[baseline=-1mm]
 \node at (0,-0.25) {$A$};
\filldraw (0,0) circle (0.1);
  \draw (-0.4,0)--(0.4,0);
 \end{tikzpicture}
 \in \mathbb{C}^D\otimes \mathbb{C}^D\cong \mathcal{M}_D,
 \end{equation*}
 where $\mathcal{M}_D$ is the space of square matrices of dimension $D$. The multiplication of a vector by a matrix is represented as
 \begin{equation*}
A\ket{\phi}\equiv 
\begin{tikzpicture}
\filldraw (0.4,0) circle (0.08);
\filldraw (0,0) circle (0.1);
 \node at (0.4,0.25) {$\phi$};
  \node at (0,-0.25) {$A$};
  \draw (-0.4,0)--(0.4,0);
 \end{tikzpicture},
 \end{equation*}
  and the trace of a matrix is represented as
  \begin{equation*}
  \tr[A]=
\begin{tikzpicture}
\filldraw (0,0) circle (0.1);
  \node at (0,0.25) {$A$};
  \draw (-0.4,0)--(0.4,0)--(0.4,-0.4)--(-0.4,-0.4)--(-0.4,0);
 \end{tikzpicture} \; .
 \end{equation*}
Let us now consider a rank-3 tensor $A\in \mathbb{C}^D\otimes \mathbb{C}^D \otimes \mathbb{C}^d$. This is equivalent to $d$ matrices belonging to $\mathcal{M}_D$. We will denote each of these matrices as
 \begin{equation*}
A^i =
\begin{tikzpicture}
\draw (0,0)--(0,0.25);
\pic at (0,0) {tensor};
 \node at (0,0.4) {$i$};
  \node at (0,-0.25) {$A$};
  \draw (-0.4,0)--(0.4,0);
 \end{tikzpicture}
 \equiv
\begin{tikzpicture}
\filldraw (0,0.4) circle (0.08);
\draw (0,0)--(0,0.4);
\filldraw (0,0) circle (0.1);
 \node at (0.1,0.25) {$i$};
  \node at (0,-0.25) {$A$};
  \draw (-0.4,0)--(0.4,0);
 \end{tikzpicture},
 \end{equation*}
where the label above a leg is meant to fix the index to that label. 
\begin{definition}
A tensor network state is a multi-partite state placed on a lattice constructed via the contraction of local tensors placed on the vertices. 
\end{definition}

The first example of a tensor network state is called Matrix Product State (MPS) \cite{Fannes92, PerezGarcia07} and it defines a state placed on a unidimensional lattice. An MPS with periodic boundary condition (the system is placed on a ring) and constructed with local tensors independent of the site, i.e. translationally invariant, is written as follows

\begin{equation}\label{eq:MPS-TI}
\ket{\psi_A}= \sum^{d}_{i_1,\dots,i_N=1}\tr[A^{i_1}A^{i_2}\cdots A^{i_N}]\ket{i_1,\cdots,i_N}
\end{equation}
and it is represented as

 \begin{equation*}
 \begin{tikzpicture}[scale=1.4]
    \pic at (0,0) {tensor};
    \pic at (0.5,0) {tensor};
    \draw (-0.2,0)--(0.7,0);
     \draw (2.2,0)--(1.3,0);
    \draw (-0.2,-0.15)--(2.2,-0.15);
       \draw (2.2,0)--(2.2,-0.15);
       \draw (-0.2,0)--(-0.2,-0.15);
    \node at (1,0) {${\cdots}$};
    \pic at (1.5,0) {tensor};
    \pic at (2,0) {tensor};
    \draw (2,0)--(1.3,0);
    \end{tikzpicture} \; .
\end{equation*}

Let  $D$ be the maximum rank of the virtual indices (those that get contracted) which is called bond dimension. Then, the state is specified by $ND^2d$ parameters instead of the previous exponential dependance ($d^N$) on the number of subsystems. \\

One key aspect here is how $D$ depends on $N$, since any state can be written as a tensor network with a bond dimension that grows exponentially with the number of particles. Indeed, to obtain a tensor network description of any one-dimensional state successive Schmidt  Decompositions (SD) can be done \cite{Vidal03}. Performing a SD between the first subsystem and the rest of the chain we obtain 
$$\ket{\psi}=\sum_{\alpha=1}^d \lambda^{[1]}_{\alpha} \ket{\alpha}^{[1]}\ket{\alpha}^{[2,\dots, N]}=\sum_{i_1=1}^d \sum_{\alpha=1}^d A^{[1]}_{i_1,\alpha}\lambda^{[1]}_{\alpha}\ket{i_1} \ket{\alpha}^{[2,\dots, N]},$$
where $A^{[1]}_{i_1,\alpha}= \langle i_1| \alpha\rangle^{[1]}$ and $\lambda^{[1]}_{\alpha}$ are the Schmidt values. The SD of the first two subsystems with the rest of the chain can be written as follows:
\begin{equation}\label{psiSD}\ket{\psi}=\sum_{\beta=1}^{d^2} \lambda^{[2]}_{\beta} \ket{\beta}^{[1,2]}\ket{\beta}^{[3,\dots, N]}=\sum_{i_1=1}^d \sum_{\alpha=1}^d \sum_{\beta=1}^{d^2}  A^{[1]}_{i_1,\alpha} A^{[2]}_{i_2,\alpha,\beta} \lambda^{[2]}_{\beta}\ket{i_1} \ket{i_2} \ket{\beta}^{[3,\dots, N]},
\end{equation}
where we have introduced the resolution of the identity when needed and $A^{[2]}_{i_2,\alpha,\beta}= (\bra{ \alpha}^{[1]}\bra{ i_2}) \ket{\beta}^{[1,2]}$. In this way we obtain the expression
\begin{equation}\label{MPS-OBC}
\ket{\psi}= \sum^{d}_{i_1,\dots,i_N=1}A^{[1]}_{i_1}A^{[2]}_{i_2}\cdots A^{[N]}_{i_N} \ket{i_1,\cdots,i_N} ,
\end{equation}
in which the bond dimension grows in the worst case to $d^{N/2}$ in the middle of the chain. 

Note that, even if the state $\ket{\psi}$ is translationally invariant, the description obtained in this way does not reflect this fact. In particular, it is not of the form \eqref{eq:MPS-TI}. This can be fixed, but in some cases at the price of growing the bond dimension with the size of the system \cite{Sanz10}.

The successive SD are graphically represented as:
\begin{align*}
\begin{tikzpicture}[baseline=-1mm,scale=1.2]
\filldraw[rounded corners=.05cm] (-1,-0.05) rectangle (1, 0.05);
  \draw (0.5,0.2)--(0.5,0);
    \draw (0.9,0.2)--(0.9,0);
      \node at (0,0.15) {${\cdots}$};
     \draw (-0.5,0.2)--(-0.5,0);
    \draw (-0.9,0.2)--(-0.9,0);
 \end{tikzpicture}
\; \stackrel{{\rm SD}}{\longrightarrow}\;  & 
 \begin{tikzpicture}[baseline=-1mm,scale=1.2]
    \pic at (-0.9,0) {tensor};
    \filldraw[rounded corners=.05cm] (-0.6,-0.05) rectangle (1, 0.05);
    \draw (0.5,0.2)--(0.5,0);
    \draw (0.9,0.2)--(0.9,0);
      \node at (0,0.15) {${\cdots}$};
     \draw (-0.5,0.2)--(-0.5,0);
     \draw (-0.9,0)--(0,0);
        \end{tikzpicture}
\; \stackrel{{\rm SD}}{\longrightarrow} \; \cdots  \\ 
\; \stackrel{{\rm SD}}{\longrightarrow}  \;&
 \begin{tikzpicture}[baseline=-1mm,scale=1.2]
    \pic at (-0.9,0) {tensor};
        \pic at (-0.5,0) {tensor};
    \filldraw[rounded corners=.05cm] (-0.3,-0.05) rectangle (1, 0.05);
     \draw (-0.2,0.2)--(-0.2,0);    
       \node at (0.2,0.15) {${\cdots}$};
    \draw (0.6,0.2)--(0.6,0);
    \draw (0.9,0.2)--(0.9,0);
     \draw (-0.9,0)--(0,0);
        \end{tikzpicture}
\; \stackrel{{\rm SD}}{\longrightarrow} \;
        \begin{tikzpicture}[baseline=-1mm,scale=1.2]
            \pic at (0,0) {tensor};
    \pic at (0.5,0) {tensor};
    \draw (0,0)--(0.7,0);
    \node at (1,0) {${\cdots}$};
    \pic at (1.5,0) {tensor};
    \pic at (2,0) {tensor};
    \draw (2,0)--(1.3,0);
    \end{tikzpicture}.
\end{align*}

Suppose now that the matrices in \eqref{MPS-OBC} have size upper bounded by $D$ independent of $N$. Then $\ket{\psi}$ satisfies the area law for any bipartition in right and left. Indeed, in that case the SD $\ket{\psi}=\sum_{\alpha=1}^D \lambda_{\alpha} \ket{\alpha}^{[R]}\ket{\alpha}^{[L]}$ has only $D$ terms, and hence the entanglement entropy $S_R(\ket{\psi})=-\sum_{\alpha=1}^D\lambda_\alpha^2\log \lambda_\alpha^2 \le \log D$ is bounded by the boundary of the bipartition, which in 1D is just a constant. \\

At the level of mathematical proofs, it has been proven that MPS approximate well any ground state of a locally interacting gapped Hamiltonian in 1D \cite{Hastings07A, Arad13}.  Also any MPS is the (essentially unique) ground state of a locally interacting gapped Hamiltonian. 

So one can claim that the set of MPS {\it essentially coincides} with the set of GS of gapped locally interacting Hamiltonians and hence gives an efficient parametrization of this set. This makes MPS  the appropriate mathematical framework to prove statements about 1D systems. In 2D dimensions, despite some promising results along the same lines \cite{Hastings06, Molnar15}, the full picture is far from being completed. 

Every MPS is the GS of a locally interacting Hamiltonian, called {\it parent Hamiltonian} $H=\sum_i \mathfrak{h}_i$. $H$ is constructed in such a way that its GS keeps the local structure of the MPS. The local terms $\mathfrak{h}_i$ are defined by ${\rm ker}(\mathfrak{h}_i)={\rm Im}(\Gamma_\mathcal{R})$, choosing that  $\mathfrak{h}_i$ is an orthogonal projector, where 
\begin{equation}\label{eq:defGammma}
\Gamma_\mathcal{R}: (\C^D)^{\otimes 2}\rightarrow (\C^d)^{\otimes |\mathcal{R}|}.
\end{equation}
$\Gamma_\mathcal{R}$ maps linearly, using the tensors ${A}$ in the region $\mathcal{R}$, operators ({\it boundary conditions}) living in the virtual space to vectors in the physical Hilbert space of the region $\mathcal{R}$:

\begin{equation*}
\Gamma_{\mathcal{R}=5}(X)=
  \begin{tikzpicture}
    \draw (-0.5,0) rectangle (4.5,-0.3);
             \foreach \x in {0,1,2,3,4}{
        \node[tensor] at (\x,0) {};
        \draw (\x,0) -- (\x,0.3); }
    \node[tensorr,label=below:$X$] at (2,-0.3) {};
  \end{tikzpicture}\; .
\end{equation*}
It is clear that the given MPS is a ground state of $H$ and that $H$ is frustration free, meaning that the ground state of $H$ minimizes the energy of each local term $\mathfrak{h}_i$, {\it i.e.} $\mathfrak{h}_i \ket{\Psi_A}=0, \; \forall i$.

To obtain that the MPS is the \emph{unique} GS of its parent Hamiltonian we have to restrict to some classes of tensors. These classes come from imposing certain conditions on $\Gamma$ and $\mathcal{R}$. In particular, a tensor $A$ is defined to be injective if $\Gamma_{\mathcal{R}=1}$ is injective as a map from virtual to physical indices:

\begin{equation}\label{eq:injec1}
  \begin{tikzpicture}
    \draw (-0.5,0) rectangle (0.5,-0.3);
    \node[tensor] (a) at (0,0) {};
    \draw (a)--++(0,0.3);
    \node[tensorr,label=below:$X$] at (0,-0.3) {};
  \end{tikzpicture} \ = 0 \quad \Rightarrow \quad X = 0.
\end{equation}
This condition is equivalent to the existence of a right inverse $A^{-1}$ such that $A A^{-1}= \id_D\times\id_D$:
\begin{equation}\label{eq:injec}
   \begin{tikzpicture}
    \draw (-0.5,0)--(0.5,0);
    \draw (-0.5,0.4)--(0.5,0.4);
    \draw (0,0)--(0,0.4);
    \node[tensor,label=below:$A$] at (0,0) {};
    \node[tensor,label=above:$A^{-1}$] at (0,0.4) {};
  \end{tikzpicture}  =
  \begin{tikzpicture}
      \clip (0.2,-0.2) rectangle (1,0.7);
      \draw (0,0) rectangle (0.5,0.4);
      \draw[shift={(0.7,0)}] (0,0) rectangle (0.5,0.4);
  \end{tikzpicture} 
    \Rightarrow
   \begin{tikzpicture}
    \draw (-0.5,0)--(1,0);
    \draw (-0.5,0.4)--(1,0.4);
    \draw (0,0)--(0,0.4);
     \draw (0.5,0)--(0.5,0.4);
    \node[tensor,label=below:$A$] at (0,0) {};
    \node[tensor,label=above:$A^{-1}$] at (0,0.4) {};
        \node[tensor,label=below:$A$] at (0.5,0) {};
    \node[tensor,label=above:$A^{-1}$] at (0.5,0.4) {};
  \end{tikzpicture}  =
  D \cdot
  \begin{tikzpicture}[baseline=1mm]
      \clip (0.2,-0.2) rectangle (1,0.7);
      \draw (0,0) rectangle (0.5,0.4);
      \draw[shift={(0.7,0)}] (0,0) rectangle (0.5,0.4);
  \end{tikzpicture},
\end{equation}  
where the last equation expresses the fact that injectivity is preserved under blocking. 

\
For injective tensors the local terms of the Hamiltonian are defined by
$${\rm ker}(\mathfrak{h}_i)={\rm Im}(\Gamma_2)=
\left \{
  \begin{tikzpicture}
    \draw (-0.5,0) rectangle (1.5,-0.2);
             \foreach \x in {0,1}{
        \node[tensor] at (\x,0) {};
        \draw (\x,0) -- (\x,0.2);  }
         \node[tensorr,label=below:$X$] at (0.5,-0.2) {};
  \end{tikzpicture},
X\in \mathcal{M}_D  \right\}
,$$
which enforces to two neighbouring sites of the GS to be generated by two $A$ tensors. One can show that \cite{Schuch10}
$${\rm ker}(H)=\bigcap_i {\rm ker}(\mathfrak{h}_i)= \ket{\Psi_A}.$$
In general if $\Gamma_n$ is injective, the local terms are chosen to be ${\rm ker}(\mathfrak{h}_i)={\rm Im}(\Gamma_{n+1})$ to obtain that the MPS is the unique GS of $H=\sum_i \mathfrak{h}_i$.
\

The graphical representation of states can be extended to operators. An operator acting on $N$ sites is represeted as follows:
\begin{equation*}
\begin{tikzpicture}[baseline=-1mm]
\filldraw[rounded corners=.05cm] (-1,-0.05) rectangle (1, 0.05);
  \draw (0.5,0.2)--(0.5,-0.2);
    \draw (0.9,0.2)--(0.9,-0.2);
      \node at (0,0.15) {${\cdots}$};
            \node at (0,-0.15) {${\cdots}$};
     \draw (-0.5,0.2)--(-0.5,-0.2);
    \draw (-0.9,0.2)--(-0.9,-0.2);
 \end{tikzpicture}.
 \end{equation*}
One can also consider operators coming from a tensor network, that is, Matrix Product Operators (MPO):
\begin{equation*}
        \begin{tikzpicture}[baseline=-1mm]
            \pic at (0,0) {tensorud};
    \pic at (0.5,0) {tensorud};
    \draw (0,0)--(0.7,0);
    \node at (1,0) {${\cdots}$};
    \pic at (1.5,0) {tensorud};
    \pic at (2,0) {tensorud};
    \draw (2,0)--(1.3,0);
    \end{tikzpicture},
\end{equation*}
where the local tensor are matrices depending on two virtual indices. 
\

In particular the local terms $\mathfrak{h}_i$ of some parent Hamiltonians $H=\sum_i \mathfrak{h}_i$ can be written as an MPO. Let us consider an isometric MPS which is defined by $A$ being an isometry: $A^{-1}=A^{\dagger}$. The local term of the parent Hamiltonian is constructed as $\mathfrak{h}=\id_2-\Pi_2$, where
\begin{equation*}
\Pi_2=
  \begin{tikzpicture}
    \draw (-0.5,0) rectangle (1.5,-0.4);
             \foreach \x in {0,1}{
        \node[tensor] at (\x,0) {};
        \draw (\x,0) -- (\x,0.3); 
          \node[tensor] at (\x,-0.4) {};
           \node at (\x+0.4,-0.4-0.2) {$A^{-1}$};
        \draw (\x,-0.4) -- (\x,-0.7); }
  \end{tikzpicture}.
\end{equation*}

These terms, which are orthogonal projectors, commute with each other so that the Hamiltonian is gapped (see \cref{gapH}). For general injective MPS the gap of the parent Hamiltonian is proven in \cite{Fannes92}.

Let us finish this section by commenting briefly on the graphical description of operators that act on the Hilbert space and expected values. For example consider an operator acting only on one site:
\begin{equation*}
\mathcal{O} \equiv
\begin{tikzpicture}
\draw (0,-0.2)--(0,0.2);
 \filldraw[draw=black, fill=red] (0,0) circle (0.06);
 \end{tikzpicture}
 \longrightarrow  \mathcal{O}_{[2]}\ket{\psi} \equiv
        \begin{tikzpicture}[baseline=-1mm]
            \pic at (0,0) {tensor};
    \pic at (0.5,0) {tensor};
      \draw (0.5,0)--(0.5,0.3);
     \filldraw[draw=black, fill=red] (0.5,0.15) circle (0.06);
    \draw (0,0)--(0.7,0);
    \node at (1,0) {${\cdots}$};
    \pic at (1.5,0) {tensor};
    \pic at (2,0) {tensor};
    \draw (2,0)--(1.3,0);
    \end{tikzpicture}.
\end{equation*}
The expectation value is then represented as:
\begin{equation}\label{ExpValue}
\langle \mathcal{O}_{[2]} \rangle \equiv
        \begin{tikzpicture}
            \pic at (0,0) {tensor};
             \pic at (0,0.3) {tensord};
    \pic at (0.5,0) {tensor};
        \pic at (0.5,0.3) {tensord};
     \filldraw[draw=black, fill=red] (0.5,0.15) circle (0.06);
    \draw (0,0)--(0.7,0);
     \draw (0,0.3)--(0.7,0.3);
    \node at (1,0) {${\cdots}$};
     \node at (1,0.3) {${\cdots}$};
    \pic at (1.5,0) {tensor};
    \pic at (1.5,0.3) {tensord};
    \pic at (2,0) {tensor};
     \pic at (2,0.3) {tensord};
    \draw (2,0)--(1.3,0);
     \draw (2,0.3)--(1.3,0.3);
    \end{tikzpicture}
    = \begin{tikzpicture}
             \pic at (-0.3,0) {tensor};
             \pic at (-0.3,-0.2) {tensord};
            \pic at (0,0) {tensor};
             \pic at (0,-0.2) {tensord};
                  \filldraw[draw=black, fill=red] (0,0.11) circle (0.05);
                 \draw (-0.3,0)--(0.15,0);
                 \draw (-0.3,-0.2)--(0.15,-0.2);
                                   \draw[preaction={draw, line width=1pt, white}][line width=0.5pt] (0,0.19) to [out=45, in=-45] (0,-0.39);
                                                     \draw[preaction={draw, line width=1pt, white}][line width=0.5pt] (-0.3,0.19) to [out=45, in=-45] (-0.3,-0.39);
                 \node at (0.35,0) {${\cdots}$};
                 \node at (0.35,-0.2) {${\cdots}$};
    \pic at (0.7,0) {tensor};
    \pic at (0.7,-0.2) {tensord};
                         \draw (0.55,0)--(1,0);
                     \draw (0.55,-0.2)--(1,-0.2);
                    \draw[preaction={draw, line width=1pt, white}][line width=0.5pt] (0.7,0.19) to [out=45, in=-45] (0.7,-0.39);
    \pic at (1,0) {tensor};
    \pic at (1,-0.2) {tensord};
                \draw[preaction={draw, line width=1pt, white}][line width=0.5pt] (1,0.19) to [out=45, in=-45] (1,-0.39);
    \end{tikzpicture}
    \equiv \tr[\mathcal{O}_{[2]} \ket{\psi_A}\bra{\psi_A}],
\end{equation}
where the tensor 

\begin{equation}\label{Abar}
   \begin{tikzpicture}
    \draw (-0.4,0.4)--(0.4,0.4);
    \draw (0,0.1)--(0,0.4);
    \node[tensor] at (0,0.4) {};
  \end{tikzpicture}
  \equiv 
  \bar{A}
\end{equation}
is the complex conjugate of $A$ and from now on we will not label it in the diagrams. This tensor represet the state $\bra{\psi_A}$. In the contrary we will write the label $A^{-1}$ when we use it.

\section{Projected Entangled Pair States}
An MPS-analogous tensor network states in two dimensions are the so-called PEPS \cite{Verstraete04}. A translational invariant PEPS for a square lattice is defined by a set of rank-5 tensors ${A}^{[v]}\in \C^d\otimes (\C^D)^{\otimes 4}$:
\begin{equation*}
({A})^{i}_{\alpha,\beta,\gamma,\delta}=
\begin{tikzpicture} 
         \node at (-1.6,0.3,1) {$i$};
       \pic at (-1.6,0,1) {3dpeps};
       \node[anchor=south] at (-1.9,0,1.2) {$\alpha$};
       \node[anchor=north] at (-1.6,0,1.1) {$\delta$};
        \node[anchor=south] at (-1,0,1.35) {$\gamma$};
       \node[anchor=south] at (-1.35,0,0.9) {$\beta$};
\end{tikzpicture}.
\end{equation*}
The PEPS is the contraction of the tensors in all sites
$$ \ket{\Psi_A}=\sum \limits_{i_1, \cdots, i_N =1} ^d \mathcal{C}\{ {A}^{i_1},\dots, {A}^{i_N}\} |i_1\cdots i_N\rangle, $$ and it is represented graphically as follows
\begin{equation*}
\ket{\Psi_A}= 
 \begin{tikzpicture}
  \foreach \z in {0,0.7,1.4}{
      \foreach \x in {0,0.5,1,1.5}{
           \pic at (\x,0,\z) {3dpeps}; } }
            \node at (-0.6,0,0.7) {$\cdots$};
  \node at (2.9-0.6,0,0.7) {$\cdots$};
   \node[rotate=45] at (1.45-0.7,0,-0.7) {$\cdots$};
  \node[rotate=45] at (1.45-0.6,0,2.1) {$\cdots$};
    \end{tikzpicture},
\end{equation*}
where we will assume periodic boundary conditions, i.e. a torus, but we will not draw it.

It is not difficult to see that, for a fixed bond dimension $D$, PEPS also fulfill the area law. This and the fact that MPS approximate well the ground state of any locally interacting gapped Hamiltonian in 1D motivate the conjecture of that fact for PEPS. This can be seen as the {\it practical} version of the Area Law Conjecture since it comes with a concrete parametrization of the set of (approximate) ground states. Indeed, many algorithms (including the ubiquitous DMRG algorithm of S. White \cite{White92A,White92B}) aiming to solve locally interacting Hamiltonians implement different types of optimization procedures to find the MPS or PEPS that minimizes the energy. They turn out to work very well in practice (see \cite{Schollwock05, Schollwock11, Verstraete08, Ran17, Orus18} for reviews on that), supporting the validity of this practical Area Law Conjecture.

Analogously to MPS, a parent Hamiltonian for PEPS can be constructed. Any PEPS is a ground state of its parent Hamiltonian which is defined via \cref{eq:defGammma} analogously as for MPS:
\begin{equation*}
X\equiv
 \begin{tikzpicture}
     \begin{scope}[canvas is zx plane at y=0]
      \draw[rounded corners=.05cm, very thick] (-0.53,-0.4) rectangle (1.9,1.9);
    \end{scope}
        \draw (-0.4,0,0.7) -- (-0.2,0,0.7);
       \draw (-0.4,0,1.4) -- (-0.2,0,1.4);   
        \draw (-0.4,0,0) -- (-0.2,0,0);   
          \draw (0,0,-0.5)--(0,0,-0.2);
            \draw (0.5,0,-0.5)--(0.5,0,-0.2);
            \draw (1,0,-0.5)--(1,0,-0.2);
            \draw (1.5,0,-0.5)--(1.5,0,-0.2);
          \draw (0,0,1.6)--(0,0,1.9);
          \draw (0.5,0,1.6)--(0.5,0,1.9);
          \draw (1,0,1.6)--(1,0,1.9);
           \draw (1.5,0,1.6)--(1.5,0,1.9);
       \draw (1.9,0,0.7) -- (1.7,0,0.7);
        \draw (1.9,0,0) -- (1.7,0,0);
         \draw (1.9,0,1.4) -- (1.7,0,1.4);
    \end{tikzpicture}
\longrightarrow
 \begin{tikzpicture}
     \begin{scope}[canvas is zx plane at y=0]
      \draw[rounded corners=.05cm, very thick] (-0.53,-0.4) rectangle (1.9,1.9);
    \end{scope}
        \pic at (0,0,0.7) {3dpeps};
        \pic at (0,0,1.4) {3dpeps};   
          \pic at (0,0,0) {3dpeps};
      \pic at (0.5,0,0) {3dpeps};
      \pic at (0.5,0,0.7) {3dpeps};
      \pic at (0.5,0,1.4) {3dpeps};
      \pic at (1,0,0) {3dpeps};
        \pic at (1,0,0.7) {3dpeps};
        \pic at (1,0,1.4) {3dpeps};
         \pic at (1.5,0,0) {3dpeps};
        \pic at (1.5,0,0.7) {3dpeps};
        \pic at (1.5,0,1.4) {3dpeps};
    \end{tikzpicture}
    \equiv \Gamma_R(X).
\end{equation*}

Different classes of tensors can be defined to relate the PEPS and the ground subspace of its parent Hamiltonian. In particular, injective PEPS are unique GS of their associated parent Hamiltonians. As in the MPS case, injective PEPS are defined as those for which the tensors have an inverse  $A^{-1}$: 
\begin{equation}\label{eq:PEPSinj}
\begin{tikzpicture}
   \pic at (0,0,0) {3dpeps};
    \pic at (0,0.4,0) {3dpepsdown};
    \node[anchor=north] at (0,0,0) {${A}$};
    \node[anchor=south] at (0,0.4,0) {$\myinv{A}$};
 \end{tikzpicture}
    = 
     \begin{tikzpicture}
       \pic at (0,0,0) {3disopeps};
\end{tikzpicture}.
\end{equation} 
This is the so-called injectivity condition and it is equivalent to $\Gamma_1$ being injective, see \cref{eq:defGammma}. The region $\mathcal{R}$ that defines the local hamiltonian is a patch of $2\times 2$ tensors such that $\Gamma_{2\times 2}: (\C^D)^{\otimes 8}\rightarrow (\C^d)^{\otimes 4}$. In the next section we deal with a more general class of tensors which includes degenerate ground spaces and topological order. 
\section{G-injective PEPS}\label{sec:introGinj}

\begin{definition} \label{def:GPEPS} A PEPS is $G$-injective, introduced in Ref. \cite{Schuch10},  if its tensor $A$ satisfies the following
\begin{itemize}
\item the $G$-invariant condition: for a given representation $u_g$ of $G$

\begin{equation}\label{Ginva}
\begin{tikzpicture}[scale=1.2]
\pic at (0.1,0,4) {3dpeps};
	\node[anchor=east] at (1,0,4) {$=$};
	\pic at (1.8,0,4) {3dpeps};
	 \begin{scope}[canvas is zx plane at y=0]
      \draw (4,1.2)--(4,2.4);
      \draw (3.2,1.8)--(4.8,1.8);
    \end{scope}
	\node[anchor= south] at (1.4,0,4) {$\myinv{u}_g$};
	\node[anchor=west] at (1.8,0,4.5) {$\myinv{u}_g$};
        \node[anchor=east] at (1.7,0,3.1) {$u_{g}$};
        \node[anchor=south] at (2.3,0,4) {$u_{g}$};
        
	\filldraw[draw=blue,fill=white]  (1.4,0,4) circle (0.04);
	\filldraw[draw=black,fill=blue]  (2.2,0,4) circle (0.04);
	\filldraw[draw=black,fill=blue]  (1.8,0,3.5) circle (0.04);
	\filldraw[draw=blue,fill=white] (1.8,0,4.5) circle (0.04);
\end{tikzpicture}  \;\; \forall g\in G,
\end{equation}

where $u_g$ contains all the irreps of $G$ in its decomposition.
\item there exists a tensor $A^{-1}$ such that
\begin{equation}\label{Ginje}
 \mathcal{P}_G \equiv
\begin{tikzpicture}[scale=1.2]
   \pic at (0,0,0) {3dpeps};
    \pic at (0,0.4,0) {3dpepsdown};
    \node[anchor=north] at (0,0,0) {${A}$};
    \node[anchor=south] at (0,0.4,0) {$\myinv{A}$};
 \end{tikzpicture}
    =  \frac{1}{|G|}\sum_{g\in G} \;  
     \begin{tikzpicture}
       \pic at (0,0,0) {3dGisopepsproj};
                      \node at (-0.35,0.2,0.35) {$\myinv{u}_g$};
              \node at (0.15,-0.1,0.3) {$\myinv{u}_g$};
        \node at (0.37,0.35,0.3) {${u_g}$};
                \node at (0.5,-0.03,0.3) {${u_g}$};
\end{tikzpicture} . 
  \end{equation}    
\end{itemize}
We have supposed a TI PEPS since the representation of $G$ in each tensor is the same.
\end{definition}
\cref{Ginva} is equivalent to
$$A=A(u_g \otimes u_g \otimes \myinv{u}_g \otimes \myinv{u}_g ),$$
where the operators are acting on the virtual d.o.f., we will denote this fact by multiplying the operators from the right of the tensor. The operators acting on the physical Hilbert space will be placed on the LHS of the tensor, {\it e. g.} $\mathcal{O}A$. We will also establish an ordering concerning the virtual legs for the equations, the legs are ordered clockwise starting from the upper leg. 
\
\cref{Ginje} is $ \mathcal{P}_G\equiv \myinv{A}A=\frac{1}{|G|}\sum_{g\in G}(u_g \otimes u_g \otimes \myinv{u}_g \otimes \myinv{u}_g )$ where $\mathcal{P}_G$ is the projector onto the subspace invariant under the action of $u_g\otimes u_g\otimes \myinv{u}_g\otimes \myinv{u}_g$

\begin{observation} \label{obs:semireg}
Any representation of a finite group $G$ can be decomposed as follows: $u_g\cong \bigoplus_\sigma \pi_\sigma(g)\otimes \id_{m_\sigma}$ where the sum runs over the irreps  $\pi_\sigma$ ($\sigma$ denotes the label) of $G$, with dimension $d_\sigma$ and $m_\sigma$ is its multiplicity. $u_g$ is semiregular if $m_\sigma \ge 1$ for all the irreps of $G$.
\end{observation}

We will simplified the graphics through the manuscript by writing $g$ instead of $u_g$. we will represent the operators as circles where the blue ones stand for $g$ and the white circles stand for $\myinv{g}$. Also we represent the elements of $G$ as black circles when a sum is involved.  

The $G$-invariance condition can be seen as a pulling through condition:
\begin{equation}\label{eq:Gmoves}
\begin{tikzpicture}[scale=1.2]
\pic at (0,0,4) {3dpeps};
	 \begin{scope}[canvas is zx plane at y=0]
      \draw (4,-0.6)--(4,0.6);
      \draw (3.2,0)--(4.8,0);
        \end{scope}
        \draw[densely dotted, blue, rounded corners] (-0.5,0,3.4)--(-0.4,0,4)--(-0.2,0,4.5)--(0,0,4.5)--(0.2,0,4.6)--(0.4,0,5);
        
        	\filldraw[draw=black,fill=blue]  (-0.4,0,4) circle (0.04);
	\filldraw[draw=black,fill=blue] (0,0,4.5) circle (0.04);
		\node[anchor=east] at (1.15,0,4) {$=$};
	\pic at (1.8,0,4) {3dpeps};
	 \begin{scope}[canvas is zx plane at y=0]
      \draw (4,1.2)--(4,2.4);
      \draw (3.2,1.8)--(4.8,1.8);
        \end{scope}
                \draw[densely dotted, blue, rounded corners] (1.5,0,3.2)--(1.6,0,3.4)--(1.8,0,3.5)--(2,0,3.6)--(2.2,0,4)--(2.2,0,4.5);
        	\filldraw[draw=black,fill=blue]  (2.2,0,4) circle (0.04);
	\filldraw[draw=black,fill=blue]  (1.8,0,3.5) circle (0.04);
	\end{tikzpicture}
	,\;
	\begin{tikzpicture}[scale=1.2]
		\pic at (0.3,0,3.5) {3dpeps};
	\filldraw [draw=black, fill=blue]  (0.05,0,3.5) circle (0.04);
	\draw[densely dotted, blue] (0.05,0,3) -- (0.05,0,4.05);
	\node[anchor=east] at (1.2,0,3.5) {$=$};
	\pic at (1.6,0,3.5) {3dpeps};
	\draw[densely dotted, blue, rounded corners] (1.35,0,3) -- 
	       (1.35,0,3.2) --  (1.85,0,3.2) -- (1.85,0,3.8) -- (1.35,0,3.8) --
	(1.37,0,4.05);
	 \filldraw [draw=black, fill=blue]  (1.6,0,3.2) circle (0.04);
	\filldraw [draw=black, fill=blue]  (1.85,0,3.5) circle (0.04);
	\filldraw [draw=blue, fill=white]  (1.6,0,3.8) circle (0.04);
\end{tikzpicture}
\end{equation}
which allows to deform strings of tensor product of $u_g$ operators in the tensor network.

The requirement of $u_g$ to be semiregular allows to use the matrix $\mathfrak{D} \cong \frac{1}{|G|}\bigoplus_\sigma \frac{d_\sigma}{m_\sigma} \id_{d_\sigma}\otimes \id_{m_\sigma}$ (with the same change of basis that the representation $u_g$), to obtain the inverse of blocked tensors:
\begin{equation}\label{conca}
\begin{tikzpicture}
   \pic at (0,0,0) {3dpeps};
   \draw (0,0,0)--(0,0.5,0);
    \pic at (0,0.5,0) {3dpepsdown};
        \node[anchor=south east] at (0,0.5,0) {$\myinv{A}$};
     \filldraw (0.3,0.5,0) circle (0.04);
      \node[anchor=south] at (0.3,0.5,0) {${\mathfrak{D}}$};
    \pic at (0.6,0,0) {3dpeps};
    \draw (0.6,0,0)--(0.6,0.5,0);
    \pic at (0.6,0.5,0) {3dpepsdown};
     \node[anchor=south west] at (0.6,0.5,0) {$\myinv{A}$};
 \end{tikzpicture}
    = \frac{1}{|G|} \sum_{g\in G}   
\begin{tikzpicture}
         \filldraw (-0.3,0,0) circle (0.04);
     \filldraw (0,0,0.3) circle (0.04);
       \filldraw (0.8,0,0) circle (0.04);
        \filldraw (0,0,-0.3) circle (0.04); 
             \filldraw (0.5,0,0.3) circle (0.04);
        \filldraw (0.5,0,-0.3) circle (0.04); 
                            \node at (-0.25,0.2,0.3) {$\myinv{g}$};
              \node at (0.15,-0.1,0.3) {$\myinv{g}$};
                \node at (0.65,-0.1,0.3) {$\myinv{g}$};
        \node at (0.25,0.37,0.3) {${g}$};
                \node at (1,0,0.3) {${g}$};
                     \node at (0.9,0.33,0.3) {${g}$};
\draw (-0.1,0,0)--(-0.1,0.4,0);
       \begin{scope}[canvas is zx plane at y=0]
          \draw (0.1,0)--(0.5,0);
        \draw (-0.5,0)--(-0.1,0);
           \draw (0.1,0.5)--(0.5,0.5);
        \draw (-0.5,0.5)--(-0.1,0.5);
    \end{scope} 
     \begin{scope}[canvas is zx plane at y=0.4]
       \draw (-0.5,0)--(-0.1,0);
      \draw[ preaction={draw, line width=1pt, white}] (0.1,0)--(0.5,0);
             \draw (-0.5,0.5)--(-0.1,0.5);
      \draw[preaction={draw, line width=1pt, white}] (0.1,0.5)--(0.5,0.5);
    \end{scope} 
     \draw (0,0,0.1)--(0,0.4,0.1);
\draw (0,0,-0.1)--(0,0.4,-0.1);
     \draw (0.5,0,0.1)--(0.5,0.4,0.1);
\draw (0.5,0,-0.1)--(0.5,0.4,-0.1);
       \draw[preaction={draw, line width=1pt, white}] (0.6,0,0)--(0.6,0.4,0);
         \begin{scope}[canvas is zx plane at y=0]
          \draw (0,0.6)--(0,0.9);
       \draw (0,-0.4)--(0,-0.1);
          \end{scope} 
                  \begin{scope}[canvas is zx plane at y=0.4]
          \draw (0,0.6)--(0,0.9);
       \draw (0,-0.4)--(0,-0.1);
          \end{scope} 
\end{tikzpicture},
\end{equation}
since $\tr[u_g\mathfrak{D}]=|G|\delta_{e,g}$ \cite{Schuch11} . The parent Hamiltonian $H$ is defined by the local term $\mathfrak{h}$, acting on a plaquette (a $2\times 2$ patch), so that ${\rm ker}(\mathfrak{h})={\rm Im}(\Gamma_{2\times 2})$. 
Let us study the GS subspace of this parent Hamiltonian. For that we consider deformations of the PEPS, placed on a torus, by inserting operators $g$ and $g^{-1}$ in the virtual d.o.f. Contractible loops of these operators do not modify the state due to the $G$-invariance, see \cref{Ginva}, {\it i.e.} a loop is absorbed by the tensors:

\begin{equation*}
		 \begin{tikzpicture}
     	   \foreach \z in {-0.7,0,0.7,1.4}{
      \foreach \x in {-0.5,0,0.5,1,1.5}{
           \pic at (\x,0,\z) {3dpeps}; 
             } } 
		 \filldraw[draw=blue, fill=white]  (1,0,1.05) circle (0.05);
		 \filldraw[draw=black, fill=blue]  (1.25,0,0.7) circle (0.05);
		 \filldraw[draw=black, fill=blue]  (1.25,0,0) circle (0.05);
		 \filldraw[draw=black, fill=blue]  (1,0,-0.35) circle (0.05);
		  \filldraw[draw=black, fill=blue]  (0.5,0,-0.35) circle (0.05);
		  \filldraw[draw=black, fill=blue]  (0,0,-0.35) circle (0.05);
		  \filldraw[draw=blue,fill=white]  (0.5,0,1.05) circle (0.05);		 
		   \filldraw[draw=blue, fill=white]  (0,0,1.05) circle (0.05);	
		 \filldraw[draw=blue, fill=white]  (-0.25,0,0.7) circle (0.05);
		 \filldraw[draw=blue, fill=white]  (-0.25,0,0) circle (0.05);
	\end{tikzpicture}
	=
			 \begin{tikzpicture}
     	   \foreach \z in {-0.7,0,0.7,1.4}{
      \foreach \x in {-0.5,0,0.5,1,1.5}{
           \pic at (\x,0,\z) {3dpeps}; 
             } } 
             \end{tikzpicture}.
\end{equation*}
But non-contractible loops are not absorbed by the $G$-invariance, so they can modify the state.   
They can be deformed due to the $G$-invariance of the tensors. Hence, non-contractible loops are locally non-detectable; the parent Hamiltonian cannot detect the presence of such loops. Graphically,
\begin{equation*}
\begin{tikzpicture}[scale=1.2]
  \foreach \z in {0,0.7,1.4,2.1}{
      \foreach \x in {0,0.5,1,1.5}{
           \pic at (\x,0,\z) {3dpeps};
           \filldraw [draw=black, fill=blue]  (0.75,0,\z) circle (0.05); } }    
          \draw[thick,densely dashed, magenta,rounded corners](0.7,0,2.25)--(1.13,0,2.25)--(1.13,0,1.25)--(0.37,0,1.25)--(0.37,0,2.25)--cycle; 
                                                \node[anchor=south] at (0.78,0,0) {${g}$};
         \draw[densely dotted, blue](0.75,0,-0.5)--(0.75,0,2.6);
        \end{tikzpicture}
=
\begin{tikzpicture}[scale=1.2]
  \foreach \z in {0,0.7,1.4,2.1}{
      \foreach \x in {0,0.5,1,1.5}{
           \pic at (\x,0,\z) {3dpeps}; } }         
          \draw[thick,densely dashed, magenta,rounded corners](0.7,0,2.25)--(1.13,0,2.25)--(1.13,0,1.25)--(0.37,0,1.25)--(0.37,0,2.25)--cycle; 
         \node[anchor=south] at (0.78,0,0) {${g}$};
                          \draw[densely dotted, blue, rounded corners] (0.75,0,-0.4)--(0.75,0,1.05) --(0.25,0,1.05)--(0.25,0,2.4)-- (0.75,0,2.4)--(0.75,0,2.7);
                   \filldraw [draw=black, fill=blue]  (0.75,0,0) circle (0.05);
                \filldraw [draw=black, fill=blue]   (0.75,0,0.7) circle (0.05);
        \filldraw [draw=black, fill=blue]  (0.25,0,1.4) circle (0.05);
        \filldraw [draw=black, fill=blue]  (0.25,0,2.1) circle (0.05);
           \filldraw [draw=blue, fill=white]  (0.5,0,1.05) circle (0.05);
        \filldraw [draw=black, fill=blue]  (0.5,0,2.4) circle (0.05);
 \end{tikzpicture},
\end{equation*}  
where the magenta dashed square represents the plaquette where $\mathfrak{h}$ acts on. Therefore the PEPS including such operators belongs to the ground subspace of the parent Hamiltonian. It turns out that, in fact, these loops exactly characterize the ground subspace. A basis for the ground subspace is given by the set of all pairs $g,h\in G$, $(g,h)$ such that $gh=hg$ with the equivalence relation $(g,h)\sim (xgx^{-1},xhx^{-1}),\; \forall x\in G$. These pairs, with the previous equivalence relation, are called pair conjugacy classes and correspond to two non-contractible loop operators acting on the tensor network:
\begin{equation}\label{eq:gsubspace}
H|\Psi_A(g,h)\rangle=0, \; \forall g,h\in G, \;gh=hg, \;\;
|\Psi_A(g,h)\rangle=
\begin{tikzpicture}[scale=1.2]
  \foreach \z in {0,0.7,1.4,2.1}{
      \foreach \x in {0,0.5,1,1.5}{
           \pic at (\x,0,\z) {3dpeps}; 
             \filldraw [draw=black,fill=blue] (0.75,0,\z) circle (0.05);
                      \filldraw [draw=black,fill=green] (\x,0,1.05) circle (0.05); } }
         \draw[densely dotted, blue](0.75,0,-0.5)--(0.75,0,2.6);
         \node[anchor=south] at (0.78,0,0) {${g}$};
	   \draw[densely dotted, green](-0.5,0,1.05) --(2,0,1.05);
	 \node[anchor=west] at (1.5,0,1.05) {${h}$};
 \end{tikzpicture},
\end{equation}
where the following holds $|\Psi_A(g,h)\rangle= |\Psi_A(xgx^{-1},xhx^{-1})\rangle$ for all $x\in G$ due to the $G$-invariance of the tensor.\\

An important case within the family of $G$-injective PEPS is when $u_g=L_g$ is the left regular representation, acting as $L_g|h\rangle=|gh\rangle$ on $\mathbb{C}[G]={\rm span}\{|h\rangle, h\in G\}$. This representation satisfies $\tr[L_g]=|G|\delta_{e,g}$ which implies that in this case $\mathfrak{D}=\id$. 
These states are the so-called $G$-isometric PEPS if moreover $A$ is unitary. Then, the tensor is unitarily equivalent to the operator $\mathcal{P}_G$ of Eq.(\ref{Ginje}) so w.l.o.g. 
$$A= \frac{1}{|G|}\sum_g L_g\otimes L_g\otimes L^\dagger_g\otimes L^\dagger_g$$ 
Then, the tensors satisfy the important property $\myinv{A}=\bar{A}$ so it follows that the contraction of two neighbouring sites reads:
\begin{equation}\label{Gisoconca}
\begin{tikzpicture}
   \pic at (0,0,0) {3dpeps};
    \pic at (0,0.4,0) {3dpepsdown};
    \node[anchor=north] at (0,0,0) {${A}$};
    \node[anchor=south] at (0,0.4,0) {$\bar{A}$};
       \pic at (0.6,0,0) {3dpeps};
    \pic at (0.6,0.4,0) {3dpepsdown};
    \node[anchor=north] at (0.6,0,0) {${A}$};
    \node[anchor=south] at (0.6,0.4,0) {$\bar{A}$};
  \end{tikzpicture}
   =  \frac{1}{|G|}\sum_{g\in G} \;  
\begin{tikzpicture}
         \filldraw (-0.3,0,0) circle (0.04);
     \filldraw (0,0,0.3) circle (0.04);
       \filldraw (0.8,0,0) circle (0.04);
        \filldraw (0,0,-0.3) circle (0.04); 
             \filldraw (0.5,0,0.3) circle (0.04);
        \filldraw (0.5,0,-0.3) circle (0.04); 
                            \node at (-0.25,0.2,0.3) {$\myinv{g}$};
              \node at (0.15,-0.1,0.3) {$\myinv{g}$};
                \node at (0.65,-0.1,0.3) {$\myinv{g}$};
        \node at (0.25,0.37,0.3) {${g}$};
                \node at (1,0,0.3) {${g}$};
                     \node at (0.9,0.33,0.3) {${g}$};
\draw (-0.1,0,0)--(-0.1,0.4,0);
       \begin{scope}[canvas is zx plane at y=0]
          \draw (0.1,0)--(0.5,0);
        \draw (-0.5,0)--(-0.1,0);
           \draw (0.1,0.5)--(0.5,0.5);
        \draw (-0.5,0.5)--(-0.1,0.5);
    \end{scope} 
     \begin{scope}[canvas is zx plane at y=0.4]
       \draw (-0.5,0)--(-0.1,0);
      \draw[ preaction={draw, line width=1pt, white}] (0.1,0)--(0.5,0);
             \draw (-0.5,0.5)--(-0.1,0.5);
      \draw[preaction={draw, line width=1pt, white}] (0.1,0.5)--(0.5,0.5);
    \end{scope} 
     \draw (0,0,0.1)--(0,0.4,0.1);
\draw (0,0,-0.1)--(0,0.4,-0.1);
     \draw (0.5,0,0.1)--(0.5,0.4,0.1);
\draw (0.5,0,-0.1)--(0.5,0.4,-0.1);
       \draw[preaction={draw, line width=1pt, white}] (0.6,0,0)--(0.6,0.4,0);
         \begin{scope}[canvas is zx plane at y=0]
          \draw (0,0.6)--(0,0.9);
       \draw (0,-0.4)--(0,-0.1);
          \end{scope} 
                  \begin{scope}[canvas is zx plane at y=0.4]
          \draw (0,0.6)--(0,0.9);
       \draw (0,-0.4)--(0,-0.1);
          \end{scope} 
\end{tikzpicture}.
 \end{equation}
 Eq.\eqref{Gisoconca} allows us to compute easily expectation values since 
 $$\langle\psi(A) |=|\psi(A)\rangle ^\dagger=\sum \mathcal{C}\{ \bar{A}^{i_1},\dots, \bar{A}^{i_N}\} \langle i_1\cdots i_N|.$$ 
 
 Concretely, considering a connected region $\mathcal{M}$, the contraction of each tensor $A$ inside $\mathcal{M}$ with its corresponding $\bar{A}$ tensor results in the boundary operator $\frac{1}{|G|}\sum_{b\in G} L_b\otimes \cdots \otimes L^{-1}_b$, where each term of the sum contains $|\partial \mathcal{M}|$ factors. The expectation value of a local operator acting on the complementary of $\mathcal{M}$, $\mathcal{M}^{c}$, is computed by applying the boundary operator $\frac{1}{|G|}\sum_{b\in G} L_b\otimes \cdots \otimes L^{-1}_b$ to the contraction between the tensors $A$ in $\mathcal{M}^{c}$, the local operator and the $\bar{A}$ in $\mathcal{M}^{c}$.
 
The norm of a $G$-isometric PEPS is $\sqrt{|G|^{\ell_h\times  \ell_v+1}}$ where $\ell_h$ and $\ell_v$ is the number of horizontal and vertical sites of the torus respectively; this is obtained by counting the number of loops coming from Eq.(\ref{Gisoconca}). We will omit this normalization in all calculations.\       
       
The parent Hamiltonian of a $G$-isometric PEPS is gapped, this is because the local terms $\mathfrak{h}=\id-\Pi_{2\times 2}$ are commuting orthogonal projectors, where 
\begin{equation*}
\Pi_{2\times 2}=
\begin{tikzpicture}
  \foreach \x in {0,0.6}{
      \foreach \z in {0,0.7}{
           \pic at (\x,0.5,\z) {3dpeps};
          \pic at (\x,0,\z) {3dpepsdown};
          \begin{scope}[canvas is xy plane at z=\z]
     	  \draw (-0.5,0.5) rectangle (1.1,0);
         \end{scope}
                   \begin{scope}[canvas is zy plane at x=\x]
     	  \draw (-0.5,0.5) rectangle (1.2,0);
         \end{scope}
            }}
            \end{tikzpicture} \; .
 \end{equation*}
Excitations of the parent Hamiltonian can be constructed as modifications of the tensors of the PEPS. The modified tensors do not belong to $ker(\mathfrak{h})$: they are eigenstates with eigenvalues greater than zero, that is, the energies of such excitations. The relevant excitations of the parent Hamiltonian of $G$-isometric PEPS are quasiparticle excitations that interact via braiding. This interaction does not depend on the distance, {\it i.e.}, it is of topological nature. The excitations, which are called anyons, are characterized by the set $\{ ([g],\alpha)\}$, where $[g]$ runs over all conjugacy classes of $G$ and $\alpha$ over the irreps of the normalizer of $g$. When $g=e$ they are called charges and they are characterized by the irreps $\{ \sigma\}$ of $G$, for $\alpha=1$ the quasiparticles are called fluxes and they are characterized solely by conjugacy classes of $G$. The combined object $\{ ([g],\alpha)\}$ with $g\neq e$ and $\alpha\neq 1$ is called dyon. We remark that there is a relation between anyons and ground states of these Hamiltonians. In particular the number of different anyons is the same as the dimension of the ground subspace of the Hamiltonian. This can be seen explicitly in the so-called Minimally Entangled States (MES) basis, see \cite{Buerschaper14}. In PEPS this basis is constructed as follows \cite{Norbertpriv}

\begin{equation*}
|\Psi_A([g],\alpha)\rangle= \sum_{n\in N_g}\chi_\alpha(n) |\Psi_A(g,n) \rangle
=\sum_{n\in N_g}\chi_\alpha(n)
\begin{tikzpicture}
        \pic at (0,0,0.7) {3dpeps};
        \pic at (0,0,1.4) {3dpeps};   
         \pic at (0,0,2.1) {3dpeps};  
          \pic at (0,0,0) {3dpeps};
      \pic at (0.5,0,0) {3dpeps};
      \pic at (0.5,0,0.7) {3dpeps};
      \pic at (0.5,0,1.4) {3dpeps};
     \pic at (0.5,0,2.1) {3dpeps};
      \pic at (1,0,0) {3dpeps};
        \pic at (1,0,0.7) {3dpeps};
        \pic at (1,0,1.4) {3dpeps};
         \pic at (1,0,2.1) {3dpeps};     
	  \pic at (1.5,0,0) {3dpeps};
        \pic at (1.5,0,0.7) {3dpeps};
        \pic at (1.5,0,1.4) {3dpeps};
         \pic at (1.5,0,2.1) {3dpeps};  
         
         \filldraw [draw=black, fill=blue]  (0.75,0,0.7) circle (0.04);
          \filldraw [draw=black, fill=blue]  (0.75,0,1.4) circle (0.04);
          \filldraw [draw=black, fill=blue]  (0.75,0,2.1) circle (0.04);
          \filldraw [draw=black, fill=blue]  (0.75,0,0) circle (0.04);
         \draw[densely dotted, blue](0.75,0,-0.5)--(0.75,0,2.6);
         \node[anchor=south] at (0.78,0,0) {${g}$};

         \filldraw [draw=black, fill=green]  (0,0,1.05) circle (0.04);
	 \filldraw [draw=black, fill=green]  (0.5,0,1.05) circle (0.04);
	 \filldraw [draw=black, fill=green]  (1,0,1.05) circle (0.04);
	 \filldraw [draw=black, fill=green]  (1.5,0,1.05) circle (0.04);
	   \draw[densely dotted, green](-0.5,0,1.05) --(2,0,1.05);
	 \node[anchor=west] at (1.5,0,1.05) {${n}$};
 \end{tikzpicture},
 \end{equation*}
 where $\chi_\alpha$ is the character of $\alpha$.

Fluxes and charges were studied in Ref.\cite{Schuch10} in the $G$-isometric PEPS framework, we will revisit their construction. For completeness we have constructed the dyons representation in the PEPS picture. This construction was developed in our work \cite{Garre17}, see Appendix \ref{ap:dyonic}. We remark that anyons are created in particle-antiparticle pairs that can be moved around unitarily. The PEPS representation of each type of anyon is the following:

\begin{itemize}
\item {\bf Fluxes.} 
The creation of a pair of fluxes associated with some conjugacy class $[g]$ (and its antiparticle $[g^{-1}]$), is described by a tensor product of operators: 
$$\bigotimes_{ i \in \gamma } (L_g)^{m_i},$$
 where $\gamma$ denotes a path of lattice edges connecting the plaquettes where each flux is located. For each link $i$, $m_i \in \{ -1, + 1 \}$, the precise value of $m_i$ depends on the path $\gamma$ via the $G$-invariance of the tensors. That is, if we introduce a string of $g$ operators in the following form
 \begin{equation*}
\begin{tikzpicture}
    \foreach \z in {0,0.7,1.4,2.1}{
   	   \foreach \x in {0,0.5,1,1.5,2,2.5}{
           	\pic at (\x,0,\z) {3dpeps}; 
             					} } 
 \draw[densely dotted, blue](0.5,0,1.05)--(2,0,1.05);
     \foreach \x in {0.5,1,1.5,2}{
             \filldraw [draw=black, fill=blue]  (\x,0,1.05) circle (0.04);
             				   }
 \end{tikzpicture}
 \; ,
\end{equation*}
the string can be modified according to the orientation chosen in \cref{Ginva} or equivalently using the moves of \cref{eq:Gmoves}.

Due to the $G$-injectivity, the chosen representative  $g\in [g]$ does not matter. Diagrammatically, we will represent such operators as a string of blue circles or circumferences on the edges of the square lattice. 
In particular the excitations, eigenstates with eigenvalue $1$ of $\mathfrak{h}_i$ and thus of $H$, are localized in the final plaquettes since the intermediate string can be moved using the $G$-invariance of the tensors:

\begin{equation*}
\begin{tikzpicture}[scale=1.2]
  \foreach \z in {0,0.7,1.4,2.1}{
      \foreach \x in {0,0.5,1}{
           \pic at (\x,0,\z) {3dpeps}; 
             } }    
                        \draw[densely dotted, blue, rounded corners] (0.75,0,0.7)--(0.75,0,0.35)--(0.25,0,0.35)--(0.25,0,0.7)-- (0.25,0,1.4)--(0.25,0,1.75)--(0,0,1.75);
	\filldraw[draw=black,fill=blue]  (0.75,0,0.7) circle (0.05);
	\filldraw[draw=black,fill=blue]  (0.5,0,0.35) circle (0.05);
	\filldraw[draw=black,fill=blue]  (0,0,1.75) circle (0.05);
		\filldraw[draw=blue,fill=white]  (0.25,0,0.7) circle (0.05);
	\filldraw[draw=blue,fill=white]  (0.25,0,1.4) circle (0.05);
 \end{tikzpicture} 
 =
\begin{tikzpicture}[scale=1.2]
  \foreach \z in {0,0.7,1.4,2.1}{
      \foreach \x in {0,0.5,1}{
           \pic at (\x,0,\z) {3dpeps}; 
             } }
             \draw[densely dotted, blue, rounded corners] (0,0,1.75)--(0.75,0,1.75)--(0.75,0,1.4);
	\filldraw[draw=blue,fill=white]  (0.75,0,1.4) circle (0.05);
		\filldraw[draw=black,fill=blue]  (0,0,1.75) circle (0.05);
	\filldraw[draw=black,fill=blue]  (0.5,0,1.75) circle (0.05);
	 \end{tikzpicture} 
\end{equation*}

\item {\bf Charges.}  The virtual operator associated with a pair of charges is 
$$\Pi_{\sigma}=\sum_{g,h\in G} \chi_\sigma(h^{-1}g)|g\rangle \langle g|\otimes |h\rangle \langle h|,$$
 where $\sigma$ is the label of an irrep of $G$,  $\chi_\sigma(\cdot)=\tr[\pi_\sigma(\cdot)]$ is its character and the tensor product represents the fact that the two particles are acting on two different virtual edges. 
If $\sigma$ is a one-dimensional irrep, $\Pi_{\sigma}$ can be factorized as follows: 
$$\Pi_{\sigma}=\left(\sum_{g\in G} \chi_\sigma(g)|g\rangle \langle g|\right)\otimes \left(\sum_{g\in G} \chi_\sigma(h^{-1}) |h\rangle \langle h|\right) \equiv C_{\sigma}\otimes \bar{C}_{\sigma}.$$ 
Otherwise $\Pi_\sigma$ is a sum of factors: $\Pi_{\sigma}\equiv\sum_{h\in G} C_{\sigma,h}\otimes \bar{C}_{\sigma,h}$. In any case, $\Pi_{\sigma}$ will be represented as two orange rectangles on two different edges of the lattice: 
 \begin{equation*} 
 \begin{tikzpicture}[scale=1.2]
  \foreach \z in {0,0.7,1.4}{
      \foreach \x in {0,0.5,1,1.5}{
           \pic at (\x,0,\z) {3dpeps}; 
             } } 
     	\filldraw[draw=black,fill=orange] (0.25,0,0.55) rectangle (0.25,0,0.85);           
	\filldraw[draw=black,fill=orange] (1.25,0,0.55) rectangle (1.25,0,0.85);
 \end{tikzpicture} .
\end{equation*}
Each term of the pair creates an excitation on the two neighbouring plaquettes of the edge where the operator is placed. We will denote by $\mathcal{O}_{ \sigma} (x,y)$  the physical operator that creates a particle-antiparticle of type $\sigma$ on the edges $x,y$ respectively.

\item {\bf Dyons.} Given $h\in G$ we denote the normalizer subgroup of this element as $N_h=\{ n\in G | nh=hn\}$. The normalizer of another element in the conjugacy class of $h$, $h^g\equiv ghg^{-1}\in [h]$, is $N_{h^g}=gN_hg^{-1}$. So the normalizers of the elements of a conjugacy class are all isomorphic, and the expression $N_{[h]}$ is meaningful. We can decompose the group $G$ in right cosets of $N_h$ with representatives $k_1=e,k_2,\cdots,k_\kappa$ where $\kappa=|G|/|N_h|$. A relation between these cosets and elements of the conjugacy class can be given by $h_j=k_jh k^{-1}_j$. \\

A dyon-antidyon pair, associated to $([h],\alpha)$ where $\alpha$ is an irrep of $N_{[h]}$, is represented virtually as (the explicit construction is given in Appendix \ref{ap:dyonic}):
$$ \sum_{n,m \in N_{h}} \chi_\alpha( nm^{-1}) L_h \left (\sum^\kappa_{j=1} |n k_j\rangle\langle n k_j|\right)\otimes L^{\otimes \ell}_h\otimes L_h \left (\sum^\kappa_{i=1} |m k_i\rangle\langle mk_i|\right),$$
where we have chosen the element $h$ as the representative of the conjugacy class, $\chi_\alpha$ is the character of the irrep $\alpha$ of $N_{[h]}$ and  the tensor products represent the different edges where the operators are placed on. Notice that we have chosen one of the equivalent virtual representations of the mobile string, in this case a straight line:

\begin{equation*}
\begin{tikzpicture}[scale=1.2]
        \pic at (0,0,2.1) {3dpeps};
        \pic at (0,0,1.4) {3dpeps}; 
         \pic at (0,0,0.7) {3dpeps}; 
      \pic at (0.5,0,0.7) {3dpeps};
      \pic at (0.5,0,1.4) {3dpeps};
     \pic at (0.5,0,2.1) {3dpeps};
       \pic at (1,0,0.7) {3dpeps};
        \pic at (1,0,1.4) {3dpeps};
         \pic at (1,0,2.1) {3dpeps};
           \pic at (1.5,0,0.7) {3dpeps};
        \pic at (1.5,0,1.4) {3dpeps};
         \pic at (1.5,0,2.1) {3dpeps};
                    \filldraw [yellow,rotate around={45:(0.805,0,2.1)}, draw=black]  (0.73,0,2) rectangle (0.88,0,2.2);
                              \filldraw [yellow,rotate around={45:(0.805,0,0.7)}, draw=black]  (0.73,0,0.6) rectangle (0.88,0,0.8);
	\filldraw [draw=black,fill=blue] (0.7,0, 2.1) circle (0.05);
	\filldraw [draw=black,fill=blue] (0.7,0,1.4) circle (0.05);
	\filldraw [draw=black,fill=blue] (0.7,0,0.7) circle (0.05);
        \draw[densely dotted, blue](0.7,0,2.1)--(0.7,0,0.7);
 \end{tikzpicture} 
\end{equation*}

This operator consists of a chain of $L_h$, corresponding to the flux part, ended in an operator representing the compatible charge part acting on the two ends. Here we focus on one end of the operator for simplicity and we define
\begin{equation}
\label{dyonend}
D^w_\alpha \equiv \sum_{n \in N_{h}} \chi_\alpha(w n) \sum^\kappa_{j=1} |nk_j\rangle\langle n k_j|,
\end{equation}
where $w\in N_h$ corresponds to the internal state of the charge part. We represent graphically this operator as a yellow rotated square at the end of the flux chain as follows:

\begin{equation*}
\begin{tikzpicture}[scale=1.2]
  \foreach \z in {0.7,1.4,2.1}{
      \foreach \x in {0,0.5,1,1.5}{
           \pic at (\x,0,\z) {3dpeps}; 
             } } 
          \filldraw [yellow,rotate around={45:(0.805,0,1.4)}, draw=black]  (0.73,0,1.3) rectangle (0.88,0,1.5);
	\filldraw [draw=black,fill=blue] (0.7,0,2.1) circle (0.05);
	\filldraw [draw=black,fill=blue] (0.7,0,1.4) circle (0.05);
        \draw[densely dotted, blue](0.7,0,2.5)--(0.7,0,1.4);
 \end{tikzpicture},
\end{equation*}
where we omit the other end that contains the inverse charge.

\end{itemize}
\subsection{Braiding properties of the anyons}\label{subsec:braiding}
Anyons can interact via braiding. We review the braiding properties of fluxes and charges shown in Ref. \cite{Schuch11}. We also introduce some properties of dyons develop in our work \cite{Garre17}.
\

The braiding of a flux with a flux is trivial, since the effect is a conjugation by another group element which does not change the conjugacy class. Braiding a charge with a charge is also trivial, because they are point-like operators in the virtual d.o.f. We will show what is the effect of braiding counter-clockwise a flux around a charge in $G$-isometric PEPS. We first create a charge -- anti-charge pair $\Pi_{\sigma}$ and a pair of fluxes, characterised by the conjugacy class $[g]$, and we focus on only one flux.
Using the $G$-injectivity of the tensors, we see that the action of this braiding is the conjugation by $L^\dagger_g$ on the operator of the charge:

	 \begin{equation*}      
	 \begin{tikzpicture}[scale=1.2]
	   \foreach \z in {0,0.7,1.4}{
      \foreach \x in {0,0.5,1,1.5,2}{
           \pic at (\x,0,\z) {3dpeps}; 
             } } 
     	\filldraw[draw=black,fill=orange] (0.75,0,0.55) rectangle (0.75,0,0.85);           
	\filldraw[draw=black,fill=orange] (1.75,0,0.55) rectangle (1.75,0,0.85);
		\filldraw[draw=black,fill=blue]  (0.75,0,1.4) circle (0.05);
		 \draw[densely dotted, blue, rounded corners] (0.75,0,2)--(0.75,0,1.05) --(1.25,0,1.05)--(1.25,0,0.35)-- (0.25,0,0.35)--(0.25,0,1.05)--(0.5,0,1.05);
		 \filldraw[draw=blue,fill=white]  (1,0,1.05) circle (0.05);
		 \filldraw[draw=black,fill=blue]  (1.25,0,0.7) circle (0.05);
		 \filldraw[draw=black,fill=blue]  (1,0,0.35) circle (0.05);
		  \filldraw[draw=blue,fill=white]  (0.5,0,1.05) circle (0.05);
		 \filldraw[draw=blue,fill=white]  (0.25,0,0.7) circle (0.05);
		 \filldraw[draw=black,fill=blue]  (0.5,0,0.35) circle (0.05);
	\end{tikzpicture}
	=
	 \begin{tikzpicture}[scale=1.2]
     	   \foreach \z in {0,0.7,1.4}{
      \foreach \x in {0,0.5,1,1.5,2}{
           \pic at (\x,0,\z) {3dpeps}; 
             } } 
  \filldraw[draw=blue,fill=white]  (0.63,0,0.7) circle (0.05);
     	\filldraw[draw=black,fill=orange] (0.75,0,0.55) rectangle (0.75,0,0.85);   
	  \filldraw[draw=black,fill=blue]  (0.87,0,0.7) circle (0.05);        
	\filldraw[draw=black,fill=orange] (1.75,0,0.55) rectangle (1.75,0,0.85);
	\draw[densely dotted, blue] (0.75,0,1.4) --(0.75,0,2) ;
		\filldraw[draw=black,fill=blue]  (0.75,0,1.4) circle (0.05);
	\end{tikzpicture}
	=
		 \begin{tikzpicture}[scale=1.2]
     	   \foreach \z in {0,0.7,1.4}{
      \foreach \x in {0,0.5,1,1.5,2}{
           \pic at (\x,0,\z) {3dpeps}; 
             } } 
     	\filldraw[draw=black,fill=orange] (0.75,0,0.55) rectangle (0.75,0,0.85);           
	\filldraw[draw=black,fill=orange] (1.75,0,0.55) rectangle (1.75,0,0.85);
		\filldraw[draw=black,fill=blue]  (0.75,0,1.4) circle (0.05);
		 \draw[densely dotted, blue, rounded corners] (0.75,0,2)--(0.75,0,1.05) --(1.25,0,1.05)--(1.25,0,-0.35)-- (-0.25,0,-0.35)--(-0.25,0,1.05)--(0.5,0,1.05);
		 \filldraw[draw=blue, fill=white]  (1,0,1.05) circle (0.05);
		 \filldraw[draw=black, fill=blue]  (1.25,0,0.7) circle (0.05);
		 \filldraw[draw=black, fill=blue]  (1.25,0,0) circle (0.05);
		 \filldraw[draw=black, fill=blue]  (1,0,-0.35) circle (0.05);
		  \filldraw[draw=black, fill=blue]  (0.5,0,-0.35) circle (0.05);
		  \filldraw[draw=black, fill=blue]  (0,0,-0.35) circle (0.05);
		  \filldraw[draw=blue,fill=white]  (0.5,0,1.05) circle (0.05);		 
		   \filldraw[draw=blue, fill=white]  (0,0,1.05) circle (0.05);	
		 \filldraw[draw=blue, fill=white]  (-0.25,0,0.7) circle (0.05);
		 \filldraw[draw=blue, fill=white]  (-0.25,0,0) circle (0.05);
	\end{tikzpicture},
\end{equation*}
where with the last drawing we want to emphasize that the braiding does not depend on the distance: only the surrounded topological charge matters.
That is, this action transforms $\Pi_{\sigma}$ as
\begin{align}\label{braidcwf}
B^{[\sigma]}_g(\Pi_{\sigma})& = \sum_{h\in G} \tau_{g^{-1}}(C_{\sigma, h})\otimes \bar{C}_{\sigma,h} = \sum_{h,t\in G}\chi_\sigma(t^{-1}h)L^{\dagger}_g |h\rangle \langle h| L_g\otimes  |t\rangle \langle t| \notag \\
  &= \sum_{h,t\in G}\chi_\sigma(t^{-1}gh) |h\rangle \langle h| \otimes |t\rangle \langle t|  ,
\end{align}
where $B^{[\sigma]}_g$ stands for braiding with $g$ on one charge $\sigma$ of the pair and $\tau_g$ for the conjugation with $g$. If any one of the anyons involved is abelian, i.e. if $\sigma$ is one-dimensional or $g$ belongs to the center of $G$, $Z(G)$, the effect of the braiding is a phase factor: $B^{[\sigma]}_g(\Pi_{\sigma}) = (\chi_\sigma(g)/d_\sigma) \Pi_{\sigma}$. Otherwise to detect the effect of the braiding, we project with the initial state (the configuration with only charges). Using $G$-injectivity, Eq.(\ref{Ginje}), we see that
 \begin{equation}\label{braid-result}
 \langle  \;  \mathcal{O}^\dagger_{ \sigma} (x,y) \; \mathcal{B}^{[\sigma]}_g \;  \mathcal{O}_{ \sigma} (x,y) \; \rangle  =  
\begin{tikzpicture}
           \pic at (0.4,0,0.7) {3dpeps};
        \pic at (0.4,0.4,0.7) {3dpepsdown};
                   \pic at (1.1,0,0.7) {3dpeps};
        \pic at (1.1,0.4,0.7) {3dpepsdown};
   \node[anchor=west] at (-0.4,0.2,0.7)  {$\cdots$};
      \node[anchor=west] at (2.7,0.2,0.7)  {$\cdots$};
         \node[rotate=45] at (1.45,0.4,0.2)  {$\cdots$};
         \node[rotate=45] at (1.45,0,1.2)  {$\cdots$};
           \pic at (1.8,0,0.7) {3dpeps};
        \pic at (1.8,0.4,0.7) {3dpepsdown};
	           \pic at (2.5,0,0.7) {3dpeps};
        \pic at (2.5,0.4,0.7) {3dpepsdown};
        	\filldraw[draw=blue,fill=white]  (0.55,0,0.7) circle (0.05);
             \filldraw[draw=black,fill=orange] (0.75,0,0.5) rectangle (0.75,0,0.9);
             	\filldraw[draw=black,fill=blue]  (0.95,0,0.7) circle (0.05);
            \filldraw[draw=black,fill=orange] (0.75,0.4,0.5) rectangle (0.75,0.4,0.9); 
                  \filldraw[draw=black,fill=orange] (2.15,0.4,0.5) rectangle (2.15,0.4,0.9 );  
            \filldraw[draw=black,fill=orange] (2.15,0,0.5) rectangle (2.15,0,0.9);	
  \end{tikzpicture}  
 \propto \; \sum_{b\in G} \;
 \begin{tikzpicture}
 \draw (-0.4,0.3) rectangle (0.4,-0.3); 
    \filldraw[draw=black,fill=orange] (-0.1,-0.4) rectangle (0.1,-0.2);
    \filldraw[draw=black,fill=orange] (-0.1,0.4) rectangle (0.1,0.2);
    \filldraw[draw=blue,fill=white]  (-0.25,-0.3) circle (0.06);
     \filldraw[draw=black,fill=blue]  (0.25,-0.3) circle (0.06);
         \node[anchor=west] at (-0.4,0)  {${b}$};
          \filldraw (-0.4,0)  circle (0.04);
 \node[anchor=west] at (0.4,0)  {$\myinv{b}$};
        \filldraw (0.4,0) circle (0.04);
 \end{tikzpicture}
\times
   \begin{tikzpicture}
    \draw (-0.4,0.3) rectangle (0.4,-0.3); 
    \filldraw[draw=black,fill=orange] (-0.1,-0.4) rectangle (0.1,-0.2);
    \filldraw[draw=black,fill=orange] (-0.1,0.4) rectangle (0.1,0.2);
      \node[anchor=west] at (-0.4,0)  {${b}$};
          \filldraw (-0.4,0)  circle (0.04);
 \node[anchor=west] at (0.4,0)  {$\myinv{b}$};
        \filldraw (0.4,0) circle (0.04);
 \end{tikzpicture},
 \end{equation}
where the sum over $b \in G$ is the only part that remains when $G$-injectivity is used to evaluate the overlap and $\mathcal{B}^{[\sigma]}_g$ stands for the physical action of the braiding. Each term of the sum of \eqref{braid-result} corresponds to
$$\sum_{h\in G} \tr[ C^{[u]}_{\sigma,h}  (bg)^{-1} C^{[d]}_{\sigma,h} (bg)  ] \times \tr[  \bar{C}^{[u]}_{\sigma,h}  b^{-1} \bar{C}^{[d]}_{\sigma,h} b ]. $$
It can be shown that
\bed
\langle  \; \mathcal{O}^\dagger_{ \sigma} (x,y) \; \mathcal{B}^{[\sigma]}_g \; \mathcal{O}_{ \sigma} (x,y) \; \rangle /\bra{\Psi_A}\Psi_A\rangle = \chi_\sigma(g)/d_\sigma.
\eed
The clockwise braiding would give $\chi_\sigma(g^{-1})/d_\sigma$. The square of this quantity is the probability of the pair of charges to fuse to the vacuum after braiding, i.e. the change in the total charge of the pair. This method allows us to identify the type of a given unknown flux using a probe charge (or a set of them)\cite{Preskill04, Schuch10}.

Let us now show some of the properties of the dyon in the framework of $G$-isometric PEPS:
\begin{itemize}

\item Self-braiding is the effect of a half exchange of a dyon and its antiparticle or equivalently a $2\pi$ rotation of one dyon. The $2\pi$ clockwise rotation of the dyon corresponds to the counter-clockwise braiding with the other quasiparticle string with the following effect: 

\begin{equation*}
\begin{tikzpicture}[scale=1.2]
     	   \foreach \z in {2.1,0.7,1.4}{
      \foreach \x in {0,0.5,1,1.5}{
           \pic at (\x,0,\z) {3dpeps}; 
             } } 
           \draw[densely dotted, blue, rounded corners ](0.65,0,2.5)--(0.65,0,1.75)--(0.25,0,1.75)--(0.25,0,1.05)--(1.25,0,1.05)--(1.25,0,1.75)--(0.75,0,1.75)--(0.75,0,1.4);
     	\filldraw [draw=black,fill=blue] (0.65,0,2.1) circle (0.04);
	\filldraw [draw=blue,fill=white] (0.5,0,1.75) circle (0.04);
	\filldraw [draw=blue,fill=white] (0.25,0,1.4) circle (0.04);
	\filldraw [draw=black,fill=blue] (0.5,0,1.05) circle (0.04);
	\filldraw [draw=black,fill=blue] (1,0,1.05) circle (0.04);
	\filldraw [draw=black,fill=blue] (1.25,0,1.4) circle (0.04);
	\filldraw [draw=blue,fill=white] (1,0,1.75) circle (0.04);
	\filldraw [draw=black,fill=blue] (0.75,0,1.4) circle (0.04);
	          \filldraw [yellow,rotate around={45:(0.845,0,1.4)}, draw=black]  (0.77,0,1.3) rectangle (0.92,0,1.5);
 \end{tikzpicture} 
=
\begin{tikzpicture}[scale=1.2]
     	   \foreach \z in {2.1,0.7,1.4}{
      \foreach \x in {0,0.5,1,1.5}{
           \pic at (\x,0,\z) {3dpeps}; 
             } } 
       	\filldraw [draw=black,fill=blue] (0.8,0,2.1) circle (0.04);
	\filldraw [draw=black,fill=blue] (0.8,0,1.4) circle (0.04);
        \draw[densely dotted, blue](0.8,0,2.5)--(0.82,0,1.4);
         \filldraw [yellow,rotate around={45:(0.705,0,1.4)}, draw=black]  (0.63,0,1.3) rectangle (0.78,0,1.5);
 \end{tikzpicture} 
\end{equation*}

In order to complete the whole $2\pi$ spin we express $ D^w_\alpha L_h$ as $L_h (L^{\dagger}_h D^w_\alpha L_h )$; the effect of this operation is the conjugation by $L^{\dagger}_h$ in the charge part of the dyon. Since $h$ is central in $N_h$ the matrix representation of $h$ is a multiple of the identity so  $L^{\dagger}_h D^w_\alpha L_h= \chi_\alpha(h)D^w_\alpha$. The corresponding topological spin is $\chi_\alpha(h)/d_\alpha$, where $d_\alpha$ is the dimension of the irrep $\alpha$.

\item Braiding with $g\in N_h$: this operation corresponds to the conjugation by $L_g$ over the string and over (\ref{dyonend}). The flux part remains invariant because $g^{-1}hg=h$ and the charge part transforms as $L^{\dagger}_g D^w_\alpha L_g= D^{g^{-1}w}_\alpha$. 

\item Braiding with $g\notin N_h$: the string gets conjugated $h\to ghg^{-1}\equiv h^{g} $ and the charge part gets also conjugated $L_g D^w_\alpha L^{\dagger}_g$. The conjugation action is given by:
$$ \sum_{n \in N_h} \chi_\alpha(w n) \sum^\kappa_{j=1} |gnk_j\rangle\langle gnk_j|.$$
In order to operate with this expression we rewrite $gnk_j=   gng^{-1} \;  (gk_jg^{-1} \;g\; \tilde{k}^{-1}_{x_j}) \;\tilde{k}_{x_j} $, where we have just inserted identities and the element $\tilde{k}_{x_j}$. To define this element let us denote the representatives of the right cosets of $G/N_{h^{g}}$ as $\tilde{k}_j=g k_jg^{-1}$ with the relation $\tilde{h}_j = \tilde{k}_j \;h^{g}\; \tilde{k}^{-1}_j $. We now denote with the index $x^{[g]}_j\in [1,\cdots,\kappa]$ the element corresponding to $\tilde{h}_{x^{[g]}_j}= \tilde{k}_{x^{[g]}_j}  h^{g}\tilde{k}^{-1}_{x^{[g]}_j}=g \tilde{h}_j g^{-1} =g \tilde{k}_j \;h^{g}\; \tilde{k}^{-1}_j g^{-1}$. By the previous definition, it is straightforward that $\tilde{n}^{-1}_{x^{[g]}_j} \equiv  \tilde{k}_j g \tilde{k}^{-1}_{x^{[g]}_j} $ belongs to $N_{h^{g}}$ and then $L_g D^w_\alpha L^{\dagger}_g$ equals

\begin{align}\label{eq:braidyon}
 \sum_{n \in N_h} \chi_\alpha(wn) \sum^\kappa_{j=1} \ket{ n^g \tilde{n}^{-1}_{x^{[g]}_j}  \tilde{k}_{x_j}  } \bra{n^g \tilde{n}^{-1}_{x^{[g]}_j}  \tilde{k}_{x_j}}  
&= \sum_{\tilde{n} \in N_{h^{g}} } \chi_\alpha(gwg^{-1}\tilde{n}) \sum^\kappa_{j=1}
 \ket{ \tilde{n} \tilde{n}^{-1}_{x^{[g]}_j}  \tilde{k}_{x_j}  } \bra{\tilde{n} \tilde{n}^{-1}_{x^{[g]}_j}  \tilde{k}_{x_j}}  \notag \\
&= \sum^\kappa_{j=1} \sum_{\tilde{n} \in N_{h^{g}} } \chi_\alpha( w^g  \tilde{n} \tilde{n}_{x^{[g]}_j} ) \ket{ \tilde{n} \tilde{k}_{x^{[g]}_j} } \bra{  \tilde{n} \tilde{k}_{x^{[g]}_j}} .
\end{align}

This action coincides with the symmetry transformations of the quantum double algebras described in \cite{Dijkgraaf91,ThesisMark}.
\end{itemize}

\subsection{Connection with quantum double models}\label{dis:qd} 

$G$-injective PEPS are the tensor network realization of the quantum doubles model of $G$, $\mathcal{D}(G)$, constructed by Kitaev \cite{Kitaev03}. Let us show this relation, from the simplest model: $G=\mathbb{Z}_2$ which is known as the Toric Code (TC), see Ref.\cite{Schuch10} for more details. Kitaev proposed his models as commuting local Hamiltonians acting on a (square) lattice:
$$H_{TC}= -\sum_s A_s-\sum_p B_p,$$
where $p,s$ run over all plaquettes and sites respectively, $ A_s=\sigma_x^{\otimes 4}$ and $B_p=\sigma_z^{\otimes 4}$: see Figure \ref{fig:QDGiso}. It can be shown that the ground subspace is four-fold degenerate. The related Hamiltonian whose GS is represented as a PEPS is the following:
$$H= \frac{1}{2}\sum_s (\id_s-A_s)+\frac{1}{2}\sum_p (\id_p-B_p).$$
It is just a constant displacement in energy of $H_{TC}$. The basis of the ground subspace can be written (up to normalization) as
$$\{ \mathcal{W}^i_x \mathcal{W}^j_z {\displaystyle \prod_{s}  }(\id_s+A_s)\ket{0}^{\otimes n}\}_{i,j}^{0,1},$$
where $\mathcal{W}^{i,j}_{x,y}$ are operators $\sigma_{x,z}$ acting on a non-contractible loop of the torus and $i,j=0$ or $1$ denote the absence or presence of these operators respectively. The state $\Pi_p(\id_p+A_s)\ket{0}^{\otimes n}$, which is the equal weight superposition of $1$-valued loops in a background of $0$'s, can be written as a PEPS with the following tensor: 
\begin{equation*}
(T)^{i}_{\alpha,\beta,\gamma,\delta}= 
\begin{tikzpicture} 
         \node at (-1.6,0.3,1) {$i$};
       \pic at (-1.6,0,1) {3dpeps};
       \node[anchor=south] at (-1.9,0,1.2) {$\alpha$};
       \node[anchor=north] at (-1.6,0,1.1) {$\sigma$};
        \node[anchor=south] at (-1,0,1.35) {$\gamma$};
       \node[anchor=south] at (-1.35,0,0.9) {$\beta$};
\end{tikzpicture}
=\delta_{i, \alpha+\beta ({\rm mod} 2)} \delta_{\beta, \gamma}\delta_{\alpha,\sigma} \; ,
\end{equation*}
where all indices take values $0$ or $1$ and the exact position on the lattice is shown in Figure \ref{fig:QDGiso}. 

\begin{figure}[ht!]
\begin{center}
\includegraphics[scale=1.5]{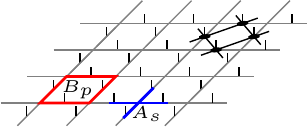}
\end{center}
\caption{The square lattice, drawn in gray, where the model of Ref. \cite{Kitaev03} is defined is shown. The physical spins live on the edges, they are draw in black. The red square represents the spins of the plaquette where the operator $B_p$ acts on. The blue cross represents the spins of the site where the operator $A_s$ acts on. The upper right corner shows how the tensor $T$ is  placed on the lattice to generate the GS of $H$.}
\label{fig:QDGiso}
\end{figure}

The tensor $T$ is invariant under the action of $\sigma_x$ on all the virtual legs: this corresponds to the $\mathbb{Z}_2$-invariance. But $T$ is also invariant under more purely virtual actions. These symmetries not coming from the $\mathbb{Z}_2$-invariance can be discarded when considering a block of four tensors. This is because that symmetry becomes local: it only acts between adjacent blocked tensors. It can be shown that the 4-block $T$ tensors are unitarily equivalent to the tensor product of a $G$-isometric  tensor and a tensor that creates a product state: the latter will not affect the topological properties.\\

Quantum doubles models of finite groups are not the most general topological models: string-net models \cite{Levin05} are believed to capture all non-chiral topological orders in 2D. They can also be represented as a tensor network, the so-called MPO-injective PEPS \cite{Sahinoglu14}. For chiral phases, that posses gapless edge modes propagating in only one direction,  models with topological order can be found in Ref.\cite{Kitaev06} or in Refs \cite{Yang15,Poilblanc15} using PEPS. Still, many questions for chiral topological PEPS are open.

One of the biggest advantages of tensor networks is that the topological order is encoded locally in the single tensors.  That encoding is a symmetry in the virtual d.o.f. of the tensor. Therefore a representation of each topological ordered phase can be given by constructing a tensor that has that symmetry.
%
%
Also if there is another property that we want to study for topological PEPS, one could wonder if there is a local encoding of  that property. Then, the interplay between the new property and the topological order could be analyzed via their local encodings. In this thesis, we are interested in the interplay between global symmetries and topological order in PEPS.  The key ingredient for that is to understand how the local tensors of two tensor networks that describe the same state must be related. This will be the content of the next chapter.

\newpage \cleardoublepage

\chapter{Fundamental theorems of tensor network states}\label{chap:FT}
A general property of TN is that two different sets of tensors can generate the same state. For instance, this happens in translational invariant MPS:
$$ \sum^{d}_{i_1,\dots,i_N=1}\tr[A^{i_1}A^{i_2}\cdots A^{i_N}]\ket{i_1,\cdots,i_N} =  \sum^{d}_{i_1,\dots,i_N=1}\tr[B^{i_1}B^{i_2}\cdots B^{i_N}]\ket{i_1,\cdots,i_N},$$
where $B^i = X A^i X^{-1}$, for any $X\in GL(D)$. This is a so-called gauge transformation: it adds an invertible matrix on each index of the tensors so that the matrices cancel out when contracting neighbouring tensors. 

A key question arises: if two sets of tensors generate the same state, is a gauge transformation the only relationship allowed between them? When this is the case, as it is the case for MPS, it has been used to define the canonical form \cite{Fannes92,PerezGarcia07,Vidal03} and to characterize global or local symmetries \cite{Sanz09,Kull17}. 
The reason for the latter is the following: if a state is symmetric, it means that an operation leaves it invariant. But it will generally change the tensors and for this particular case, it will do it with a gauge transformation. This implies that symmetries in the quantum states can be captured by transformations in the tensors. 

This question is also decisive in many other situations dealing with string order \cite{PerezGarcia08B}, topological order \cite{Sahinoglu14}, renormalization \cite{Cirac17A}, or time evolution \cite{Cirac17B}. Theorems answering such fundamental questions about the structure of TNs are typically referred to as Fundamental Theorems.

There is no Fundamental Theorem for the most general TNS. Ref. \cite{Scarpa18} shows that even for two tensors generating translationally invariant PEPS in an $N\times N$ square lattice, there cannot exist an algorithm to decide whether they will generate the same state for all $N$ or not. Therefore, some restrictions have to be imposed on the TN to avoid that no-go theorem: these restrictions are mainly imposed on the properties of the tensors.

For MPS, Fundamental Theorems have been proven for translationally invariant states \cite{Cirac17A,Cuevas17} when the tensors generate the same state for any system size. A Fundamental Theorem for MPS that are not necessarily translationally invariant has been proven in \cite{PerezGarcia10} for a fixed (but large enough) system size in the case of injective and normal tensors: the ones that become injective after blocking a few sites. In 2D, Ref. \cite{PerezGarcia10} also proves the result for normal tensors and Ref.\cite{Andras18B} proves a Fundamental Theorem for the so-called semi-injective tensors. Those theorems require only a fixed (but large enough) system size but they rely on the lattice structure to reduce to 1D techniques so they do not generalize to other geometries.  

In this chapter we prove a Fundamental Theorem for $G$-injective tensors that generate the same PEPS. There were no previous results in this direction for this class of tensors. This result is crucial for the rest of the thesis since it allows us to characterize on-site symmetries of $G$-injective PEPS. The precise statement is the following:

\begin{tcolorbox}
\begin{theorem}[Fundamental theorem for $G$-injective PEPS]\label{theo:FTGinjective}
 Suppose two G-injective tensors $A$ and $B$  generate the same PEPS for every system size on the square lattice; $|\Psi_A\rangle=|\Psi_B\rangle$. Then, there exist invertible matrices $X$ and $Y$ such that 
  \begin{equation*}
    \begin{tikzpicture}
      \pic at (0,0,0) {3dpeps};
      \node at (0.07,-0.27,0) {$A$};
      \draw (0,0,0) -- (0,0.35,0);
    \end{tikzpicture} =
    \begin{tikzpicture}
      \draw (-0.7,0,0) -- (0.7,0,0);
      \draw (0,0,-0.9) -- (0,0,0.9);
      \pic at (0,0,0) {3dpepsres};
      \node at (0.1,-0.27,0) {$B$};
      \filldraw [draw=black, fill=red] (-0.5,0,0) circle (0.05);
      \filldraw [draw=black, fill=red] (0.5,0,0) circle (0.05);
      \node[anchor=south] at (-0.5,0,0) {$\myinv{X}$};
      \node[anchor=north] at (0.5,0,0) {$X$};
      \filldraw [draw=black, fill=red] (0,0,-0.6) circle (0.05);
      \filldraw [draw=black, fill=red] (0,0,0.6) circle (0.05);
       \node[anchor=south] at (0,0,-0.6) {$Y$};
      \node[anchor=north] at (0,0,0.6) {$\myinv{Y}$};
    \end{tikzpicture}.
  \end{equation*}
\end{theorem}
\end{tcolorbox}

Before writing the proof for $G$-injective PEPS, we show the Fundamental Theorem for normal PEPS in arbitrary lattices (geometries and dimensions). This case will establish the basic tools for the proof of \cref{theo:FTGinjective}. The main ingredient of the proof is a new technique introduced in Ref.\cite{Molnar18A} which is a local reduction to the MPS case,  instead of a \emph{slice} reduction along one dimension. This Fundamental Theorem generalizes the previous results as follows. First, we do not require equality of states for all system sizes, as required for example in Ref. \cite{Cirac17A}, our result is valid for a fixed (but large enough) size which is smaller than the one required in Ref. \cite{PerezGarcia10}. Second, the TN considered here do not need to be translationally invariant, which is important when applying the results to local symmetries. 
Third, the results hold for any geometry (including, for instance, hyperbolic as it is used in the constructions of AdS/CFT correspondence \cite{Pastawski15,Hayden16}) and dimension.

\section{Injective tensor network states}

In this section we prove the two main lemmas leading to the Fundamental Theorem for non-translational invariant MPS. In the following, we consider two injective tensor networks generating the same state. The defining tensors of the two MPS are labeled by $A_s$ and $B_s$, where $s$ denote the lattice site. 

The first lemma assigns a special gauge transformation to each edge of one of the tensor networks; the second lemma shows that once these gauge transformations are absorbed into the tensors $B_s$, the resulting tensors are equal to $A_s$.

\begin{lemma}\label{lem:inj_isomorph}
 Suppose $\{A_1,A_2,A_3\}$,  $\{B_1,B_2,B_3\}$ are injective non-necessarily translational invariant tensors that generate the same tripartite MPS. Then for every edge and for every matrix $X$  there is a matrix $Y$ such that
  \begin{equation*}
    \begin{tikzpicture}
      \draw (0.5,0) rectangle (3.5,-0.5);
      \foreach \x in {1,2,3}{
        \node[tensor,label=below:$A_\x$] (t\x) at (\x,0) {};
        \draw (t\x) --++ (0,0.5);
      }
      \node[tensorr, draw=black,fill=red,label=above:$X$] (x) at (1.5,0) {};
    \end{tikzpicture} = 
    \begin{tikzpicture}
      \draw (0.5,0) rectangle (3.5,-0.5);
      \foreach \x in {1,2,3}{
        \node[tensor,label=below:$B_\x$] (t\x) at (\x,0) {};
        \draw (t\x) --++ (0,0.5);
      }
      \node[tensorr, draw=black,fill=red,label=above:$Y$] (x) at (1.5,0) {};
    \end{tikzpicture} \ .
  \end{equation*} 
  Moreover, $X$ and $Y$ have the same dimension and there is an invertible matrix $Z$ such that $Y=Z^{-1}X Z$. This $Z$ is uniquely defined up to multiplication with a constant.
\end{lemma}

This Lemma will be used to assign a local gauge transformation to all edges on one of two tensor networks generating the same state. These local gauges will then be incorporated in the defining tensors; doing so will lead us to two tensor networks where inserting any matrix $X$ on any bond simultaneously in the two networks gives two new states that are still equal. 

The proof of \cref{lem:inj_isomorph} is based on the observation that any operator acting on one edge, at the virtual level, can be realized by a physical operation on either of the neighboring sites.  And vice versa: two physical operations on neighboring sites that transform the state in the same way correspond to the same virtual operator acting on the bond connecting the two sites. Given two tensor networks generating the same state, this correspondence establishes an isomorphism between the algebra of virtual operations. The basis change realizing this isomorphism is the local gauge relating the two tensors.

Before proceeding to the proof, notice that due to injectivity of the tensors using \eqref{eq:injec}, if 
\begin{equation}\label{eq:equalbininj}
  \begin{tikzpicture}
    \draw (0.5,0) rectangle (3.5,-0.5);
    \foreach \x in {1,2,3}{
      \node[tensor,label=below:$A_\x$] (t\x) at (\x,0) {};
      \draw (t\x) --++ (0,0.5);
    }
    \node[tensorr,  draw=black,fill=red,label=above:$X_1$] (x) at (1.5,0) {};
  \end{tikzpicture} = 
  \begin{tikzpicture}
    \draw (0.5,0) rectangle (3.5,-0.5);
    \foreach \x in {1,2,3}{
      \node[tensor,label=below:$A_\x$] (t\x) at (\x,0) {};
      \draw (t\x) --++ (0,0.5);
    }
    \node[tensorr,  draw=black,fill=red,label=above:$X_2$] (x) at (1.5,0) {};
  \end{tikzpicture} \ ,
\end{equation}
then $X_1=X_2$.

\begin{proof}[Proof of \cref{lem:inj_isomorph}]
Consider now a deformation of the TN by inserting a matrix $X$ on one of the bonds. This deformation can be realized by physical operations acting on either of the two neighbouring sites:
\begin{equation*}
  \begin{tikzpicture}[baseline=-0.1cm]
    \draw (0.5,0) rectangle (3.5,-0.5);
    \foreach \x in {1,2,3}{
      \node[tensor,label=below:$A_\x$] (t\x) at (\x,0) {};
      \draw[ ] (t\x) --++ (0,0.5);
    }
    \node[tensorr,  draw=black,fill=red,label=above:$X$] (x) at (1.5,0) {};
  \end{tikzpicture} = 
  \begin{tikzpicture}[baseline=-0.1cm]
    \draw (0.5,0) rectangle (3.5,-0.5);
    \foreach \x in {1,2,3}{
      \node[tensor,label=below:$A_\x$] (t\x) at (\x,0) {};
      \draw[ ] (t\x) --++ (0,0.5);
    }
    \node[tensorr,  draw=black,fill=red,label=left:$O_1(X)$] (o) at (1,0.5) {};
    \draw[ ] (o)--++(0,0.5); 
  \end{tikzpicture} = 
  \begin{tikzpicture}[baseline=-0.1cm]
    \draw (0.5,0) rectangle (3.5,-0.5);
    \foreach \x in {1,2,3}{
      \node[tensor,label=below:$A_\x$] (t\x) at (\x,0) {};
      \draw[ ] (t\x) --++ (0,0.5);
    }
    \node[tensorr,  draw=black,fill=red,label=left:$O_2(X)$] (o) at (2,0.5) {};
    \draw[ ] (o)--++(0,0.5); 
  \end{tikzpicture} \ ,
\end{equation*} 
where
\begin{equation}\label{eq:X->O}
   O_1(X) = 
  \begin{tikzpicture}
    \draw (-0.5,0) rectangle (0.5,1.2);
    \node[tensor,label=below:$A_1$] (t) at (0,1.2) {};
    \node[tensor,label=above:$A_1^{-1}$] (b) at (0,0) {};
    \node[tensorr,draw=black,fill=red,label=right:$X$] (x) at (0.5,0.6) {};
    \draw[ ] (t)--++(0,0.5);
    \draw[ ] (b)--++(0,-0.5);
  \end{tikzpicture} \quad \text{and} \quad 
   O_2(X) = 
  \begin{tikzpicture}
    \draw (-0.5,0) rectangle (0.5,1.2);
    \node[tensor,label=below:$A_2$] (t) at (0,1.2) {};
    \node[tensor,label=above:$A_2^{-1}$] (b) at (0,0) {};
    \node[tensorr,draw=black,fill=red,label=left:$X$] (x) at (-0.5,0.6) {};
    \draw[ ] (t)--++(0,0.5);
    \draw[ ] (b)--++(0,-0.5);
  \end{tikzpicture} \ .     
\end{equation}
This can be checked easily using \eqref{eq:injec}. It is important to note that the mappings $X\mapsto  O_1(X)$ and $X\mapsto O_2^T(X)$ are algebra homomorphisms as they are linear, $O_1(\alpha[X_1+X_2])=\alpha [O_1(X_1)+O_1(X_2)]$, and satisfy $O_1(X_1X_2)=O_1(X_1)O_1(X_2)$. The virtual bonds of the tensors $A_s$ read from left to right, thus the loops in \cref{eq:X->O} read clockwise; hence the transpose in the mapping $X\mapsto O_2^T(X)$. Also these mappings do not depend on $A_3$.

Consider now the converse: two physical operations on neighboring sites that maps the MPS to the same state:
\begin{equation}\label{eq:resonate}
  \begin{tikzpicture}[baseline=-0.1cm]
    \draw (0.5,0) rectangle (3.5,-0.5);
    \foreach \x in {1,2,3}{
      \node[tensor,label=below:$B_\x$] (t\x) at (\x,0) {};
      \draw[ ] (t\x) --++ (0,0.5);
    }
    \node[tensorr,  draw=black,fill=red,label=left:$S_1$] (o) at (1,0.5) {};
    \draw[ ] (o)--++(0,0.5); 
  \end{tikzpicture} = 
  \begin{tikzpicture}[baseline=-0.1cm]
    \draw (0.5,0) rectangle (3.5,-0.5);
    \foreach \x in {1,2,3}{
      \node[tensor,label=below:$B_\x$] (t\x) at (\x,0) {};
      \draw[ ] (t\x) --++ (0,0.5);
    }
    \node[tensorr,  draw=black,fill=red,label=left:$S_2$] (o) at (2,0.5) {};
    \draw[ ] (o)--++(0,0.5); 
  \end{tikzpicture} \ .
\end{equation}
 We apply $B^{-1}_2$ and $B^{-1}_3$ on both sides of equality \eqref{eq:resonate} and we arrive at
\begin{equation}\label{eq:inj_O->X_argument}
  \begin{tikzpicture}
    \node[tensor,label=below:$B_1$] (l) at (0,-0.6) {};
      \draw[ ] (l)--++(0,0.6);
    \draw (l)++(-0.6,0)--(l)--++(0.6,0);
    \node[tensorr, draw=black,fill=red,label=left:$S_1$] (o) at ($(l)+(0,0.5)$) {};
    \draw[  ] (o) --++(0,0.5);
  \end{tikzpicture} = 
  \begin{tikzpicture}
    \node[tensor,label=below:$B_1$] (l) at (0,-0.6) {};
    \node[tensor,label=below:$B_2$] (r) at (1,-0.6) {};
    \foreach \x in {l,r}{
      \draw[ ] (\x)--++(0,0.6);
    }
    \draw (l)++(-0.6,0)--(l)--(r)--++(0.6,0);
    \node[tensorr, draw=black,fill=red,label=left:$S_2$] (o) at ($(r)+(0,0.5)$) {};
    \draw[ ] (o) --++(0,0.5);      
    \node[tensor,label=above:$B_2^{-1}$] (i) at ($ (r) + (0,1)$) {};
    \draw (i)++(-0.6,0)--(i)--++(0.6,0)--++(0,-1);
        \draw[rounded corners=2mm,red] (0.2,-1) rectangle (2.3,1);
         \node[anchor=north] at (1.9,0) {$ D_{23}^{-1}$};
    \node[red] at (1,-1.3) {$W$};
  \end{tikzpicture}  =  
  \begin{tikzpicture}
    \draw (-0.5,0)--(1.5,0);
    \node[tensor,label=below:$B_1$] (a) at (0,0) {};
    \draw[ ] (a)--++(0,0.5);
    \node[tensorr, draw=black,fill=red,label=below:$W$] (x) at (1,0) {};
  \end{tikzpicture} \ ,
\end{equation}  
for some matrix $W$ and where $D_{23}$ is the dimension of the vector space on the bond between sites $2$ and $3$.  Similarly, inverting $B_1$ and $B_3$, we arrive at
\begin{equation*}
  \begin{tikzpicture}
    \draw (-0.5,0)--(0.5,0);
    \node[tensor,label=below:$B_2$] (a) at (0,0) {};
    \draw[ ] (a)--++(0,0.5);
    \node[tensorr,  draw=black,fill=red,label=left:$S_2$] (o) at (0,0.5) {};
    \draw[ ] (o)--++(0,0.5); 
  \end{tikzpicture} = 
  \begin{tikzpicture}
    \draw (-1.5,0)--(0.5,0);
    \node[tensor,label=below:$B_2$] (a) at (0,0) {};
    \draw[ ] (a)--++(0,0.5);
    \node[tensorr, draw=black,fill=red,label=below:$V$] (x) at (-1,0) {};
  \end{tikzpicture} \ ,
\end{equation*}
for some matrix $V$.  Therefore
\begin{equation*}
  \begin{tikzpicture}[baseline=-0.1cm]
    \draw (0.5,0) rectangle (3.5,-0.5);
    \foreach \x in {1,2,3}{
      \node[tensor,label=below:$B_\x$] (t\x) at (\x,0) {};
      \draw[ ] (t\x) --++ (0,0.5);
    }
    \node[tensorr,  draw=black,fill=red,label=above:$W$] (x) at (1.5,0) {};
  \end{tikzpicture} = 
  \begin{tikzpicture}[baseline=-0.1cm]
    \draw (0.5,0) rectangle (3.5,-0.5);
    \foreach \x in {1,2,3}{
      \node[tensor,label=below:$B_\x$] (t\x) at (\x,0) {};
      \draw[ ] (t\x) --++ (0,0.5);
    }
    \node[tensorr,  draw=black,fill=red,label=left:$S_1$] (o) at (1,0.5) {};
    \draw[ ] (o)--++(0,0.5); 
  \end{tikzpicture} = 
  \begin{tikzpicture}[baseline=-0.1cm]
    \draw (0.5,0) rectangle (3.5,-0.5);
    \foreach \x in {1,2,3}{
      \node[tensor,label=below:$B_\x$] (t\x) at (\x,0) {};
      \draw[ ] (t\x) --++ (0,0.5);
    }
    \node[tensorr,  draw=black,fill=red,label=left:$S_2$] (o) at (2,0.5) {};
    \draw[ ] (o)--++(0,0.5); 
  \end{tikzpicture} \ = \ 
  \begin{tikzpicture}[baseline=-0.1cm]
    \draw (0.5,0) rectangle (3.5,-0.5);
    \foreach \x in {1,2,3}{
      \node[tensor,label=below:$B_\x$] (t\x) at (\x,0) {};
      \draw[ ] (t\x) --++ (0,0.5);
    }
    \node[tensorr,  draw=black,fill=red,label=above:$V$] (x) at (1.5,0) {};
  \end{tikzpicture} \ ,
\end{equation*}
and thus by injectivity, using \eqref{eq:equalbininj}, $V=W$. Therefore
\begin{equation} \label{eq:O->X2}
  \begin{tikzpicture}
    \draw (-0.5,0)--(0.5,0);
    \node[tensor,label=below:$B_1$] (a) at (0,0) {};
    \draw[ ] (a)--++(0,0.5);
    \node[tensorr,  draw=black,fill=red,label=left:$S_1$] (o) at (0,0.5) {};
    \draw[ ] (o)--++(0,0.5); 
  \end{tikzpicture} = 
  \begin{tikzpicture}
    \draw (-0.5,0)--(1.5,0);
    \node[tensor,label=below:$B_1$] (a) at (0,0) {};
    \draw[ ] (a)--++(0,0.5);
    \node[tensorr, draw=black,fill=red, label=below:$W$] (x) at (1,0) {};
  \end{tikzpicture} \quad \text{and} \quad 
  \begin{tikzpicture}
    \draw (-0.5,0)--(0.5,0);
    \node[tensor,label=below:$S_2$] (a) at (0,0) {};
    \draw[ ] (a)--++(0,0.5);
    \node[tensorr,  draw=black,fill=red,label=left:$S_2$] (o) at (0,0.5) {};
    \draw[ ] (o)--++(0,0.5); 
  \end{tikzpicture} = 
  \begin{tikzpicture}
    \draw (-1.5,0)--(0.5,0);
    \node[tensor,label=below:$B_2$] (a) at (0,0) {};
    \draw[ ] (a)--++(0,0.5);
    \node[tensorr,draw=black,fill=red, label=below:$W$] (x) at (-1,0) {};
  \end{tikzpicture} \ ,
\end{equation}
and the maps $S_1\mapsto W$ and $S_2^T \mapsto W$ are uniquely defined and are  algebra homomorphisms. 

Consider now two three-site non-necessarily translational invariant injective MPS generating the same state:
\begin{equation*}
  \begin{tikzpicture}
    \draw (0.5,0) rectangle (3.5,-0.5);
    \foreach \x in {1,2,3}{
      \node[tensor,label=below:$A_\x$] (t\x) at (\x,0) {};
      \draw[ ] (t\x) --++ (0,0.5);
    }
  \end{tikzpicture} = 
  \begin{tikzpicture}
    \draw (0.5,0) rectangle (3.5,-0.5);
    \foreach \x in {1,2,3}{
      \node[tensor,label=below:$B_\x$] (t\x) at (\x,0) {};
      \draw[ ] (t\x) --++ (0,0.5);
    }
  \end{tikzpicture} \ .
\end{equation*}
Deform the MPS on the LHS by inserting a matrix $X$ on one of the bonds. By the above arguments, this deformation is equivalent to any of the two physical operations:
\begin{equation*}
  \begin{tikzpicture}[baseline=-0.1cm]
    \draw (0.5,0) rectangle (3.5,-0.5);
    \foreach \x in {1,2,3}{
      \node[tensor,label=below:$A_\x$] (t\x) at (\x,0) {};
      \draw[ ] (t\x) --++ (0,0.5);
    }
    \node[tensorr,  draw=black,fill=red,label=above:$X$] (x) at (1.5,0) {};
  \end{tikzpicture} = 
  \begin{tikzpicture}[baseline=-0.1cm]
    \draw (0.5,0) rectangle (3.5,-0.5);
    \foreach \x in {1,2,3}{
      \node[tensor,label=below:$A_\x$] (t\x) at (\x,0) {};
      \draw[ ] (t\x) --++ (0,0.5);
    }
    \node[tensorr,  draw=black,fill=red,label=left:$O_1(X)$] (o) at (1,0.5) {};
    \draw[  ] (o)--++(0,0.5); 
  \end{tikzpicture} = 
  \begin{tikzpicture}[baseline=-0.1cm]
    \draw (0.5,0) rectangle (3.5,-0.5);
    \foreach \x in {1,2,3}{
      \node[tensor,label=below:$A_\x$] (t\x) at (\x,0) {};
      \draw[ ] (t\x) --++ (0,0.5);
    }
    \node[tensorr,  draw=black,fill=red,label=left:$O_2(X)$] (o) at (2,0.5) {};
    \draw[ ] (o)--++(0,0.5); 
  \end{tikzpicture} \ .
\end{equation*} 
As the MPS defined by the $A$ and $B$ tensors is the same state, these physical operators also satisfy
\begin{equation*}
  \begin{tikzpicture}[baseline=-0.1cm]
    \draw (0.5,0) rectangle (3.5,-0.5);
    \foreach \x in {1,2,3}{
      \node[tensor,label=below:$A_\x$] (t\x) at (\x,0) {};
      \draw[ ] (t\x) --++ (0,0.5);
    }
    \node[tensorr,  draw=black,fill=red,label=above:$X$] (x) at (1.5,0) {};
  \end{tikzpicture} = 
  \begin{tikzpicture}[baseline=-0.1cm]
    \draw (0.5,0) rectangle (3.5,-0.5);
    \foreach \x in {1,2,3}{
      \node[tensor,label=below:$B_\x$] (t\x) at (\x,0) {};
      \draw[ ] (t\x) --++ (0,0.5);
    }
    \node[tensorr,  draw=black,fill=red,label=left:$O_1(X)$] (o) at (1,0.5) {};
    \draw[ ] (o)--++(0,0.5); 
  \end{tikzpicture} = 
  \begin{tikzpicture}[baseline=-0.1cm]
    \draw (0.5,0) rectangle (3.5,-0.5);
    \foreach \x in {1,2,3}{
      \node[tensor,label=below:$B_\x$] (t\x) at (\x,0) {};
      \draw[ ] (t\x) --++ (0,0.5);
    }
    \node[tensorr,  draw=black,fill=red,label=left:$O_2(X)$] (o) at (2,0.5) {};
    \draw[ ] (o)--++(0,0.5); 
  \end{tikzpicture} \ ,
\end{equation*} 
and thus, by \cref{eq:O->X2}, for every $X$ there is a matrix $Y$ such that 
\begin{equation*}
  \begin{tikzpicture}
    \draw (0.5,0) rectangle (3.5,-0.5);
    \foreach \x in {1,2,3}{
      \node[tensor,label=below:$A_\x$] (t\x) at (\x,0) {};
      \draw[ ] (t\x) --++ (0,0.5);
    }
    \node[tensorr,  draw=black,fill=red,label=above:$X$] (x) at (1.5,0) {};
  \end{tikzpicture} = 
  \begin{tikzpicture}
    \draw (0.5,0) rectangle (3.5,-0.5);
    \foreach \x in {1,2,3}{
      \node[tensor,label=below:$B_\x$] (t\x) at (\x,0) {};
      \draw[ ] (t\x) --++ (0,0.5);
    }
    \node[tensorr,  draw=black,fill=red,label=above:$Y$] (x) at (1.5,0) {};
  \end{tikzpicture} \ .
\end{equation*}
Due to injectivity of the $B$ tensors, the mapping $X\mapsto Y$ is uniquely defined. Due to injectivity of the $A$ tensors, it is an injective map. As the argument is symmetric with respect of the exchange of the $A$ and $B$ tensors, it also has to be surjective (for every $Y$ there is a corresponding $X$) and therefore the map $X\mapsto Y$ is a bijection (a one-to-one mapping). Moreover, it is clear from the construction that it is an algebra homomorphism, as both $X\mapsto O_1$ and $O_1\mapsto Y$ are algebra homomorphisms. Therefore the mapping $X\mapsto Y$ is an algebra isomorphism. 

Since the isomorphism is between finite dimensional spaces and $X$ (and $Y$) can be any matrix this implies that the bond dimensions on the edges of the two different tensor network are the same. Moreover the algebra is simple, full matrix algebra, so the algebra isomorphism is of the form  $X\mapsto Y=ZXZ^{-1}$ for some invertible $Z$  where $Z$ is uniquely defined (up to a multiplicative constant); this is a consequence of the Skolem-Noether Theorem.
\end{proof}

\begin{lemma}\label{lem:inj_equal_tensors}
Let $A_1,A_2$ and $B_1,B_2$ be injective MPS tensors. Suppose that for all $X$ and $Y$
  \begin{equation*}
    \begin{tikzpicture}
      \draw (0.5,0) rectangle (2.5,-0.5);
      \foreach \x in {1,2}{
        \node[tensor,label=below:$A_\x$] (t\x) at (\x,0) {};
        \draw[ ] (t\x) --++ (0,0.5);
      }
      \node[tensorr,  draw=black,fill=red,label=above:$X$] (x) at (1.5,0) {};
    \end{tikzpicture} = 
    \begin{tikzpicture}
      \draw (0.5,0) rectangle (2.5,-0.5);
      \foreach \x in {1,2}{
        \node[tensor,label=below:$B_\x$] (t\x) at (\x,0) {};
        \draw[ ] (t\x) --++ (0,0.5);
      }
      \node[tensorr,  draw=black,fill=red,label=above:$X$] (x) at (1.5,0) {};
    \end{tikzpicture} \quad \text{and} \quad
    \begin{tikzpicture}
      \draw (0.5,0) rectangle (2.5,-0.5);
      \foreach \x in {1,2}{
        \node[tensor,label=below:$A_\x$] (t\x) at (\x,0) {};
        \draw[ ] (t\x) --++ (0,0.5);
      }
      \node[tensorr,  draw=black,fill=red,label=below:$Y$] (x) at (1.5,-0.5) {};
    \end{tikzpicture} = 
    \begin{tikzpicture}
      \draw (0.5,0) rectangle (2.5,-0.5);
      \foreach \x in {1,2}{
        \node[tensor,label=below:$B_\x$] (t\x) at (\x,0) {};
        \draw[ ] (t\x) --++ (0,0.5);
      }
      \node[tensorr,  draw=black,fill=red,label=below:$Y$] (x) at (1.5,-0.5) {};
    \end{tikzpicture} \ .    
  \end{equation*}
  Then $A_1 = \lambda B_1$ and $A_2 = \lambda^{-1} B_2$ for some constant $\lambda$.
\end{lemma}

\begin{proof}
  From the first equation as $X$ can be any matrix and in particular any projection to the space basis:
  \begin{equation*}
    \begin{tikzpicture}
      \draw (0.5,0)--(2.5,0);
      \node[tensor,label=below:$A_2$] (t1) at (1,0) {};
      \node[tensor,label=below:$A_1$] (t2) at (2,0) {};
      \draw[ ] (t1) --++ (0,0.5);
      \draw[ ] (t2) --++ (0,0.5);
    \end{tikzpicture} \ = \ 
    \begin{tikzpicture}
      \draw (0.5,0)--(2.5,0);
      \node[tensor,label=below:$B_2$] (t1) at (1,0) {};
      \node[tensor,label=below:$B_1$] (t2) at (2,0) {};
      \draw[ ] (t1) --++ (0,0.5);
      \draw[ ] (t2) --++ (0,0.5);
    \end{tikzpicture} \ .
  \end{equation*}
  Similarly, from the second equation,
  \begin{equation*}
    \begin{tikzpicture}
      \draw (0.5,0)--(2.5,0);
      \node[tensor,label=below:$A_1$] at (1,0) {};
      \node[tensor,label=below:$A_2$] at (2,0) {};
      \draw[ ] (t1) --++ (0,0.5);
      \draw[ ] (t2) --++ (0,0.5);
    \end{tikzpicture} \ = \ 
    \begin{tikzpicture}
      \draw (0.5,0)--(2.5,0);
      \node[tensor,label=below:$B_1$] at (1,0) {};
      \node[tensor,label=below:$B_2$] at (2,0) {};
      \draw[ ] (t1) --++ (0,0.5);
      \draw[ ] (t2) --++ (0,0.5);
    \end{tikzpicture} \ .
  \end{equation*}  
  Therefore, applying $A_2^{-1}$ to both equations, we get that
  \begin{equation*}
    \begin{tikzpicture}
      \draw (-0.5,0)--(0.5,0);
      \node[tensor,label=below:$A_1$] (a) at (0,0) {};
      \draw[ ] (a)--++(0,0.5);
    \end{tikzpicture} \ = \ 
    \begin{tikzpicture}
      \draw (-0.5,0)--(1.5,0);
      \node[tensor,label=below:$B_1$] (a) at (0,0) {};
      \draw[ ] (a)--++(0,0.5);
      \node[tensorr, draw=black,fill=red, label=below:$Z$] (x) at (1,0) {};
    \end{tikzpicture} \ = \ 
    \begin{tikzpicture}
      \draw (-0.5,0)--(1.5,0);
      \node[tensor,label=below:$B_1$] (a) at (1,0) {};
      \draw[ ] (a)--++(0,0.5);
      \node[tensorr, draw=black,fill=red, label=below:$W$] (x) at (0,0) {};
    \end{tikzpicture} \ ,
  \end{equation*}
  for some matrices $Z$ and $W$. Applying the inverse of $B_1$, we conclude that both $Z$ and $W$ are proportional to identity and hence $A_1 = \lambda B_1$. Using the same arguments we can conclude that $A_2 = \mu B_2$ for some other constant $\mu$ and $\mu = 1/\lambda$.
\end{proof}

\subsection{Injective MPS}

We now show how to use the previous lemmas for injective MPS to prove the Fundamental Theorem. This is a special case of the next section, but we present it here to explain the main ideas.

\begin{theorem}[Fundamental theorem for injective MPS]\label{thm:inj_MPS}
  Let the tensors $A_s$ and $B_s$ define two injective, non-necessarily translational invariant MPS on at least three particles. Suppose they generate the same state:
  \begin{equation*}
    \ket{\Psi} = 
    \begin{tikzpicture}
      \draw (0.5,0) rectangle (4.5,-0.5);
      \foreach \x/\t in {1/1,2/2,4/n}{
        \node[tensor,label=below:$A_\t$] (t\x) at (\x,0) {};
        \draw[ ] (t\x) --++ (0,0.5);
      }
      \node[fill=white] at (3,0) {$\dots$};
    \end{tikzpicture} = 
    \begin{tikzpicture}
      \draw (0.5,0) rectangle (4.5,-0.5);
      \foreach \x/\t in {1/1,2/2,4/n}{
        \node[tensor,label=below:$B_\t$] (t\x) at (\x,0) {};
        \draw[ ] (t\x) --++ (0,0.5);
      }
      \node[fill=white] at (3,0) {$\dots$};
    \end{tikzpicture} \ .
  \end{equation*}  
  Then there are invertible matrices $Z_s$ ($s=1,...,n+1$, $Z_{n+1}=Z_1$) such that 
    \begin{equation*}
      \begin{tikzpicture}
        \draw (-0.5,0)--(0.5,0);
        \node[tensor,label=below:$B_s$] (t) at (0,0) {};
        \draw[ ] (t)--(0,0.5);
      \end{tikzpicture}  = 
      \begin{tikzpicture}
        \draw (-1,0)--(1,0);
        \node[tensorr,draw=black,fill=red, label=below:$Z_s^{-1}$] at (-0.5,0) {};
        \node[tensorr,draw=black,fill=red,label=below:\ $Z_{s+1}$\vphantom{$Z_s^{-1}$}] at (0.5,0) {};
        \node[tensor,label=below:$A_s$\vphantom{$Z_s^{-1}$}] (t) at (0,0) {};
        \draw[ ] (t)--(0,0.5);
      \end{tikzpicture} \ .
    \end{equation*}  
    Moreover, the matrices $Z_s$ are unique up to a multiplicative constant.
\end{theorem}

\begin{proof}
  First let us choose any edge, for example the edge between sites $1$ and $2$. Let us block the tensors $A_3,\dots, A_n$ (and $B_3,\dots ,B_n$) into one tensor:
  \begin{align*}
    \begin{tikzpicture}
      \draw (0.5,0) -- (1.5,0);
      \node[tensor,label=below:$a$] (t) at (1,0) {};
      \draw[ ] (t) --++ (0,0.5);
    \end{tikzpicture} &= 
    \begin{tikzpicture}
      \draw (2.5,0) -- (6.5,0);
      \foreach \x/\t in {3/3,4/4,6/n}{
        \node[tensor,label=below:$A_\t$] (t\x) at (\x,0) {};
        \draw[ ] (t\x) --++ (0,0.5);
      }
      \node[fill=white] at (5,0) {$\dots$};
    \end{tikzpicture} \; , \\
    \begin{tikzpicture}
      \draw (0.5,0) -- (1.5,0);
      \node[tensor,label=below:$b$] (t) at (1,0) {};
      \draw[ ] (t) --++ (0,0.5);
    \end{tikzpicture} &= 
    \begin{tikzpicture}
      \draw (2.5,0) -- (6.5,0);
      \foreach \x/\t in {3/3,4/4,6/n}{
        \node[tensor,label=below:$B_\t$] (t\x) at (\x,0) {};
        \draw[ ] (t\x) --++ (0,0.5);
      }
      \node[fill=white] at (5,0) {$\dots$};
    \end{tikzpicture} \ .
  \end{align*}
  As injectivity is preserved under blocking, both $a$ and $b$ are injective tensors. With this notation, the MPS can be written as a non-translational invariant MPS on three sites:
  \begin{equation*}
    \ket{\Psi} = 
    \begin{tikzpicture}
      \draw (0.5,0) rectangle (3.5,-0.5);
      \foreach \x/\t in {1/A_1,2/A_2,3/a}{
        \node[tensor,label=below:$\t$] (t\x) at (\x,0) {};
        \draw[ ] (t\x) --++ (0,0.5);
      }
    \end{tikzpicture} = 
    \begin{tikzpicture}
      \draw (0.5,0) rectangle (3.5,-0.5);
      \foreach \x/\t in {1/B_1,2/B_2,3/b}{
        \node[tensor,label=below:$\t$] (t\x) at (\x,0) {};
        \draw[ ] (t\x) --++ (0,0.5);
      }
    \end{tikzpicture} \ .
  \end{equation*}    
Therefore \cref{lem:inj_isomorph} can be applied leading to a gauge transformation realized by the invertible matrix $Z_2$ on the edge between sites $1$ and $2$ that, for all $X$ with  $Y = Z_2^{-1} X Z_2$, satisfies
  \begin{equation*}
    \begin{tikzpicture}
      \draw (0.5,0) rectangle (3.5,-0.5);
      \foreach \x/\t in {1/$A_1$,2/$A_2$,3/$a$}{
        \node[tensor,label=below:\t] (t\x) at (\x,0) {};
        \draw[ ] (t\x) --++ (0,0.5);
      }
      \node[tensorr,  draw=black,fill=red,label=above:$X$] (x) at (1.5,0) {};
    \end{tikzpicture} = 
    \begin{tikzpicture}
      \draw (0.5,0) rectangle (3.5,-0.5);
      \foreach \x/\t in {1/$B_1$,2/$B_2$,3/$b$}{
        \node[tensor,label=below:\t] (t\x) at (\x,0) {};
        \draw[ ] (t\x) --++ (0,0.5);
      }
      \node[tensorr,  draw=black,fill=red,label=above:$Y$] (x) at (1.5,0) {};
    \end{tikzpicture} \ .
  \end{equation*}
The lemma can be applied to all edges leading to gauge transformations realized by the matrices $Z_s$ on the edge between sites $s-1$ and $s$. After incorporating these gauges into the tensors $B_s$, to define the new tensors $\tilde{B_s}$
  \begin{equation}\label{eq:B tilde}
    \begin{tikzpicture}
      \draw (-0.5,0)--(0.5,0);
      \node[tensor,label=below:$\tilde{B}_s$] (t) at (0,0) {};
      \draw[ ] (t)--(0,0.5);
    \end{tikzpicture}  = 
    \begin{tikzpicture}
      \draw (-1,0)--(1,0);
      \node[tensorr,draw=black,fill=red,label=below:$Z_s$\vphantom{$Z_{i+1}^{-1}$}] at (-0.5,0) {};
      \node[tensorr,draw=black,fill=red, label=below:\ $Z_{i+1}^{-1}$] at (0.5,0) {};
      \node[tensor,label=below:$B_s$\vphantom{$Z_{s+1}^{-1}$}] (t) at (0,0) {};
      \draw[ ] (t)--(0,0.5);
    \end{tikzpicture} \ ,
  \end{equation}  
  we arrive at two MPS with the property that on every bond for every matrix $X$
  \begin{equation*} 
    \begin{tikzpicture}
      \draw (0.5,0) rectangle (4.5,-0.6);
      \foreach \x/\t in {1/1,2/2,4/n}{
        \node[tensor,label=below:$A_\t$] (t\x) at (\x,0) {};
        \draw[ ] (t\x) --++ (0,0.5);
      }
      \node[fill=white] at (3,0) {$\dots$};
      \node[tensorr,  draw=black,fill=red,label=above:$X$] (x) at (1.5,0) {};
    \end{tikzpicture} = 
    \begin{tikzpicture}
      \draw (0.5,0) rectangle (4.5,-0.6);
      \foreach \x/\t in {1/1,2/2,4/n}{
        \node[tensor,label=below:$\tilde{B}_\t$] (t\x) at (\x,0) {};
        \draw[ ] (t\x) --++ (0,0.5);
      }
      \node[fill=white] at (3,0) {$\dots$};
      \node[tensorr,  draw=black,fill=red,label=above:$X$] (x) at (1.5,0) {};
    \end{tikzpicture} \ .
  \end{equation*}
  In particular,
  \begin{equation*} 
    \begin{tikzpicture}
      \draw (0.5,0) rectangle (4.5,-0.6);
      \foreach \x/\t in {1/1,2/2,4/n}{
        \node[tensor,label=below:$A_\t$] (t\x) at (\x,0) {};
        \draw[ ] (t\x) --++ (0,0.5);
      }
      \node[fill=white] at (3,0) {$\dots$};
      \node[tensorr,  draw=black,fill=red,label=below:$Y$] (x) at (1.5,-0.6) {};
    \end{tikzpicture} = 
    \begin{tikzpicture}
      \draw (0.5,0) rectangle (4.5,-0.6);
      \foreach \x/\t in {1/1,2/2,4/n}{
        \node[tensor,label=below:$\tilde{B}_\t$] (t\x) at (\x,0) {};
        \draw[ ] (t\x) --++ (0,0.5);
      }
      \node[fill=white] at (3,0) {$\dots$};
      \node[tensorr,  draw=black,fill=red,label=below:$Y$] (x) at (1.5,-0.6) {};
    \end{tikzpicture} \ .
  \end{equation*}
  Let us now block the MPS into a bipartite MPS:
  \begin{equation*}
    \ket{\Psi} = 
    \begin{tikzpicture}
      \draw (0.5,0) rectangle (2.5,-0.5);
      \foreach \x/\t in {1/A_1,2/a}{
        \node[tensor,label=below:$\t$] (t\x) at (\x,0) {};
        \draw[ ] (t\x) --++ (0,0.5);
      }
    \end{tikzpicture} = 
    \begin{tikzpicture}
      \draw (0.5,0) rectangle (2.5,-0.5);
      \foreach \x/\t in {1/\tilde{B}_1,2/b}{
        \node[tensor,label=below:$\t$] (t\x) at (\x,0) {};
        \draw[ ] (t\x) --++ (0,0.5);
      }
    \end{tikzpicture}\ ,  
  \end{equation*}
  with 
  \begin{align*} 
    \begin{tikzpicture}
      \node[tensor,label=below:$a$] (t) at (1,0) {};
      \draw[ ] (t) --++ (0,0.5);
      \draw (0.5,0)--(1.5,0);
    \end{tikzpicture} & = 
    \begin{tikzpicture}
      \draw (1.5,0) -- (4.5,0);
      \foreach \x/\t in {2/2,4/n}{
        \node[tensor,label=below:$A_\t$] (t\x) at (\x,0) {};
        \draw[ ] (t\x) --++ (0,0.5);
      }
      \node[fill=white] at (3,0) {$\dots$};
    \end{tikzpicture} \\
    \begin{tikzpicture}
      \node[tensor,label=below:$b$] (t) at (1,0) {};
      \draw[ ] (t) --++ (0,0.5);
      \draw (0.5,0)--(1.5,0);
    \end{tikzpicture} & = 
    \begin{tikzpicture}
      \draw (1.5,0) -- (4.5,0);
      \foreach \x/\t in {2/2,4/n}{
        \node[tensor,label=below:$\tilde{B}_\t$] (t\x) at (\x,0) {};
        \draw[ ] (t\x) --++ (0,0.5);
      }
      \node[fill=white] at (3,0) {$\dots$};
    \end{tikzpicture} \ . 
  \end{align*}  
  After this blocking, the requirements of \cref{lem:inj_equal_tensors} are satisfied, therefore $A_1 = \lambda_1 \tilde{B}_1$. Using the same argument for all sites $s$,  $A_s = \lambda_s \tilde{B}_s$ and $\prod_s \lambda_s = 1$. Notice that these $\lambda_s$ can be sequentially absorbed in the matrices $Z_s$ in \cref{eq:B tilde}. 
\end{proof}

Notice that if the two MPS are translationally invariant, i.e. the tensors at each site are the same, then the gauges relating them are also translationally invariant (up to a constant), as 
\begin{equation*}
  \begin{tikzpicture}
    \draw (-1,0)--(1,0);
    \node[tensorr,draw=black,fill=red,label=below:$Z_{s-1}^{-1}$] at (-0.5,0) {};
    \node[tensorr,draw=black,fill=red,label=below:$Z_{s}$\vphantom{$Z_{s-1}^{-1}$}] at (0.5,0) {};
    \node[tensor,label=below:$A$\vphantom{$Z_{s-1}^{-1}$}] (t) at (0,0) {};
    \draw[ ] (t)--(0,0.5);
  \end{tikzpicture} \ =
  \begin{tikzpicture}
    \draw (-1,0)--(1,0);
    \node[tensorr,draw=black,fill=red,label=below:$Z_s^{-1}$] at (-0.5,0) {};
    \node[tensorr,draw=black,fill=red,label=below:$Z_{s+1}$\vphantom{$Z_s^{-1}$}] at (0.5,0) {};
    \node[tensor,label=below:$A$\vphantom{$Z_s^{-1}$}] (t) at (0,0) {};
    \draw[ ] (t)--(0,0.5);
  \end{tikzpicture} \ \Rightarrow \ Z_s \propto Z_{s+1},
\end{equation*}  
which can be seen by inverting the tensor $A$. We conclude therefore that
\begin{corollary}
  Let the tensors $A$ and $B$ define two injective, translationally invariant MPS on $n\geq 3$ sites. Suppose they generate the same state:
  \begin{equation*}
    \ket{\Psi} = 
    \begin{tikzpicture}
      \draw (0.5,0) rectangle (4.5,-0.5);
      \foreach \x/\t in {1/1,2/2,4/n}{
        \node[tensor,label=below:$A$] (t\x) at (\x,0) {};
        \draw[ ] (t\x) --++ (0,0.5);
      }
      \node[fill=white] at (3,0) {$\dots$};
    \end{tikzpicture} = 
    \begin{tikzpicture}
      \draw (0.5,0) rectangle (4.5,-0.5);
      \foreach \x/\t in {1/1,2/2,4/n}{
        \node[tensor,label=below:$B$] (t\x) at (\x,0) {};
        \draw[ ] (t\x) --++ (0,0.5);
      }
      \node[fill=white] at (3,0) {$\dots$};
    \end{tikzpicture} \ .
  \end{equation*}  
  Then there is an invertible matrix $Z$ and a constant $\lambda\in\mathbb{C}$, $\lambda^n = 1$, such that 
    \begin{equation*}
      \begin{tikzpicture}
        \draw (-0.5,0)--(0.5,0);
        \node[tensor,label=below:$B$] (t) at (0,0) {};
        \draw[ ] (t)--(0,0.5);
      \end{tikzpicture}  =  \lambda \cdot\ 
      \begin{tikzpicture}
        \draw (-1,0)--(1,0);
        \node[tensorr,draw=black,fill=red,label=below:$Z^{-1}$] at (-0.5,0) {};
        \node[tensorr,draw=black,fill=red,label=below:$Z$\vphantom{$Z^{-1}$}] at (0.5,0) {};
        \node[tensor,label=below:$A$\vphantom{$Z^{-1}$}] (t) at (0,0) {};
        \draw[ ] (t)--(0,0.5);
      \end{tikzpicture} \ .
    \end{equation*}    
Moreover, the matrix $Z$ is unique up to a multiplicative constant.    
\end{corollary}

\subsection{Injective PEPS}\label{sec:injPEPS}
In this section we will work with injective TNS placed in arbitrary lattices. We say that the tensor network is \emph{injective} if all tensors interpreted as maps from the virtual space to the physical space are injective. This is just the generalization of the MPS definition of Eq.\eqref{eq:injec1}. Again, injectivity is equivalent to the tensor having an inverse, as in the MPS case -Eq.\eqref{eq:injec}- or the square lattice PEPS case -Eq.\eqref{eq:PEPSinj}-. Similarly, the contraction of two injective tensors results in an injective tensor.
 
An example of a PEPS placed on a non-square lattice is the following:
\begin{equation*}
  \begin{tikzpicture}
  \node[tensor] (c1) at (0,0) {};
    \node[tensor] (c2) at (1,0) {};
    \node[tensor] (c3) at (1.5,1) {}; 
    \node[tensor] (c4) at (0.6,1.5) {};
    \node[tensor] (c5) at (-0.6,1.5) {};
    \foreach \x in {1,2,3,4,5}{
      \draw[ ] (c\x) --++(0,0.5);
    }
    \draw (c1)--(c2)--(c3)--(c4)--(c5)--(c1);
    \draw (c2)--(c5);
  \end{tikzpicture} \ .
\end{equation*}

One can group sites of the PEPS together treating the corresponding tensors as one bigger tensor. This regrouping can naturally be reflected in PEPS to block tensor networks to a three site MPS as follows. Choose one edge of the PEPS and group together all vertices except the two that connects the chosen edge. This regrouped tensor together with the two endpoints of the chosen edge forms a three-site MPS. 

For the example above we can do the following:  
\begin{equation}\label{eq:block_to_mps}
  \begin{tikzpicture}
    \node[tensor] (c1) at (0,0) {};
    \node[tensor] (c2) at (1,0) {};
    \node[tensor] (c3) at (1.5,1) {}; 
    \node[tensor] (c4) at (0.6,1.5) {};
    \node[tensor] (c5) at (-0.6,1.5) {};
    \foreach \x in {1,2,3,4,5}{
      \draw[ ] (c\x) --++(0,0.5);
    }
    \draw (c1)--(c2)--(c3)--(c4)--(c5)--(c1);
    \draw (c2)--(c5);
%
    \draw[red] (c1) circle (0.3cm);
    \draw[red] (c5) circle (0.3cm);
    \draw[rounded corners=2mm,red] (0.4,-0.2) rectangle (1.7,1.7);
    \node[red] at (-0.45,-0.35) {$A'_2$};
    \node[red] at (-1.1,1.8) {$A'_1$};
    \node[red] at (1.7,1.8) {$A'_3$};
  \end{tikzpicture} \ \Rightarrow \ 
  \begin{tikzpicture}
    \draw (0.5,0) rectangle (3.5,-0.6);
    \foreach \x in {1,2,3}{
      \node[tensor,red, label=below:$A'_\x$] (t\x) at (\x,0) {};
      \draw[ ] (t\x) --++ (0,0.5);
    }
  \end{tikzpicture}  \,
\end{equation}
where the resulting MPS is injective. We consider now two injective PEPS defined on the same graph that generate the same state,
\begin{equation}\label{eq:TN_5_particle_eq}
  \begin{tikzpicture}
    \node[tensor, label=below:$A_1$] (c1) at (0,0) {};
    \node[tensor, label=below:$A_2$] (c2) at (1,0) {};
    \node[tensor, label=below right:$A_3$] (c3) at (1.5,1) {}; 
    \node[tensor, label=below:$A_4$] (c4) at (0.6,1.5) {};
    \node[tensor, label=below left:$A_5$] (c5) at (-0.6,1.5) {};
    \foreach \x in {1,2,3,4,5}{
      \draw[ ] (c\x) --++(0,0.5);
    }
    \draw (c1)--(c2)--(c3)--(c4)--(c5)--(c1);
    \draw (c2)--(c5);
  \end{tikzpicture} \ = \ 
  \begin{tikzpicture}
    \node[tensor, label=below:$B_1$] (c1) at (0,0) {};
    \node[tensor, label=below:$B_2$] (c2) at (1,0) {};
    \node[tensor, label=below right:$B_3$] (c3) at (1.5,1) {};
    \node[tensor, label=below:$B_4$] (c4) at (0.6,1.5) {};
    \node[tensor, label=below left:$B_5$] (c5) at (-0.6,1.5) {};
    \foreach \x in {1,2,3,4,5}{
      \draw[ ] (c\x) --++(0,0.5);
    }
    \draw (c1)--(c2)--(c3)--(c4)--(c5)--(c1);
    \draw (c2)--(c5);
  \end{tikzpicture} \ .
\end{equation}

After blocking to a three-site MPS as described above, we arrive at two injective MPS generating the same state; hence \cref{lem:inj_isomorph} can be applied. This establishes a gauge transformation on the edge between sites $1$ and $5$ of the original PEPS. Similar regrouping can be done around every edge; applying then \cref{lem:inj_isomorph} results in a gauge transformation assigned to every edge. Define now the tensors $\tilde{B}_s$ by absorbing the matrices corresponding to these gauge transformations into the tensors $B_s$. For the resulting PEPS, we have that for every edge and matrix $X$ the following is satisfied
\begin{equation}\label{eq:inj_equal_edge}
  \begin{tikzpicture}
    \node[tensor, label=below:$A_1$] (c1) at (0,0) {};
    \node[tensor, label=below:$A_2$] (c2) at (1,0) {};
    \node[tensor, label=below right:$A_3$] (c3) at (1.5,1) {};
    \node[tensor, label=below:$A_4$] (c4) at (0.6,1.5) {};
    \node[tensor, label=below left:$A_5$] (c5) at (-0.6,1.5) {};
    \foreach \x in {1,2,3,4,5}{
      \draw[ ] (c\x) --++(0,0.5);
    }
    \draw (c1)--(c2)--(c3)--(c4)--(c5)--(c1);
    \draw (c2)--(c5);
    \node[draw=black,fill=red,tensorr, label=right:$X$] at ($(c2)!0.5!(c5)$) {};
  \end{tikzpicture} \ = \ 
  \begin{tikzpicture}
    \node[tensor, label=below:$\tilde{B}_1$] (c1) at (0,0) {};
    \node[tensor, label=below:$\tilde{B}_2$] (c2) at (1,0) {};
    \node[tensor, label=below right:$\tilde{B}_3$] (c3) at (1.5,1) {};
    \node[tensor, label=below:$\tilde{B}_4$] (c4) at (0.6,1.5) {};
    \node[tensor, label=below left:$\tilde{B}_5$] (c5) at (-0.6,1.5) {};
    \foreach \x in {1,2,3,4,5}{
      \draw[ ] (c\x) --++(0,0.5);
    }
    \draw (c1)--(c2)--(c3)--(c4)--(c5)--(c1);
    \draw (c2)--(c5);
    \node[draw=black,fill=red,tensorr, label=right:$X$] at ($(c2)!0.5!(c5)$) {};
  \end{tikzpicture} \ .
\end{equation}
To conclude that $A_s = \lambda_s \tilde{B}_s$ as in the previous section, we will need to use a more general version of \cref{lem:inj_equal_tensors}:
\begin{lemma}\label{lem:inj_equal_tensors_2}
  Let $A_1,A_2$ and $B_1,B_2$ be injective tensors. Suppose that for all $X$ and for all edges the following is true
  \begin{equation}\label{eq:lem_inj_eq_ten_2}
    \begin{tikzpicture}
      \foreach \x/\t in {1/1,3/2}{
        \node[tensor,label=below:$A_\t$] (t\t) at (\x,0) {};
        \draw[ ] (t\x) --++ (0,0.5);
      }
      \draw (t1) to[bend left=30] node[midway,tensorr, draw=black,fill=red,label=above:$X$] {} (t2);
      \draw (t1) to[bend right=30] (t2);
      \draw (t1) to[bend right=45] (t2);
      \node at (2,0.05) {$\vdots$};
    \end{tikzpicture} = 
    \begin{tikzpicture}
      \foreach \x/\t in {1/1,3/2}{
        \node[tensor,label=below:$B_\t$] (t\t) at (\x,0) {};
        \draw[ ] (t\x) --++ (0,0.5);
      }
      \draw (t1) to[bend left=30] node[midway,tensorr, draw=black,fill=red,label=above:$X$] {} (t2);
      \draw (t1) to[bend right=30] (t2);
      \draw (t1) to[bend right=45] (t2);
      \node at (2,0.05) {$\vdots$};
    \end{tikzpicture}  \ .    
  \end{equation}
  Then $A_1 = \lambda B_1$ and $A_2 = \lambda^{-1} B_2$ for some constant $\lambda$.  
\end{lemma}
\begin{proof}
  W.l.o.g.\ suppose that there are three lines connecting the tensors. Similar to the proof of \cref{lem:inj_equal_tensors}, if \cref{eq:lem_inj_eq_ten_2} holds for all $X$, then
  \begin{align*}
    \begin{tikzpicture}
      \foreach \x/\t in {1/1,3/2}{
        \node[tensor,label=below:$A_\t$] (t\t) at (\x,0) {};
        \draw[ ] (t\x) --++ (0,0.5);
      }
      \draw (t1) to[bend left=30]  (t2);
      \draw (t1) to[bend right=10] (t2);
      \draw (t1) to[bend right=45] (t2);
      \fill[white] (1.3,0.1) rectangle (2.7,0.4);
    \end{tikzpicture} &= 
    \begin{tikzpicture}
      \foreach \x/\t in {1/1,3/2}{
        \node[tensor,label=below:$B_\t$] (t\t) at (\x,0) {};
        \draw[ ] (t\x) --++ (0,0.5);
      }
      \draw (t1) to[bend left=30]  (t2);
      \draw (t1) to[bend right=10] (t2);
      \draw (t1) to[bend right=45] (t2);
      \fill[white] (1.3,0.1) rectangle (2.7,0.4);
    \end{tikzpicture}  \\
    \begin{tikzpicture}
      \foreach \x/\t in {1/1,3/2}{
        \node[tensor,label=below:$A_\t$] (t\t) at (\x,0) {};
        \draw[ ] (t\x) --++ (0,0.5);
      }
      \draw (t1) to[bend left=30]  (t2);
      \draw (t1) to[bend right=10] (t2);
      \draw (t1) to[bend right=45] (t2);
      \fill[white] (1.3,-0.15) rectangle (2.7,-0.5);
    \end{tikzpicture} &= 
    \begin{tikzpicture}
      \foreach \x/\t in {1/1,3/2}{
        \node[tensor,label=below:$B_\t$] (t\t) at (\x,0) {};
        \draw[ ] (t\x) --++ (0,0.5);
      }
      \draw (t1) to[bend left=30]  (t2);
      \draw (t1) to[bend right=10] (t2);
      \draw (t1) to[bend right=45] (t2);
      \fill[white] (1.3,-0.15) rectangle (2.7,-0.5);
    \end{tikzpicture}  \\ 
    \begin{tikzpicture}
      \foreach \x/\t in {1/1,3/2}{
        \node[tensor,label=below:$A_\t$] (t\t) at (\x,0) {};
        \draw[ ] (t\x) --++ (0,0.5);
      }
      \draw (t1) to[bend left=30]  (t2);
      \draw (t1) to[bend right=10] (t2);
      \draw (t1) to[bend right=45] (t2);
      \fill[white] (1.3,-0.15) rectangle (2.7,0.1);
    \end{tikzpicture} &= 
    \begin{tikzpicture}
      \foreach \x/\t in {1/1,3/2}{
        \node[tensor,label=below:$B_\t$] (t\t) at (\x,0) {};
        \draw[ ] (t\x) --++ (0,0.5);
      }
      \draw (t1) to[bend left=30]  (t2);
      \draw (t1) to[bend right=10] (t2);
      \draw (t1) to[bend right=45] (t2);
      \fill[white] (1.3,-0.15) rectangle (2.7,0.1);
    \end{tikzpicture}  \ .              
  \end{align*} 
  Applying now the inverse of $A_2$, we conclude that 
  \begin{equation*}
    \begin{tikzpicture}
      \node[tensor, label=below:$A_1$] (a) {};
      \draw[ ] (a) --++ (0,0.5);
      \draw (a) -- (30:0.8);
      \draw (a) -- (-10:0.8);
      \draw (a) -- (-45:0.8);
    \end{tikzpicture} = 
    \begin{tikzpicture}
      \node[tensor, label=below:$B_1$] (a) {};
      \draw[ ] (a) --++ (0,0.5);
      \draw (a) -- (30:0.8);
      \draw (a) -- (-10:0.8);
      \draw (a) -- (-45:0.8);
       \draw[black, line width=1.2mm] (-57:0.6) arc (-57:2:0.6);
      \draw[red, line width=1mm] (-55:0.6) arc (-55:0:0.6);
      \node at (-27.5:0.9) {$Z$};
    \end{tikzpicture} = 
    \begin{tikzpicture}
      \node[tensor, label=below:$B_1$] (a) {};
      \draw[ ] (a) --++ (0,0.5);
      \draw (a) -- (30:0.8);
      \draw (a) -- (-10:0.8);
      \draw (a) -- (-45:0.8);
      \draw[red, line width=1mm] (-20:0.6) arc (-20:40:0.6);
      \node at (10:0.9) {$U$};
    \end{tikzpicture} = 
    \begin{tikzpicture}
      \node[tensor, label=below:$B_1$] (a) {};
      \draw[ ] (a) --++ (0,0.5);
      \draw (a) -- (30:0.8);
      \draw (a) -- (-10:0.4);
      \draw (a) -- (-45:0.8);
      \draw[red, line width=1mm] (-55:0.6) arc (-55:40:0.6);
      \node at (-7.5:0.9) {$W$};
    \end{tikzpicture} \ . 
  \end{equation*}
  Inverting $B_1$ we arrive at the following equations that satisfy the matrices $Z,U,W$: 
  \begin{equation*}
    \sum_j \id \otimes Z^{(1)}_j \otimes Z^{(2)}_j= \sum_j  U^{(1)}_j \otimes U^{(2)}_j \otimes \id = \sum_j  W^{(1)}_j \otimes \id  \otimes W^{(2)}_j , 
  \end{equation*}
  where we have written 
  \begin{align*}
    Z &=\sum_j Z_j^{(1)}\otimes Z_j^{(2)},\\
    U &=\sum_j U_j^{(1)}\otimes U_j^{(2)},\\ 
    W &=\sum_j W_j^{(1)}\otimes W_j^{(2)} \ .
  \end{align*}
Since each matrix acts trivially in one different factor of the tensor product and all of them are equal, by comparing pair by pair we conclude that the three matrices are proportional to the identity. Thus $A_1 = \lambda B_1$. In the same way we get $A_2 = 1/\lambda B_2$.
\end{proof}

Let us now block the PEPS in \cref{eq:inj_equal_edge} into two injective tensors: select one tensor and block all the others into another injective tensor. These PEPS now satisfy the requirements of  \cref{lem:inj_equal_tensors_2} and thus for all $s$, $A_s = \lambda_s \tilde{B}_s$ for some constant $\lambda_s$, giving the Fundamental Theorem for general injective PEPS (the constants $\lambda_s$ can be incorporated into the matrices of the gauge transformations):

\begin{theorem}[Fundamental theorem for injective PEPS]\label{thm:inj}
  Two injective PEPS  -- defined on a graph that does not contain double edges and self-loops -- generate the same state if and only if the generating tensors are related by a gauge transformation. The matrices defining the gauge transformation are unique up to a multiplicative constant. 
\end{theorem}

As the defining graph cannot contain double edges and self-loops, the theorem is applicable for MPS of size $N$ only if $N\geq 3$, and for 2D PEPS of size $N\times M$ only if both $N\geq 3$ and $M\geq 3$. As an illustration of the theorem, for the two PEPS in \cref{eq:TN_5_particle_eq} there are matrices $Z_{12},Z_{23},Z_{34},Z_{45},Z_{51}$ and $Z_{25}$ such that
\begin{align*}
  \begin{tikzpicture}
    \node[tensor, label=below:$A_1$] (c1) at (0,0) {};
    \coordinate (c2) at (1,0) {};
    \coordinate (c5) at (-0.6,1.5) {};
    \draw[ ] (c1) --++(0,0.5);
    \draw (c1)--(c2);
    \draw ($(c1)!0.5!(c5)$)--(c1);
    \node[ draw=black,fill=red,tensorr,label=below:$Z_{12}$] (z12) at ($(c1)!0.5!(c2)$) {}; 
    \node[ draw=black,fill=red,tensorr,label=below left:$Z_{51}$] (z15) at ($(c1)!0.25!(c5)$) {}; 
  \end{tikzpicture}    =
  \begin{tikzpicture}
    \node[tensor, label=below:$B_1$] (c1) at (0,0) {};
    \coordinate (c2) at (1,0) {};
    \coordinate (c5) at (-0.6,1.5) {};
    \draw[ ] (c1) --++(0,0.5);
    \draw (c1)--(c2);
    \draw ($(c1)!0.5!(c5)$)--(c1);
  \end{tikzpicture}  \quad &\text{and} \quad
  \begin{tikzpicture}
    \coordinate (c1) at (0,0);
    \node[tensor, label=below:$A_2$] (c2) at (1,0) {};
    \coordinate (c3) at (1.5,1);
    \coordinate (c5) at (-0.6,1.5) {};
    \draw[ ] (c2) --++(0,0.5);
    \draw (c1)--(c2);
    \draw (c2)--(c3);
    \draw (c2)--($(c2)!0.5!(c5)$);
    \node[ draw=black,fill=red,tensorr,label=below:$Z_{12}^{-1}$] (z12) at ($(c1)!0.5!(c2)$) {}; 
    \node[ draw=black,fill=red,tensorr,label=right:$Z_{23}$] (z23) at ($(c3)!0.5!(c2)$) {}; 
    \node[ draw=black,fill=red,tensorr,label=left:$Z_{25}$] (z25) at ($(c2)!0.25!(c5)$) {}; 
  \end{tikzpicture}   = 
  \begin{tikzpicture}
    \coordinate (c1) at (0,0);
    \node[tensor, label=below:$B_2$] (c2) at (1,0) {};
    \coordinate (c3) at (1.5,1);
    \coordinate (c5) at (-0.6,1.5) {};
    \draw[ ] (c2) --++(0,0.5);
    \draw (c1)--(c2);
    \draw (c2)--(c3);
    \draw (c2)--($(c2)!0.5!(c5)$);
  \end{tikzpicture}   \ , \\
  \begin{tikzpicture}
    \coordinate (c2) at (1,0);
    \node[tensor, label=right:$A_3$] (c3) at (1.5,1) {};
    \coordinate (c4) at (0.6,1.5);
    \draw[ ] (c3) --++(0,0.5);
    \draw (c2)--(c3);
    \draw (c3)--(c4);
    \node[ draw=black,fill=red,tensorr,label=right:$Z_{23}^{-1}$] (z23) at ($(c3)!0.5!(c2)$) {}; 
    \node[ draw=black,fill=red,tensorr,label=left:$Z_{34}$] (z34) at ($(c3)!0.5!(c4)$) {}; 
  \end{tikzpicture}   = 
  \begin{tikzpicture}
    \coordinate (c2) at (1,0);
    \node[tensor, label=right:$B_3$] (c3) at (1.5,1) {};
    \coordinate (c4) at (0.6,1.5);
    \draw[ ] (c3) --++(0,0.5);
    \draw (c2)--(c3);
    \draw (c3)--(c4);
  \end{tikzpicture}  \quad &\text{and} \quad
  \begin{tikzpicture}
    \coordinate (c3) at (1.5,1);
    \node[tensor, label=below:$A_4$] (c4) at (0.6,1.5) {};
    \coordinate (c5) at (-0.6,1.5);
    \draw[ ] (c4) --++(0,0.5);
    \draw (c3)--(c4);
    \draw (c4)--(c5);
    \node[ draw=black,fill=red,tensorr,label=right:$Z_{34}^{-1}$] (z34) at ($(c3)!0.5!(c4)$) {}; 
    \node[ draw=black,fill=red,tensorr,label=below:$Z_{45}$] (z45) at ($(c4)!0.5!(c5)$) {}; 
  \end{tikzpicture}  = 
  \begin{tikzpicture}
    \coordinate (c3) at (1.5,1);
    \node[tensor, label=below:$B_4$] (c4) at (0.6,1.5) {};
    \coordinate (c5) at (-0.6,1.5);
    \draw[ ] (c4) --++(0,0.5);
    \draw (c3)--(c4);
    \draw (c4)--(c5);
  \end{tikzpicture} \ ,  \\
  \begin{tikzpicture}
    \coordinate (c1) at (0,0);
    \coordinate (c2) at (1,0);
    \coordinate (c4) at (0.6,1.5);
    \node[tensor, label=left:$A_5$] (c5) at (-0.6,1.5) {};
    \draw[ ] (c5) --++(0,0.5);
    \draw (c4)--(c5);
    \draw (c5)--($(c1)!0.5!(c5)$);
    \draw ($(c2)!0.5!(c5)$)--(c5);
    \node[ draw=black,fill=red,tensorr,label=left:$Z_{51}^{-1}$] (z51) at ($(c5)!0.25!(c1)$) {}; 
    \node[ draw=black,fill=red,tensorr,label=right:$Z_{25}^{-1}$] (z25) at ($(c5)!0.25!(c2)$) {}; 
    \node[ draw=black,fill=red,tensorr,label=above:$Z_{45}^{-1}$] (z45) at ($(c5)!0.5!(c4)$) {}; 
  \end{tikzpicture} & = 
  \begin{tikzpicture}
    \coordinate (c1) at (0,0);
    \coordinate (c2) at (1,0);
    \coordinate (c4) at (0.6,1.5);
    \node[tensor, label=left:$B_5$] (c5) at (-0.6,1.5) {};
    \draw[ ] (c5) --++(0,0.5);
    \draw (c4)--(c5);
    \draw (c5)--($(c1)!0.5!(c5)$);
    \draw ($(c2)!0.5!(c5)$)--(c5);
  \end{tikzpicture} \ .
\end{align*}

\section{Normal PEPS}\label{sec:normal}

\begin{definition}\label{def:nomalTN}
A PEPS is normal if after blocking regions the resulting blocked tensors are injective.
\end{definition}

To derive a Fundamental theorem for this kind of PEPS, we use the same arguments as above after blocking to a scale large enough to get injectivite tensors. For simplicity, we consider TI normal PEPS on a square lattice, but it can easily be generalized to the non-translational invariant case on any geometry. We will need the following result:

\begin{lemma}\label{lem:injective_union}
  For any tensor network, the union of two injective regions is injective.
\end{lemma}

\begin{proof}
  Let $A$ and $B$ be two injective regions. W.l.o.g.\ the TN can be blocked as follows :
    \begin{equation*}
    \begin{tikzpicture}
      \node[tensor,label=below left:$A\backslash B$] (1) at (0,0) {};
      \node[tensor,label=right:$A\cap B$] (2) at (0.8,0.8) {};
      \node[tensor,label=below right:$B\backslash A$] (3) at (2,0) {};
      \node[tensor,label=below:$(A\cup B)^{c}$] (4) at (1.3,-0.7) {};
      \draw[ ] (1)--++(0,0.5);
      \draw[ ] (2)--++(0,0.5);
      \draw[ ] (3)--++(0,0.5);
      \draw[ ] (4)--++(0,0.5);
      \draw (1)--(2)--(3)--(4)--(1)--(3);
      \draw (2) -- (4);
    \end{tikzpicture} \ .
  \end{equation*}
  Notice that $A\cup B = (A\backslash B) \cup (A\cap B) \cup (B\backslash A)$ and $ \emptyset = (A\backslash B) \cap (A\cap B) \cap (B\backslash A)$ . Let $X$ now be a tensor such that 
  \begin{equation*}
    \begin{tikzpicture}
      \node[tensor,label=below left:$A\backslash B$] (1) at (0,0) {};
      \node[tensor,label=right:$A\cap B$] (2) at (0.7,0.7) {};
      \node[tensor,label=below right:$B\backslash A$] (3) at (2,0) {};
      \node[tensor,label=below:$X$] (4) at (1.3,-0.7) {};
      \draw[ ] (1)--++(0,0.5);
      \draw[ ] (2)--++(0,0.5);
      \draw[ ] (3)--++(0,0.5);
      \draw (1)--(2)--(3)--(4)--(1)--(3);
      \draw (2) -- (4);
    \end{tikzpicture} \ = 0.
  \end{equation*}
  As the region $A = (A\backslash B) \cup (A\cap B)$ is injective,
  \begin{equation*}
    \begin{tikzpicture}
      \coordinate (1) at (0,0);
      \coordinate (2) at (0.7,0.7) {};
      \node[tensor,label=below right:$B\backslash A$] (3) at (2,0) {};
      \node[tensor,label=below:$X$] (4) at (1.3,-0.7) {};
      \draw[ ] (3)--++(0,0.5);
      \draw (3)--(4);
      \draw (4) -- ($(4)!0.3!(2)$);
      \draw (4) -- ($(4)!0.3!(1)$);
      \draw (3) -- ($(3)!0.3!(2)$);
      \draw (3) -- ($(3)!0.3!(1)$);
    \end{tikzpicture} \ = 0.
  \end{equation*}
  Plugging back the tensor over the region $A\cap B$,
  \begin{equation*}
    \begin{tikzpicture}
      \coordinate (1) at (0,0);
      \node[tensor,label=right:$A\cap B$] (2) at (0.7,0.7) {};
      \node[tensor,label=below right:$B\backslash A$] (3) at (2,0) {};
      \node[tensor,label=below:$X$] (4) at (1.3,-0.7) {};
      \draw[ ] (2)--++(0,0.5);
      \draw[ ] (3)--++(0,0.5);
      \draw (3)--(4);
      \draw (4) -- (2);
      \draw (4) -- ($(4)!0.3!(1)$);
      \draw (3) -- (2);
      \draw (3) -- ($(3)!0.3!(1)$);
      \draw (2) -- ($(2)!0.3!(1)$);
    \end{tikzpicture} \ = 0.
  \end{equation*}
  Finally, the region $B = (A\cap B) \cup (B\backslash A)$ is injective, hence inverting the tensor over that region gives
  \begin{equation*}
    \begin{tikzpicture}
      \coordinate (1) at (0,0);
      \coordinate (2) at (0.7,0.7);
      \coordinate (3) at (2,0);
      \node[tensor,label=below:$X$] (4) at (1.3,-0.7) {};
      \draw (4) -- ($(4)!0.3!(1)$);
      \draw (4) -- ($(4)!0.3!(2)$);
      \draw (4) -- ($(4)!0.3!(3)$);
    \end{tikzpicture} \ = 0,
  \end{equation*}
  which means that the region $A\cup B$ is injective. 
\end{proof}
\begin{observation}\label{ob1:RS}
For example, if in a TI 2D PEPS every $2\times 3$ and $3\times 2$ region is injective, then the following regions:
\begin{equation*}
  \begin{tikzpicture}
          \foreach \x in {0,0.5,1,1.5,2}{
      \foreach \z in {1,1.5,2,2.5,3}{
        \node[tensor] at (\x,0,\z) {};
        \draw (\x,0,\z)--(\x,0.15,\z);
      }}
           \begin{scope}[canvas is xz plane at y=0]       
    \clip (-0.25,0.75) rectangle (2.25,3.25);
    \draw[step=0.5] (-1,-1) grid (4,4);
       \draw[red,fill=red!40, thick, rounded corners] (0.3,1.3) -- (1.2,1.3) --(1.2,1.8)--(1.7,1.8)--(1.7,2.7)--(0.3,2.7)--cycle;
              \end{scope}  
           \node  at (0.75,0,2) {$R$}; 
  \end{tikzpicture}  
  \ {\rm and} \     
  \begin{tikzpicture}
            \foreach \x in {0,0.5,1,1.5,2}{
      \foreach \z in {1,1.5,2,2.5,3}{
        \node[tensor] at (\x,0,\z) {};
        \draw (\x,0,\z)--(\x,0.15,\z);
      }}
       \begin{scope}[canvas is xz plane at y=0]  
    \clip (-0.25,0.75) rectangle (2.25,3.25);
    \draw[step=0.5] (-1,-1) grid (4,4);
    \draw[red,fill=red!40, thick, rounded corners] (0.3,1.3) rectangle (1.7,2.7);
      \end{scope}  
    \node  at (1,0,2) {$S$}; 
  \end{tikzpicture}  
\end{equation*}
are unions of smaller injective regions, and they are thus injective. 
\end{observation}
\begin{observation}\label{ob2:T}
If the size of the PEPS is at least $5\times 6$ and every $2\times 3$ and $3\times 2$ region is injective, the region $T$ depicted below is injective:

\begin{equation*}
  \begin{tikzpicture}
            \foreach \x in {0.5,1,1.5,2,2.5}{
      \foreach \z in {1,1.5,2,2.5,3}{
        \node[tensor] at (\x,0,\z) {};
        \draw (\x,0,\z)--(\x,0.15,\z);
      } }
         \begin{scope}[canvas is xz plane at y=0]  
             \clip (-0.25,0.25) rectangle (3.25,3.25);
    \draw[step=0.5] (-1,-1) grid (4,4);
          \filldraw[draw=red, fill=red!40, thick, rounded corners,even odd rule] (0.3,2.7) -- (1.2,2.7) -- (1.2,1.7) -- (2.7,1.7)--(2.7,0.8)--(1.3,0.8)--(1.3,1.3)--(0.3,1.3)--cycle
    (-0.25,0.25) rectangle (3.25,3.25);   
            \end{scope}  
    \node  at (2,0,2.3) {$T$}; 
  \end{tikzpicture} \ .
\end{equation*} 
\end{observation}
In the following we prove the Fundamental Theorem for a normal TI 2D PEPS. In particular, we prove it in detail for the case where every region of size $2\times 3$ and $3\times 2$ is injective as in the examples above. Then, we generalize the proof for any normal PEPS that is big enough to allow the necessary blockings.

\begin{theorem}\label{thm:3}
Let $A$ and $B$ be two normal 2D PEPS tensors such that every $2\times 3$ and $3\times 2$ region is injective. Suppose they generate the same state on some region $n\times m$ with $n,m\geq 7$. Then $A$ and $B$ are related to each other by a gauge transformation: 
  \begin{equation*}
    \begin{tikzpicture}[baseline=-0.1cm]
      \draw (-0.5,0)--(0.5,0);
      \draw (0,0,-0.5)--(0,0,0.5);
      \node[tensor,label=below:$B$] at (0,0) {};
      \draw[ ] (0,0)--++(0,0.5);
    \end{tikzpicture} = \lambda \cdot 
    \begin{tikzpicture}[baseline=-0.1cm]
      \draw (-1.3,0)--(1.3,0);
      \draw (0,0,-1.6)--(0,0,1.6);
      \node[tensor,label=below:$A$] at (0,0) {};
      \node[tensorr,draw=black,fill=red,label=right:$Y^{-1}$] at (0,0,-1) {};
      \node[tensorr,draw=black,fill=red,label=below:$Y$] at (0,0,1) {};
      \node[tensorr,draw=black,fill=red,label=below:$X^{-1}$] at (0.8,0,0) {};
      \node[tensorr,draw=black,fill=red,label=below:$X$] at (-0.8,0,0) {};
      \draw[ ] (0,0)--++(0,0.5);
    \end{tikzpicture} \ , 
  \end{equation*}
  with $\lambda^{n\cdot m} = 1$ and $X,Y$ invertible matrices. $X$ and $Y$ are unique up to a multiplicative constant.
\end{theorem}

\begin{proof}
  Let us block the TN into three injective parts around an edge. This can be done with e.g. the following choice of regions:
  \begin{equation*}
    \begin{tikzpicture}
          \foreach \x in {0,0.5,...,3}{
        \foreach \z in {0,0.5,...,3}{
          \node[tensor] at (\x,0,\z) {};
          \draw (\x,0,\z)--(\x,0.15,\z);
        } }
           \begin{scope}[canvas is xz plane at y=0]  
      \clip (-0.25,-0.25) rectangle (3.25,3.25);
      \draw[step=0.5] (-1,-1) grid (4,4);
      \draw[line width=0.6mm,green] (1.1,1.5)--(1.4,1.5);
            \draw[red,thick, rounded corners] (0.3,1.3) rectangle (1.2,2.7);
      \draw[blue,thick, rounded corners] (1.3,1.7) rectangle (2.7,0.8);
      \end{scope}  
    \end{tikzpicture} 
    \ \Rightarrow \ 
    \begin{tikzpicture}
      \draw (0.5,0) rectangle (3.5,-0.5);
      \draw[thick,green] (1,0)--(2,0);
      \foreach \x/\c in {1/red,2/blue,3/black}{
        \node[tensor,\c,label=below:$\color{\c}A_\x$] (t\x) at (\x,0) {};
        \draw[ ] (t\x) --++ (0,0.5);
      }
    \end{tikzpicture}  \ ,    
  \end{equation*} 
  where $A_1$ corresponds to the red region, $A_2$ to the blue and $A_3$ to the rest. By \cref{ob2:T} the region $A_3$ is injective as long as the size of the PEPS is at least $5\times 7$. A PEPS of size $5\times 6$ is not enough since there would be regions $A_3$ that are not unions of injective regions and then one cannot conclude that $A_3$ is injective using \cref{lem:injective_union}. Therefore a $7\times 7$ PEPS can be blocked to form an injective tripartite MPS where the green can be any edge (including the vertical edges that require a PEPS size at least $7 \times 5$). Therefore \cref{lem:inj_isomorph} can be applied giving a gauge transformation on every edge.  Due to translation invariance, these gauges are described by the same matrix $X$ ($Y$) on all horizontal (vertical) edges. 
  
Define now $\tilde{B}$ by incorporating the local gauges into the tensors $B$:
  \begin{equation*}
    \begin{tikzpicture}[baseline=-0.1cm]
      \draw (-0.5,0)--(0.5,0);
      \draw (0,0,-0.5)--(0,0,0.5);
      \node[tensor,label=below:$\tilde{B}$] at (0,0) {};
      \draw[ ] (0,0)--++(0,0.5);
    \end{tikzpicture} = 
    \begin{tikzpicture}[baseline=-0.1cm]
      \draw (-1.3,0)--(1.3,0);
      \draw (0,0,-1.6)--(0,0,1.6);
      \node[tensor,label=below:$B$] at (0,0) {};
      \node[tensorr,draw=black,fill=red,label=right:$Y$] at (0,0,-1) {};
      \node[tensorr,draw=black,fill=red,label=below:$Y^{-1}$] at (0,0,1) {};
      \node[tensorr,draw=black,fill=red,label=below:$X$] at (0.8,0,0) {};
      \node[tensorr,draw=black,fill=red,label=below:$X^{-1}$] at (-0.8,0,0) {};
      \draw[ ] (0,0)--++(0,0.5);
    \end{tikzpicture} \ . 
  \end{equation*}
The two PEPS tensors $A$ and $\tilde{B}$ generate the same state. Moreover, inserting a matrix $Z$ on any bond of the first PEPS gives the same state as inserting the same matrix $Z$ on the corresponding bond of the second PEPS. By \cref{ob1:RS},
  \begin{equation*}
  \begin{tikzpicture}
          \foreach \x in {0,0.5,1,1.5,2}{
      \foreach \z in {1,1.5,2,2.5,3}{
        \node[tensor] at (\x,0,\z) {};
        \draw (\x,0,\z)--(\x,0.15,\z);
      }}
           \begin{scope}[canvas is xz plane at y=0]       
    \clip (-0.25,0.75) rectangle (2.25,3.25);
    \draw[step=0.5] (-1,-1) grid (4,4);
       \draw[red,fill=red!40, thick, rounded corners] (0.3,1.3) -- (1.2,1.3) --(1.2,1.8)--(1.7,1.8)--(1.7,2.7)--(0.3,2.7)--cycle;
              \end{scope}  
           \node  at (0.75,0,2) {$R$}; 
  \end{tikzpicture}  
  \ {\rm and} \     
  \begin{tikzpicture}
            \foreach \x in {0,0.5,1,1.5,2}{
      \foreach \z in {1,1.5,2,2.5,3}{
        \node[tensor] at (\x,0,\z) {};
        \draw (\x,0,\z)--(\x,0.15,\z);
      }}
       \begin{scope}[canvas is xz plane at y=0]  
    \clip (-0.25,0.75) rectangle (2.25,3.25);
    \draw[step=0.5] (-1,-1) grid (4,4);
    \draw[red,fill=red!40, thick, rounded corners] (0.3,1.3) rectangle (1.7,2.7);
      \end{scope}  
    \node  at (1,0,2) {$S$}; 
  \end{tikzpicture}  
\end{equation*}
are injective regions and notice that the two regions differ in a single site. Moreover, if the PEPS is at least $5\times 5$, their complement regions $R^c$ and $S^c$ are also injective. Let us denote the tensors of the two networks on region $R$ as $A_R$, $\tilde{B}_R$ and on region $S$ as $A_S$, $\tilde{B}_S$. Then, by \cref{lem:inj_equal_tensors_2}, $A_R \propto \tilde{B}_R$ and $A_S \propto \tilde{B}_S$. This can be represented as
  \begin{equation*}
    \begin{tikzpicture}
      \draw (0.5,0)--(2.5,0);
      \node[tensor,label=below:$A_{R}$] (a1) at (1,0) {};
      \draw[ ] (a1)--++(0,0.5);
      \node[tensor,label=below:$A$] (a2) at (2,0) {};
      \draw[ ] (a2)--++(0,0.5);
    \end{tikzpicture} = 
    \begin{tikzpicture}
      \draw (0.5,0)--(1.5,0);
      \node[tensor,label=below:$A_{S}$] (a1) at (1,0) {};
      \draw[ ] (a1)--++(0,0.5);
    \end{tikzpicture} \propto
    \begin{tikzpicture}
      \draw (0.5,0)--(1.5,0);
      \node[tensor,label=below:$\tilde{B}_{S}$] (a1) at (1,0) {};
      \draw[ ] (a1)--++(0,0.5);
    \end{tikzpicture} = 
    \begin{tikzpicture}
      \draw (0.5,0)--(2.5,0);
      \node[tensor,label=below:$\tilde{B}_{R}$] (a1) at (1,0) {};
      \draw[ ] (a1)--++(0,0.5);
      \node[tensor,label=below:$\tilde{B}$] (a2) at (2,0) {};
      \draw[ ] (a2)--++(0,0.5);
    \end{tikzpicture}.
  \end{equation*}
  Applying the inverse of $A_R\propto \tilde{B}_R$ on the two ends of the equation, we get that the tensors $A$ and $\tilde{B}$ are proportional.  
\end{proof}
 
The above proof can be repeated for any PEPS as long as it is possible to form blocks of injective regions as required by \cref{lem:inj_isomorph} and \cref{lem:inj_equal_tensors_2}. This leads to: 
\begin{theorem}[Fundamental theorem of normal PEPS]\label{thm:normal}
  Suppose two normal PEPS generating the same state satisfy the following:
  \begin{itemize}
  \item they can be blocked into tripartite injective MPS around every edge,
  \item for every site, there are injective regions that differ only in the given site and also their complements are injective.
  \end{itemize}
Then the defining tensors are related with a local gauge transformation, {\it i.e.} invertible matrices $Z_i$ that cancel out with the neighbour tensor when the contraction is carried out:
$$ A_s= B_s (\bigotimes_i Z_i).$$
Moreover, the matrices of the gauge transformation are unique up to a multiplicative constant.
\end{theorem}
Notice that this statement holds for a fixed system size (which is big enough to allow blocking into injective MPS), and translational invariance is not required. In case of a translational invariant system, the matrices $X, Y$ are also translational invariant.  Remarkably the gauge transformations relate tensor by tensor and not only the injective patch after blocking. In the following we present some particular cases, as normal MPS. For non-translational invariant MPS, the statement reads as
\begin{corollary}
  Let $\{A_s\}_{s=1}^n$ and $\{B_s\}_{s=1}^n$ two normal MPS on $n\geq 3L$ sites with the property that blocking any $L$ consecutive sites results in an injective tensor. Suppose they generate the same state:
  \begin{equation*}
    \ket{\Psi} = 
    \begin{tikzpicture}
      \draw (0.5,0) rectangle (4.5,-0.5);
      \foreach \x/\t in {1/1,2/2,4/n}{
        \node[tensor,label=below:$A_\t$] (t\x) at (\x,0) {};
        \draw[ ] (t\x) --++ (0,0.5);
      }
      \node[fill=white] at (3,0) {$\dots$};
    \end{tikzpicture} = 
    \begin{tikzpicture}
      \draw (0.5,0) rectangle (4.5,-0.5);
      \foreach \x/\t in {1/1,2/2,4/n}{
        \node[tensor,label=below:$B_\t$] (t\x) at (\x,0) {};
        \draw[ ] (t\x) --++ (0,0.5);
      }
      \node[fill=white] at (3,0) {$\dots$};
    \end{tikzpicture} \ .
  \end{equation*}  
  Then there are invertible matrices $Z_s$ (for $s=1 \dots n$, $n+1\equiv 1$) such that for all $s=1\dots n$
    \begin{equation*}
      \begin{tikzpicture}
        \draw (-0.5,0)--(0.5,0);
        \node[tensor,label=below:$B_s$] (t) at (0,0) {};
        \draw[ ] (t)--(0,0.5);
      \end{tikzpicture}  = 
      \begin{tikzpicture}
        \draw (-1,0)--(1,0);
        \node[tensorr,draw=black,fill=red,label=below:$Z_s^{-1}$] at (-0.5,0) {};
        \node[tensorr,draw=black,fill=red,label=below:$\ Z_{s+1}$\vphantom{$Z_s^{-1}$}] at (0.5,0) {};
        \node[tensor,label=below:$A_s$\vphantom{$Z_s{-1}$}] (t) at (0,0) {};
        \draw[ ] (t)--(0,0.5);
      \end{tikzpicture} \ .
    \end{equation*}    
    Moreover, the matrices $Z_s$ are unique up to a multiplicative constant.
\end{corollary}
In the appendix of Ref. \cite{Molnar18A} we strengthen the statement to include system sizes $n\geq 2L +1$. For TI MPS, the statement reads as
\begin{corollary}
  Let $A$ and $B$ be two normal TI MPS on $n\geq 3L$ sites with the property that blocking $L$ consecutive sites results in an injective tensor. Suppose they generate the same state:
  \begin{equation*}
    \ket{\Psi} = 
    \begin{tikzpicture}
      \draw (0.5,0) rectangle (4.5,-0.5);
      \foreach \x/\t in {1/1,2/2,4/n}{
        \node[tensor,label=below:$A$] (t\x) at (\x,0) {};
        \draw[ ] (t\x) --++ (0,0.5);
      }
      \node[fill=white] at (3,0) {$\dots$};
    \end{tikzpicture} = 
    \begin{tikzpicture}
      \draw (0.5,0) rectangle (4.5,-0.5);
      \foreach \x/\t in {1/1,2/2,4/n}{
        \node[tensor,label=below:$B$] (t\x) at (\x,0) {};
        \draw[ ] (t\x) --++ (0,0.5);
      }
      \node[fill=white] at (3,0) {$\dots$};
    \end{tikzpicture} \ .
  \end{equation*}  
  Then there is an invertible matrix $Z$ and a constant $\lambda$ with $\lambda^n=1$ such that 
    \begin{equation*}
      \begin{tikzpicture}
        \draw (-0.5,0)--(0.5,0);
        \node[tensor,label=below:$B$] (t) at (0,0) {};
        \draw[ ] (t)--(0,0.5);
      \end{tikzpicture}  = \lambda \cdot
      \begin{tikzpicture}
        \draw (-1,0)--(1,0);
        \node[tensorr,draw=black,fill=red,label=below:$Z^{-1}$] at (-0.5,0) {};
        \node[tensorr,draw=black,fill=red,label=below:$Z$\vphantom{$Z^{-1}$}] at (0.5,0) {};
        \node[tensor,label=below:$A$\vphantom{$Z^{-1}$}] (t) at (0,0) {};
        \draw[ ] (t)--(0,0.5);
      \end{tikzpicture} \ .
    \end{equation*}    
    Moreover the matrix $Z$ is unique up to a multiplicative constant.
\end{corollary}
In the appendix of Ref. \cite{Molnar18A} we strengthen the statement to include system sizes $n\geq 2L +1$. For 2D TI PEPS, the statement reads as 
\begin{corollary}
Let $A$ and $B$ be two normal 2D PEPS tensors such that every $L \times K$ region ($L,K>1$) is injective. Suppose they generate the same state on some region $n\times m$ with $n \geq 3 L$ and $m\geq 3K$. Then $A$ and $B$ are related to each other by a gauge transformation: 
  \begin{equation*}
    \begin{tikzpicture}[baseline=-0.1cm]
      \draw (-0.5,0)--(0.5,0);
      \draw (0,0,-0.5)--(0,0,0.5);
      \node[tensor,label=below:$B$] at (0,0) {};
      \draw[ ] (0,0)--++(0,0.5);
    \end{tikzpicture} = \lambda \cdot 
    \begin{tikzpicture}[baseline=-0.1cm]
      \draw (-1.3,0)--(1.3,0);
      \draw (0,0,-1.6)--(0,0,1.6);
      \node[tensor,label=below:$A$] at (0,0) {};
      \node[tensorr,draw=black,fill=red,label=right:$Y^{-1}$] at (0,0,-1) {};
      \node[tensorr,draw=black,fill=red,label=below:$Y$] at (0,0,1) {};
      \node[tensorr,draw=black,fill=red,label=below:$X^{-1}$] at (0.8,0,0) {};
      \node[tensorr,draw=black,fill=red,label=below:$X$] at (-0.8,0,0) {};
      \draw[ ] (0,0)--++(0,0.5);
    \end{tikzpicture} \ , 
  \end{equation*}
where $\lambda^{n\cdot m} = 1$ and $X,Y$ are invertible matrices. Moreover these matrices $X,Y$ are unique up to a multiplicative constant.
\end{corollary}
In the appendix of \cite{Molnar18A} we strengthen the statement to include system sizes $n\geq 2L +1$ and $m\geq 2K +1$. Similar statements can be made for the non-translational invariant case as well as for other situations, including PEPS in 3 and higher dimensions, other lattices (e.g. triangular, honeycomb, Kagome), and other geometries (e.g. hyperbolic, as it is used in the AdS/CFT constructions \cite{Hayden16,Pastawski15}). 

Furthermore, the results hold for general tensor networks as well (including tensors that do not have physical index), provided that the TN satisfies the conditions in \cref{thm:normal}.  An example for TNs that contain tensors without physical index is the class of Tree Tensor Network (TTN) States \cite{Shi06}. For this particular class, our proof method works: given two normal or injective TTNs generating the same state, the generating tensors are related to each other by local gauge transformations. A sufficient criterion for a binary TTN to be normal is that the tensors are of minimal bond dimension \cite{Singh13}. MERA \cite{Vidal07} is another class of TNs that contain tensors without physical index. For this class, however, we did not find a simple way to block to tripartite injective MPS due to the particular geometry of the network. Therefore, our proof method is not directly applicable for MERA.

\section{Application to symmetries}

Consider a normal TN on $n$ sites describing the state $\ket{\Psi}$. Suppose that $\ket{\Psi}$ has a global on-site symmetry: $U^{\otimes n} \ket{\Psi_A} = \ket{\Psi_A}$. Then, if the TN satisfies the conditions in \cref{thm:normal}, the symmetry operators transform the individual tensors as 
$$U A_s= A_s (\bigotimes_i Z_i),$$
up to multiplicative constants. For example, in TI MPS, this is reflected as follows:
\begin{equation*}
  \begin{tikzpicture}[baseline=-0.1cm]
    \draw (-0.5,0)--(0.5,0);
    \node[tensor,label=below:$A$] (t) at (0,0) {};
    \draw (t)--(0,0.7);
    \node[tensorr,draw=black, fill=red, label=left:$U$] (t) at (0,0.35) {};
  \end{tikzpicture}  =  \lambda \cdot\ 
  \begin{tikzpicture}[baseline=-0.1cm]
    \draw (-1,0)--(1,0);
    \node[tensorr,draw=black, fill=red,label=below:$Z^{-1}$] at (-0.5,0) {};
    \node[tensorr,draw=black, fill=red, label=below:$Z$\vphantom{$Z^{-1}$}] at (0.5,0) {};
    \node[tensor,label=below:$A$\vphantom{$Z^{-1}$}] (t) at (0,0) {};
    \draw (t)--(0,0.5);
  \end{tikzpicture} \ ,
\end{equation*}    
with $\lambda^n = 1$. Similar statements are true in the non TI case (in which case the matrices of the gauges might be different on every edge) and for any geometry. An important situation happens when the symmetry operators are a representation of a group $G$:
$$U_gU_h=U_{gh},\;\; \forall g,h\in G.$$
Then, the matrices on the virtual indices are also labelled by elements of $G$: $Z_g$. But these matrices are not required to be a linear representation, they could form a projective representation (see \cref{Ap:projrep}):
$$Z_g Z_h= e^{i \omega(g,h)}Z_{gh},$$
this is because $U_g$ is translated into $Z_g\otimes Z^{-1}_g$ in the virtual level, so phase factors from $Z_g$ can cancel out with the phase factors from $Z^{-1}_g$. It turns out, that the different (non-equivalent) projective representations classify the 1D quantum phases under a symmetry -see \cite{Pollmann10, Chen11, Schuch11}.

\newpage \cleardoublepage

\section{Proof of the Fundamental theorem for $G$-injective PEPS}

We would like to follow the same route of \cref{sec:injPEPS}, the injective PEPS case, to prove a fundamental theorem for $G$-injective PEPS. One of the key steps in that proof is \cref{lem:inj_isomorph}, stated in terms of MPS, which establishes an isomorphism between the bonds of the tensor networks that generate the same state. We will work out the corresponding result for $G$-injective PEPS and we will show the differences with the injective case.

For that purpose we introduce the analogous MPS class of $G$-injective PEPS -see in \cref{def:GPEPS} in \cref{chapter:Intro}.

\begin{definition} An MPS is $G$-injective if its tensor $A$ satisfies the following
\begin{itemize}
\item the $G$-invariant condition: for a given representation $u_g$ of $G$
\begin{equation}\label{eq:Ginjmps}
  \begin{tikzpicture}[baseline=-0.1cm]
    \draw (-0.5,0)--(0.5,0);
    \node[tensor,label=below:$A$] (t) at (0,0) {};
    \draw (t)--(0,0.5);
  \end{tikzpicture}  = 
  \begin{tikzpicture}[baseline=-0.1cm]
    \draw (-0.8,0)--(0.8,0);
    \node[tensorr,draw=black, fill=red,label=below:$u_g^{-1}$] at (-0.5,0) {};
    \node[tensorr,draw=black, fill=red, label=below:$u_g$\vphantom{$u_g^{-1}$}] at (0.5,0) {};
    \node[tensor,label=below:$A$\vphantom{$u_g^{-1}$}] (t) at (0,0) {};
    \draw (t)--(0,0.5);
  \end{tikzpicture}
  \;\; \forall g\in G,
\end{equation}    
where $u_g$ contains all the irreps of $G$ in its decomposition.
\item there exists a tensor $A^{-1}$ such that
\begin{equation*}
 \begin{tikzpicture}
    \node[tensor,label=below:$A$] at (0,0) {};
     \node[tensor,label=above:$A^{-1}$] at (0,0.7) {};
     \draw (-0.5,0)--(0,0)--(0,0.7)--(-0.5,0.7);
      \draw (0.5,0)--(0,0)--(0,0.7)--(0.5,0.7);
   \end{tikzpicture}
   =\frac{1}{|G|}\sum_g
  \begin{tikzpicture}
    \draw (-0.8,0) -- (-0.3,0) -- (-0.3,0.7) -- (-0.8,0.7);
    \draw ( 0.8,0) -- ( 0.3,0) -- ( 0.3,0.7) -- ( 0.8,0.7);
    \node[tensorr, label=left:$u_g^{-1}$]    at (-0.3,0.35) {};
    \node[tensorr, label=right:$u_g$]  at ( 0.3,0.35) {};
  \end{tikzpicture} . 
  \end{equation*}    
\end{itemize}
We note that this definition can be generalized to non-necessarily TI MPS where the representation of $G$ on each side of the tensor in Eq.\eqref{eq:Ginjmps} can be different.
\end{definition}

We can block the $G$-injective PEPS in a tripartite MPS:
 \begin{equation*}
   \begin{tikzpicture}
         
    \draw[canvas is xz plane at y=0] (-0.5,-0.5) grid (3.5,2.5);
    \foreach \x in {0,1,2,3}{
      \foreach \z in {0,1,2}{
          \draw[green,  ] (1,0,1)--(2,0,1);
          \draw (\x,0,\z)--(\x,0.3,\z);
       \node[tensor] at (\x,0,\z) {}; 
      }
        \node[tensor, red] at (1,0,1) {};
            \node[anchor=north] at (1.1,0,1) {$A_1$};
        \node[tensor, blue]  at (2,0,1) {};
                    \node[anchor=north] at (2.1,0,1) {$A_2$};
    }
  \end{tikzpicture} \; \equiv \;
    \begin{tikzpicture}
    \draw (0.5,0) rectangle (3.5,-0.6);
    \draw[green,  ] (1,0)--(2,0);
    \foreach \x in {1,2,3}{
      \draw (\x,0) --++ (0,0.3);
    }    %
    \node[red,tensor,label=below:$A_1$] at (1,0) {}; 
    \node[blue,tensor,label=below:$A_2$] at (2,0) {}; 
    \node[ tensor,label=below:$A_3$] at (3,0) {}; 
  \end{tikzpicture},
\end{equation*}
where the single edge 1-2 (green) can be any edge of the lattice. From the $G$-injectivity of the PEPS tensors we can conclude that the single tensors $A_1$ and $A_2$ are $G$-injective MPS tensors. This is also the case for $A_3$; the $G$-invariance is guaranteed by the $G$-invariance of the PEPS tensors. Also there exists an inverse, which is the concatenation of the inverse PEPS tensors plus the matrix $\mathfrak{D}$ on each bond. 

\begin{observation}
We recall that the representation can be decomposed as $u_g\cong \sum_\sigma \pi_\sigma(g)\otimes \id_{m_\sigma}$ (see \cref{obs:semireg}). The block decomposition of the matrix algebra generated by this representation is $\mathcal{A}^A\cong \sum_\sigma \mathcal{M}_{d_\sigma}\otimes \id_{m_\sigma}$, where the super-index denotes the tensor, $A$ or $B$, which the representation comes from. Then the matrix algebra that commute with this $\mathcal{A}^A$, the centralizer $\mathcal{C}^A$, has the structure $\mathcal{C}^A \cong \sum_\sigma \id_{d_\sigma}\otimes \mathcal{M}_{m_\sigma}$. This centralizer is associated to each edge of the tensor network and it can be different for each edge if the tensor network is not translationally invariant, so we will denote it by $\mathcal{C}_e^A$.
\end{observation}

\begin{lemma}\label{lem:Ginj_isomorph}
  Suppose $A,B$ are two G-injective, non-necessarily translational invariant MPS on three sites that generate the same state. Then for every edge $e$ and for every matrix $X\in\mathcal{C}_e^A$ (the centralizer of the representation on $A$ of the bond/edge $e$) there is a matrix $Y\in\mathcal{C}_e^B$ such that
  \begin{equation*}
    \begin{tikzpicture}
      \draw (0.5,0) rectangle (3.5,-0.5);
      \foreach \x in {1,2,3}{
        \node[tensor,label=below:$A_\x$] (t\x) at (\x,0) {};
        \draw (t\x) --++ (0,0.3);
      }
      \filldraw[draw=black, fill=red] (1.5,0) circle (0.06);
       \node[anchor=south] at (1.5,0) {$X$};
    \end{tikzpicture} = 
    \begin{tikzpicture}
      \draw (0.5,0) rectangle (3.5,-0.5);
      \foreach \x in {1,2,3}{
        \node[tensor,label=below:$B_\x$] (t\x) at (\x,0) {};
        \draw (t\x) --++ (0,0.3);
      }
      \filldraw[draw=black, fill=red] (1.5,0) circle (0.06);
       \node[anchor=south] at (1.5,0) {$Y$};  
         \end{tikzpicture}\; ,
  \end{equation*} 
and the mapping $X\mapsto Y$ is an algebra isomorphism between the corresponding centralizers. 
\end{lemma}

\begin{proof}[Proof of \cref{lem:Ginj_isomorph}] 
Consider an MPS of three sites with a matrix $X\in\mathcal{C}_{12}^A$ inserted on the bond $(1,2)$. This state can be realized by physical operations acting on either of the two neighboring sites of the MPS:
\begin{equation*}
  \begin{tikzpicture}
    \draw (0.5,0) rectangle (3.5,-0.5);
    \foreach \x in {1,2,3}{
      \node[tensor,label=below:$A_\x$] (t\x) at (\x,0) {};
      \draw (t\x) --++ (0,0.3);
    }
      \filldraw[draw=black, fill=red] (1.5,0) circle (0.06);
       \node[anchor=south] at (1.5,0) {$X$};
         \end{tikzpicture} = 
  \begin{tikzpicture}
    \draw (0.5,0) rectangle (3.5,-0.5);
    \foreach \x in {1,2,3}{
      \node[tensor,label=below:$A_\x$] (t\x) at (\x,0) {};
      \draw (t\x) --++ (0,0.3); }
        \draw (1,0.3)--(1,0.5);
    \filldraw[draw=black, fill=red] (1,0.3) circle (0.06);
       \node[anchor=east] at (1,0.3) {$O_1$};
  \end{tikzpicture} 
  = 
  \begin{tikzpicture}
    \draw (0.5,0) rectangle (3.5,-0.5);
    \foreach \x in {1,2,3}{
      \node[tensor,label=below:$A_\x$] (t\x) at (\x,0) {};
      \draw (t\x) --++ (0,0.3);
    }
          \draw (2,0.3)--(2,0.5);
    \filldraw[draw=black, fill=red] (2,0.3) circle (0.06);
       \node[anchor=east] at (2,0.3) {$O_2$}; 
  \end{tikzpicture} \ ,
\end{equation*} 
where
\begin{equation}\label{eq:gX->O}
   O_1 = 
  \begin{tikzpicture}
    \draw (-0.5,0) rectangle (0.5,1);
    \node[tensor,label=below:$A_1$] (t) at (0,1) {};
    \node[tensor,label=above:$A_1^{-1}$] (b) at (0,0) {};
    \filldraw[draw=black, fill=red] (0.5,0.5) circle (0.06);
       \node[anchor=west] at (0.5,0.5) {$X$};
    \draw (t)--++(0,0.3);
    \draw (b)--++(0,-0.3);
  \end{tikzpicture} \quad \text{and} \quad 
   O_2 = 
  \begin{tikzpicture}
    \draw (-0.5,0) rectangle (0.5,1);
    \node[tensor,label=below:$A_2$] (t) at (0,1) {};
    \node[tensor,label=above:$A_2^{-1}$] (b) at (0,0) {};
     \filldraw[draw=black, fill=red] (-0.5,0.5) circle (0.06);
       \node[anchor=east] at (-0.5,0.5) {$X$};
    \draw (t)--++(0,0.3);
    \draw (b)--++(0,-0.3);
  \end{tikzpicture} \ .     
\end{equation}
The mappings $X\mapsto O_1$ and $X\mapsto O_2^T$ are algebra homomorphisms as they are linear and satisfy for $X,Y\in \mathcal{C}_{12}^A$
 \begin{equation*}
  O_1(X) \cdot O_1(Y) = 
  \begin{tikzpicture}
    \draw (-0.5,0) rectangle (0.5,1);
    \node[tensor,label=below:$A_1$] (t) at (0,1) {};
    \node[tensor,label=above:$A_1^{-1}$] (b) at (0,0) {};
       \filldraw[draw=black, fill=red] (0.5,0.5) circle (0.06);
       \node[anchor=west] at (0.5,0.5) {$Y$};
    \draw (t)--++(0,0.3);
    \draw (b)--++(0,-0.3);
    \begin{scope}[shift={(0,1.5)}]
      \draw (-0.5,0) rectangle (0.5,1);
      \node[tensor,label=below:$A_1$] (t) at (0,1) {};
      \node[tensor,label=above:$A_1^{-1}$] (b) at (0,0) {};
         \filldraw[draw=black, fill=red] (0.5,0.5) circle (0.06);
       \node[anchor=west] at (0.5,0.5) {$X$};
      \draw (t)--++(0,0.3);
      \draw (b)--++(0,-0.3);
    \end{scope}
  \end{tikzpicture}  = \frac{1}{|G|}\sum_{g\in G}
  \begin{tikzpicture}
    \draw (-0.5,0) rectangle (0.5,2.5);
    \node[tensor,label=below:$A_1$] (t) at (0,2.5) {};
    \node[tensor,label=above:$A_1^{-1}$] (b) at (0,0) {};
       \filldraw[draw=black, fill=red] (0.5,0.5) circle (0.06);
       \node[anchor=west] at (0.5,0.5) {$Y$};
          \filldraw[draw=black, fill=red] (0.5,2.0) circle (0.06);
       \node[anchor=west] at (0.5,2.0) {$X$};
\filldraw (0.5,1.25) circle (0.06);
       \node [anchor=west] at (0.5,1.25) {$g$};
       \filldraw  (-0.5,1.25) circle (0.06);
       \node [anchor=east] at (-0.5,1.25) {$g^{-1}$};
    \draw (t)--++(0,0.3);
    \draw(b)--++(0,-0.3);
  \end{tikzpicture} = \frac{1}{|G|}\sum_{g\in G}
    \begin{tikzpicture}
      \draw (-0.5,0) rectangle (0.5,2.5);
      \node[tensor,label=below:$A_1$] (t) at (0,2.5) {};
      \node[tensor,label=above:$A_1^{-1}$] (b) at (0,0) {};
         \filldraw[draw=black, fill=red] (0.5,0.5) circle (0.06);
       \node[anchor=west] at (0.5,0.5) {$Y$};
          \filldraw[draw=black, fill=red] (0.5,1.25) circle (0.06);
       \node[anchor=west] at (0.5,1.25) {$X$};
  \filldraw (0.5,2.0) circle (0.06);
       \node[anchor=west] at (0.5,2.0) {$g$};
  \filldraw (-0.5,2.0) circle (0.06);
       \node[anchor=east] at (-0.5,2.0) {$g^{-1}$};
      \draw (t)--++(0,0.3);
      \draw (b)--++(0,-0.3);
    \end{tikzpicture} = O_1(XY) \ ,   
\end{equation*}
where the third equation holds because $gX = Xg$ for all $g\in G$ ($X\in\mathcal{C}_{12}^A$) and the last equation holds since $A_1$ is $G$-invariant. We consider now the oposite situation where two physical operations on neighbouring sites give the same modified MPS:
\begin{equation}\label{eq:resonate}
   \begin{tikzpicture}
    \draw (0.5,0) rectangle (3.5,-0.5);
    \foreach \x in {1,2,3}{
      \node[tensor,label=below:$B_\x$] (t\x) at (\x,0) {};
      \draw (t\x) --++ (0,0.3); }
        \draw (1,0.3)--(1,0.5);
    \filldraw[draw=black, fill=red] (1,0.3) circle (0.06);
       \node[anchor=east] at (1,0.3) {$O_1$};
  \end{tikzpicture} 
  = 
  \begin{tikzpicture}
    \draw (0.5,0) rectangle (3.5,-0.5);
    \foreach \x in {1,2,3}{
      \node[tensor,label=below:$B_\x$] (t\x) at (\x,0) {};
      \draw (t\x) --++ (0,0.3);
    }
          \draw (2,0.3)--(2,0.5);
    \filldraw[draw=black, fill=red] (2,0.3) circle (0.06);
       \node[anchor=east] at (2,0.3) {$O_2$}; 
  \end{tikzpicture}.
\end{equation} 

Inverting $B_2$, $B_3$  and inserting $\mathfrak{D}$,  the l.h.s. becomes
\begin{equation*}
  \begin{tikzpicture}
    \draw (0.5,0) rectangle (3.5,-0.5);
    \foreach \x in {1,2,3}{
      \node[tensor,label=below:$B_\x$] (t\x) at (\x,0) {};
      \draw (t\x) --++ (0,0.5);
    }
     \filldraw[draw=black, fill=red] (1,0.3) circle (0.06);
       \node[anchor=east] at (1,0.3) {$O_1$};
    \draw (1.5,0.5)--(3.5,0.5);
    \node[tensor,label=above:$B_2^{-1}$] (i2) at (2,0.5) {};
    \node[tensor,label=above:$B_3^{-1}$] (i3) at (3,0.5) {};
         \filldraw (2.5,0.5) circle (0.06);
       \node[anchor=north] at (2.5,0.5) {$\mathfrak{D}$};
  \end{tikzpicture} = \sum_{g\in G} 
  \begin{tikzpicture}
    \draw (1.5,0.5)--(2,0.5)--(2,0)--(0.5,0) --(0.5,-0.5)--(3.5,-0.5)--(3.5,0)--(3,0)--(3,0.5)--(3.5,0.5);
    \node[tensor,label=below:$B_1$] (t1) at (1,0) {};
    \draw(t1) --++ (0,0.5);
      \filldraw[draw=black, fill=red] (1,0.3) circle (0.06);
       \node[anchor=east] at (1,0.3) {$O_1$};
    \node[tensorr,label=right:$g^{-1}$]  at (3,0.25) {};
    \node[tensorr,label=left:$g$]  at (2,0.25) {};
  \end{tikzpicture} \ = 
  \begin{tikzpicture}
    \node[tensor,label=below:$B_1$] (l) at (0,-0.6) {};
      \draw (l)--++(0,0.6);
    \draw (l)++(-0.6,0)--(l)--++(0.6,0);
    \node[tensorr, draw=black, fill=red,label=left:$O_1$] (o) at ($(l)+(0,0.3)$) {};
  \end{tikzpicture} \ ,  
\end{equation*}
where in the last equation we have used the $G$-invariance of $B_1$. Similarly,
\begin{equation*}
  \begin{tikzpicture}
    \draw (0.5,0) rectangle (3.5,-0.5);
    \foreach \x in {1,2,3}{
      \node[tensor,label=below:$B_\x$] (t\x) at (\x,0) {};
    }
    \draw (t1) --++ (0,0.5);
    \draw (t2) --++ (0,1);
    \draw (t3) --++ (0,1);
    \draw (1.5,1)--(3.5,1);
    \node[tensorr, draw=black, fill=red,label=left:$O_2$] (o) at (2,0.5) {};
    \node[tensor,label=above:$B_2^{-1}$] (i2) at (2,1) {};
    \node[tensor,label=above:$B_3^{-1}$] (i3) at (3,1) {};
    \node[tensorr,label=above:$\mathfrak{D}$] (d) at (2.5,1) {};
  \end{tikzpicture} \ = \sum_{g\in G} \ 
  \begin{tikzpicture}
    \draw (1.5,1)--(3,1)--(3,0)--(0.5,0) --(0.5,-0.5)--(4,-0.5)--(4,0)--(3.5,0)--(3.5,1)--(4,1);
    \foreach \x in {1,2}{
      \node[tensor,label=below:$B_\x$] (t\x) at (\x,0) {};
    }
    \draw(t1) --++ (0,0.5);
    \draw (t2) --++ (0,1);
    \node[tensorr, draw=black, fill=red,label=left:$O_2$] (o) at (2,0.5) {};
    \node[tensor,label=above:$B_2^{-1}$] (i2) at (2,1) {};
    \node[tensorr,label=above:$\mathfrak{D}$] (d) at (2.5,1) {};
    \node[tensorr,label=right:$g^{-1}$]  at (3.5,0.5) {};
    \node[tensorr,label=left:$g$]  at (3,0.5) {};
  \end{tikzpicture} \ = 
  \begin{tikzpicture}
    \node[tensor,label=below:$B_1$] (l) at (0,-0.6) {};
    \node[tensor,label=below:$B_2$] (r) at (1,-0.6) {};
    \foreach \x in {l,r}{
      \draw (\x)--++(0,0.6);
    }
    \draw (l)++(-0.6,0)--(l)--(r)--++(0.6,0);
    \node[tensorr, draw=black, fill=red,label=left:$O_2$] (o) at ($(r)+(0,0.5)$) {};
    \draw (o) --++(0,0.5);      
    \node[tensor,label=above:$B_2^{-1}$] (i) at ($ (r) + (0,1)$) {};
    \draw (i)++(-0.6,0)--(i)--++(0.6,0)--++(0,-1);
        \node[tensorr,label=left:$\mathfrak{D}$]  at (1.6,0) {};
  \end{tikzpicture}  =  
  \begin{tikzpicture}
    \draw (-0.5,0)--(1.5,0);
    \node[tensor,label=below:$B_1$] (a) at (0,0) {};
    \draw (a)--++(0,0.4);
    \node[tensorr, draw=black, fill=red, label=below:$W$] (x) at (1,0) {};
  \end{tikzpicture} \ ,
\end{equation*}
where in the second equation we have used the $G$-invariance of the tensors $B_1$ and $B_2$. As both $B_2^{-1}$ and $B_2$ are $G$-invariant, $W\in \mathcal{C}_{12}$. Therefore
\begin{equation}\label{eq:inj_O->X_argument}
   \begin{tikzpicture}
    \node[tensor,label=below:$B_1$] (l) at (0,-0.6) {};
      \draw (l)--++(0,0.6);
    \draw (l)++(-0.6,0)--(l)--++(0.6,0);
    \node[tensorr, draw=black, fill=red,label=left:$O_1$] (o) at ($(l)+(0,0.3)$) {};
  \end{tikzpicture} = 
  \begin{tikzpicture}
    \draw (-0.5,0)--(1.5,0);
    \node[tensor,label=below:$B_1$] (a) at (0,0) {};
    \draw (a)--++(0,0.4);
    \node[tensorr, draw=black, fill=red, label=below:$W$] (x) at (1,0) {};
  \end{tikzpicture} \ ,
\end{equation}  
for some matrix $W\in \mathcal{C}_{12}$.  Similarly, inverting $B_1$ and $B_3$, we arrive at the identity 
\begin{equation*}
  \begin{tikzpicture}
    \node[tensor,label=below:$B_2$] (l) at (0,-0.6) {};
      \draw (l)--++(0,0.6);
    \draw (l)++(-0.6,0)--(l)--++(0.6,0);
    \node[tensorr, draw=black, fill=red,label=left:$O_2$] (o) at ($(l)+(0,0.3)$) {};
  \end{tikzpicture} = 
  \begin{tikzpicture}
    \draw (-1.5,0)--(0.5,0);
    \node[tensor,label=below:$B_2$] (a) at (0,0) {};
    \draw (a)--++(0,0.4);
    \node[tensorr, draw=black, fill=red, label=below:$V$] (x) at (-1,0) {};
  \end{tikzpicture} \ ,
\end{equation*}
for some matrix $V\in\mathcal{C}_{12}$.  Therefore
\begin{equation*}
  \begin{tikzpicture}
    \draw (0.5,0) rectangle (3.5,-0.5);
    \foreach \x in {1,2,3}{
      \node[tensor,label=below:$B_\x$] (t\x) at (\x,0) {};
      \draw (t\x) --++ (0,0.3);
    }
    \node[tensorr, draw=black, fill=red, label=above:$W$] (x) at (1.5,0) {};
  \end{tikzpicture} = 
     \begin{tikzpicture}
    \draw (0.5,0) rectangle (3.5,-0.5);
    \foreach \x in {1,2,3}{
      \node[tensor,label=below:$B_\x$] (t\x) at (\x,0) {};
      \draw (t\x) --++ (0,0.3); }
        \draw (1,0.3)--(1,0.5);
    \filldraw[draw=black, fill=red] (1,0.3) circle (0.06);
       \node[anchor=east] at (1,0.3) {$O_1$};
  \end{tikzpicture} 
  = 
  \begin{tikzpicture}
    \draw (0.5,0) rectangle (3.5,-0.5);
    \foreach \x in {1,2,3}{
      \node[tensor,label=below:$B_\x$] (t\x) at (\x,0) {};
      \draw (t\x) --++ (0,0.3);
    }
          \draw (2,0.3)--(2,0.5);
    \filldraw[draw=black, fill=red] (2,0.3) circle (0.06);
       \node[anchor=east] at (2,0.3) {$O_2$}; 
  \end{tikzpicture} \ = \ 
  \begin{tikzpicture}
    \draw (0.5,0) rectangle (3.5,-0.5);
    \foreach \x in {1,2,3}{
      \node[tensor,label=below:$B_\x$] (t\x) at (\x,0) {};
      \draw (t\x) --++ (0,0.3);
    }
    \node[tensorr, draw=black, fill=red , label=above:$V$] (x) at (1.5,0) {};
  \end{tikzpicture} \ ,
\end{equation*}
and thus by the $G$-injectivity, as both $V\in\mathcal{C}_{12}$ and $W\in\mathcal{C}_{12}$, $V=W$. Therefore the maps $O_1\mapsto W$ and $O_2^T \mapsto W$ are uniquely defined and are algebra homomorphisms.

Consider now two three-site $G$-injective MPS generating the same state:
\begin{equation*}
  \begin{tikzpicture}
    \draw (0.5,0) rectangle (3.5,-0.5);
    \foreach \x in {1,2,3}{
      \node[tensor,label=below:$A_\x$] (t\x) at (\x,0) {};
      \draw (t\x) --++ (0,0.3);
    }
  \end{tikzpicture} = 
  \begin{tikzpicture}
    \draw (0.5,0) rectangle (3.5,-0.5);
    \foreach \x in {1,2,3}{
      \node[tensor, label=below:$B_\x$] (t\x) at (\x,0) {};
      \draw (t\x) --++ (0,0.3);
    }
  \end{tikzpicture} \ .
\end{equation*}
Deform the MPS on the LHS by inserting a matrix $X\in\mathcal{C}_{12}$ on the bond $(1,2)$. By the above arguments, this deformation is equivalent to either of two physical operations:
\begin{equation*}
  \begin{tikzpicture}
    \draw (0.5,0) rectangle (3.5,-0.5);
    \foreach \x in {1,2,3}{
      \node[tensor,label=below:$A_\x$] (t\x) at (\x,0) {};
      \draw[ ] (t\x) --++ (0,0.3);
    }
    \node[tensorr, draw=black, fill=red, label=above:$X$] (x) at (1.5,0) {};
  \end{tikzpicture} = 
\begin{tikzpicture}
    \draw (0.5,0) rectangle (3.5,-0.5);
    \foreach \x in {1,2,3}{
      \node[tensor,label=below:$A_\x$] (t\x) at (\x,0) {};
      \draw (t\x) --++ (0,0.3); }
        \draw (1,0.3)--(1,0.5);
    \filldraw[draw=black, fill=red] (1,0.3) circle (0.06);
       \node[anchor=east] at (1,0.3) {$O_1$};
  \end{tikzpicture} 
  = 
  \begin{tikzpicture}
    \draw (0.5,0) rectangle (3.5,-0.5);
    \foreach \x in {1,2,3}{
      \node[tensor,label=below:$A_\x$] (t\x) at (\x,0) {};
      \draw (t\x) --++ (0,0.3);
    }
          \draw (2,0.3)--(2,0.5);
    \filldraw[draw=black, fill=red] (2,0.3) circle (0.06);
       \node[anchor=east] at (2,0.3) {$O_2$}; 
  \end{tikzpicture}  \ .
\end{equation*} 
As the MPS defined by the $A$ and $B$ tensors is the same state, these physical operators also satisfy
\begin{equation*}
  \begin{tikzpicture}
    \draw (0.5,0) rectangle (3.5,-0.5);
    \foreach \x in {1,2,3}{
      \node[tensor,label=below:$A_\x$] (t\x) at (\x,0) {};
      \draw[ ] (t\x) --++ (0,0.3);
    }
    \node[tensorr, draw=black, fill=red, label=above:$X$] (x) at (1.5,0) {};
  \end{tikzpicture} = 
\begin{tikzpicture}
    \draw (0.5,0) rectangle (3.5,-0.5);
    \foreach \x in {1,2,3}{
      \node[tensor,label=below:$B_\x$] (t\x) at (\x,0) {};
      \draw (t\x) --++ (0,0.3); }
        \draw (1,0.3)--(1,0.5);
    \filldraw[draw=black, fill=red] (1,0.3) circle (0.06);
       \node[anchor=east] at (1,0.3) {$O_1$};
  \end{tikzpicture} 
  = 
  \begin{tikzpicture}
    \draw (0.5,0) rectangle (3.5,-0.5);
    \foreach \x in {1,2,3}{
      \node[tensor,label=below:$B_\x$] (t\x) at (\x,0) {};
      \draw (t\x) --++ (0,0.3);
    }
          \draw (2,0.3)--(2,0.5);
    \filldraw[draw=black, fill=red] (2,0.3) circle (0.06);
       \node[anchor=east] at (2,0.3) {$O_2$}; 
  \end{tikzpicture}  \ .
\end{equation*} and thus for every $X\in\mathcal{C}_{12}$ there is a matrix $Y\in\mathcal{C}_{12}$ such that 
\begin{equation*}
  \begin{tikzpicture}
    \draw (0.5,0) rectangle (3.5,-0.5);
    \foreach \x in {1,2,3}{
      \node[tensor,label=below:$A_\x$] (t\x) at (\x,0) {};
      \draw (t\x) --++ (0,0.3);
    }
    \node[tensorr, draw=black, fill=red,label=above:$X$] (x) at (1.5,0) {};
  \end{tikzpicture} = 
  \begin{tikzpicture}
    \draw (0.5,0) rectangle (3.5,-0.5);
    \foreach \x in {1,2,3}{
      \node[tensor,label=below:$B_\x$] (t\x) at (\x,0) {};
      \draw (t\x) --++ (0,0.3);
    }
    \node[tensorr, draw=black, fill=red,label=above:$Y$] (x) at (1.5,0) {};
  \end{tikzpicture} \ .
\end{equation*}
Due to $G$-injectivity of the $B$ tensors, the mapping $X\mapsto Y$ is uniquely defined. Due to $G$-injectivity of the $A$ tensors, it is an injective map. As the argument is symmetric with respect of the exchange of the $A$ and $B$ tensors, it also has to be surjective and therefore the map $X\mapsto Y$ is a bijection. Moreover, it is clear from the construction that it is an algebra homomorphism, as both $X\mapsto O_1$ and $O_1\mapsto Y$ are algebra homomorphisms. Therefore the mapping $X\mapsto Y$ is an algebra isomorphism between the corresponding centralizers.\\

\end{proof}

\begin{observation}
Notice that in the case of $G$-injective MPS we cannot conclude the existence of an invertible matrix $Z$ so that the isomorphism is $Z(\; \cdot \;) Z^{-1}$ as in \cref{lem:inj_isomorph} for injective MPS (which does not allow to follow the same strategy). This is because the algebra now is not simple: it is semi-simple, where the block decomposition corresponds to the irreps. For example an isomorphism of the algebra $( \id_2 \otimes \mathcal{M}_D)\oplus \mathcal{M}_D$ is $a\oplus a \oplus b \mapsto b\oplus b \oplus a$ for any $a,b\in \mathcal{M}_D$ which cannot be realized by a gauge transformation since the traces do not match.
\end{observation}

Therefore with the blocking procedure showed above we conclude from \cref{lem:Ginj_isomorph} that there exists an isomorphism on each edge of two $G$-injective PEPS that generate the same state. Since we are focusing on translational invariant PEPS the isomorphism is the same for every single edge. Moreover due to its uniqueness this isomorphism has to be compatible under blocking. That is, the tensor product of isomorphism of singles edges has to be equal to the isomorphism of the tensor product of the edges. We will come back to this point later.\\

The proof of the fundamental theorem for $G$-injective PEPS can be separated mainly in two steps. The first one is \cref{prop:groupalgebra1} below which shows that a global property, equality of states in TN, is reflected locally; there is a relation between the tensors at each site. This result is achieved using the fundamental theorem of Ref.\cite{PerezGarcia07} which assumes the equality of states for every system size: this is the main limitation compared to \cref{thm:inj}. The local relation between the tensors in \cref{prop:groupalgebra1} is not a gauge transformation. The second step tackles this issue by pushing this local relation to a gauge transformation using \cref{lem:Ginj_isomorph}. This step will be separated in three propositions for the sake of readability.\\

Let us first prove the following lemma:
\begin{lemma}\label{lem:local_op}
Let $\mathcal{O}_{\mathcal{R}}$  be an operator acting on a compact and contractible region $\mathcal{R}$ of a lattice $\Lambda$ and let us denote by $\partial \mathcal{R}$ the sites surrounding $\mathcal{R}$.  
If a $G$-injective PEPS of size at least $\mathcal{R} \cup \partial \mathcal{R}$ is left invariant by $\mathcal{O}_{\mathcal{R}}$, then $\mathcal{O}_{\mathcal{R}}$ leaves invariant the $\mathcal{R}$ patch of tensors.
\end{lemma}
\begin{proof}
Applying the inverse of the tensor $\myinv{A}$ on the complementary region of $\mathcal{R}$ to the equation $\mathcal{O}_{\mathcal{R}} |\Psi(A)_\Lambda\rangle= |\Psi(A)_\Lambda\rangle$ we end up with $\myinv{A}_{\Lambda\setminus \mathcal{R}} \mathcal{O}_{\mathcal{R}} |\Psi(A)_\Lambda\rangle= \myinv{A}_{\Lambda\setminus \mathcal{R}}|\Psi(A)_\Lambda\rangle= A_{ \mathcal{R}}$, where $A_{ \mathcal{R}}$ denotes the contraction of the tensors $A$ in the region $\mathcal{R}$ with OBC. Since $[\mathcal{O}_{\mathcal{R}},\myinv{A}_{\Lambda\setminus \mathcal{R}}]=0$ (they act on non-overlapping regions) then $\myinv{A}_{\Lambda\setminus \mathcal{R}} \mathcal{O}_{\mathcal{R}} |\Psi(A)_\Lambda\rangle= \mathcal{O}_{\mathcal{R}} \myinv{A}_{\Lambda\setminus \mathcal{R}}|\Psi(A)_\Lambda\rangle =\mathcal{O}_{\mathcal{R}}A_{ \mathcal{R}}= A_{ \mathcal{R}}$.
\end{proof}

\begin{proposition}\label{prop:groupalgebra1}
 Suppose two $G$-injective tensors $A$ and $B$ generate the same PEPS for every system size ( $\ket{\Psi(A)}=\ket{\Psi(B)}$)
 
 \begin{equation*}
  \begin{tikzpicture}
    \draw[canvas is xz plane at y=0] (-0.5,-0.5) grid (3.5,2.7);
    \foreach \x in {0,1,2,3}{
      \foreach \z in {0,1,2}{
       \node[tensorB] at (\x,0,\z) {}; 
        \draw (\x,0,\z)--(\x,0.3,\z);
      }
            \node[anchor=north] at (0.05,0.05) {$ {B}$};
    }
  \end{tikzpicture} \ = \ 
  \begin{tikzpicture}
    \draw[canvas is xz plane at y=0] (-0.5,-0.5) grid (3.5,2.7);
    \foreach \x in {0,1,2,3}{
      \foreach \z in {0,1,2}{
        \node[tensor] at (\x,0,\z) {}; 
        \draw (\x,0,\z)--(\x,0.3,\z);
      }
            \node[anchor=north] at (0.05,0.05) {$A$};
    }
  \end{tikzpicture},
\end{equation*}
where we represent each tensor with a different shape ($A$ rounded and $B$ squared):
\begin{equation*}
  \begin{tikzpicture}
    \pic at (0,0,0) {3dpeps};
    \node at (-0.2,0.2) {A};
  \end{tikzpicture}
\quad  ,\quad
  \begin{tikzpicture} 
    \pic at (0,0,0) {3dpepsres};
    \node at (-0.2,0.2) {B};
  \end{tikzpicture}\ .
\end{equation*}  
 Then, there are invertible matrices $X,Y$ and $T \in (\mathcal{A}^A)^{\otimes 3}$ such that 
\begin{equation}
  \begin{tikzpicture}
    \draw (-1.2,0,0)--(1.2,0,0);
    \draw (0,0,-1.4)--(0,0,1.4);
          \foreach \x in {(1,0,0),(-1,0,0),(0,0,1.1),(0,0,-1.1)}{
        \filldraw[draw=black, fill=red] \x circle (0.05);
      }      
   
       \draw[canvas is xz plane at y=0,double=red, double distance=0.7mm,line cap=round]  (-0.1,0.8) -- (0,0.8) arc (90:-90:0.8)-- (-0.1,-0.8); 

    \node[anchor=north] at (-1.2,0,0) {$X$};
    \node[anchor=north west] at (0,0,1.2) {$Y$};
    \node[anchor=north] at (1.2,0,0) {$\myinv{X}$};
    \node[anchor=west] at (0,0,-1.2) {$\myinv{Y}$};
    \pic at (0,0,0) {3dpeps};
    \node[anchor=north west] at (0.3,0,0.6) {T};
    \node at (-0.2,0.2) {A};
  \end{tikzpicture}=
  \begin{tikzpicture} [baseline=-1mm]
    \pic at (0,0,0) {3dpepsres};
    \node at (-0.2,0.2) {B};
  \end{tikzpicture}\ .
\end{equation}  
\end{proposition}

\begin{proof}[Proof of \cref{prop:groupalgebra1}]
Let us denote by $A_+$ the tensor constructed with the concatenation of $5$ tensors sharing a vertex of the square lattice:
  \begin{equation*}
 A_+ = 
	\begin{tikzpicture}
             \pic at (0,-0.4,0) {pepsplus};
        \end{tikzpicture}
       \; .
  \end{equation*}
$B_+$ is defined analogously. The projector onto the physical subspace generated by $A_+$ is $P^{[A]}_+$:
  \begin{equation*}
    P^{[A]}_+ = 
	\begin{tikzpicture}
              \pic at (0,0,0) {PAplus};
        \end{tikzpicture}
       \; .
  \end{equation*}
Since $ P^{[A]}_+ A_+=A_+$, we have that $ P^{[A]}_+\ket{\Psi(A)}=\ket{\Psi(A)}$ so $P^{[A]}_+\ket{\Psi(B)}=\ket{\Psi(B)}$.  Using \cref{lem:local_op} we obtain  $ B_+= P^{[A]}_+ B_+$. That is,

  \begin{equation*}
    \begin{tikzpicture}
     \pic at (0,0,0) {pepsplusres};
    \end{tikzpicture} 
    =
    \begin{tikzpicture}
      \pic at (0,-0.8,0) {pepsplusres}; 
      \pic at (0,0,0) {PAplus};
       \foreach \x/\z in {0/-1.5,0/0,-1.5/0,1.5/0,0/1.5}{
        \draw (\x,-0.8,\z)--(\x,0,\z);
      }
    \end{tikzpicture}
 \ .
  \end{equation*}
  
Now we apply the tensor $\myinv{A}$ on each site without contracting their virtual indices. We carry out the product between $A$ and $A^{-1}$ so:

 \begin{align}\label{eq:PAtensors}
    \begin{tikzpicture}
    \foreach \x/\z in {0/-1.5,0/0,-1.5/0,1.5/0,0/1.5}{
        \pic at (\x,0,\z) {3dpepsdown};
        \draw (\x,-0.5,\z)--(\x,0,\z);
      }
      \pic at (0,-0.5,0) {pepsplusres}; 
    \end{tikzpicture}
      = &\sum_{g_1,g_2,g_3,g_4,g_5}
    \begin{tikzpicture}
           \foreach \x/\z in {0/-1.5,0/0,-1.5/0,1.5/0,0/1.5}{
        \draw (\x,-0.7,\z)--(\x,0,\z);
      }         
     \pic at (0,0,0) {pepsplusdown};
      \foreach \x in {(1,0,0),(-1,0,0),(0,0,1),(0,0,-1)}{
          \filldraw[draw=black, fill=black] \x circle (0.05);
      }
                   \node[anchor=north] at (-1,0,0) {$\myinv{g}_1$};
          \node[anchor=north] at(1,0,0) {$g_3$};
           \node[anchor= north] at (0,0,-1) {$g_2 $};
            \node[anchor=west] at (0,0,1) {$\myinv{g}_4$};
     \pic at (0,-0.7,0) {pepsplusres};
    \end{tikzpicture}   \\ 
    & \times 
    \begin{tikzpicture}
     \pic at (0,0.5,0) {openrectcirc};
             \node[anchor=north] at (-0.9,0.5,0) {$g_1\myinv{g}_5$};
          \node[anchor=north] at (0.9,0.5,0) {$g_5 \myinv{g}_3$};
           \node[anchor= east] at (0,0.5,-0.9) {$g_5 \myinv{g}_2$};
            \node[anchor=west] at (0,0.5,0.9) {$g_4\myinv{g}_5$};
        \end{tikzpicture}  
    \ . \notag 
  \end{align}

Notice that because the tensors $A$ and $B$ generate the same state for any system size and in particular for 1D systems, a $1\times n$ torus, the following holds for an invertible $X$ acting on the horizontal bonds
  \begin{equation*}
    \begin{tikzpicture}[baseline=-1mm]
      \draw [canvas is zy plane at x=0] (-0.5,-0.15) rectangle (0.5,0);
      \pic at (0,0,0) {3dpeps};
    \end{tikzpicture}  = 
    \begin{tikzpicture}[baseline=-1mm]
      \draw (-0.7,0)--(0.7,0);
      \draw [canvas is zy plane at x=0] (-0.5,-0.15) rectangle (0.5,0);
      \pic at (0,0,0) {3dpepsres};
      \filldraw[draw=black, fill=red] (0.4,0,0) circle (0.05);
      \filldraw[draw=black, fill=red] (-0.4,0,0) circle (0.05);
      \node[anchor=south] at (0.4,-0.4) {$X$};
      \node[anchor=south] at (-0.4,-0.4) {$\myinv{X}$};
    \end{tikzpicture} \; \Rightarrow \; 
        \begin{tikzpicture}
      \pic at (0,0,0) {3dpepsres};
      \draw (0,0,0)--(0,0.5,0);
      \pic at (0,0.5,0) {3dpepsdown};
      \filldraw[draw=black, fill=red] (0.3,0,0) circle (0.05);
      \node[ anchor=south] at (0.3,-0.4) {$X\mathfrak{D}$};
      \draw (0.4,0)--(0.4,0.5);
      \draw[canvas is zy plane at x=0] (-0.5,-0.15) rectangle (0.5,0);
      \draw[canvas is zy plane at x=0] (-0.5,0.5) rectangle (0.5,0.65);
    \end{tikzpicture} 
    =
           \begin{tikzpicture}
                 \pic at (0,0,0) {3dpeps};
      \draw (0,0,0)--(0,0.5,0);
      \pic at (0,0.5,0) {3dpepsdown};
            \filldraw[draw=black, fill=red] (-0.3,0,0) circle (0.05);
      \node[ anchor=south] at (-0.3,-0.4) {$X$};
      \filldraw[draw=black, fill=red] (0.3,0,0) circle (0.05);
      \node[ anchor=south] at (0.3,-0.4) {$\mathfrak{D}$};
      \draw (0.4,0)--(0.4,0.5);
      \draw[canvas is zy plane at x=0] (-0.5,-0.15) rectangle (0.5,0);
      \draw[canvas is zy plane at x=0] (-0.5,0.5) rectangle (0.5,0.65);
    \end{tikzpicture}
        = 
       \begin{tikzpicture}
      \draw (-0.3,0,0)--(0,0,0)--(0,0.5,0)--(-0.3,0.5,0);
            \filldraw[draw=black, fill=red] (0,0.25,0) circle (0.05);
      \node[ anchor=west] at (0,0.25) {$X$};
    \end{tikzpicture} .
  \end{equation*}  
We have used the fundamental theorem of Ref \cite{PerezGarcia07} since the tensors
  \begin{equation*}
    \begin{tikzpicture}
      \draw [canvas is zy plane at x=0] (-0.5,-0.15) rectangle (0.5,0);
      \pic at (0,0,0) {3dpeps};
    \end{tikzpicture}
    \quad, \quad
        \begin{tikzpicture}
      \draw [canvas is zy plane at x=0] (-0.5,-0.15) rectangle (0.5,0);
      \pic at (0,0,0) {3dpepsres};
    \end{tikzpicture}
  \end{equation*}
are G-injective, and then in canonical form, generating the same state. We can repeat the same argument above in the other directions, including an invertible $Y$ acting on the vertical bonds, to conclude that:
\begin{equation*}
   \begin{tikzpicture}
      \pic at (0,0,0) {3dpepsres};
      \draw (0,0,0)--(0,0.5,0);
      \pic at (0,0.5,0) {3dpepsdown};
       \filldraw[draw=black, fill=red] (-0.3,0,0) circle (0.05);
      \node[ anchor=south] at (-0.5,-0.5) {$\mathfrak{D} \myinv{X}$};
      \draw (-0.4,0)--(-0.4,0.5);
      \draw[canvas is zy plane at x=0] (-0.5,-0.15) rectangle (0.5,0);
      \draw[canvas is zy plane at x=0] (-0.5,0.5) rectangle (0.5,0.65);
    \end{tikzpicture}= \myinv{X}, \;
     \begin{tikzpicture}
      \pic at (0,0,0) {3dpepsres};
      \draw (0,0,0)--(0,0.5,0);
      \pic at (0,0.5,0) {3dpepsdown};
   \draw[canvas is zy plane at x=0] (0,0.5) rectangle (-0.7,0);
      \draw[canvas is xy plane at z=0] (-0.5,-0.15) rectangle (0.5,0);
      \draw[canvas is xy plane at z=0] (-0.5,0.5) rectangle (0.5,0.65);
       \draw (0,0,0)--(0,0,1);
        \draw (0,0.5,0)--(0,0.5,1);
        \filldraw[draw=black, fill=red] (0.17,0.17,0) circle (0.05);
      \node[ anchor=west] at (0.17,0.17,0) {$Y\mathfrak{D}$};
    \end{tikzpicture}=Y,\;
      \begin{tikzpicture}
      \pic at (0,0,0) {3dpepsres};
      \draw (0,0,0)--(0,0.5,0);
      \pic at (0,0.5,0) {3dpepsdown};
   \draw[canvas is zy plane at x=0] (0,0.5) rectangle (1,0);
      \draw[canvas is xy plane at z=0] (-0.5,-0.15) rectangle (0.5,0);
      \draw[canvas is xy plane at z=0] (-0.5,0.5) rectangle (0.5,0.65);
       \draw (0,0,0)--(0,0,-1);
        \draw (0,0.5,0)--(0,0.5,-1);
       \filldraw[draw=black, fill=red] (-0.27,-0.27,0) circle (0.05);
      \node[ anchor=west] at (-0.3,-0.3,0) {$\myinv{Y}\mathfrak{D}$};
    \end{tikzpicture}= \myinv{Y} .
    \end{equation*}

We now apply $X\mathfrak{D},Y\mathfrak{D}, X^{-1}\mathfrak{D}, Y^{-1}\mathfrak{D}$  in Eq.\eqref{eq:PAtensors} to use the previous relations and we also contract the rest of the open  virtual indices as follows:

  \begin{equation*}
    \begin{tikzpicture}
           \draw (-0.7,0,0)--(0.7,0,0);
      \draw (0,0,-0.8)--(0,0,0.8);
            \pic at (0,0,0) {3dpepsres};
       \draw (0,0,0)--(0,0.5,0);
      \pic at (0,0.5,0) {3dpepsdown};
      \foreach \x in {(0.5,0,0),(-0.5,0,0),(0,0,0.5),(0,0,-0.5)}{
          \filldraw[draw=black, fill=red] \x circle (0.05);
      }
      \node[anchor=north] at (-0.5,0,0) {$\myinv{X}$};
      \node[anchor=north west] at (0,0,0.5) {$\myinv{Y}$};
      \node[anchor=north] at (0.5,0,0) {$ X$};
      \node[anchor=west] at (0,0,-0.5) {$Y$};
    \end{tikzpicture} 
    =  \sum_{g_1,g_2,g_3,g_4,g_5}
    \begin{tikzpicture}
            \pic at (-3.5,0,0) {openrectcirc1};
             \node[anchor=north] at (-4.6,0,0) {$g_1\myinv{g}_5$};
          \node[anchor=north east] at(-2.3,0,0) {$g_5 \myinv{g}_3$};
           \node[anchor= north] at (-3.5,0,-0.9) {$g_5 \myinv{g}_2$};
            \node[anchor=north] at (-3.5,0,0.9) {$g_4\myinv{g}_5$};
    \pic at (0,0,0) {pepsplusrescontrac};
          \foreach \x in {(1,0.7,0),(-1,0.7,0),(0,0.7,1),(0,0.7,-1)}{
          \filldraw[draw=black, fill=black] \x circle (0.05);
      }
                   \node[anchor=south] at (-1,0.7,0) {$\myinv{g}_1$};
          \node[anchor=south west] at(1,0.7,0) {$g_3$};
           \node[anchor= east] at (0,0.7,-1) {$g_2 $};
            \node[anchor=south] at (0,0.7,1) {$\myinv{g}_4$};
                 \pic at (0,0.7,0) {pepspluscontracdown};
     \draw (1.5,0,0)--(1.5,0.7,0);
     \draw (2,0,0)--(2,0.7,0);
       \node[tensorr, draw=black, fill=red , label=below:$X\mathfrak{D}$] at (1.8,0,0) {};
   \node[tensorr, draw=black, fill=red , label=below:$\mathfrak{D}\myinv{X}$] at (-1.8,0,0) {};
   \node[tensorr, draw=black, fill=red , label=above:$Y\mathfrak{D}$] at (0,0,-1.8) {};
  \node[tensorr, draw=black, fill=red , label=below:$\mathfrak{D} \myinv{Y}$] at (0,0,1.8) {};
                               \draw (-1.5,0,0)--(-1.5,0.7,0);
     \draw (-2,0,0)--(-2,0.7,0);
          \draw (0,0,1.5)--(0,0.7,1.5);
     \draw (0,0,2)--(0,0.7,2);       
        \draw (0,0,-1.5)--(0,0.7,-1.5);
     \draw (0,0,0)--(0,0.7,0);
       \draw (0,0,-2)--(0,0.7,-2);
    \end{tikzpicture}
    \in \mathcal{A}^{\otimes 4}.
  \end{equation*}
  
%

We apply $A$ to the previous equation in the top layer. The LHS is $P^{[A]} B(X\otimes Y\otimes X^{-1}\otimes Y^{-1})$ which is equal to $B(X\otimes Y\otimes X^{-1}\otimes Y^{-1})$. The RHS is an operator $W\in \mathcal{A}^{\otimes 4}$ acting on $A$. That is,

\begin{align}\label{eq:defT}
 \begin{tikzpicture}
      \draw (-0.7,0,0) -- (0.7,0,0);
      \draw (0,0,-0.9) -- (0,0,0.9);
      \pic at (0,0,0) {3dpepsres};
      \node at (0.1,-0.27,0) {$B$};
      \filldraw [draw=black, fill=red] (-0.5,0,0) circle (0.05);
      \filldraw [draw=black, fill=red] (0.5,0,0) circle (0.05);
      \node[anchor=south] at (-0.5,0,0) {$\myinv{X}$};
      \node[anchor=north] at (0.5,0,0) {$X$};
      \filldraw [draw=black, fill=red] (0,0,-0.6) circle (0.05);
      \filldraw [draw=black, fill=red] (0,0,0.6) circle (0.05);
       \node[anchor=south] at (0,0,-0.6) {$Y$};
      \node[anchor=north] at (0,0,0.6) {$\myinv{Y}$};
    \end{tikzpicture}  =
    &  \sum_{g_1,g_2,g_3, g_4\in G} W_{g_1,g_2,g_3, g_4}
  \begin{tikzpicture}
   \draw (-0.7,0,0) -- (0.7,0,0);
      \draw (0,0,-0.9) -- (0,0,0.9);
      \pic at (0,0,0) {3dpeps};
      \filldraw [draw=black, fill=red] (-0.5,0,0) circle (0.05);
      \filldraw [draw=black, fill=red] (0.5,0,0) circle (0.05);
      \node[anchor=south] at (-0.5,0,0) {$g_2$};
      \node[anchor=north] at (0.5,0,0) {$g_4$};
      \filldraw [draw=black, fill=red] (0,0,-0.6) circle (0.05);
      \filldraw [draw=black, fill=red] (0,0,0.6) circle (0.05);
       \node[anchor=south] at (0,0,-0.6) {${g_1}$};
      \node[anchor=north] at (0,0,0.6) {$g_3$};
  \end{tikzpicture} 
   = \sum_{g_1,g_2,g_3, g_4\in G} W_{g_1,g_2,g_3, g_4}
     \begin{tikzpicture}
   \draw (-0.7,0,0) -- (0.7,0,0);
      \draw (0,0,-0.9) -- (0,0,0.9);
      \pic at (0,0,0) {3dpeps};
     \filldraw [draw=black, fill=red] (0.5,0,0) circle (0.05);
      \node[anchor=north] at (0.5,0,0) {$g_2g_4$};
      \filldraw [draw=black, fill=red] (0,0,-0.6) circle (0.05);
      \filldraw [draw=black, fill=red] (0,0,0.6) circle (0.05);
       \node[anchor=south] at (0,0,-0.6) {${g_2g_1}$};
      \node[anchor=north] at (0,0,0.6) {$g_3\myinv{g}_2$};
  \end{tikzpicture} \notag
 \\
 \equiv &
   \sum_{g_1,g_2, g_3\in G} T_{g_1,g_2, g_3}
     \begin{tikzpicture}
   \draw (-0.7,0,0) -- (0.7,0,0);
      \draw (0,0,-0.9) -- (0,0,0.9);
      \pic at (0,0,0) {3dpeps};
     \filldraw [draw=black, fill=red] (0.5,0,0) circle (0.05);
      \node[anchor=north] at (0.5,0,0) {$g_2$};
      \filldraw [draw=black, fill=red] (0,0,-0.6) circle (0.05);
      \filldraw [draw=black, fill=red] (0,0,0.6) circle (0.05);
       \node[anchor=south] at (0,0,-0.6) {${g_1}$};
      \node[anchor=north] at (0,0,0.6) {$g_3$};
  \end{tikzpicture}  
 \equiv
  \begin{tikzpicture}
    \draw (-0.8,0,0) -- (0.8,0,0);
    \draw(0,0,-1) -- (0,0,1);
        \pic at (0,0,0) {3dpeps};
\draw[canvas is xz plane at y=0,double=red, double distance=0.6mm,line cap=round]  (-0.1,0.7) -- (0,0.7) arc (90:-90:0.7)-- (-0.1,-0.7);
    \node at (0.3,-0.4) {$T$};
  \end{tikzpicture} \ ,
\end{align}
where we have used the $G$-invariance of the tensor $A$ to define the operator $T\in\mathcal{A}^{\otimes 3}$. 
\end{proof}

Let us introduce the tensor $\tilde{B}$ defined by 
 \begin{equation}\label{eq:defBtilde}
   \begin{tikzpicture}[baseline=-1mm]
 \pic at (0,0,0) {3dpepsres};
      \node at (0.1,-0.27,0) {$\tilde{B}$};
      \end{tikzpicture} =
  \begin{tikzpicture}
      \draw (-0.7,0,0) -- (0.7,0,0);
      \draw (0,0,-0.9) -- (0,0,0.9);
      \pic at (0,0,0) {3dpepsres};
      \node at (0.1,-0.27,0) {$B$};
      \filldraw [draw=black, fill=red] (-0.5,0,0) circle (0.05);
      \filldraw [draw=black, fill=red] (0.5,0,0) circle (0.05);
      \node[anchor=south] at (-0.5,0,0) {$\myinv{X}$};
      \node[anchor=north] at (0.5,0,0) {$X$};
      \filldraw [draw=black, fill=red] (0,0,-0.6) circle (0.05);
      \filldraw [draw=black, fill=red] (0,0,0.6) circle (0.05);
       \node[anchor=south] at (0,0,-0.6) {$Y$};
      \node[anchor=north] at (0,0,0.6) {$\myinv{Y}$};
    \end{tikzpicture}
    =
     \begin{tikzpicture}
    \draw (-0.8,0,0) -- (0.8,0,0);
    \draw(0,0,-1) -- (0,0,1);
        \pic at (0,0,0) {3dpeps};
\draw[canvas is xz plane at y=0,double=red, double distance=0.6mm,line cap=round]  (-0.1,0.7) -- (0,0.7) arc (90:-90:0.7)-- (-0.1,-0.7);
    \node at (0.3,-0.4) {$T$};
  \end{tikzpicture},
 \end{equation}
 where the invertible matrices $X,Y$ are the ones obtained in \cref{prop:groupalgebra1}. Note that the representation of the $G$-invariance of $\tilde{B}$ is the one of $B$ but conjugated by $X$ or $Y$ depending on the direction. It is clear that the tensor $\tilde{B}$ generates the same state as the tensor $B$ and therefore the same state as $A$. The next three propositions show some properties of $T$ which will end up in the conclusion that the operator $T$ is actually the identity operator. The first one is about the normalization of $T$:

\begin{proposition}[Normalization of $T$] \label{lem:Toperator} The operator $T\in\mathcal{A}^{\otimes 3}$ satisfies the following:
\begin{equation*}
  \begin{tikzpicture}
        \draw (0.8,0,0) -- (0.4,0,0);
    \draw(0,0,0.4) -- (0,0,1);
        \draw(0,0,-0.4) -- (0,0,-1);
\draw[canvas is xz plane at y=0,double=red, double distance=0.6mm,line cap=round]  (-0.1,0.7) -- (0,0.7) arc (90:-90:0.7)-- (-0.1,-0.7);
        \node[draw=black,fill=orange,inner sep=1pt, canvas is xz plane at y=0,shape=semicircle,  label={[shift={(0.3,0)}]$\ket{w}$}]  at (0,0,-1) {};
                \node[draw=black,fill=orange,inner sep=1pt,canvas is xz plane at y=0,shape=semicircle, label={[shift={(-0.2,-0.1)}]$\ket{w}$}]  at (0,0,0.4) {};
            \node[draw=black,fill=orange,inner sep=1pt,rotate=180,canvas is xz plane at y=0,shape=semicircle, label={[shift={(-0.3,-0.2)}]$\bra{w}$}]  at (0,0,-0.4) {}; 
                        \node[draw=black,fill=orange,inner sep=1pt,rotate=180,canvas is xz plane at y=0,shape=semicircle, label={[shift={(0.3,-0.4)}]$\bra{w}$}]  at (0,0,1) {};      
    \node at (0.4,-0.35) {$T$};
  \end{tikzpicture} 
  =
  \tikz  \draw (0.8,0,0) -- (0.4,0,0);
\; \longleftrightarrow \;
\left(\bra{w}\otimes \id \otimes \bra{w}\right) T\left(\ket{w}\otimes \id \otimes \ket{w}\right)=\id,
  \end{equation*}
where $\ket{w}$ is a unit vector from some one-dimensional irrep sector of the representation of $\tilde{B}$.
\end{proposition}

 \begin{proof}[Proof of \cref{lem:Toperator}]
 
The physical operator that corresponds to the insertion of an operator $X\in \mathcal{C}_A$ on a virtual bond, 

 \begin{equation*}
   \begin{tikzpicture}
    \draw[canvas is xz plane at y=0] (-0.5,-0.5) grid (3.5,2.5);
    \foreach \x in {0,1,2,3}{
      \foreach \z in {0,1,2}{
          \draw (\x,0,\z)--(\x,0.35,\z);
       \node[tensor] at (\x,0,\z) {}; 
      }
        \node[tensorr, draw=black, fill=red] at (1,0.2,1) {};
            \node[anchor=east] at (1.05,0.2,1) {$O(X)$};
    }
  \end{tikzpicture} \; = \;
   \begin{tikzpicture}
    \draw[canvas is xz plane at y=0] (-0.5,-0.5) grid (3.5,2.5);
    \foreach \x in {0,1,2,3}{
      \foreach \z in {0,1,2}{
          \draw (\x,0,\z)--(\x,0.35,\z);
       \node[tensor] at (\x,0,\z) {}; 
      }
        \node[tensorr, draw=black, fill=red] at (1.5,0,1) {};
            \node[anchor=north] at (1.5,0,1) {$X$};
    }
  \end{tikzpicture}, 
\end{equation*}
is
\begin{equation*}
   O(X)=
    \begin{tikzpicture}
      \draw (-1,0.25)--(-1,0)--(1,0)--(1,0.25);
      \draw[canvas is zy plane at x=0] (-1,0.25)--(-1,0)--(1,0)--(1,0.25);
      \draw (-1,0.25)--(-1,0.5)--(1,0.5)--(1,0.25);
      \draw [canvas is zy plane at x=0] (-1,0.25)--(-1,0.5)--(1,0.5)--(1,0.25);
      \filldraw [draw=black, fill=red] (1,0.25,0) circle (0.05);
      \pic at (0,0,0) {3dpepsdown};
      \pic at (0,0.5,0) {3dpeps};
      \node[anchor=west] at (1,0.25,0) {$ X$};
    \end{tikzpicture}.
  \end{equation*}

Let us consider some unit vector $\ket{v}$ from the trivial irrep sector, $g\ket{v}=\ket{v}\; \forall g\in G$, of the representation of $A$. There can be more than one depending on the multiplicity of the trivial irrep. All the irreps are in the decomposition of $u_g$ since it is semi-regular. Then, the rank-one projector $V=\ket{v}\bra{v}$ belongs to $\mathcal{C}_A$, in fact $V$ satisfies $gV=V=Vg$. Therefore using \cref{lem:Ginj_isomorph}, for the PEPS defined by $\tilde{B}$ there is a $W\in\mathcal{C}_{\tilde{B}}$ such that
\begin{equation}\label{eq:creaWV}
  \begin{tikzpicture}
    \draw[canvas is xz plane at y=0] (-0.5,-0.5) grid (3.5,2.7);
    \foreach \x in {0,1,2,3}{
      \node[tensorr, draw=black, fill=orange,label=left:$W$] at (\x, 0, 0.55) {};
      \node[tensorr, draw=black, fill=orange, label=left:$W$] at (\x, 0, 1.55) {};
          \node[tensorr, draw=black, fill=orange, label=left:$W$] at (\x, 0, 2.4) {};

      \foreach \z in {0,1,2}{
       \node[tensorB] at (\x,0,\z) {}; 
        \draw (\x,0,\z)--(\x,0.3,\z);
      }
            \node[anchor=north] at (0.05,0.05) {$ \tilde{B}$};
    }
  \end{tikzpicture} \ = \ 
  \begin{tikzpicture}
    \draw[canvas is xz plane at y=0] (-0.5,-0.5) grid (3.5,2.7);
    \foreach \x in {0,1,2,3}{
      \node[tensorr, draw=black, fill=red ,label=left:$V$] at (\x, 0, 0.55) {};
      \node[tensorr, draw=black, fill=red ,label=left:$V$] at (\x, 0, 1.55) {};
       \node[tensorr, draw=black, fill=red ,label=left:$V$] at (\x, 0, 2.4) {};
      \foreach \z in {0,1,2}{
        \node[tensor] at (\x,0,\z) {}; 
        \draw (\x,0,\z)--(\x,0.3,\z);
      }
            \node[anchor=north] at (0.05,0.05) {$A$};
    }
  \end{tikzpicture} .
\end{equation}

If we block $n$ bonds of the two tensor networks, the generated algebras of the representation of $G$ are $\mathcal{A}^{\otimes n}_A$ and $\mathcal{A}^{\otimes n}_{\tilde{B}}$. The corresponding centralizers are denoted by $\mathcal{C}^{n}_{A}$ and $\mathcal{C}^n_{\tilde{B}}$. Using \cref{lem:Ginj_isomorph} there exists also an isomorphism between such centralizer algebras. It can be shown that for a large $n$ that isomorphism is a gauge transformation, see \cref{sec:isols} for the proof. Then, if we denote by 
$\xi(\cdot )=C(\cdot )C^{-1}$ the isomorphism from $\mathcal{C}^{n}_{A}$ to $\mathcal{C}^n_{\tilde{B}}$ of \cref{lem:Ginj_isomorph} we see that
$$1=\tr[V^{\otimes n}]=\tr[\xi(W^{\otimes n})]\Rightarrow 1=\tr[W]^{n}.$$
This implies that $W$ is a rank-one projector. Therefore $W=\ket{w}\bra{w}$, where $\ket{w}$ is some unit vector that belongs to the eigenspace of a one-dimensional irrep of the representation of $G$ in $B$  because $W\in\mathcal{C}_{\tilde{B}}$.\\ 

With the charaterization of $V$ and $W$ the PEPS in Eq.\eqref{eq:creaWV} factorizes into a product of $G$-injective MPS with the following generating tensors
\begin{equation*}
  \begin{tikzpicture}[scale=1.2]
    \draw (-0.5,0,0) -- (0.5,0,0);
    \draw (0,0,-0.5) -- (0,0,0.5);
    \node[draw=black,fill=orange,inner sep=1pt,canvas is xz plane at y=0,shape=semicircle,label=above:$\ket{w}$]  at (0,0,-0.5) {};
        \node[draw=black,fill=orange,inner sep=1pt,rotate=180,canvas is xz plane at y=0,shape=semicircle,label=below:$\bra{w}$]  at (0,0,0.5) {};
    \node[tensorB] (b) at (0,0) {};
    \draw (b)--(0,0.2,0);
  \end{tikzpicture}
  \;  {\rm and} \;
  \begin{tikzpicture}[scale=1.2]
    \draw (-0.5,0,0) -- (0.5,0,0);
    \draw (0,0,-0.5) -- (0,0,0.5);
    \node[draw=black,fill=red,inner sep=1pt,canvas is xz plane at y=0,shape=semicircle,label=above:$\ket{v}$]  at (0,0,-0.5) {};
        \node[draw=black,fill=red,inner sep=1pt,rotate=180,canvas is xz plane at y=0,shape=semicircle,label=below:$\bra{v}$]  at (0,0,0.5) {};
    \node[tensor] (b) at (0,0) {};
    \draw (b)--(0,0.2,0);
  \end{tikzpicture},
\end{equation*}
which are $G$-invariant
\begin{equation*}
  \begin{tikzpicture}[scale=1.2]
     \draw (-0.5,0,0) -- (0.5,0,0);
    \draw (0,0,-0.5) -- (0,0,0.5);
    \node[draw=black,fill=orange,inner sep=1pt,canvas is xz plane at y=0,shape=semicircle]  at (0,0,-0.5) {};
        \node[draw=black,fill=orange,inner sep=1pt,rotate=180,canvas is xz plane at y=0,shape=semicircle]  at (0,0,0.5) {};
    \node[tensorB] (b) at (0,0) {};
    \draw (b)--(0,0.2,0);
          \filldraw[black] (0.3,0) circle (0.04);
          \node[anchor=north] at (0.3,0) {$g$};
          \filldraw[black] (-0.3,0) circle (0.04);
          \node[anchor=south] at (-0.3,0) {$g^{-1}$};
  \end{tikzpicture} =
  \begin{tikzpicture}[scale=1.2]
    \draw (-0.5,0,0) -- (0.5,0,0);
    \draw (0,0,-0.5) -- (0,0,0.5);
    \node[draw=black,fill=orange,inner sep=1pt,canvas is xz plane at y=0,shape=semicircle]  at (0,0,-0.5) {};
        \node[draw=black,fill=orange,inner sep=1pt,rotate=180,canvas is xz plane at y=0,shape=semicircle]  at (0,0,0.5) {};
    \node[tensorB] (b) at (0,0) {};
    \draw (b)--(0,0.2,0);
              \filldraw[black] (0,0,0.25) circle (0.04);
          \node[anchor=east] at (0,0,0.25) {$g$};
          \filldraw[black] (0,0,-0.25) circle (0.04);
          \node[anchor=west] at (0,0,-0.25) {$g^{-1}$};
  \end{tikzpicture} =
  \begin{tikzpicture}[scale=1.2]
    \draw (-0.5,0,0) -- (0.5,0,0);
    \draw (0,0,-0.5) -- (0,0,0.5);
    \node[draw=black,fill=orange,inner sep=1pt,canvas is xz plane at y=0,shape=semicircle]  at (0,0,-0.5) {};
        \node[draw=black,fill=orange, inner sep=1pt,rotate=180,canvas is xz plane at y=0,shape=semicircle]  at (0,0,0.5) {};
    \node[tensorB] (b) at (0,0) {};
    \draw (b)--(0,0.2,0);
  \end{tikzpicture} \ .
\end{equation*}
They also have pseudo inverses:
\begin{equation*}
  \begin{tikzpicture}
  
    \draw (-0.5,0,0) -- (0.5,0,0);
    \draw (0,0,-0.5) -- (0,0,0.5);
    \draw (-0.5,0.5,0) -- (0.5,0.5,0);
    \draw (0,0.5,-0.5) -- (0,0.5,0.5);
    \node[draw=black,fill=orange,inner sep=1pt,canvas is xz plane at y=0,shape=semicircle]  at (0,0,-0.5) {};
        \node[draw=black,fill=orange, inner sep=1pt,rotate=180,canvas is xz plane at y=0,shape=semicircle]  at (0,0,0.5) {};
            \node[draw=black,fill=orange,inner sep=1pt,canvas is xz plane at y=0,shape=semicircle]  at (0,0.5,-0.5) {};
        \node[draw=black,fill=orange, inner sep=1pt,rotate=180,canvas is xz plane at y=0,shape=semicircle]  at (0,0.5,0.5) {};
    \node[tensorB] (b) at (0,0,0) {};
    \node[tensorB] (c) at (0,0.5,0) {};
    \draw (b)--(c);
  \end{tikzpicture} \ =  \sum_g 
  \begin{tikzpicture}
    \draw (0,0,-0.8) -- (0,0,-0.2) -- (0,0.5,-0.2) -- (0,0.5,-0.8);
    \draw (0,0, 0.8) -- (0,0, 0.2) -- (0,0.5, 0.2) -- (0,0.5, 0.8);
    \draw (-0.8,0,0) -- (-0.3,0,0) -- (-0.3,0.5,0) -- (-0.8,0.5,0);
    \draw ( 0.8,0,0) -- ( 0.3,0,0) -- ( 0.3,0.5,0) -- ( 0.8,0.5,0);
        \node[draw=black,fill=orange,inner sep=1pt,canvas is xz plane at y=0,shape=semicircle]  at (0,0,-0.8) {};
        \node[draw=black,fill=orange, inner sep=1pt,rotate=180,canvas is xz plane at y=0,shape=semicircle]  at (0,0,0.8) {};
            \node[draw=black,fill=orange,inner sep=1pt,canvas is xz plane at y=0,shape=semicircle]  at (0,0.5,-0.8) {};
        \node[draw=black,fill=orange, inner sep=1pt,rotate=180,canvas is xz plane at y=0,shape=semicircle]  at (0,0.5,0.8) {};
 \filldraw[black] (0,0.25,-0.2) circle (0.05);
          \node[anchor=south west] at (0,0.25,-0.2) {$g$};    
       \filldraw[black] (0,0.25,0.2) circle (0.05);
          \node[anchor=north east] at (0,0.25,0.2) {$g^{-1}$};
 \filldraw[black] (-0.3,0.35,0) circle (0.05);
          \node[anchor= east] at (-0.3,0.35,0) {$g^{-1}$};    
       \filldraw[black] ( 0.3,0.15,0) circle (0.05);
          \node[anchor= west] at ( 0.3,0.15,0) {$g$};
  \end{tikzpicture} \ = \sum_g 
  \begin{tikzpicture}
    \draw (-0.8,0,0) -- (-0.3,0,0) -- (-0.3,0.5,0) -- (-0.8,0.5,0);
    \draw ( 0.8,0,0) -- ( 0.3,0,0) -- ( 0.3,0.5,0) -- ( 0.8,0.5,0);
    \node[tensorr, label=left:$g^{-1}$]       at (-0.3,0.25,0) {};
    \node[tensorr, label=right:$g$]  at ( 0.3,0.25,0) {};
  \end{tikzpicture} \ .
\end{equation*}
Both generated MPSs are the same for every system size. Therefore, using here the fundamental theorem of Ref. \cite{PerezGarcia07}, the generating tensors are related to each other by an invertible matrix:
\begin{equation*}
  \begin{tikzpicture}[scale=1.2]
    \draw (-0.5,0,0) -- (0.5,0,0);
    \draw (0,0,-0.5) -- (0,0,0.5);
    \node[draw=black,fill=orange,inner sep=1pt,canvas is xz plane at y=0,shape=semicircle,label=above:$\ket{w}$]  at (0,0,-0.5) {};
        \node[draw=black,fill=orange,inner sep=1pt,rotate=180,canvas is xz plane at y=0,shape=semicircle,label=below:$\bra{w}$]  at (0,0,0.5) {};
    \node[tensorB] (b) at (0,0) {};
    \draw (b)--(0,0.2,0);
  \end{tikzpicture}
=
  \begin{tikzpicture}[scale=1.2]
    \draw (-0.5,0,0) -- (0.5,0,0);
    \draw (0,0,-0.5) -- (0,0,0.5);
    \node[draw=black,fill=red,inner sep=1pt,canvas is xz plane at y=0,shape=semicircle,label=above:$\ket{v}$]  at (0,0,-0.5) {};
        \node[draw=black,fill=red,inner sep=1pt,rotate=180,canvas is xz plane at y=0,shape=semicircle,label=below:$\bra{v}$]  at (0,0,0.5) {};
    \node[tensor] (b) at (0,0) {};
        \filldraw[draw=black, fill=red] (0.3,0) circle (0.06);
          \node[anchor=north] at (0.3,0) {$Z$};
          \filldraw[draw=black, fill=red] (-0.3,0) circle (0.06);
          \node[anchor=south] at (-0.3,0) {$Z^{-1}$};
    \draw (b)--(0,0.2,0);
  \end{tikzpicture},
\end{equation*}

Writing out the relation between the $A$ and $\tilde{B}$ tensors, we get

\begin{equation*}
  \begin{tikzpicture}
      \node[tensor] (b) at (0,0) {};
          \draw (-0.5,0,0) -- (0.5,0,0);
    \draw (0,0,-0.5) -- (0,0,0.5);
        \draw (b)--(0,0.2,0);
        \draw (0.8,0,0) -- (0.4,0,0);
    \draw(0,0,0.4) -- (0,0,1);
        \draw(0,0,-0.4) -- (0,0,-1);
\draw[canvas is xz plane at y=0,double=red, double distance=0.6mm,line cap=round]  (-0.1,0.7) -- (0,0.7) arc (90:-90:0.7)-- (-0.1,-0.7);
        \node[draw=black,fill=orange,inner sep=1pt, canvas is xz plane at y=0,shape=semicircle,  label={[shift={(0.3,0)}]$\ket{w}$}]  at (0,0,-1) {};
                        \node[draw=black,fill=orange,inner sep=1pt,rotate=180,canvas is xz plane at y=0,shape=semicircle, label={[shift={(0.3,-0.4)}]$\bra{w}$}]  at (0,0,1) {};      
    \node at (0.4,-0.35) {$T$};
  \end{tikzpicture} 
  =
  \begin{tikzpicture}[scale=1.2]
    \draw (-0.5,0,0) -- (0.5,0,0);
    \draw (0,0,-0.5) -- (0,0,0.5);
    \node[draw=black,fill=red,inner sep=1pt,canvas is xz plane at y=0,shape=semicircle,label=above:$\ket{v}$]  at (0,0,-0.5) {};
        \node[draw=black,fill=red,inner sep=1pt,rotate=180,canvas is xz plane at y=0,shape=semicircle,label=below:$\bra{v}$]  at (0,0,0.5) {};
    \node[tensor] (b) at (0,0) {};
        \filldraw[draw=black, fill=red] (0.3,0) circle (0.06);
          \node[anchor=north] at (0.3,0) {$Z$};
          \filldraw[draw=black, fill=red] (-0.3,0) circle (0.06);
          \node[anchor=south] at (-0.3,0) {$Z^{-1}$};
    \draw (b)--(0,0.2,0);
  \end{tikzpicture}.
\end{equation*}
The operator acting on $A$ in the LHS, using \cref{eq:defT}, is 

\begin{align}
T\left(\ket{w}\otimes \id \otimes \ket{w}\right)=&  \sum_{g_1,g_2g_3}T_{g_1,g_2,g_3}g_1\ket{w}\otimes g_2\otimes g_3\ket{w} \notag\\
=&   \sum_{g_1,g_2g_3}T_{g_1,g_2,g_3}\ket{w}\otimes g_2\otimes \ket{w} \notag \\
=&  \left(\ket{w}\otimes \id \otimes \ket{w}\right) \sum_{g_1,g_2g_3}T_{g_1,g_2,g_3} \id \otimes g_2\otimes \id \notag \\
=&  \left(\ket{w}\otimes \id \otimes \ket{w}\right)  \id \otimes \left( \sum_{g_1,g_2g_3}T_{g_1,g_2,g_3}g_2 \right)\otimes \id \notag .
\end{align}

That is,

\begin{equation*}
  \begin{tikzpicture}[scale=1.2]
    \draw (-0.5,0,0) -- (0.5,0,0);
    \draw (0,0,-0.5) -- (0,0,0.5);
    \node[draw=black,fill=orange,inner sep=1pt,canvas is xz plane at y=0,shape=semicircle,label=above:$\ket{w}$]  at (0,0,-0.5) {};
        \node[draw=black,fill=orange,inner sep=1pt,rotate=180,canvas is xz plane at y=0,shape=semicircle,label=below:$\bra{w}$]  at (0,0,0.5) {};
            \node[tensorr, draw=black, fill=red , label=below:$T|_2$] at (0.35,0) {};
    \node[tensor] (b) at (0,0) {};
    \draw (b)--(0,0.2,0);
  \end{tikzpicture}
=
  \begin{tikzpicture}[scale=1.2]
    \draw (-0.5,0,0) -- (0.5,0,0);
    \draw (0,0,-0.5) -- (0,0,0.5);
    \node[draw=black,fill=red,inner sep=1pt,canvas is xz plane at y=0,shape=semicircle,label=above:$\ket{v}$]  at (0,0,-0.5) {};
        \node[draw=black,fill=red,inner sep=1pt,rotate=180,canvas is xz plane at y=0,shape=semicircle,label=below:$\bra{v}$]  at (0,0,0.5) {};
    \node[tensor] (b) at (0,0) {};
        \filldraw[draw=black, fill=red] (0.3,0) circle (0.06);
          \node[anchor=north] at (0.3,0) {$Z$};
          \filldraw[draw=black, fill=red] (-0.3,0) circle (0.06);
          \node[anchor=south] at (-0.3,0) {$Z^{-1}$};
    \draw (b)--(0,0.2,0);
  \end{tikzpicture},
\end{equation*}
where $T|_2=\left(\bra{w}\otimes \id \otimes \bra{w}\right) T\left(\ket{w}\otimes \id \otimes \ket{w}\right)=\sum_{g_1,g_2g_3}T_{g_1,g_2,g_3}g_2$. We now apply the inverse of $A$

\begin{equation*}
\sum_g 
  \begin{tikzpicture}[baseline=1mm]
    \draw (-0.8,0,0) -- (-0.3,0,0) -- (-0.3,0.5,0) -- (-0.8,0.5,0);
    \draw ( 0.8,0,0) -- ( 0.3,0,0) -- ( 0.3,0.5,0) -- ( 0.8,0.5,0);
    \node[tensorr, label=left:$g^{-1}$]       at (-0.3,0.25,0) {};
    \node[tensorr, label=right:$g$]  at ( 0.3,0.25,0) {};
        \node[tensorr, draw=black, fill=red, label=below:$T|_2$]  at (0.5,0,0) {};
  \end{tikzpicture} 
  =   \sum_g 
  \begin{tikzpicture}[baseline=1mm]
    \draw (-0.8,0,0) -- (-0.3,0,0) -- (-0.3,0.5,0) -- (-0.8,0.5,0);
    \draw ( 0.8,0,0) -- ( 0.3,0,0) -- ( 0.3,0.5,0) -- ( 0.8,0.5,0);
    \node[tensorr, label=left:$g^{-1}$]       at (-0.3,0.25,0) {};
    \node[tensorr, label=right:$g$]  at ( 0.3,0.25,0) {};
        \node[tensorr, draw=black, fill=red, label=below:$Z$]  at (0.5,0,0) {};
          \node[tensorr, draw=black, fill=red, label=below:$Z^{-1}$]  at (-0.5,0,0) {};
  \end{tikzpicture}  .
\end{equation*}
Applying $Z$ on the left bottom leg and tracing out the right term, we get 
\begin{equation*}
\sum_g 
  \begin{tikzpicture}[baseline=1mm]
    \draw (-0.8,0,0) -- (-0.3,0,0) -- (-0.3,0.5,0) -- (-0.8,0.5,0);
    \draw ( 0.8,0,0) -- ( 0.3,0,0) -- ( 0.3,0.5,0) -- ( 0.8,0.5,0)--( 0.8,0,0) ;
    \node[tensorr, label=left:$g^{-1}$]       at (-0.3,0.25,0) {};
    \node[tensorr, label=right:$g$]  at ( 0.3,0.25,0) {};
        \node[tensorr, draw=black, fill=red, label=below:$T|_2$]  at (0.5,0,0) {};
                  \node[tensorr, draw=black, fill=red, label=below:$Z$]  at (-0.5,0,0) {};

  \end{tikzpicture} 
  =   \sum_g 
  \begin{tikzpicture}[baseline=1mm]
    \draw (-0.8,0,0) -- (-0.3,0,0) -- (-0.3,0.5,0) -- (-0.8,0.5,0);
    \draw ( 0.8,0,0) -- ( 0.3,0,0) -- ( 0.3,0.5,0) -- ( 0.8,0.5,0)--( 0.8,0,0) ;
    \node[tensorr, label=left:$g^{-1}$]       at (-0.3,0.25,0) {};
    \node[tensorr, label=right:$g$]  at ( 0.3,0.25,0) {};
        \node[tensorr, draw=black, fill=red, label=below:$Z$]  at (0.5,0,0) {};
  \end{tikzpicture}\ ,
\end{equation*}
which implies that  $Z=\sum_g \zeta_g g \in \mathcal{A}$ for some coefficients $\zeta_g\in \mathbb{C}$.
Since $Z\in \mathcal{A}$, it satisfies $(Z\otimes \id)(\sum_g g^{-1}\otimes g)=(\sum_g g^{-1}\otimes g) (\id \otimes Z)$ so
\begin{equation*}
\sum_g 
  \begin{tikzpicture}[baseline=1mm]
    \draw (-0.8,0,0) -- (-0.3,0,0) -- (-0.3,0.5,0) -- (-0.8,0.5,0);
    \draw ( 0.8,0,0) -- ( 0.3,0,0) -- ( 0.3,0.5,0) -- ( 0.8,0.5,0);
    \node[tensorr, label=left:$g^{-1}$]       at (-0.3,0.25,0) {};
    \node[tensorr, label=right:$g$]  at ( 0.3,0.25,0) {};
        \node[tensorr, draw=black, fill=red, label=below:$T|_2$]  at (0.5,0,0) {};
  \end{tikzpicture} 
  =   \sum_g 
  \begin{tikzpicture}
    \draw (-0.8,0,0) -- (-0.3,0,0) -- (-0.3,0.5,0) -- (-0.8,0.5,0);
    \draw ( 0.8,0,0) -- ( 0.3,0,0) -- ( 0.3,0.5,0) -- ( 0.8,0.5,0);
    \node[tensorr, label=left:$g^{-1}$]       at (-0.3,0.25,0) {};
    \node[tensorr, label=right:$g$]  at ( 0.3,0.25,0) {};
  \end{tikzpicture} .
\end{equation*}
Therefore tracing out the left leg leads to the conclusion that $T|_2=\id$ 
\end{proof}

We are now going to consider the blocking of two tensors. This results in $G$-injective PEPS tensors that generate the same state for every system size. The representation of the group under the blocking is given by the linear map $\Delta:\mathcal{A}\mapsto \mathcal{A}\otimes \mathcal{A} $ which is defined by
$$\Delta(g)=g\otimes g,\; \forall g\in G,$$
where by linearity $\Delta(\sum_g\alpha_g g)= \sum_g\alpha_g g\otimes g$. The composition of the map will be denoted by $\Delta^n:\mathcal{A}\mapsto \mathcal{A}^{\otimes n}; \Delta^n(g)=g^{\otimes n}$ for $g \in G$. 
\

The map $\Delta$ is borrowed here from the axioms of Hopf algebras, see \cite{HAbook}, and it is called coproduct. In fact $\mathcal{A}$ is a Hopf algebra if we also consider the map $\epsilon:\mathcal{A}\mapsto \mathbb{C}$ and the anti-linear map $S:\mathcal{A}\mapsto \mathcal{A}$ defined by $\epsilon(g)=1$ and $S(g)=g^{-1}$ for all $g\in G$ respectively.

The next proposition shows how the operator $T$ behaves when concatenating two $\tilde{B}$ tensors, see \cref{eq:defBtilde}. For the sake of clarity the order of the factors in the tensor product in the equations would be:

\begin{equation*}  
  \begin{tikzpicture}
   \draw (0,0,0)--(1,0,0);
         \pic at (0,0,0) {3dpeps};
      \pic at (1,0,0) {3dpeps};
   \node[anchor=west] at (1.5,0,0) {$3 $};
    \draw (0,0,-1.1)--(0,0,-0.3);
       \node[anchor=south] at (0,0,-1.1) {$1$};
    \draw (0,0,0.3)--(0,0,1.1);
           \node[anchor=north] at (0,0,1.1) {$5$};
    \draw (1,0,-0.3)--(1,0,-1.1);
              \node[anchor=south] at (1,0,-1.1) {$2$};
    \draw (1,0,0.3)--(1,0,1.1);
                  \node[anchor=north] at (1,0,1.1) {$4$};
       \end{tikzpicture}
\end{equation*}

We find that two operators $T$ acting on each tensor are equal to the action of one $T$ acting on both tensors, using $\Delta$, plus a gauge transformation in the concatenating direction:

\begin{proposition}[Concatenation of $T$] \label{propYT}
  There is a $Y\in \mathcal{A}^{\otimes 2}$ invertible such that
  \begin{equation}\label{Tconca}
    \begin{tikzpicture}
      \draw (-0.7,0,0)--(2.7,0,0);
      \draw (0,0,-1.7)--(0,0,1.7);
      \draw (1.7,0,-1.7)--(1.7,0,1.7);
\draw[canvas is xz plane at y=0,double=red, double distance=0.6mm,line cap=round]  (-0.1,0.7) -- (0,0.7) arc (90:-90:0.7)-- (-0.1,-0.7);

\draw[canvas is xz plane at y=0,double=red, double distance=0.6mm,line cap=round]  (-0.1+1.7,0.7) -- (0+1.7,0.7) arc (90:-90:0.7)-- (-0.1+1.7,-0.7);

      \pic at (0,0,0) {3dpeps};
      \pic at (1.7,0,0) {3dpeps};
       \node at (0.3,-0.4) {$T$};
        \node at (2,-0.4) {$T$};
    \end{tikzpicture}   = 
    \begin{tikzpicture}
      \draw (0,0,0)--(2.1,0,0);
      \draw (0,0,-1.9)--(0,0,1.9);
      \draw (1,0,-1.9)--(1,0,1.9);
      \pic at (0,0,0) {3dpeps};
      \pic at (1,0,0) {3dpeps};
      \draw[rounded corners=.05cm,black, fill=red, canvas is xz plane at y=0] (-0.3,-1.3) rectangle (1.3,-1.5);
      \draw[rounded corners=.05cm,black, fill=red, canvas is xz plane at y=0] (-0.3,1.3) rectangle (1.3,1.5);
      
      \draw[canvas is xz plane at y=0,double=red, double distance=0.6mm,line cap=round]  (-0.1,0.8) -- (1,0.8) arc (90:-90:0.8)-- (-0.1,-0.8);

       \node at (1.7,-0.4) {$(\Delta\otimes \id\otimes \Delta)T$};
       \node at (-1.1,-0.4) {$Y^{-1}$};
       \node at (0.6,0.8) {$Y$};
    \end{tikzpicture},
  \end{equation}  
 where $(\Delta\otimes \id\otimes \Delta)T=\sum_{g_1,g_2g_3}T_{g_1,g_2,g_3}g^{\otimes 2}_1\otimes g_2\otimes g^{\otimes 2}_3$
\end{proposition}

\begin{proof}
The tensors
\begin{equation*}
  \begin{tikzpicture}
    \draw (-0.7,0,0)--(2.7,0,0);
    \draw (0,0,-1.7)--(0,0,1.7);
    \draw (1.7,0,-1.7)--(1.7,0,1.7);
\draw[canvas is xz plane at y=0,double=red, double distance=0.6mm,line cap=round]  (-0.1,0.7) -- (0,0.7) arc (90:-90:0.7)-- (-0.1,-0.7);

\draw[canvas is xz plane at y=0,double=red, double distance=0.6mm,line cap=round]  (-0.1+1.7,0.7) -- (0+1.7,0.7) arc (90:-90:0.7)-- (-0.1+1.7,-0.7);
    \pic at (0,0,0) {3dpeps};
    \pic at (1.7,0,0) {3dpeps};
      \node at (0.3,-0.4) {$T$};
        \node at (2,-0.4) {$T$};
  \end{tikzpicture}  \quad \text{and} \quad 
  \begin{tikzpicture}
    \pic at (0,0,0) {3dpeps};
    \draw (0,0,0)--(1,0,0);
    \pic at (1,0,0) {3dpeps};
  \end{tikzpicture}
\end{equation*} 
are both $G$-injective and they generate the same state because of Eq.(\ref{eq:defT}). 
Using \cref{prop:groupalgebra1} there are invertible operators $X$, acting on one bond, and $\tilde{Y}$ acting on two bonds. Also there is an operator $T'$ acting on the bonds 1,2,3,4 and 5. Notice that since this operator belongs to the algebra generated by the representation of two blocked tensors $A$, $T'$ has a special form. That is, there exist an $F \in \mathcal{A}^{\otimes 3}$ such that $T'=(\Delta\otimes \id\otimes \Delta)F$ so
\begin{equation}\label{eq:XinG}
  \begin{tikzpicture}
    \draw (-0.7,0,0)--(2.7,0,0);
    \draw (0,0,-1.7)--(0,0,1.7);
    \draw (1.7,0,-1.7)--(1.7,0,1.7);
\draw[canvas is xz plane at y=0,double=red, double distance=0.6mm,line cap=round]  (-0.1,0.7) -- (0,0.7) arc (90:-90:0.7)-- (-0.1,-0.7);
\draw[canvas is xz plane at y=0,double=red, double distance=0.6mm,line cap=round]  (-0.1+1.7,0.7) -- (0+1.7,0.7) arc (90:-90:0.7)-- (-0.1+1.7,-0.7);
    \pic at (0,0,0) {3dpeps};
    \pic at (1.7,0,0) {3dpeps};
  \end{tikzpicture}   = 
  \begin{tikzpicture}
    \draw (-1.0,0,0)--(2.6,0,0);
    \draw (0,0,-1.9)--(0,0,1.9);
    \draw (1,0,-1.9)--(1,0,1.9);
    \pic at (0,0,0) {3dpeps};
    \pic at (1,0,0) {3dpeps};
    \draw[rounded corners=.05cm,black, fill=red, canvas is xz plane at y=0] (-0.3,-1.3) rectangle (1.3,-1.5);
    \draw[rounded corners=.05cm,black, fill=red, canvas is xz plane at y=0] (-0.3,1.3) rectangle (1.3,1.5);
      \draw[canvas is xz plane at y=0,double=red, double distance=0.6mm,line cap=round]  (-0.1,0.8) -- (1,0.8) arc (90:-90:0.8)-- (-0.1,-0.8);
    \filldraw[draw=black, fill=red]  (-0.5,0,0) circle (0.06);
    \filldraw[draw=black, fill=red]  (2.2,0,0) circle (0.06);
    \node at (1.7,-0.4) {$(\Delta\otimes \id\otimes \Delta)F$};
    \node at (-0.5,0.2) {$X$};
    \node at (2.5,0.2) {$X^{-1}$};
     \node at (-1.1,-0.4) {$\tilde{Y}^{-1}$};
       \node at (0.6,0.8) {$\tilde{Y}$};
  \end{tikzpicture}.
\end{equation}
After rearranging, we get

\begin{equation*}
  \begin{tikzpicture}
    \draw (-0.7,0,0)--(2.9,0,0);
    \draw (0,0,-1.7)--(0,0,1.7);
    \draw (1.7,0,-1.7)--(1.7,0,1.7);
    \pic at (0,0,0) {3dpeps};
    \pic at (1.7,0,0) {3dpeps};
\draw[canvas is xz plane at y=0,double=red, double distance=0.5mm,line cap=round]  (-0.1,0.5) -- (1.7,0.5) arc (90:-90:0.5)-- (-0.1,-0.5);
\draw[canvas is xz plane at y=0,double=red, double distance=0.6mm,line cap=round]  (-0.1+1.7,0.8) -- (0+1.7,0.8) arc (90:-90:0.8)-- (-0.1+1.7,-0.8);
    \filldraw[draw=black, fill=red] (2.7,0,0) circle (0.06);
     \node at (0.3,-0.4) {$(\id\otimes \Delta^2 \otimes \id)T$};
     \node at (2,-0.4) {$T$};
  \end{tikzpicture}   = 
  \begin{tikzpicture}
    \draw (-1.0,0,0)--(2.3,0,0);
    \draw (0,0,-1.9)--(0,0,1.9);
    \draw (1,0,-1.9)--(1,0,1.9);
    \pic at (0,0,0) {3dpeps};
    \pic at (1,0,0) {3dpeps};
    \draw[rounded corners=.05cm,black, fill=red, canvas is xz plane at y=0] (-0.3,-1.3) rectangle (1.3,-1.5);
    \draw[rounded corners=.05cm,black, fill=red, canvas is xz plane at y=0] (-0.3,1.3) rectangle (1.3,1.5);
      \draw[canvas is xz plane at y=0,double=red, double distance=0.6mm,line cap=round]  (-0.1,0.8) -- (1,0.8) arc (90:-90:0.8)-- (-0.1,-0.8);

             \node at (1.7,-0.4) {$(\Delta\otimes \id\otimes \Delta)F$};
  \end{tikzpicture},
\end{equation*} 
where $ (\id\otimes \Delta^2 \otimes \id)T$ is the operator $T$ after using the $G$-injectivity of $A$ to go through the second tensor. On the LHS, the leftmost virtual leg is not acted by any operator. We apply the inverse of the two concatenated tensors of \cref{eq:XinG} using \cref{conca}  and we trace out the leftmost virtual leg. Then, the tensors are removed from the LHS together with the leftmost leg. 
On the RHS, the same operation yields an element from the group algebra, $\tilde{X}= \sum_g \tr{g^{-1}X} g^{\otimes 3}\otimes {g^{-1}}^{\otimes 3}$. $\tilde{X}$ can be incorporated in $F$ which allow us to write: 
\begin{equation*}
  \begin{tikzpicture}
    \draw (2.0,0,0)--(2.9,0,0);
    \draw (0,0,-0.8)--(0,0,-0.3);
    \draw (0,0,0.3)--(0,0,0.8);
    \draw (1.7,0,-0.3)--(1.7,0,-1.1);
    \draw (1.7,0,0.3)--(1.7,0,1.1);
  \draw[canvas is xz plane at y=0,double=red, double distance=0.5mm,line cap=round]  (-0.1,0.5) -- (1.7,0.5) arc (90:-90:0.5)-- (-0.1,-0.5);
\draw[canvas is xz plane at y=0,double=red, double distance=0.6mm,line cap=round]  (-0.1+1.7,0.8) -- (0+1.7,0.8) arc (90:-90:0.8)-- (-0.1+1.7,-0.8);
     \filldraw[draw=black, fill=red] (2.7,0,0) circle (0.06);
    \node at (0.3,-0.4) {$(\id\otimes \Delta^2 \otimes \id)T$};
     \node at (2,-0.4) {$T$};
  \end{tikzpicture}   = 
  \begin{tikzpicture}
    \draw (1.6,0,0)--(2.0,0,0);
    \draw (0,0,-1.9)--(0,0,-0.4);
    \draw (0,0,0.4)--(0,0,1.9);
    \draw (1,0,-1.9)--(1,0,-0.4);
    \draw (1,0,0.4)--(1,0,1.9);
    \draw[rounded corners=.05cm,black, fill=red, canvas is xz plane at y=0] (-0.3,-1.3) rectangle (1.3,-1.5);
    \draw[rounded corners=.05cm,black, fill=red, canvas is xz plane at y=0] (-0.3,1.3) rectangle (1.3,1.5);
      \draw[canvas is xz plane at y=0,double=red, double distance=0.6mm,line cap=round]  (-0.1,0.8) -- (1,0.8) arc (90:-90:0.8)-- (-0.1,-0.8);
  \node at (1.7,-0.4) {$(\Delta\otimes \id\otimes \Delta)F$};
  \end{tikzpicture}.
\end{equation*}  
If we project the leftmost four legs on the trivial irrep, using $\ket{w}$, the LHS becomes $X$ according to \cref{lem:Toperator}. The RHS becomes an operator of $\mathcal{A}$. Therefore $X\in \mathcal{A}$, and the same happens for $X^{-1}$, so from \cref{eq:XinG} we can define a new operator $\tilde{F}$ incorporating $X,X^{-1}$ and $F$ so that the following holds:
\begin{equation}\label{eq:TF}
  \begin{tikzpicture}
    \draw (2.0,0,0)--(2.9,0,0);
    \draw (0,0,-0.8)--(0,0,-0.3);
    \draw (0,0,0.3)--(0,0,0.8);
    \draw (1.7,0,-0.3)--(1.7,0,-1.1);
    \draw (1.7,0,0.3)--(1.7,0,1.1);
   \draw[canvas is xz plane at y=0,double=red, double distance=0.5mm,line cap=round]  (-0.1,0.5) -- (1.7,0.5) arc (90:-90:0.5)-- (-0.1,-0.5);
\draw[canvas is xz plane at y=0,double=red, double distance=0.6mm,line cap=round]  (-0.1+1.7,0.8) -- (0+1.7,0.8) arc (90:-90:0.8)-- (-0.1+1.7,-0.8);
      \node at (0.3,-0.4) {$(\id\otimes \Delta^2 \otimes \id)T$};
     \node at (2,-0.4) {$T$};
  \end{tikzpicture}   = 
  \begin{tikzpicture}
    \draw (1.6,0,0)--(2.0,0,0);
    \draw (0,0,-1.9)--(0,0,-0.4);
    \draw (0,0,0.4)--(0,0,1.9);
    \draw (1,0,-1.9)--(1,0,-0.4);
    \draw (1,0,0.4)--(1,0,1.9);
    \draw[rounded corners=.05cm,black, fill=red, canvas is xz plane at y=0] (-0.3,-1.3) rectangle (1.3,-1.5);
    \draw[rounded corners=.05cm,black, fill=red, canvas is xz plane at y=0] (-0.3,1.3) rectangle (1.3,1.5);
      \draw[canvas is xz plane at y=0,double=red, double distance=0.6mm,line cap=round]  (-0.1,0.8) -- (1,0.8) arc (90:-90:0.8)-- (-0.1,-0.8);

       \node at (1.7,-0.4) {$(\Delta\otimes \id\otimes \Delta )\tilde{F}$};
  \end{tikzpicture}.
\end{equation} 
 Now we project the legs 1 and 5 onto the trivial irrep  and using again \cref{lem:Toperator} we obtain that
\begin{equation}\label{eq:TY1F}
  T = (\tilde{Y}|_1\otimes \id \otimes \myinv{\tilde{Y}}|_1)\tilde{F} ,
\end{equation} 
where we have defined $\tilde{Y}|_1=( \id \otimes \bra{w})\tilde{Y} ( \id \otimes \ket{w})$. 

We can now define $Y=\tilde{Y}(\myinv{\tilde{Y}}|_1\otimes \myinv{\tilde{Y}}|_1)$ provided that $(\tilde{Y}|_1)^{-1}=\myinv{\tilde{Y}}|_1$. 

%

With this transformation it is clear that now ${Y}$ satisfies ${Y}|_1=\id$. Finally, we can write
\begin{equation*}
  \begin{tikzpicture}
    \draw (2.0,0,0)--(2.9,0,0);
    \draw (0,0,-0.8)--(0,0,-0.3);
    \draw (0,0,0.3)--(0,0,0.8);
    \draw (1.7,0,-0.3)--(1.7,0,-1.1);
    \draw (1.7,0,0.3)--(1.7,0,1.1);
  \draw[canvas is xz plane at y=0,double=red, double distance=0.5mm,line cap=round]  (-0.1,0.5) -- (1.7,0.5) arc (90:-90:0.5)-- (-0.1,-0.5);
\draw[canvas is xz plane at y=0,double=red, double distance=0.6mm,line cap=round]  (-0.1+1.7,0.8) -- (0+1.7,0.8) arc (90:-90:0.8)-- (-0.1+1.7,-0.8);
      \node at (0.3,-0.4) {$(\id\otimes \Delta^2 \otimes \id)T$};
     \node at (2,-0.4) {$T$};
  \end{tikzpicture}   = 
  \begin{tikzpicture}
       \node at (-1.1,-0.4) {${Y}^{-1}$};
       \node at (0.6,0.8) {${Y}$};
    \draw (1.6,0,0)--(2.0,0,0);
    \draw (0,0,-1.9)--(0,0,-0.4);
    \draw (0,0,0.4)--(0,0,1.9);
    \draw (1,0,-1.9)--(1,0,-0.4);
    \draw (1,0,0.4)--(1,0,1.9);
    \draw[rounded corners=.05cm,black, fill=red, canvas is xz plane at y=0] (-0.3,-1.3) rectangle (1.3,-1.5);
    \draw[rounded corners=.05cm,black, fill=red, canvas is xz plane at y=0] (-0.3,1.3) rectangle (1.3,1.5);
      \draw[canvas is xz plane at y=0,double=red, double distance=0.6mm,line cap=round]  (-0.1,0.8) -- (1,0.8) arc (90:-90:0.8)-- (-0.1,-0.8);

       \node at (1.7,-0.4) {$(\Delta\otimes \id\otimes \Delta)T$};
  \end{tikzpicture}.
\end{equation*} 
Applying $\ket{w}$ on the two leftmost legs we get
\begin{equation*}
  T  = ({Y}|_2\otimes \id \otimes {Y}|_2^{-1})T,
\end{equation*}
and therefore ${Y}|_2=\id$.

%
\end{proof}
We study further the properties of the operator $T$. We show that ${Y}$ is an object also defined in the context of Hopf algebras, a so-called twist, see \cite{Andruskiewitsch17}.

\begin{definition}
A twist is an element ${Y}$ of $\mathcal{A}\otimes \mathcal{A}$ that satisfies the condition:
$$\id\otimes {Y}(\id\otimes \Delta) {Y} = {Y}\otimes \id (\Delta\otimes \id) {Y}.$$
\end{definition}
With this caracterization of $Y$ we can obtain how $T$ grows from one tensor to two, {\it i.e.} how the operator $T$ is pushed to the boundary when blocking tensors.

\begin{proposition}[The growth of $T$] The operator $Y$ satisfies the following:
\begin{itemize}
\item it is symmetric under transposition,
\item it is a twist.
\end{itemize}
We have
\begin{equation*}
  \begin{tikzpicture}
    \draw (-0.7,0,0)--(2.7,0,0);
    \draw (0,0,-1.7)--(0,0,1.7);
    \draw (1.7,0,-1.7)--(1.7,0,1.7);
    
    \draw[canvas is xz plane at y=0,double=red, double distance=0.6mm,line cap=round]  (-0.1,0.7) -- (0,0.7) arc (90:-90:0.7)-- (-0.1,-0.7);
\draw[canvas is xz plane at y=0,double=red, double distance=0.6mm,line cap=round]  (-0.1+1.7,0.7) -- (0+1.7,0.7) arc (90:-90:0.7)-- (-0.1+1.7,-0.7);

    \pic at (0,0,0) {3dpeps};
    \pic at (1.7,0,0) {3dpeps};
       \node at (0.3,-0.4) {$T$};
        \node at (2,-0.4) {$T$};
  \end{tikzpicture}   = 
  \begin{tikzpicture}
    \draw (0,0,0)--(2.1,0,0);
    \draw (0,0,-1.9)--(0,0,1.9);
    \draw (1,0,-1.9)--(1,0,1.9);
    \pic at (0,0,0) {3dpeps};
    \pic at (1,0,0) {3dpeps};
      \draw[canvas is xz plane at y=0,double=red, double distance=0.6mm,line cap=round]  (-0.1,0.8) -- (1,0.8) arc (90:-90:0.8)-- (-0.1,-0.8);

              \node at (1.7,-0.4) {$(\Delta\otimes \id\otimes \Delta)T$};
  \end{tikzpicture}.
\end{equation*}
\end{proposition}

\begin{proof}
If we close one direction of the PEPS with two sites, we obtain two $G$-injective MPS that generate the same state for all system sizes. They are thus related by an invertible matrix $X$:

  \begin{equation*}
   \begin{tikzpicture}
      \draw (0,0,0)--(1.9,0,0);
      \draw (-0.7,0) rectangle (1.9,-0.23);
      \draw (0,0,-1.9)--(0,0,1.9);
      \draw (1,0,-1.9)--(1,0,1.9);
      \pic at (0,0,0) {3dpeps};
      \pic at (1,0,0) {3dpeps};
      \draw[rounded corners=.05cm,black, fill=red, canvas is xz plane at y=0] (-0.3,-1.3) rectangle (1.3,-1.5);
      \draw[rounded corners=.05cm,black, fill=red, canvas is xz plane at y=0] (-0.3,1.3) rectangle (1.3,1.5);
       \node at (-1.1,-0.4) {$X$};
       \node at (0.6,0.8) {$X^{-1}$};
    \end{tikzpicture}
 =
    \begin{tikzpicture}
     \draw (-0.7,0) rectangle (2.7,-0.23);
      \draw (-0.7,0,0)--(2.7,0,0);
      \draw (0,0,-1.7)--(0,0,1.7);
      \draw (1.7,0,-1.7)--(1.7,0,1.7);
\draw[canvas is xz plane at y=0,double=red, double distance=0.6mm,line cap=round]  (-0.1,0.7) -- (0,0.7) arc (90:-90:0.7)-- (-0.1,-0.7);
\draw[canvas is xz plane at y=0,double=red, double distance=0.6mm,line cap=round]  (-0.1+1.7,0.7) -- (0+1.7,0.7) arc (90:-90:0.7)-- (-0.1+1.7,-0.7);
      \pic at (0,0,0) {3dpeps};
      \pic at (1.7,0,0) {3dpeps};
       \node at (0.3,-0.4) {$T$};
        \node at (2,-0.4) {$T$};
    \end{tikzpicture}   = 
    \begin{tikzpicture}
      \draw (0,0,0)--(2.1,0,0);
      \draw (-0.7,0) rectangle (2.1,-0.23);
      \draw (0,0,-1.9)--(0,0,1.9);
      \draw (1,0,-1.9)--(1,0,1.9);
      \pic at (0,0,0) {3dpeps};
      \pic at (1,0,0) {3dpeps};
      \draw[rounded corners=.05cm,black, fill=red, canvas is xz plane at y=0] (-0.3,-1.3) rectangle (1.3,-1.5);
      \draw[rounded corners=.05cm,black, fill=red, canvas is xz plane at y=0] (-0.3,1.3) rectangle (1.3,1.5);
      \draw[canvas is xz plane at y=0,double=red, double distance=0.6mm,line cap=round]  (-0.1,0.8) -- (1,0.8) arc (90:-90:0.8)-- (-0.1,-0.8);

       \node at (1.7,-0.4) {$(\Delta\otimes \id\otimes \Delta)T  $};
       \node at (-1.1,-0.4) {$Y^{-1}$};
       \node at (0.6,0.8) {$Y$};
    \end{tikzpicture},
  \end{equation*} 
 We apply now the inverse of the two blocked tensors. After rearranging operators we end up with
  \begin{equation*} 
  \sum_{g\in G}
  \begin{tikzpicture}
          \filldraw[draw=black, fill=black] (0,0,0.8) circle (0.06);  
            \filldraw[draw=black, fill=black] (1,0,0.8) circle (0.06);  
                      \filldraw[draw=black, fill=black] (0,0,-0.8) circle (0.06);  
            \filldraw[draw=black, fill=black] (1,0,-0.8) circle (0.06);  
          \node at (0.4,0,0.4) {$\Delta(\myinv{g})$};
          \node at (0.5,0,-1.2) {$\Delta(g)$};
    \draw (0,0,-1.2)--(0,0,-0.4);
    \draw (0,0,0.4)--(0,0,1.9);
    \draw (1,0,-1.2)--(1,0,-0.4);
    \draw (1,0,0.4)--(1,0,1.9);
        \draw[rounded corners=.05cm,black, fill=red, canvas is xz plane at y=0] (-0.3,1.3) rectangle (1.3,1.5);
   \node at (-1.1,-0.4) {$XY$};
  \end{tikzpicture}=
  \sum_{g\in G}
   \begin{tikzpicture}
                \filldraw[draw=black, fill=black] (0,0,0.4) circle (0.06);  
                \filldraw[draw=black, fill=black] (1,0,0.4) circle (0.06);  
                \filldraw[draw=black, fill=black] (0,0,-0.4) circle (0.06);  
                \filldraw[draw=black, fill=black] (1,0,-0.4) circle (0.06);  
          \node at (0.5,0,0.4) {$\Delta(\myinv{g})$};
          \node at (0.5,0,-0.4) {$\Delta(g)$};
    \draw (1.6,0,0) rectangle (2,0,-0.3);
    \draw (0,0,-1.9)--(0,0,-0.2);
    \draw (0,0,0.2)--(0,0,1.4);
    \draw (1,0,-1.9)--(1,0,-0.2);
    \draw (1,0,0.2)--(1,0,1.4);
    \draw[rounded corners=.05cm,black, fill=red, canvas is xz plane at y=0] (-0.3,-1.3) rectangle (1.3,-1.5);
          \node at (0.6,0.8) {$YX$};
      \draw[canvas is xz plane at y=0,double=red, double distance=0.6mm,line cap=round]  (-0.1,0.8) -- (1,0.8) arc (90:-90:0.8)-- (-0.1,-0.8);

         \node at (1.7,-0.4) {$(\Delta\otimes \id\otimes \Delta)T  $};
  \end{tikzpicture}.
  \end{equation*}
  Tracing the legs 1 and 2, it follows that $XY\in {\rm Diag}(\mathcal{A}^{\otimes 2})$ and it can be absorbed in the blocked $G$-injective MPS. So we can write:
  \begin{equation*}
    \begin{tikzpicture}
      \draw (-0.7,0) rectangle (2.7,-0.23);
      \draw (0,0,-1.7)--(0,0,1.7);
      \draw (1.7,0,-1.7)--(1.7,0,1.7);
 \draw[canvas is xz plane at y=0,double=red, double distance=0.6mm,line cap=round]  (-0.1,0.7) -- (0,0.7) arc (90:-90:0.7)-- (-0.1,-0.7);
\draw[canvas is xz plane at y=0,double=red, double distance=0.6mm,line cap=round]  (-0.1+1.7,0.7) -- (0+1.7,0.7) arc (90:-90:0.7)-- (-0.1+1.7,-0.7);
      \pic at (0,0,0) {3dpeps};
      \pic at (1.7,0,0) {3dpeps};
    \end{tikzpicture} \  = \ 
    \begin{tikzpicture}
      \draw (-0.5,0) rectangle (1.5,-0.24);
      \draw (0,0,-1.9)--(0,0,1.9);
      \draw (1,0,-1.9)--(1,0,1.9);
      \pic at (0,0) {3dpeps};
      \pic at (1,0) {3dpeps};
      \draw[rounded corners=.05cm,black, fill=red, canvas is xz plane at y=0] (-0.3,-1.3) rectangle (1.3,-1.5);
      \draw[rounded corners=.05cm,black, fill=red, canvas is xz plane at y=0] (-0.3,1.3) rectangle (1.3,1.5);
          \node at (-1.1,-0.4) {$Y^{-1}$};
       \node at (0.6,0.8) {$Y$};
    \end{tikzpicture} .
  \end{equation*}
The LHS is invariant under transposition of all open indices (translational invariant) so the RHS has this invariance too, which implies the following relation for $Y$ when  $A^{-1}$ is applied on each tensor:
  \begin{equation*}
    \sum_g Y (g\otimes g) \otimes Y^{-1}(g^{-1}\otimes g^{-1}) = \sum_{g} \mathcal{T}(Y) (g\otimes g) \otimes \mathcal{T}(Y^{-1})(g^{-1}\otimes g^{-1}),
  \end{equation*}
 where $\mathcal{T}$ denotes the transposition operator. We now use $\ket{w}$ on the third factor of the tensor product obtaining $(\bra{w} \otimes \id) Y^{-1} (\ket{w} \otimes \id)=\id$. We now trace out the fourth factor so as to obtain the desired equation; $Y=\mathcal{T}(Y)$.

The twist condition comes from considering the concatenation of three $A$ tensors and exploiting the associativity of concatenation by grouping pairs of tensors using \cref{propYT}. Grouping the leftmost two tensors and using \cref{Tconca} on them
  \begin{equation*}
    \begin{tikzpicture}
      \draw (-0.7,0,0)--(4.0,0,0);
      \draw (0,0,-1.5)--(0,0,1.5);
      \draw (1.5,0,-1.5)--(1.5,0,1.5);
      \draw (3,0,-1.5)--(3,0,1.5);
      
      \draw[canvas is xz plane at y=0,double=red, double distance=0.6mm,line cap=round]  (-0.1,0.7) -- (0,0.7) arc (90:-90:0.7)-- (-0.1,-0.7);
\draw[canvas is xz plane at y=0,double=red, double distance=0.6mm,line cap=round]  (-0.1+1.5,0.7) -- (0+1.5,0.7) arc (90:-90:0.7)-- (-0.1+1.5,-0.7);
\draw[canvas is xz plane at y=0,double=red, double distance=0.6mm,line cap=round]  (-0.1+3,0.7) -- (0+3,0.7) arc (90:-90:0.7)-- (-0.1+3,-0.7);

      \pic at (0,0,0) {3dpeps};
      \pic at (1.5,0,0) {3dpeps};
      \pic at (3.0,0) {3dpeps};
    \end{tikzpicture} \  = 
    \begin{tikzpicture}
      \draw (0,0,0)--(3.5,0,0);
      \draw (0,0,-1.9)--(0,0,1.9);
      \draw (1,0,-1.9)--(1,0,1.9);
      \draw (2.5,0,-1.9)--(2.5,0,1.9);
      \pic at (0,0,0) {3dpeps};
      \pic at (1,0,0) {3dpeps};
      \draw[rounded corners=.05cm,black, fill=red, canvas is xz plane at y=0] (-0.3,-1.3) rectangle (1.3,-1.5);
      \draw[rounded corners=.05cm,black, fill=red, canvas is xz plane at y=0] (-0.3,1.3) rectangle (1.3,1.5);
      \draw[canvas is xz plane at y=0,double=red, double distance=0.6mm,line cap=round]  (-0.1,0.8) -- (1,0.8) arc (90:-90:0.8)-- (-0.1,-0.8);

\draw[canvas is xz plane at y=0,double=red, double distance=0.6mm,line cap=round]  (-0.1+2.5,0.7) -- (0+2.5,0.7) arc (90:-90:0.7)-- (-0.1+2.5,-0.7);
      \pic at (2.5,0) {3dpeps};
             \node[anchor=east] at (0,0,0.6) {$(\Delta\otimes \id\otimes \Delta)T$};
       \node at (-1.1,-0.5) {$Y^{-1}$};
       \node at (0.6,0.8) {$Y$};
    \end{tikzpicture},
  \end{equation*}  
and applying Proposition \ref{propYT} to the blocking of the leftmost two tensors and the rightmost one we obtain
  \begin{equation*}
    \begin{tikzpicture}
      \draw (-0.7,0,0)--(4.0,0,0);
      \draw (0,0,-1.5)--(0,0,1.5);
      \draw (1.5,0,-1.5)--(1.5,0,1.5);
      \draw (3,0,-1.5)--(3,0,1.5);
      \draw[canvas is xz plane at y=0,double=red, double distance=0.6mm,line cap=round]  (-0.1,0.7) -- (0,0.7) arc (90:-90:0.7)-- (-0.1,-0.7);
\draw[canvas is xz plane at y=0,double=red, double distance=0.6mm,line cap=round]  (-0.1+1.5,0.7) -- (0+1.5,0.7) arc (90:-90:0.7)-- (-0.1+1.5,-0.7);
\draw[canvas is xz plane at y=0,double=red, double distance=0.6mm,line cap=round]  (-0.1+3,0.7) -- (0+3,0.7) arc (90:-90:0.7)-- (-0.1+3,-0.7);
      \pic at (0,0,0) {3dpeps};
      \pic at (1.5,0,0) {3dpeps};
      \pic at (3.0,0) {3dpeps};
    \end{tikzpicture} \  = 
    \begin{tikzpicture}
      \draw (0,0,0)--(3.5,0,0);
      \draw (0,0,-1.9)--(0,0,1.9);
      \draw (1,0,-1.9)--(1,0,1.9);
      \draw (2.5,0,-1.9)--(2.5,0,1.9);
      \pic at (0,0,0) {3dpeps};
      \pic at (1,0,0) {3dpeps};
        
            \draw[rounded corners=.05cm,black, fill=red, canvas is xz plane at y=0] (-0.3,-1.3) rectangle (1.3,-1.5);
         \draw[rounded corners=.05cm,black, fill=red, canvas is xz plane at y=0] (-0.3,-0.9) rectangle (2.8,-1.1);   
         \node[anchor=east] at (-0.3,0,-0.9) {$(\Delta\otimes \id)Y$};
      \draw[rounded corners=.05cm,black, fill=red, canvas is xz plane at y=0] (-0.3,1.3) rectangle (1.3,1.5);
      \draw[rounded corners=.05cm,black, fill=red, canvas is xz plane at y=0] (-0.3,0.9) rectangle (2.8,1.1);

          \draw[canvas is xz plane at y=0,double=red, double distance=0.6mm,line cap=round]  (-0.1,0.5) -- (2.5,0.5) arc (90:-90:0.5)-- (-0.1,-0.5);

        \node[anchor=east] at (0,0,0.5) {$(\Delta^2\otimes \id \otimes \Delta^2)T$};
            \pic at (2.5,0) {3dpeps};
    \end{tikzpicture},
  \end{equation*}  
  where the boundary operator is $(Y\otimes \id) \cdot (\Delta\otimes \id)Y$. The same reasoning starting for the rightmost two tensors gives the boundary operator $ (1\otimes Y) \cdot (\id \otimes \Delta)Y$. The two boundaries acting on the blocked three tensors with $T$ are the same object and then one can obtain that the two boundary operators have to be the same; this gives the desired twist condition:
$\id\otimes Y(\id\otimes \Delta) Y =   Y\otimes \id (\Delta\otimes \id) Y.$

\

Any twist $Y\in \mathcal{A}^{\otimes 2}$ invariant under transposition, which are called symmetric, is trivial (see Corollary 3.3 of \cite{Etingof17}). A twist  $Y$ is trivial if there exists an invertible $ y\in \mathcal{A}$ such that $Y=(y\otimes y) \cdot \Delta(y^{-1})$. Therefore the boundary operator $Y$ has the form $Y=(y\otimes y) \cdot \Delta(y^{-1})$ and $T$ can be redefine consistently such that $T\to T(y^{-1}\otimes \id \otimes y)$.

\end{proof}

\begin{proposition}[Triviality of $T$]
  Suppose 
\begin{equation}\label{Tonetotwo}
  \begin{tikzpicture}
    \draw (-0.7,0,0)--(2.7,0,0);
    \draw (0,0,-1.7)--(0,0,1.7);
    \draw (1.7,0,-1.7)--(1.7,0,1.7);
          \draw[canvas is xz plane at y=0,double=red, double distance=0.6mm,line cap=round]  (-0.1,0.7) -- (0,0.7) arc (90:-90:0.7)-- (-0.1,-0.7);
\draw[canvas is xz plane at y=0,double=red, double distance=0.6mm,line cap=round]  (-0.1+1.5,0.7) -- (0+1.5,0.7) arc (90:-90:0.7)-- (-0.1+1.5,-0.7);

    \pic at (0,0,0) {3dpeps};
    \pic at (1.7,0,0) {3dpeps};
       \node at (0.3,-0.4) {$T$};
        \node at (2,-0.4) {$T$};
  \end{tikzpicture}   = 
  \begin{tikzpicture}
    \draw (0,0,0)--(2.1,0,0);
    \draw (0,0,-1.9)--(0,0,1.9);
    \draw (1,0,-1.9)--(1,0,1.9);
    \pic at (0,0,0) {3dpeps};
    \pic at (1,0,0) {3dpeps};
        \draw[canvas is xz plane at y=0,double=red, double distance=0.6mm,line cap=round]  (-0.1,0.8) -- (1,0.8) arc (90:-90:0.8)-- (-0.1,-0.8);

              \node at (1.7,-0.4) {$(\Delta\otimes \id\otimes \Delta)T  $};
  \end{tikzpicture}.
\end{equation}
Then there is an operator $C\in \mathcal{A}$ such that $T = (\id \otimes C\otimes \id) (\id \otimes \Delta)\Delta(C^{-1})$. Therefore the relation between the tensors $A$ and $B$ is the following:
  \begin{equation*}
    \begin{tikzpicture}
      \pic at (0,0,0) {3dpeps};
      \node at (0.07,-0.27,0) {$A$};
      \draw (0,0,0) -- (0,0.35,0);
    \end{tikzpicture} =
    \begin{tikzpicture}
      \draw (-0.7,0,0) -- (0.7,0,0);
      \draw (0,0,-0.9) -- (0,0,0.9);
      \pic at (0,0,0) {3dpepsres};
      \node at (0.1,-0.27,0) {$B$};
      \filldraw [draw=black, fill=red] (-0.5,0,0) circle (0.05);
      \filldraw [draw=black, fill=red] (0.5,0,0) circle (0.05);
      \node[anchor=south] at (-0.5,0,0) {$\myinv{X}$};
      \node[anchor=north] at (0.5,0,0) {$X$};
      \filldraw [draw=black, fill=red] (0,0,-0.6) circle (0.05);
      \filldraw [draw=black, fill=red] (0,0,0.6) circle (0.05);
       \node[anchor=south] at (0,0,-0.6) {$Y$};
      \node[anchor=north] at (0,0,0.6) {$\myinv{Y}$};
    \end{tikzpicture}.
  \end{equation*}

\end{proposition}

\begin{proof}
We use the $G$-invariance of the rightmost tensor of the RHS of \cref{Tonetotwo} to obtain the operator $(\id\otimes \Delta^2 \otimes \id)T\equiv (\id\otimes \Delta^2 \otimes \id)T$ acting on the boundary. We apply the inverse of the two tensors and trace out the leftmost leg  in \cref{Tonetotwo}, so we arrive at
\begin{equation*} 
  \begin{tikzpicture}
    \draw (2.0,0,0)--(2.9,0,0);
    \draw (0,0,-0.8)--(0,0,-0.3);
    \draw (0,0,0.3)--(0,0,0.8);
    \draw (1.7,0,-0.3)--(1.7,0,-1.1);
    \draw (1.7,0,0.3)--(1.7,0,1.1);
  \draw[canvas is xz plane at y=0,double=red, double distance=0.5mm,line cap=round]  (-0.1,0.5) -- (1.7,0.5) arc (90:-90:0.5)-- (-0.1,-0.5);
\draw[canvas is xz plane at y=0,double=red, double distance=0.6mm,line cap=round]  (-0.1+1.7,0.8) -- (0+1.7,0.8) arc (90:-90:0.8)-- (-0.1+1.7,-0.8);
      \node at (0.3,-0.4) {$(\id\otimes \Delta^2 \otimes \id)T$};
     \node at (2,-0.4) {$T$};
  \end{tikzpicture} 
 = 
  \begin{tikzpicture}
   \draw (1.5,0,0)--(2,0,0);
   \node[anchor=west] at (2,0,0) {$3 $};
    \draw (0,0,-1.1)--(0,0,-0.3);
       \node[anchor=south] at (0,0,-1.1) {$1$};
    \draw (0,0,0.3)--(0,0,1.1);
           \node[anchor=north] at (0,0,1.1) {$5$};
    \draw (1,0,-0.3)--(1,0,-1.1);
              \node[anchor=south] at (1,0,-1.1) {$2$};
    \draw (1,0,0.3)--(1,0,1.1);
                  \node[anchor=north] at (1,0,1.1) {$4$};
      \draw[canvas is xz plane at y=0,double=red, double distance=0.6mm,line cap=round]  (-0.1,0.8) -- (1,0.8) arc (90:-90:0.8)-- (-0.1,-0.8);

              \node at (1.7,-0.4) {$(\Delta\otimes \id\otimes \Delta)T  $};
  \end{tikzpicture},
\end{equation*}
where we have labelled the legs. The previous picture is equivalent to the equation
$$(\id \otimes T\otimes \id) \cdot(\id\otimes \Delta^2 \otimes \id)T=(\Delta\otimes \id\otimes \Delta)T$$
 We now trace the legs composing with $W=\ket{w}\bra{w}$ to arrive to some equations involving the components of $T$. We will denote by $T_{i}$ the result of tracing the two legs different from $i$, analogously $T_{ij}$ is the result of tracing the leg that is not $i$ neither $j$. The relations are the followings:

\begin{align}
T_{23} = (T_3^{-1} \otimes \id) \cdot \Delta (T_3) \; \; &\; \;  {\rm Tracing}\; 1,2,3 \label{Ttr1}\\
T_{12}= ( \id \otimes T_1^{-1}) \cdot \Delta (T_1)\; \;  & \; \; {\rm Tracing}\; 3,4,5\label{Ttr2}\\
T_{13}= ( \id \otimes T_3) T_{12} \; \;  & \; \; {\rm Tracing}\; 2,3,5 \label{Ttr3}\\
T_{13}= ( T_1 \otimes \id) T_{23} \; \;  & \; \; {\rm Tracing}\; 1,3,4\label{Ttr4}\\	
T= (\id \otimes T_{23}) (\id \otimes \Delta)(T_{12}) \; \;  & \; \; {\rm Tracing}\; 2,5 \label{Ttr5}
\end{align}
Using \cref{Ttr1} on \cref{Ttr4} and \cref{Ttr2} on \cref{Ttr3} we can conclude that

$$T_{13}=(T_1 \otimes \id) (T^{-1}_3\otimes \id ) \Delta(T_3) =(\id\otimes T_3)(\id\otimes T^{-1}_1) \Delta(T_1) ,$$
which implies that $\Delta(T_1T_3^{-1})=T_1T_3^{-1}\otimes T_1T_3^{-1}$,  so $T_1 T_3^{-1} = g$ for an element $g\in G$.  We can now use \cref{Ttr5} to obtain that
\begin{equation*}
    T = ( g\otimes T_3^{-1} \otimes 1) \cdot (\id\otimes \Delta)\Delta(T_3).
  \end{equation*}  
Using the $G$-injectivity we can conclude that, for $C\equiv T_3$

  \begin{equation*} 
    \begin{tikzpicture}
      \draw (-1.7,0,0)--(1.7,0,0);
      \draw (0,0,-1.7)--(0,0,1.7);
          \draw[canvas is xz plane at y=0,double=red, double distance=0.6mm,line cap=round]  (-0.1,0.7) -- (0,0.7) arc (90:-90:0.7)-- (-0.1,-0.7);
      \pic at (0,0,0) {3dpeps};
    \end{tikzpicture} =
    \begin{tikzpicture}
      \draw (-1.7,0,0)--(1.7,0,0);
      \draw (0,0,-1.7)--(0,0,1.7);
      \foreach \x in {(1.2,0),(-1.2,0)}{
          \filldraw[draw=black, fill=red] \x circle (0.07);      }
           \filldraw[draw=black, fill=black] (0,0,-1.2) circle (0.06);  
      \node[anchor=north] at (-1.2,0,0) {$C$};
            \node[anchor=north] at (0,0,-1.2) {$g$};
      \node[anchor=north] at (1.2,0,0) {$C^{-1}$};
      \node[anchor=north west] at (0,0) {A};
      \pic at (0,0,0) {3dpeps};
    \end{tikzpicture}
    \; .
  \end{equation*}  
  
The tensor $\tilde{B}$ is equal to $A(g\otimes C\otimes \id \otimes C^{-1})$. Moreover, $\tilde{B}$ generates the same state as $A$ for all system sizes. This is true in particular for the MPS constructed by closing the horizontal direction (where the matrix $C$ is placed). Thus, the following is true for all system sizes $n\in \mathbb{N}$:

\begin{equation*}
 \begin{tikzpicture}
      \draw[ canvas is yz plane at x=0] (-0.6,-2.5) rectangle  (0,2);
      \foreach \x/\z in {0/-1.5,0/0,0/1.5}{
       \draw (\x-0.4,-0.15,\z) rectangle  (\x+0.4,0,\z);
        \pic at (\x,0,\z) {3dpeps};
      }
  \end{tikzpicture}
  =
   \begin{tikzpicture}
      \draw[ canvas is yz plane at x=0] (-0.6,-2.5) rectangle  (0,2);
      \foreach \x/\z in {0/-1.5,0/0,0/1.5}{
       \draw (\x-0.4,-0.15,\z) rectangle  (\x+0.4,0,\z);
        \pic at (\x,0,\z) {3dpeps};
         \node[tensorr] at (\x,0,\z-0.6) {};
              \node[anchor=south east] at (\x,0,\z-0.6) {$g$};
      }
  \end{tikzpicture}
  = \begin{tikzpicture}
      \draw[ canvas is yz plane at x=0] (-0.6,-2.5) rectangle  (0,2);
      \foreach \x/\z in {0/-1.5,0/0,0/1.5}{
       \draw (\x-0.4,-0.15,\z) rectangle  (\x+0.4,0,\z);
        \pic at (\x,0,\z) {3dpeps};
      }
      \node[tensorr,label=above:$g^n$] at (0,0,-1.5-0.6) {};
  \end{tikzpicture}
  \end{equation*}
Applying the inverse tensor to all the sites we find that $\sum_h hg^nh^{-1}=\id$ for all $n$ so we can conclude that $g=e$.
\end{proof}

\subsection{Isomorphism for a large number of edges}\label{sec:isols}

In this section we prove that the isomorphism relating the centers of the algebras of the two tensor networks is a gauge transformation when a large number of edges is considered. This is used in the proof of \cref{lem:Toperator} to conclude that $V$ is mapped to a rank-one projector, $W$, which is needed to slice the PEPS into MPSs. An important point here is that the isomorphism has to be compatible with blocking so that the isomorphism of two edges has to be the same as the tensor product of the two isomorphism of each edge. To obtain the desired result we analyze how isomorphisms act in a direct sum of irreps.\\

A finite dimensional representation of a finite group $G$ decomposes into a direct sum of irreducible representations. The corresponding algebra decomposes into a direct sum of full matrix algebras: $\mathcal{A}_G = \bigoplus_i M_{d(i)}\otimes \id_{m(i)}$, where $i$ labels the irrep, $d(i)$ is the dimension of the $i^{th}$ irrep and $m(i)$ is its multiplicity. Therefore the centralizer of $\mathcal{A}_G$, denoted here as $\mathcal{C}_G$, also decomposes into a direct sum of full matrix algebras: $\mathcal{C}_G = \bigoplus_i \id_{d(i)}\otimes M_{m(i)}$.  
\

Minimal projectors in $\mathcal{C}_G$ have the form $P_\alpha\otimes \id_{m(i)}$ where $P_\alpha$ is a rank-one projector with $\alpha\in 1,\cdots, d(i)$, they are exactly projectors onto an individual irrep of the representation.
In the same way minimal projectors in $\mathcal{A}_G$ are $ \id_{d(i)}\otimes P_\beta$ where $P_\beta$ is a rank-one projector with $\beta \in 1,\cdots, m(i)$. Then, minimal projectors on $\mathcal{A}_G\cap\mathcal{C}_G$ are $\id_{d(i)}\otimes \id_{m(i)}$, these are projectors onto an irrep sector with its multiplicity. 
\

Let $A$ and $B$ be two faithful representations of $G$. Suppose $\Phi:\mathcal{C}_G^A\to\mathcal{C}_G^B$ is an isomorphism. It is clear that $\Phi$ maps projectors onto projectors and maps elements of  $Z(\mathcal{C}_G^A)$ (center of $\mathcal{C}_G^A$) to elements of $Z(\mathcal{C}_G^B)$. Since $\mathcal{A}_G\cap\mathcal{C}_G=Z(\mathcal{C}_G)\cap\mathcal{C}_G=Z(\mathcal{C}_G)$, minimal projectors of $Z(\mathcal{C}_G^A)$ goes to minimal projectors of $Z(\mathcal{C}_G^B)$.
\
Therefore $\Phi$ maps projectors onto irrep sectors (with its multiplicity) of $A$ to projectors onto irrep sectors (with its multiplicity) of $B$: $\Phi$ implements a permutation of the blocks corresponding to the irreps. 

\begin{lemma}\label{lemperm}
Let $\Phi: \mathcal{C}^A\to \mathcal{C}^B$ be an isomorphism. If $\Phi$ does not change the dimension of the irreps of $G$, there is an invertible matrix $Z$ such that $\Phi(X)= ZXZ^{-1}$.
\end{lemma}

\begin{proof}
Let us denote as $\sigma$ the permutation of the irrep labels. Since the isomorphism acts on a full matrix algebra with multiplicity and it is finite dimensional, it can perform a gauge transformation together with a change in the multiplicity. That is $\Phi( \id_{d(i)} \otimes M_{m(i)} ) \cong \id_{d(\sigma[i])} \otimes M_{m(\sigma[i])} \cong  \id_{d(\sigma[i])} \otimes M_{m(i)} $ but by hypothesis $d(\sigma[i])=d(i)$. As $\Phi$ is an automorphism of multiple copies of the same full matrix algebra, $\Phi(\id_{d(i)} \otimes M_{m(i)}) = Z_i (\id_{d(i)} \otimes M_{m(i)}) Z_i^{-1}$ for some $Z_i$ and then $\Phi(X)= ZXZ^{-1}$ for all $X$.
\end{proof}

The tensor product of irreducible representations decomposes into a direct sum of irreducible representations. The decomposition will be characterized by the fusion rules: $i\otimes j = \bigoplus_k N_{ij}^k k$ which implies the following equation for the dimensions: $d_i \cdot d_j = \sum_k N_{ij}^k d_k$. For the trivial representation $0$, the fusion rules are trivial: $N_{0j}^k = N_{j0}^k = \delta_{j,k}$.

\begin{lemma}
  Let $\mathcal{A}_\alpha^X$, $\alpha=1,2$, $X=A,B$ be four faithful finite dimensional representations of a finite group $G$ and let us denote $\mathcal{A}^X_3=\mathcal{A}^X_1\otimes\mathcal{A}^X_2$. Let $\mathcal{C}_\alpha^X$ be the centralizer of $\mathcal{A}^X_\alpha$. Suppose $\Phi_\alpha:\mathcal{C}^A_\alpha\to\mathcal{C}^B_\alpha$, $\alpha=1,2,3$ are isomorphisms such that $\Phi_3(X\otimes Y) = \Phi_1(X)\otimes \Phi_2(Y)$ for all $X\in \mathcal{C}_1$ and $Y\in\mathcal{C}_2$. Let $\sigma_\alpha$ be the corresponding permutations of the irreps. Then $d_{\sigma_1(i)} = d_{\sigma_2(i)} = d_{\sigma_3(i)}$ for all irrep $i$.
\end{lemma}
\begin{proof}
We denote the permutation of irrep labels of  $\Phi_\alpha$ as $\sigma_\alpha$. 
The isomorphisms satisfy $\Phi_3(X\otimes Y) = \Phi_1(X)\otimes \Phi_2(Y)$ for all $X\in \mathcal{C}_1$ and $Y\in\mathcal{C}_2$ so the irrep permutations associated are related by $\sigma_1(i)\otimes \sigma_2(j)= \sigma_3(i\otimes j)$. This implies that $\sigma_1(i)\otimes \sigma_2(j) = \bigoplus_k N_{ij}^k \sigma_3(k)$ so the permutations, $\sigma_\alpha$, respect the fusion rules. 
Notice that the isomorphism, recall the proof of previous lemma, can perform a gauge transformation together with a change in the  dimensions of the irreps. The multiplicity is not modified, since its defined the full matrix algebra of the block, then $N_{ij}^k$ is not affected by $\sigma_3$.

This implies that $d_{\sigma_1(i)}\cdot d_{\sigma_2(j)} = \sum_k N_{ij}^k d_{\sigma_3(k)}$. For $i=0$, the trivial irrep,  $d_{\sigma_1(0)}\cdot d_{\sigma_2(j)} = \sum_k N_{0j}^k d_{\sigma_3(k)} = d_{\sigma_3(j)}$. Therefore, by choosing an irrep $j$ such that $\sigma_2(j)$ has the largest dimension, we conclude that for this $j$ $d_{\sigma_2(j)}=d_{\sigma_3(j)}$ and then $d_{\sigma_1(0)}=1$. This means that for all $i$, $d_{\sigma_2(i)}=d_{\sigma_3(i)}$ and by similar arguments, we can also conclude that $d_{\sigma_1(i)}=d_{\sigma_3(i)}$.
 \end{proof}

\begin{lemma}
  Let $\mathcal{A}_\alpha^X$ ($\alpha=1,\dots, \kappa$, $X=A,B$) be representations of $G$. We denote by $I$ a colection of $n$ of those representations $I=[\alpha_1,\alpha_2,\dots, \alpha_n]$, let $\mathcal{A}_I^X$ be the tensor product representation and $\mathcal{C}_I^X$ be its centralizer. Consider also an isomorphism $\Phi_I:\mathcal{C}_I^A\to \mathcal{C}_I^B$ such that if $I=J\cup K$ for disjoint continuous regions $J$ and $K$, it satisfies that for any $X\in\mathcal{C}_J^A$ and $X\in\mathcal{C}_K^A$, $\Phi_I(X\otimes Y) = \Phi_J(X)\otimes\Phi_K(Y)$. If $n$ is big enough, then for all $\alpha$ there is an invertible matrix $Z_\alpha$ such that $\Phi_\alpha(X) = Z_\alpha^{-1} X Z_\alpha$.
\end{lemma}

\begin{proof}
The isomorphism $\Phi_I$ implements permutation $\sigma_I$ of the irreps. By repeated application of the previous lemma, for all $I$ and $J$, $d_{\sigma_I(k)} = d_{\sigma_J(k)}$ for any irrep $k$. Let us consider now the tensor product of $n$ copies of an irrep $i$. In the representations $\mathcal{A}_\alpha^A$, this is $i\otimes \dots \otimes i = \bigoplus_k N_{i \dots i}^k k$. After the application of the isomorphisms, this maps to $\sigma_1(i)\otimes \dots \otimes\sigma_n(i) = \bigoplus_k N_{i \dots i}^k \sigma_{[1\dots n]}(k)$. For the dimensions, this gives the equation $d_{\sigma_1(i)} \cdots d_{\sigma_n(i)} = \sum_k N_{i \dots i}^k d_{\sigma_{[1\dots n]}(k)}$. As the dimensions coincide, this means $d_{\sigma_1(i)}^n = \sum_k N_{i \dots i}^k d_{\sigma_{1}(k)}$ but also $d_{i}^n = \sum_k N_{i \dots i}^k d_{k}$. Then, we can obtain the following bounds:
  \begin{equation*}
   \frac{d_{\sigma_1(i)}^n}{d_{i}^n}=\frac{\sum_k N_{i \dots i}^k d_{\sigma_{1}(k)}}{\sum_q N_{i \dots i}^q d_{q}} \leq \frac{\sum_k N_{i \dots i}^k d_{max}}{\sum_q N_{i \dots i}^q } \leq d_{max},
  \end{equation*}
    \begin{equation*}
   \frac{d_{\sigma_1(i)}^n}{d_{i}^n}=\frac{\sum_k N_{i \dots i}^k }{\sum_q N_{i \dots i}^q d_{max}} \geq \frac{1}{d_{max}},
  \end{equation*}
  where $d_{max}$ is the biggest irrep dimension. Therefore, the following equation has to be satisfied for all $n$
  $$ \frac{1}{d_{max}}\leq \left(\frac{d_{\sigma_1(i)}}{d_{i}}\right)^n\leq d_{max}.$$
This implies that for $n$ big enough, $d_i = d_{\sigma_\alpha(i)}$ for all irrep $i$ and representation $\alpha$. That is, there always exists an $n$ such that the dimensions of the irreps do not change. Then, using \cref{lemperm} we conclude that there is an invertible matrix realizing the isomorphism as a conjugation.
\end{proof}

\section{Discussion} 

In this chapter we have proven fundamental theorems for (injective and) normal PEPS respectively: two such TNs generate the same state if and only if the defining tensors are related through a local gauge transformation. Moreover, the gauges relating the two sets of tensors are uniquely defined up to a multiplicative constant. This result holds for a fixed (but large enough) system size. It is valid for any geometry, TI and non-TI setting, including 1D (MPS), 2D PEPS, higher- dimensional PEPS, and other lattice geometries such as the honeycomb lattice, the Kagom\'e lattice, tree tensors networks, and the hyperbolic lattice used in the AdS/CFT correspondence \cite{Pastawski15, Hayden16}. The proof method, however, is not applicable for MERA, where we did not find a simple way to apply lemma 1 due to the particular geometry of the network.\\

Second we have proven a fundamental theorem for $G$-injective PEPS. The proof is valid for the square lattice and for TN that are equal for all system sizes. We obtain a local gauge relation between the tensors. This opens the possibility to obtain a fundamental theorem for more general families of PEPS such as MPO-injective PEPS. It is left for future work to relax the hypotheses of the theorem, in particular the square lattice dependence since topological models can be defined in any lattice embedded in an orientable surface. 

\chapter{Classification of symmetric $G$-injective PEPS}\label{chap:classsymGPEPS}

In this chapter we characterize global on-site symmetries of $G$-injective PEPS. Topological phases with a global symmetry acting non-trivially on their anyons and ground subspace are referred to as Symmetry Enriched Topological (SET) phases -see \cite{Barkeshli14,Chang15}. 
The global symmetry could permute between the anyons and between the ground states. Moreover, the symmetry could act projectively on the individual anyons. This effect is called Symmetry Fractionalization (SF). 
We show what the possible patterns of permutation and symmetry fractionalization are on  $G$-injective PEPS under a global on-site symmetry. We also prove that the different patterns correspond to different quantum phases.\\

A phase in quantum many-body systems is usually defined as the set of gapped locally-interacting Hamiltonians that can be deformed into each other without closing the spectral gap (see \cref{def:localhamil} and \cref{gapH}). 
Moreover, when a symmetry is imposed in the systems, the deformation in the Hamiltonians is required to preserve the symmetry. It is usually the case that the phase classification is translated into the classification of symmetric quantities unchanged by the Hamiltonian deformation.
Due to the difficulty of proving whether a Hamiltonian is gapped or not in 2D systems, see \cite{Cubitt15}, we will not consider the previous definition. 
Here we will define phases within the set of G-injective PEPS by focusing on the G-isometric ones (which, as renormalization fixed points, are the natural representatives). Two G-isometric PEPS, together with an on-site symmetry action in each of them, will be said to be in the same phase if they can be connected to each other in finite systems with a continuous path of Hamiltonians that keeps the symmetry (no gap assumption). We will show that this is possible if and only if both share the same invariants that connect the symmetry action and the topological order (the maps $\phi$ and $\omega$ of \cref{theo:class} which will be defined in the next section together with their equivalence relation). 
This is the content of the following two theorems:

\begin{tcolorbox}
\begin{theorem}[Classification of symmetric $G$-isometric PEPS]\label{theo:class}
Given two finite groups $G$ and $Q$ and a $G$-injective PEPS with $Q$ as a global on-site symmetry, one can define (an equivalence class of) a homomorphism $\phi:Q\rightarrow {\rm Aut} (G)$ and a 2-cocycle $\omega: Q\times Q\rightarrow G$ so that the pair $(\phi,\omega)$ is constant in a neighbourhood of any $G$-isometric PEPS, when perturbed with a {\it natural} perturbation (those that correspond to a continuous deformation of the parent Hamiltonian). 
\end{theorem}
\end{tcolorbox}

\begin{tcolorbox}
\begin{theorem}[Continuos path of $G$-injective PEPS]\label{theo:path}
Given two $G$-injective PEPS invariant under a global on-site symmetry of $Q$, there is a continuous path connecting both if the class of the maps $(\phi, \omega)$ are the same for both $G$-injective PEPS.
\end{theorem}
\end{tcolorbox}


%

\

We recall that $G$-injective PEPS describe the topological order associated to quantum doubles models of a finite group $G$.
But the topological order in $G$-injective PEPS is not guaranteed solely by the $G$-invariance of the tensors. Under local and continuous transformations, they can suffer phase transitions in the thermodynamic limit driven by boson condensations \cite{Duivenvoorden17}. 
To avoid these phase transitions we have restricted the classification to $G$-isometric PEPS, the renormalization fixed points, whose parent Hamiltonians are commuting (gapped in the thermodynamic limit), and which have zero correlation length. These points are in the same topological order as $\mathcal {D}(G)$ in the thermodynamic limit. In the case where the transformation preserves the topological order, that is, when the transformation does not close the gap, we end up with a separation of topological phases invariant under symmetries in the PEPS framework. Ref.\cite{Schuch11} shows some bounds for the gap of the parent Hamiltonian when these transformations are considered.

We do not construct interpolating paths in our classification since only one representative of each phase is considered. However, we will consider in \cref{sec:finitepath} an interpolation between two symmetric $G$-injective PEPS at finite sizes which gives us the desired condition of equality of the classes without gaps considerations. For the sake of completeness in \cref{sec:gauging} we will also study how the global symmetry can be gauged in these models.



%

We now explain the important connection between our classification of \cref{theo:class} and the theory of group extensions. The maps $(\phi, \omega)$ also appear when characterizing the possible group extensions of $G$ by $Q$. These extensions $E$ are defined by the short exact sequence
\begin{equation*}
1\rightarrow G \rightarrow E \rightarrow Q \rightarrow 1,
\end{equation*}
which relates the involved groups: $G\triangleleft E$ and $Q=E/G$. Since this connection will be crucial for our work, we review the notion of group extensions in Appendix \ref{ap:ext}.

\section{Proof of \cref{theo:class}}
 
In this section we prove \cref{theo:class}. First, in \cref{sec:defmaps},  we show how the maps $\phi$ and $\omega$ are defined and what are their equivalence classes. Second, in \cref{sec:deform}, we prove the robustness of the class in the maps $\phi$ and $\omega$  within a phase of  PEPS. This leads us to \cref{theo:class}. 

Let us prove a lemma which will be used to define the maps $(\phi,\omega)$. This lemma just states that the $G$-invariance is the only virtual symmetry of the considered tensors:

 \begin{lemma} \label{Glemma}
 Given a $G$-injective tensor $A$ and two invertible matrices $X$ and $Y$ such that
  $$  \begin{tikzpicture}
      \pic at (0,0,0) {3dpeps};
    \end{tikzpicture} 
    =
    \begin{tikzpicture}
      \draw (-0.7,0,0) -- (0.7,0,0);
      \draw (0,0,-0.9) -- (0,0,0.9);
      \pic at (0,0,0) {3dpeps};
      \filldraw[draw=black, fill=red] (-0.5,0,0) circle (0.05);
      \filldraw[draw=black, fill=red] (0.5,0,0) circle (0.05);
      \node[anchor=south] at (-0.5,0,0) {$\myinv{X}$};
      \node[anchor=north] at (0.5,0,0) {$X$};
      \filldraw[draw=black, fill=red] (0,0,-0.6) circle (0.05);
      \filldraw[draw=black, fill=red] (0,0,0.6) circle (0.05);
       \node[anchor=south] at (0,0,-0.6) {$Y$};
      \node[anchor=north] at (0,0,0.6) {$\myinv{Y}$}; 
    \end{tikzpicture},
$$
the matrices satisfy $X=Y=u_g$ (up to a constant) for some $g \in G$ where $u_g$ is the representation of the $G$-invariance. 
 \end{lemma}

 \begin{proof} 
 Apply the inverse of the tensor on both sides, apply $\mathfrak{D} Y$ on the south leg, where $\mathfrak{D}$ is the matrix which makes the group elements orthogonal for a semi-regular representation ($\tr{[u_g\mathfrak{D}]}=\delta_{e,g}|G|$ defined in \cref{sec:introGinj}). We close the virtual indices as
   $$  \begin{tikzpicture}
      \node[anchor=south] at (0,0.7,0) {$A^{-1}$}; 
   \pic at (0,0.7,0) {3dpepsdown};
   \draw (0,0,0)--(0,0.7,0)--(0,0.7,0.9)--(0,0,0.9)--(0,0,0);
   \node[anchor=north] at (0,0,0.6) {${\mathfrak{D}Y}$}; 
    \pic at (0,0,0) {3dpeps};
          \filldraw[draw=black, fill=red] (0,0,0.6) circle (0.05);
    \end{tikzpicture} 
    =
    \begin{tikzpicture}
          \node[anchor=south] at (0,0.7,0) {$A^{-1}$}; 
       \pic at (0,0.7,0) {3dpepsdown};
   \draw (0,0,0)--(0,0.7,0)--(0,0.7,1)--(0,0,1)--(0,0,0);
      \draw (-0.7,0,0) -- (0.7,0,0);
      \draw (0,0,-0.9) -- (0,0,0.9);
      \pic at (0,0,0) {3dpeps};
      \filldraw[draw=black, fill=red] (-0.5,0,0) circle (0.05);
      \filldraw[draw=black, fill=red] (0.5,0,0) circle (0.05);
      \node[anchor=south east] at (-0.5,0,0) {$\myinv{X}$};
      \node[anchor=north] at (0.5,0,0) {$X$};
      \filldraw[draw=black, fill=red] (0,0,-0.6) circle (0.05);
      \filldraw[draw=black, fill=red] (0,0,0.6) circle (0.05);
       \node[anchor=south] at (0,0,-0.6) {$Y$};
      \node[anchor=north] at (0,0,0.6) {$\mathfrak{D}$}; 
   
    \end{tikzpicture}.
$$
By $G$-injectivity, this results in the identity
\begin{equation}\label{eq:gXY}
\sum_g   \tr{ [{g}^{-1} Y \mathfrak{D}]}  g\otimes  g\otimes   {g}^{-1} = Y\otimes X \otimes {X }^{-1}.
\end{equation}
Because the RHS is not zero and the elements $g$ of $G$ are linearly independent there must exist an element $s\in G$ such that $\tr{ [ Y\mathfrak{D} s ]}\neq0$. Therefore contracting the first factor of the tensor product of \cref{eq:gXY} with $\mathfrak{D} s$ we obtain:
$$   \sum_g   \tr{ [{g}^{-1} Y \mathfrak{D}]}  \tr{ [\mathfrak{D} s g]}\otimes  g\otimes  {g}^{-1}= \tr{ [s Y \mathfrak{D}]}  {s}^{-1} \otimes  s =  \tr{ [s Y\mathfrak{D} ]}  X \otimes {X}^{-1}, $$
which implies that  $X\in G\times \mathbb{C}$ and similarly for $Y$.  Following this argument we find an element $r\in G$ such that $\tr{ [ X\mathfrak{D} r ]}\neq0$ so that from \cref{eq:gXY} we arrive at 
$$    r\otimes r=  \frac{\tr{ [{X}^{-1} r \mathfrak{D}]} }  {\tr{ [Y {r}^{-1} \mathfrak{D}]} }  X \otimes Y, $$
which implies that $X=Y\in G$ up to an arbitrary complex number that can be renormalized to a phase factor. We will drop w.l.o.g. the complex phase dependence.  \end{proof}

Given a $G$-injective PEPS  $|\Psi_A\rangle$ placed on a square lattice $\Lambda$ with periodic boundary conditions.  We consider the case where $|\Psi_A\rangle$ has a global on-site symmetry given by some finite group $Q$:
$$\bigotimes_{v\in \Lambda}U^{[v]}_q |\Psi_A\rangle=|\Psi_A\rangle; \;\forall q\in Q,$$
where $U_q$ is a unitary (linear) representation of $Q$ and $v$ are the vertices of the lattice $\Lambda$. We now can apply \cref{theo:FTGinjective} and conclude that there exist invertible matrices $v_q$ and $w_q$ acting on the virtual d.o.f. such that

  \begin{equation} \label{locsym}
    \begin{tikzpicture}[baseline=-1mm]
      \pic at (0,0,0) {3dpeps};
      \draw (0,0,0) -- (0,0.35,0);
      \filldraw[draw=black, fill=purple] (0,0.2,0) circle (0.06);
      \node[anchor=east] at (0,0.3,0) {$U_q$};
    \end{tikzpicture} =
    \begin{tikzpicture}
      \draw (-0.7,0,0) -- (0.7,0,0);
      \draw (0,0,-0.9) -- (0,0,0.9);
      \pic at (0,0,0) {3dpeps};
      \filldraw [draw=black, fill=red] (-0.5,0,0) circle (0.05);
      \filldraw [draw=black, fill=red] (0.5,0,0) circle (0.05);
      \node[anchor=south] at (-0.5,0,0) {$\myinv{v_q}$};
      \node[anchor=north] at (0.5,0,0) {$v_q$};
      \filldraw [draw=black, fill=red] (0,0,-0.6) circle (0.05);
      \filldraw [draw=black, fill=red] (0,0,0.6) circle (0.05);
       \node[anchor=south] at (0,0,-0.6) {$w_q$};
      \node[anchor=north] at (0,0,0.6) {$\myinv{w_q}$};
    \end{tikzpicture}
    \; \forall q\in Q.
  \end{equation}
Note that \cref{theo:FTGinjective} is stated for a relation between every system size: this is a physically meaningful situation for global symmetries.

 \begin{remark}[Gauge freedom] \label{equivdef} From \cref{locsym}, the operators $w_q,v_q$ are defined up to an arbitrary element of $G$, for each $q\in Q$. Due to the $G$-injectivity of the tensor:
 $$ \begin{tikzpicture}
      \draw (-0.9,0,0) -- (0.9,0,0);
      \draw (0,0,-1.1) -- (0,0,1.1);
      \pic at (0,0,0) {3dpeps};
      \filldraw [draw=black, fill=red] (-0.7,0,0) circle (0.05);
      \filldraw [draw=black, fill=red] (0.7,0,0) circle (0.05);
      \node[anchor=south] at (-0.7,0,0) {$\myinv{v_q}$};
      \node[anchor=north] at (0.7,0,0) {$v_q$};
      \filldraw [draw=black, fill=red] (0,0,-0.8) circle (0.05);
      \filldraw [draw=black, fill=red] (0,0,0.8) circle (0.05);
       \node[anchor=south] at (0,0,-0.8) {$w_q$};
      \node[anchor=north] at (0,0,0.8) {$\myinv{w_q}$};
    \end{tikzpicture} =
    \begin{tikzpicture}
      \draw (-0.9,0,0) -- (0.9,0,0);
      \draw (0,0,-1.1) -- (0,0,1.1);
      \pic at (0,0,0) {3dpeps};
      \filldraw [draw=black, fill=red] (-0.7,0,0) circle (0.05);
      \filldraw [draw=black, fill=red] (0.7,0,0) circle (0.05);
      \node[anchor=south] at (-0.7,0,0) {$\myinv{v_q}$};
      \node[anchor=north] at (0.7,0,0) {$v_q$};
      \filldraw [draw=black, fill=red] (0,0,-0.8) circle (0.05);
      \filldraw [draw=black, fill=red] (0,0,0.8) circle (0.05);
       \node[anchor=west] at (0,0,-0.8) {$w_q$};
      \node[anchor=east] at (0,0,0.8) {$\myinv{w_q}$};
            \filldraw (-0.4,0,0) circle (0.05);
      \filldraw  (0.4,0,0) circle (0.05);
      \node[anchor=south] at (-0.4,0,0) {$\myinv{g}$};
      \node[anchor=north] at (0.4,0,0) {$g$};
      \filldraw (0,0,-0.5) circle (0.05);
      \filldraw (0,0,0.5) circle (0.05);
       \node[anchor= south east] at (0,0,-0.5) {$g$};
      \node[anchor=west] at (0,0,0.5) {$\myinv{g}$};
    \end{tikzpicture} . $$
 For all $g\in G$,  the pairs $(v_q,w_q)$ and $(gv_q ,gw_q)$ have to be considered equivalent since the action on the tensor is the same. This is the gauge freedom of the virtual symmetry operators that has to be considered when we define maps in terms of $v_q$ and $w_q$.
\end{remark}
 
\subsection{Definitions of the maps $\phi$ and $\omega$}\label{sec:defmaps}

Using \cref{locsym} and the $G$-injectivity of the tensor we get for all $ g\in G$ that
$$ \begin{tikzpicture}
      \draw (-0.7,0,0) -- (0.7,0,0);
      \draw (0,0,-0.9) -- (0,0,0.9);
      \pic at (0,0,0) {3dpeps};
      \filldraw [draw=black, fill=red] (-0.5,0,0) circle (0.05);
      \filldraw [draw=black, fill=red] (0.5,0,0) circle (0.05);
      \node[anchor=south] at (-0.5,0,0) {$\myinv{v_q}$};
      \node[anchor=north] at (0.5,0,0) {$v_q$};
      \filldraw [draw=black, fill=red] (0,0,-0.6) circle (0.05);
      \filldraw [draw=black, fill=red] (0,0,0.6) circle (0.05);
       \node[anchor=south] at (0,0,-0.6) {$w_q$};
      \node[anchor=north] at (0,0,0.6) {$\myinv{w_q}$};
    \end{tikzpicture} =
    \begin{tikzpicture}
      \draw (-0.7,0,0) -- (0.7,0,0);
      \draw (0,0,-0.9) -- (0,0,0.9);
      \pic at (0,0,0) {3dpeps};
      \filldraw [draw=black, fill=red] (-0.5,0,0) circle (0.05);
      \filldraw [draw=black, fill=red] (0.5,0,0) circle (0.05);
      \node[anchor=south] at (-0.5,0,0) {$\myinv{(v_q g)}$};
      \node[anchor=north] at (0.5,0,0) {$v_q g$};
      \filldraw [draw=black, fill=red] (0,0,-0.6) circle (0.05);
      \filldraw [draw=black, fill=red] (0,0,0.6) circle (0.05);
       \node[anchor=south] at (0,0,-0.6) {$w_q g$};
      \node[anchor=north] at (0,0,0.6) {$\myinv{(w_q g)}$};
    \end{tikzpicture}, $$
for each $q\in Q$. This implies 

$$ \begin{tikzpicture}
      \draw (-0.7,0,0) -- (0.7,0,0);
      \draw (0,0,-0.9) -- (0,0,0.9);
      \pic at (0,0,0) {3dpeps};
    \end{tikzpicture} =
    \begin{tikzpicture}
      \draw (-0.7,0,0) -- (0.7,0,0);
      \draw (0,0,-0.9) -- (0,0,0.9);
      \pic at (0,0,0) {3dpeps};
      \filldraw [draw=black, fill=red] (-0.5,0,0) circle (0.05);
      \filldraw [draw=black, fill=red] (0.5,0,0) circle (0.05);
      \node[anchor=south] at (-0.5,0,0) {$v_q \myinv{g}\myinv{v_q}$};
      \node[anchor=north west] at (0.5,0,0) {$v_q g \myinv{v_q}$};
      \filldraw [draw=black, fill=red] (0,0,-0.6) circle (0.05);
      \filldraw [draw=black, fill=red] (0,0,0.6) circle (0.05);
       \node[anchor=south west] at (0,0,-0.6) {$w_q g \myinv{w_q}$};
      \node[anchor=north] at (0,0,0.6) {$w_q \myinv{g}\myinv{w_q}$};
    \end{tikzpicture}, $$
so we associate $ v_q g v^{-1}_q\equiv X$ and $ w_q g w^{-1}_q\equiv Y$  in \cref{Glemma} and then $v_q g v^{-1}_q= w_q g w^{-1}_q\in G$. 

\begin{definition}[Definition of $\phi$]\label{def:phi}
For each $q\in Q$ the permutation map $\phi$ is defined as follows
\begin{align}\label{autmap}
  \phi_q \colon G &\to G \notag\\
  g &\mapsto \phi_q(g)=v_q g v^{-1}_q.
\end{align}
The map $ \phi_q$ is invertible, and by \cref{Glemma} is equal to the map $\tilde{\phi}_q(g)=w_q g w^{-1}_q$. It also satisfies $\phi_q(g) \phi_q(h)=\phi_q(gh)$. So $\phi_q$ is a map from $Q$ to Aut$(G)$ for each $q\in Q$.\\
\end{definition}
A linear representation satisfies $U_kU_q=U_{kq}$. Thus, by \cref{locsym}, we have
$$ \begin{tikzpicture}
      \draw (-0.7,0,0) -- (0.7,0,0);
      \draw (0,0,-0.9) -- (0,0,0.9);
      \pic at (0,0,0) {3dpeps};
      \filldraw [draw=black, fill=red] (-0.5,0,0) circle (0.04);
      \filldraw [draw=black, fill=red] (0.5,0,0) circle (0.04);
      \node[anchor=south] at (-0.5,0,0) {$\myinv{v_{kq}}$};
      \node[anchor=north] at (0.5,0,0) {$v_{kq}$};
      \filldraw [draw=black, fill=red] (0,0,-0.6) circle (0.04);
      \filldraw [draw=black, fill=red] (0,0,0.6) circle (0.04);
       \node[anchor=south] at (0,0,-0.6) {$w_{kq}$};
      \node[anchor=north] at (0,0,0.6) {$\myinv{w_{kq}}$};
    \end{tikzpicture} =
    \begin{tikzpicture}
      \draw (-0.7,0,0) -- (0.7,0,0);
      \draw (0,0,-0.9) -- (0,0,0.9);
      \pic at (0,0,0) {3dpeps};
      \filldraw [draw=black, fill=red] (-0.5,0,0) circle (0.04);
      \filldraw [draw=black, fill=red] (0.5,0,0) circle (0.04);
      \node[anchor=south] at (-0.5,0,0) {$\myinv{(v_k v_q)}$};
      \node[anchor=north] at (0.5,0,0) {$v_k v_q$};
      \filldraw [draw=black, fill=red] (0,0,-0.6) circle (0.04);
      \filldraw [draw=black, fill=red] (0,0,0.6) circle (0.04);
       \node[anchor=south] at (0,0,-0.6) {$w_k w_q$};
      \node[anchor=north] at (0,0,0.6) {$\myinv{(w_k w_q)}$};
    \end{tikzpicture}
    \; \forall q,k\in Q.$$
Again, using \cref{Glemma} it follows that $ v_k v_q v^{-1}_{kq}= w_k w_q w^{-1}_{kq} \in G$. 

\begin{definition}[Definition of $\omega$]\label{def:omega}
The map $\omega$ is defined as follows:
\begin{align}
  \omega \colon Q\times Q  &\to G \notag\\
 (k,q) &\mapsto\omega(k,q)=v_k v_q v^{-1}_{kq}.\notag
\end{align}
\end{definition}
\begin{proposition}
The map $\omega$ is a $2$-cocycle, see \cite{Adem04}, i.e. it satisfies the following $2$-cocycle condition:
\begin{equation}\label{eq:2cocyclecond}
\omega(k,q) \omega(kq,p)= \phi_k ( \omega(q,p))  \omega(k,qp).
\end{equation}
\end{proposition}
\begin{proof}
The $2$-cocycle condition satisfied by $\omega$ comes from the associativity of the matrices. We can decompose the virtual action of $U_k U_q U_p$ in two ways as follows:
$$v_k v_q v_p= \omega(k,q) v_{kq}v_p = \omega(k,q) \omega(kq,p)v_{kpq}\; \;{\rm or}$$
$$v_k v_q v_p= v_k  \omega(q,p) v_{qp}= v_k  \omega(q,p)v^{-1}_k v_k v_{qp}= \phi_k ( \omega(q,p))  \omega(k,qp) v_{kpq}.$$
Then, applying $v^{-1}_{kpq}$, we obtain the $2$-cocycle condition.
\end{proof}

It is important to note the following relation between the maps $\phi$ and $\omega$ defined above:
\begin{equation}\label{eq:quaihomo}
 v_k v_q= \omega(k,q)v_{kq}\Rightarrow \phi_k\circ \phi_q=\tau_{\omega(k,q)}\circ \phi_{kq},
\end{equation}
where $\tau_{g}$ denotes the conjugation by $g \in G$.
Let us show how \cref{eq:quaihomo} allows us to show that $\phi$ can define an homomorphism from $Q$ to ${\rm Aut}(G)$. 
In the case where $G$ is abelian, $\phi$ is directly a homomorphism from $Q$ to Aut($G$) since $\tau_{\omega(k,q)}$ is trivial on elements of $G$:
$$\phi_k\circ \phi_q|_G= \phi_{kq}|_G.$$

In the non-abelian case to define a homomorphism with $\phi$ we have to consider the  group of outer automorphisms of $G$, ${\rm Out}(G)$. That group is defined by quotienting the automorphism group with the conjugation by elements of $G$. The conjugations by $G$ formed a group, the so-called inner automorphism group ${\rm Inn}(G)$, which is normal in Aut($G$). That is,
$$ {\rm Out}(G) =   {\rm Aut}(G)/{\rm Inn}(G) .$$
Therefore, we can define $\psi$, analogous to $\phi$ in Eq.(\ref{autmap}), as the homomorphism from $Q$ to ${\rm Out}(G)$ quotienting the RHS of \cref{eq:quaihomo} by ${\rm Inn}(G)$. Therefore,
$$\psi_k\circ \psi_q|_G= \psi_{kq}|_G.$$ 

\begin{definition}[Equivalence relation of $(\phi,\omega)$]\label{def:eqrelop}
We say that two pairs $(\phi,\omega)$ and $(\phi',\omega')$ are equivalent, $(\phi,\omega)\sim (\phi',\omega')$, if the following holds
\begin{equation}\label{cocyclefreedom} 
\omega'(k,q)=g_k \phi_k(g_q)\omega(k,q)g^{-1}_{kq} \; \; {\rm and}
\end{equation}
$$\phi'_k= \tau_{g_k} \circ \phi_k,$$
for some $g_q, g_k, g_{kq}\in G$.
\end{definition}
The equivalence relation between of the pair $(\phi,\omega)$ comes from the redundancy in the definition of $v_k$; the gauge freedom commented in \cref{equivdef}. That is, if we modified $v_q$ by $v'_q=g v_q$ for all $q\in Q$ we arrive to \cref{cocyclefreedom}.

\begin{proposition}[Classification of $\omega$]\label{prop:H2}
Given $\phi$, the $2$-cocycle $\omega$ is classified, under the equivalence relation of \cref{def:eqrelop},  by the group $H_\phi^2(Q,G)$ when $G$ is abelian. The group $H_\phi^2(Q,G)$ is defined as the quotient between $2$-cocycles and $2$-coboundaries, see Appendix \ref{ap:ext}. $\rho:Q\times Q\mapsto G$ is a $2$-coboundary if there exists a map from $Q$ to $G$: $q\mapsto g_q$ such that
$$ \rho(q,k)= g_k \phi_k(g_q)g^{-1}_{kq}, \; {\it for \; any}\; k,q\in Q.$$
\end{proposition}
\begin{proof}
If $G$ is abelian $\phi'_k|_G=  \phi_k|_G$ and
$$\omega'(k,q)= \rho(k,q)\omega'(k,q),$$
where  $\rho(k,q)=g_k \phi_k(g_q)g^{-1}_{kq}$ is a $2$-coboundary since it satisfies the $2$-cocycle. Quotienting $2$-cocycles by $2$-coboundaries we obtain the second cohomology group $H^2_\phi(Q,G)$. This group is finite due to the finiteness of $Q$ and $G$.
\end{proof}%
\begin{observation}
The class of the pair $(\phi,\omega)$ is robust under blocking since they would act equivalently on the underlying representation of the group: the tensor product representation of the blocked tensors.
\end{observation}

Let us notice that another equivalence relation has to be added in our classification:

\begin{observation}\label{obs:relabelling}
We consider two systems equivalent if their maps  $(\omega,\phi)$ are related by a relabelling of the elements of $Q$. This comes from the ambiguity the label in the elements of the group that defines the symmetry operators. 
The pair $(\phi, \omega)$ is related to $(\phi', \omega')$ by a relabelling if $\omega(q,k)= \omega'(\rho(q),\rho(k))$ and $\phi_q=\phi'_{\rho(q)}$, where $\rho \in$Aut($Q$). We notice that two system has to be consider equivalent even when $\omega$ and $\omega'$ could be inequivalent as $2$-cocycles, i.e. elements of $H^2_\phi(Q,G)$. 
\end{observation}

One example is given by $G=Q=\mathbb{Z}_p$ with $p$ prime because $H^2(\mathbb{Z}_p,\mathbb{Z}_p)\cong \mathbb{Z}_p$ but incorporating the relation of Aut($\mathbb{Z}_p$)$=\mathbb{Z}_{p-1}$ on $\omega$ we only find two distinct classes.
Some remarks are in order:

\begin{remark} \label{SPT label}
The fact that the operators act in a tensor product form and that they are defined up to the phase factor mention in \cref{Glemma} allows them to be a projective representation of $Q$. This would assign a discrete label, when considering the freedom of the phase factor, from $H^2(Q,U(1))$ in each direction. We point out that this label is not stable under blocking, that is why we dropped it in \cref{Glemma}, so we will not consider it in this thesis (besides it could matter in finite size systems). 
\end{remark}

\begin{remark}
Since the maps $\omega$ and $\phi$ are the same using $v_q$ or $w_q$ in their definitions and also by \cref{SPT label}  we will write w.l.o.g. the following
 \begin{equation} \label{locsym2}
    \begin{tikzpicture}[baseline=-1mm]
      \pic at (0,0,0) {3dpeps};
      \draw (0,0,0) -- (0,0.35,0);
      \filldraw[draw=black, fill=purple] (0,0.2,0) circle (0.06);
      \node[anchor=east] at (0,0.3,0) {$U_q$};
    \end{tikzpicture} =
    \begin{tikzpicture}
      \draw (-0.7,0,0) -- (0.7,0,0);
      \draw (0,0,-0.9) -- (0,0,0.9);
      \pic at (0,0,0) {3dpeps};
      \filldraw [draw=black, fill=red] (-0.5,0,0) circle (0.05);
      \filldraw [draw=black, fill=red] (0.5,0,0) circle (0.05);
      \node[anchor=south] at (-0.5,0,0) {$\myinv{v_q}$};
      \node[anchor=north] at (0.5,0,0) {$v_q$};
      \filldraw [draw=black, fill=red] (0,0,-0.6) circle (0.05);
      \filldraw [draw=black, fill=red] (0,0,0.6) circle (0.05);
       \node[anchor=south] at (0,0,-0.6) {$v_q$};
      \node[anchor=north] at (0,0,0.6) {$\myinv{v_q}$};
    \end{tikzpicture}
    \; \forall q\in Q.
  \end{equation}
\end{remark}

\begin{remark} \label{Symform}
Consider a $G$-injective tensor decomposed as a product of an invertible matrix $Y$ and the projector onto the $G$-symmetric space, see Eq.\eqref{Ginje}, $A=Y\mathcal{P}_G$. The symmetry operator acts as follows:
\begin{align*}
U_q (Y\mathcal{P}_G)= (Y\mathcal{P}_G) (v_q\otimes v_q\otimes v^{-1}_q \otimes v^{-1}_q) &= Y  (v_q\otimes v_q\otimes v^{-1}_q \otimes v^{-1}_q)[\phi^{-1}_q\otimes \phi^{-1}_q\otimes \phi_q \otimes \phi_q ] (\mathcal{P}_G)\\
&=Y  (v_q\otimes v_q\otimes v^{-1}_q \otimes v^{-1}_q) \mathcal{P}_G.
\end{align*}
That is, the symmetry operator $U_q$ is projected on the symmetric subspace as $v_q\otimes v_q\otimes v^{-1}_q \otimes v^{-1}_q$ up to an invertible matrix.  If $A=\mathcal{P}_G$ the physical operator $U_q$ is projected to $v_q\otimes v_q\otimes v^{-1}_q \otimes v^{-1}_q$ when considering the subspace generated by the tensor:
\end{remark}

\begin{equation}
\sum_{g\in G}
     \begin{tikzpicture}[scale=1.5]
          \begin{scope}[canvas is zx plane at y=0.32]
    \filldraw [draw=black, fill=purple] (0,0) circle (0.12);
    \end{scope} 
    \node[anchor=east] at (0,0.32,0) {$U_q$};
       \draw (0.1,0,0)--(0.1,0.4,0);
	\draw (-0.1,0,0)--(-0.1,0.4,0);
	\draw (0,0,0.1)--(0,0.4,0.1);
	\draw (0,0,-0.1)--(0,0.4,-0.1);
  \filldraw  (0.1,0.17,0) circle (0.03);
  \filldraw  (-0.1,0.17,0) circle (0.03);
  \filldraw  (0,0.16,0.1) circle (0.03);
  \filldraw  (0,0.18,-0.1) circle (0.03);
    \begin{scope}[canvas is zx plane at y=0]
      \draw (0.1,0)--(0.6,0);
        \draw (-0.6,0)--(-0.1,0);
      \draw (0,0.1)--(0,0.45);
       \draw (0,-0.45)--(0,-0.1);
    \end{scope} 
\end{tikzpicture}
=\sum_{g\in G}
     \begin{tikzpicture}[scale=1.5]
       \draw (0.1,0,0)--(0.1,0.4,0);
	\draw (-0.1,0,0)--(-0.1,0.4,0);
	\draw (0,0,0.1)--(0,0.4,0.1);
	\draw (0,0,-0.1)--(0,0.4,-0.1);
  \filldraw  (0.1,0.25,0) circle (0.03);
  \filldraw  (-0.1,0.25,0) circle (0.03);
  \filldraw  (0,0.24,0.1) circle (0.03);
  \filldraw  (0,0.26,-0.1) circle (0.03);
    \begin{scope}[canvas is zx plane at y=0]
      \draw (0.1,0)--(0.6,0);
        \draw (-0.6,0)--(-0.1,0);
      \draw (0,0.1)--(0,0.45);
       \draw (0,-0.45)--(0,-0.1);
    \end{scope} 
  \filldraw [draw=black, fill=red] (0.3,0,0) circle (0.04);
  \filldraw [draw=black, fill=red] (-0.3,0,0) circle (0.04);
  \filldraw [draw=black, fill=red] (0,0,0.4) circle (0.04);
  \filldraw [draw=black, fill=red] (0,0,-0.4) circle (0.04);
\end{tikzpicture}
=\sum_{g\in G}
     \begin{tikzpicture}[scale=1.5]
       \draw (0.1,0,0)--(0.1,0.4,0);
	\draw (-0.1,0,0)--(-0.1,0.4,0);
	\draw (0,0,0.1)--(0,0.4,0.1);
	\draw (0,0,-0.1)--(0,0.4,-0.1);
  \filldraw  (0.1,0.16,0) circle (0.03);
  \filldraw  (-0.1,0.16,0) circle (0.03);
  \filldraw  (0,0.15,0.1) circle (0.03);
  \filldraw  (0,0.17,-0.1) circle (0.03);
    \begin{scope}[canvas is zx plane at y=0]
      \draw (0.1,0)--(0.6,0);
        \draw (-0.6,0)--(-0.1,0);
      \draw (0,0.1)--(0,0.45);
       \draw (0,-0.45)--(0,-0.1);
    \end{scope} 
  \filldraw [draw=black, fill=red] (0.1,0.3,0) circle (0.035);
  \filldraw [draw=black, fill=red] (-0.1,0.3,0) circle (0.035);
  \filldraw [draw=black, fill=red] (0,0.29,0.1) circle (0.035);
  \filldraw [draw=black, fill=red] (0,0.31,-0.1) circle (0.035);
\end{tikzpicture},
\end{equation}
where the black dots represent the elements $g\in G$ and the red circles the matrices $v_q$.

\subsection{Robustness of the class under smooth deformation}\label{sec:deform}
In this subsection we perturbe $G$-isometric PEPS as it was mentioned already in the introduction of this chapter. We consider the so-called natural perturbations of PEPS \cite{Cirac13} where local operators $R(\epsilon)$ satisfying ${\rm lim}_{\epsilon \to 0} R(\epsilon)=\id$ are applied to the state $\ket{\Psi_A(\epsilon)}=R^{\otimes n}(\epsilon)\ket{\Psi_A}$. 
These transformations correspond to smooth perturbation of the parent Hamiltonian since $R(\epsilon)$ is invertible for small $\epsilon$. After some blocking, depending on the support of $R(\epsilon)$, we can consider these transformations as on-site operations, that is
$$\ket{\Psi_A(\epsilon)}=\ket{\Psi_{R(\epsilon)A}}= \ket{\Psi_{A(\epsilon)}},$$
where we have denoted $A(\epsilon)=R(\epsilon)A$. The required symmetry condition, for $\epsilon$ in some neighbourhood of $0$, can be imposed mainly in two ways:
\begin{enumerate}
\item[(Strong)] The symmetry operators commute with the perturbation: $[U_q,R(\epsilon)]=0$.  Then $U_q\ket{\Psi_{A(\epsilon)}}= \ket{\Psi_{A(\epsilon)}}$ since $U_q{A(\epsilon)}={A(\epsilon)}(v_q\otimes v_q\otimes v^{-1}_q \otimes v^{-1}_q)$, the class of $(\phi, \omega)$ does not change.
\item[(Weak)] The new PEPS satisfies $U^{\otimes n}_q \ket{\Psi_{A(\epsilon)}}= \ket{\Psi_{A(\epsilon)}}$. Then 
$$U_q{A(\epsilon)}={A(\epsilon)} (w_q(\epsilon)\otimes w_q(\epsilon)\otimes w_q(\epsilon)^{-1} \otimes w_q(\epsilon)^{-1}),$$
where $w_q(\epsilon)$ defines a $\tilde{\phi}_q^{[\epsilon]}\in Aut(G)$. Since $A$ is the projector onto the $G$-isometric subspace by \cref{Symform} we can write:
\begin{equation}\label{eq:RUrelw}
 [R(\epsilon)^{-1} U_q R(\epsilon)] A= [ w_q(\epsilon)\otimes w_q(\epsilon)\otimes w_q(\epsilon)^{-1} \otimes w_q(\epsilon)^{-1}] A,
 \end{equation}
with $U_q= v_q\otimes v_q\otimes v_q^{-1} \otimes v_q^{-1}$ which implies that $w_q(\epsilon)\otimes w_q(\epsilon)\otimes w_q(\epsilon)^{-1} \otimes w_q(\epsilon)^{-1}$ is continuous. We can invert the operators in the RHS of Eq.\eqref{eq:RUrelw} and using \cref{Glemma} the following holds:
\begin{equation*}
[w_q(\epsilon)^{-1}\otimes w_q(\epsilon)^{-1}\otimes w_q(\epsilon)\otimes w_q(\epsilon)]\cdot[R(\epsilon)^{-1} (v_q\otimes v_q\otimes v_q^{-1} \otimes v_q^{-1}) R(\epsilon)] = g\otimes g \otimes g^{-1}\otimes g^{-1}
\end{equation*}
for some $g\in G$. Therefore
\begin{equation}\label{contw}
[R(\epsilon)^{-1} (v_q\otimes v_q\otimes v_q^{-1} \otimes v_q^{-1}) R(\epsilon)] =g_qw_q(\epsilon)\otimes g_q w_q(\epsilon)\otimes  g_q^{-1} w_q(\epsilon)^{-1}\otimes g_q^{-1}w_q(\epsilon)^{-1},
\end{equation}
where we have labelled the element of $G$ with the subindex $q$. Since $R(\epsilon)$ and $R(\epsilon)^{-1}$ converge to $\id$ as $\epsilon \rightarrow 0$, the product also converges so ${\rm lim}_{\epsilon \to 0} w_q(\epsilon)=g v_q \sim v_q$ and ${\rm lim}_{\epsilon \to 0} \tilde{\phi}_q^{[\epsilon]} \sim \phi_q$.
\end{enumerate}

We notice that the previous analysis is also valid for $A$ equal to the projector onto the $G$-injective subspace. But, recalling the introduction of this chapter, the topological phase is only well defined for $G$-isometric PEPS since the Hamiltonian is commuting and then gapped. \\

We now suppose that $R(\epsilon)$ is continuous in some neighbourhood $0<\epsilon<\epsilon_0$. 
By \cref{contw} this implies that $w_q(\epsilon)\otimes w_q(\epsilon)\otimes w_q(\epsilon)^{-1} \otimes w_q(\epsilon)^{-1}$ is continuous. Therefore, (contracting indices keeps the continuity by linearity) the functions $(\tilde{\phi}_q^{[\epsilon]}(h))_{m,n}=\delta_{m,\tilde{\phi}_q^{[\epsilon]}(h)n}$ are continuous in $\epsilon$ for fixed $m,n,h\in G$. Since the previous delta function is zero or one for a given $m,n,h\in G$ for all $\epsilon \in (0,\epsilon_0)$ and ${\rm lim}_{\epsilon \to 0} \tilde{\phi}_q^{[\epsilon]} \sim \phi_q$, continuity implies that: $\tilde{\phi}_q^{[\epsilon]} \sim \phi_q,\; \forall \epsilon$. This means that continuous natural perturbations on $G$-isometric PEPS do not change the class of $\phi$, the permutation pattern of the anyons.\\

It is easy to see that Eq.\eqref{contw} is valid for any element of $Q$, in particular for $q$,$k$ and $kq$. We multiply the corresponding expressions, but we invert the one of $kq$, obtaining the following:
\begin{multline*}
R(\epsilon)^{-1} [\omega(k,q)^{\otimes 2} \otimes {\omega(k,q)^{-1}}^{\otimes 2} ] R(\epsilon) = \\
(g_k \tilde{\phi}_k^{[\epsilon]}(g_q)\tilde{\omega}^{[\epsilon]}(k,q)g^{-1}_{kq})^{\otimes 2} \otimes
(g^{-1}_k \tilde{\phi}_k^{[\epsilon]}(g^{-1}_q){\tilde{\omega}^{[\epsilon]}(k,q)}^{-1}g_{kq})^{\otimes 2}.
\end{multline*}
We fix the indices by applying $\bra{n,n,m,m}\cdot \ket{m,m,n,n}$, to both parts of this identity, where $m,n\in G$. We obtain
\begin{multline*}
\bra{n,n,m,m} R(\epsilon)^{-1} [\omega(k,q)^{\otimes 2} \otimes {\omega(k,q)^{-1}}^{\otimes 2} ] R(\epsilon) \ket{m,m,n,n} = \\
\bra{n}g_k \tilde{\phi}_k^{[\epsilon]}(g_q)\tilde{\omega}^{[\epsilon]}(k,q)g^{-1}_{kq}  \ket{m} .
\end{multline*}
Therefore, using the same arguments we used for the case of $\phi$, $\tilde{\omega}^{[\epsilon]}\sim \omega, \; \forall \epsilon$.

This allows us to assert that local continuous transformations preserve the class of $\phi$  and $\omega$ and then \cref{theo:class} is proven.

\section{Proof of \cref{theo:path} }\label{sec:finitepath}

In this section we will connect two $G$-injective PEPS in a finite size system with a smooth path which preserves the symmetry. 
We will see that they can be connected if the two pairs of maps $( \phi^0_q, \omega^0(k,q))$ , $(\phi^1_q, \omega^1(k,q))$ are in the same class, see \cref{def:eqrelop}. Moreover, in \cref{sec:gauging} we show how two PEPS with equivalents ($\phi,\omega$) can be mapped to the same  $E$-injective PEPS. \\

Since in this construction we do not consider the gap of the path we can restrict the form of the $G$-injective tensors. In particular, we consider two $G$-injective PEPS with tensors of the form $\mathcal{P}^0_G$ and $\mathcal{P}^1_G$, the projector into the symetric subspace -see \cref{Symform}-, with representations $u^0_g$ and $u^1_g$ respectively. 

We also suppose that these PEPS have a symmetry realized by the virtual operators $v^0_q$ and $v^1_q$ respectively. Let us define another semi-regular representation:
$$u_g \equiv u^0_g \oplus u^1_g,$$ 
with which we construct the tensor:
$$A(\lambda)=\frac{M(\lambda)}{|G|} \sum_g u_g \otimes u_g \otimes u^{-1}_g\otimes u^{-1}_g; \; \lambda \in [0,1],$$
where $ M(\lambda)= \begin{bmatrix}
    \lambda \id_{D^0} & 0 \\
    0 & (1-\lambda) \id_{D^1}
  \end{bmatrix} ^{\otimes 4}$  is an invertible matrix for $\lambda \in (0,1)$. Then, it is clear that $A(\lambda)$ is a $G$-injective tensor whose extreme points are the two symmetric $G$-injective tensors that we have considered before. 
Consider now the following operator:
$$v_q\equiv v^0_q\oplus v^1_q,$$
which commutes with $M(\lambda)$. It defines the symmetry operator:
$$U_q=v_q\otimes  v_q\otimes v^{-1}_q \otimes v^{-1}_q= U^0_q\oplus U^1_q\oplus U^{\rm path}_q.$$  
It is easy to see that $U^{\rm path}_q$ is only linear if  $\omega^0(k,q)= \omega^1(k,q)$. The operator $U_q$ acts over the tensor $A(\lambda)$ as follows:
$$U_q  A(\lambda)= A'(\lambda)(v_q\otimes  v_q\otimes v^{-1}_q \otimes v^{-1}_q).$$
The tensor $A'(\lambda)$ is the one in which the representation $u_g$ is modified as follows:  

\begin{align}
u_g \rightarrow v_q u_g v^{-1}_q & = [v^0_q\oplus v^1_q] \left(  u^0_g \oplus  u^1_g \right )[v^0_q\oplus v^1_q]^{-1}  \notag \\
& =   u^0_{\phi^0_q(g)} \oplus u^1_{\phi^1_q(g)}\notag.
\end{align}
Therefore $U_q$ is a symmetry if and only if  $A'(\lambda)=A(\lambda)$ and this only happens if $\phi^0_q=\phi^1_q \; \forall q\in Q$.
The condition of linear representation on $U_q$ enforces the map
$$ \omega(k,q)\equiv v_k v_q v^{-1}_{kq} =  v^0_k v^0_q {v^0}^{-1}_{kq} \oplus v^1_k v^1_q {v^1}^{-1}_{kq}=\omega^0(k,q) \oplus \omega^1(k,q)$$
to belong to $G$ and this only holds if $\omega^0 = \omega^1$. We notice that all the above identitites are satisfied in the case  $(\omega^0(k,q), \phi^0_k)\sim (\omega^1(k,q), \phi^1_k)$ choosing the proper gauge for the symmetry operators.

\subsection{Gauging the global symmetry}\label{sec:gauging}
The mathematical procedure to promote a global symmetry into a local (gauge) symmetry is called gauging. On a lattice, the procedure adds new terms to the Hamiltonian of the system which allows to change the character of the symmetry operators from global to local \cite{Barkeshli14}. Gauging can also be formulated at the level of states, in particular in PEPS \cite{Haegeman14, Williamson16}, where the procedure connects SPT phases and topologically ordered phases.
Also in Ref. \cite{Williamson17} the authors generalized the procedure to map a SET phase to a purely topological ordered phase. This is done by modifying the local tensors in such a way that the physical global symmetry becomes a local symmetry and also a virtual invariance. We will follow Ref. \cite{Haegeman14} to transform a $G$-injective tensor with a symmetry characterized by $(\phi,\omega)$ into an $E$-invariant tensor, where $E$ is the group extension of $G$ by $Q$ characterized by $(\phi,\omega)$.  We consider the $G$-injective tensor $A=A(u_g\otimes u_g\otimes u^{-1}_g\otimes u^{-1}_g)$ with an on-site global symmetry given by $Q$: $U_qA=A(v_q\otimes v_q\otimes v^{-1}_q\otimes v^{-1}_q)$.
Let us construct the tensor $B$ from $A$ as follows:
$$A\to B=\sum_{q\in Q} U_q A\otimes |q,q)(q,q|,$$
where $ |q)\in \mathbb{C}[Q]$ are the new virtual d.o.f. that we add. Notice that the virtual space has been enlarged from a $D$ dimensional space to a $D\times |Q|$ dimensional one. The tensor $B$ has the following virtual symmetries:
$$B \left([u_g\otimes \id_{|Q|}]^{\otimes 2} \otimes [\id_D \otimes \id_{|Q|}]^{\otimes 2} \right)=B \left([\id_D \otimes \id_{|Q|}]^{\otimes 2} \otimes [u_g\otimes \id_{|Q|}]^{\otimes 2} \right),$$
$$B \left( [v_q\otimes R_q]^{\otimes 2} \otimes [\id_D \otimes \id_{|Q|}]^{\otimes 2} \right) = B \left( [\id_D \otimes \id_{|Q|}]^{\otimes 2} \otimes [v_q\otimes R_q]^{\otimes 2} \right).$$
This means that the set $E\equiv  G\times Q$ is a gauge symmetry of $B$. 
Moreover $E=G\times Q$ is a group. We identify element $(g,k)$ with the matrix $[u_g\otimes \id_{|Q|}]\cdot[v_k\otimes R_k]$. This is well defined since $u$ and $R$ are faithful representations of $G$ and $Q$ respectively. Then
\begin{align}
(g,k)\cdot (h,q)&=[u_gv_k\otimes R_k]\cdot[u_hv_q\otimes R_q]=[u_gv_ku_hv_q\otimes R_kR_q] \notag \\
&= [u_gv_ku_hv^{-1}_k (v_kv_qv^{-1}_{kq})v_{kq}\otimes R_{kq}] \notag \\
&=(g\phi_k(h)\omega(k,q),kq)\in E. \notag
\end{align}
It is straightforward to show the associativity of this multiplication rule (using that $\omega$ is a $2$-cocycle) and also that the inverse is
$$(g,k)^{-1}=(\phi_{k^{-1}}[g^{-1}\omega^{-1}(k,k^{-1})],k^{-1})\in E.$$
Then, $B$ is an $E$-invariant tensor. The group $G\cong \{(g,e); g\in G\}$ is a normal subgroup of $E$ because $(h,k)(g,e)(h,k)^{-1}=(\phi_k(g),e)$. 
It remains to show that the virtual operators are a semiregular representation of $E$. This is the case if $u_g=L_g\otimes \id_p$ for some $p\in \mathbb{N}$ since  
$$\frac{1}{|E|}\sum_{\epsilon\in E}\bar{\chi}_\alpha(\epsilon)\tr{[u_gv_k\otimes R_k]}= d_\alpha p\neq0,$$
which means that $[u_gv_k\otimes R_k]$ contains all the irreps of $E$. In the case that the constructed representation of $E$ is faithful, the semiregularity is obtained after a finite number of blocking iterations. The $E$-injectivity and the locality of the parent Hamiltonian is proven in \cite{Haegeman14}. Also note the following action on $B$
\begin{align*}
U_kB&=\sum_{q\in Q} U_kq A\otimes |q,q)(q,q|=\sum_{q\in Q} U_q A\otimes |k^{-1}q,k^{-1}q)(k^{-1}q,k^{-1}q|\\
&= B \left( [\id_D \otimes L^{-1}_q]^{\otimes 2} \otimes [\id_D \otimes L_q]^{\otimes 2} \right).
\end{align*}
This implies that the previous action is a global symmetry but, it is disconnected from the topological part since
$$\left[ ( \id_D \otimes L_q), (u_gv_k\otimes R_k)\right]=0, \; {\rm for \; all} \; q,k\in Q \; {\rm and} \; g\in G.$$
Ref. \cite{Haegeman14} also shows that inserting additional tensors in the bonds connecting the $B$ tensors, the previous global symmetry can be mapped to a local symmetry.

With this procedure we have materialized in a new tensor the extension group associated with the global symmetry. 
That is, we have carried out a transformation from a tensor that describes a SET phase to a purely topologically ordered tensor.  Concretely from the quantum double model of $g$, $\mathcal {D}(G)$, plus a global symmetry of the group $Q$ charaterized by $(\phi,\omega)$, to $\mathcal {D}(E)$. 

It is of particular interest that equivalent SET phases are gauged into the same topological ordered phase. This was proposed in Ref.\cite{Barkeshli14} in the abstract language of unitary modular tensor categories. 

\section{Symmetry action over the anyons and ground subspace}\label{sect:anyons}

In this section we focus on $G$-isometric PEPS which allow to construct a gapped Hamiltonian in the same phase as the $\mathcal {D}(G)$, i.e. with the same GS topological degeneracy and same topological excitations \cite{Schuch10}.
We characterize the action of the symmetry, via the maps $\phi$ and $\omega$, over the anyons and ground subspace. The results are stated in the following propositions:

 \begin{tcolorbox}
\begin{proposition}[Permutation]\label{prop:perm}
The map $\phi$ describes the permutation of the anyons and ground states. Given the characterization of both in terms of a pair $([g],\alpha)$, where $[g]$ is a conjugacy class and $\alpha$ is an irrep of the normalizer of $[g]$, see \cref{sec:introGinj}, the global symmetry has the following effect:
$$ \left([g],\alpha \right) \stackrel{U_q^{\otimes n}}{\longrightarrow} \left([\phi_q(g)],\alpha^{[q]} \right),$$
where $\alpha^{[q]}$ is the label associated with the irrep $\pi_\alpha \circ \phi_q$ of the group $N_{\phi_q(g)}$.
\end{proposition}
\end{tcolorbox}

 \begin{tcolorbox}
\begin{proposition}[Symmetry Fractionalization]\label{prop:sf}
The $2$-cocycle $\omega$ characterizes the projective action of the symmetry of $Q$ on the charges. The different projective actions, given $\phi$, correspond to $H^2_\phi(Q,G)$. This projective action on a charge is equivalent to the braiding of the charge with the flux corresponding to the element $\omega(q,k)\in G$.
\end{proposition}
\end{tcolorbox}

We first deal with the action associated with the permutation map $\phi$. Afterwards we address the action associated with the cocycle $\omega$.

\subsection{Proof of \cref{prop:perm}}
We will prove \cref{prop:perm} analyzing the action of the global symmetry on each anyon type:
\begin{itemize}
\item {\bf Fluxes.} 

Consider the flux, characterized by the conjugacy class $[g]$ (see \cref{sec:introGinj}), placed on the edges of the path $\gamma$ represented by the operator $\bigotimes_{i\in \gamma}L_g^{m_i}$, where $L_g$ is the left regular representation of $G$.
 According to \cref{locsym}, the effect of  the global on-site operator  $\bigotimes_{v\in \Lambda}U^{[v]}_q$ over the $G$-injective PEPS with a flux is the action of $\phi_q$ (or $\phi_q^{-1} $ depending on $\gamma$) over each factor of $\bigotimes_{i\in \gamma}L_g^{m_i}$:
 $$\bigotimes_{i\in \gamma}L_g^{m_i}\mapsto \bigotimes_{i\in \gamma}L_{\phi_q(g)}^{m_i}, $$
 so the symmetry maps the flux-type $[g]$ to $[\phi_q(g)]$. The action of the symmetry is represented graphically as follows:
\begin{equation}
\begin{tikzpicture}
 \draw[densely dotted, blue,rounded corners](0.25,0,2.5)--(0.25,0,1.75)--(0.75,0,1.75)--(0.75,0,0.7);
        \pic at (0,0,2.1) {3dpeps};
        \pic at (0,0,1.4) {3dpeps}; 
      \pic at (0.5,0,0.7) {3dpeps};
      \pic at (0.5,0,1.4) {3dpeps};
     \pic at (0.5,0,2.1) {3dpeps};
       \pic at (1,0,0.7) {3dpeps};
        \pic at (1,0,1.4) {3dpeps};
         \pic at (1,0,2.1) {3dpeps};
	\filldraw [draw=black, fill=blue]  (0.25,0,2.1) circle (0.04);
	\filldraw [draw=blue, fill=white]  (0.5,0,1.75) circle (0.04);
	\filldraw [draw=black, fill=blue]  (0.75,0,1.4) circle (0.04);
	\filldraw [draw=black, fill=blue]  (0.75,0,0.7) circle (0.04);
	  \node[anchor=south] at (0.8,0,0.7)  {$g$};
 \end{tikzpicture} 
\stackrel{U_q^{\otimes n}}{\longrightarrow}
  \begin{tikzpicture}
   \draw[densely dotted, blue,rounded corners](0.25,0,2.5)--(0.25,0,1.75)--(0.75,0,1.75)--(0.75,0,0.7);
        \pic at (0,0,2.1) {3dpeps};
        \pic at (0,0,1.4) {3dpeps}; 
      \pic at (0.5,0,0.7) {3dpeps};
      \pic at (0.5,0,1.4) {3dpeps};
     \pic at (0.5,0,2.1) {3dpeps};
       \pic at (1,0,0.7) {3dpeps};
        \pic at (1,0,1.4) {3dpeps};
         \pic at (1,0,2.1) {3dpeps};
	\filldraw [draw=black, fill=blue]  (0.25,0,2.1) circle (0.04);
	\filldraw [draw=blue, fill=white]  (0.5,0,1.75) circle (0.04);
	\filldraw [draw=black, fill=blue]  (0.75,0,1.4) circle (0.04);
	\filldraw [draw=black, fill=blue]  (0.75,0,0.7) circle (0.04);
	 \filldraw [draw=black, fill=purple] (0,0.13,2.1) circle (0.04);
	 \filldraw [draw=black, fill=purple] (0,0.13,1.4) circle (0.04);
	 \filldraw [draw=black, fill=purple] (0.5,0.13,0.7) circle (0.04);
	 \filldraw [draw=black, fill=purple] (0.5,0.13,1.4) circle (0.04);
	 \filldraw [draw=black, fill=purple] (0.5,0.13,2.1) circle (0.04);
	 \filldraw [draw=black, fill=purple] (1,0.13,0.7) circle (0.04);
	  \filldraw [draw=black, fill=purple] (1,0.13,1.4) circle (0.04);
	   \filldraw [draw=black, fill=purple] (1,0.13,2.1) circle (0.04);
	    \node[anchor=south] at (0.8,0,0.7)  {$g$};
 \end{tikzpicture} 
 =
 \begin{tikzpicture}
 \draw[densely dotted, blue,rounded corners](0.25,0,2.5)--(0.25,0,1.75)--(0.75,0,1.75)--(0.75,0,0.7);
        \pic at (0,0,2.1) {3dpeps};
        \pic at (0,0,1.4) {3dpeps}; 
      \pic at (0.5,0,0.7) {3dpeps};
      \pic at (0.5,0,1.4) {3dpeps};
     \pic at (0.5,0,2.1) {3dpeps};
       \pic at (1,0,0.7) {3dpeps};
        \pic at (1,0,1.4) {3dpeps};
         \pic at (1,0,2.1) {3dpeps};
         \filldraw [draw=black, fill=red]  (0.12,0,2.1) circle (0.04);
	\filldraw [draw=black, fill=blue]  (0.25,0,2.1) circle (0.04);
	\filldraw [draw=black, fill=red]  (0.38,0,2.1) circle (0.04);
	
	\filldraw [draw=black, fill=red]  (0.5,0,1.55) circle (0.04);
	\filldraw [draw=blue, fill=white]  (0.5,0,1.75) circle (0.04);
	\filldraw [draw=black, fill=red]  (0.5,0,1.95) circle (0.04);
	
	\filldraw [draw=black, fill=red]  (0.62,0,1.4) circle (0.04);
	\filldraw [draw=black, fill=blue]  (0.75,0,1.4) circle (0.04);
	\filldraw [draw=black, fill=red]  (0.88,0,1.4) circle (0.04);
	
	\filldraw [draw=black, fill=red]  (0.62,0,0.7) circle (0.04);
	\filldraw [draw=black, fill=blue]  (0.75,0,0.7) circle (0.04);
	\filldraw [draw=black, fill=red]  (0.88,0,0.7) circle (0.04);
	 \node[anchor=south] at (0.8,0,0.5)  {$\phi_q(g)=v_q g v^{-1}_q$};
 \end{tikzpicture} .
 \end{equation}
It is clear that $\phi_k$ acts linearly in the class of fluxes and it permutes between conjugacy classes with the same number of elements, that is, between fluxes with the same quantum dimension. Consider the map $\psi$, which exactly captures the class of the fluxes because ${\rm Inn}(G)$ is the freedom of those fluxes. It is a homomorphism from $Q$ to the automorphisms of the conjugacy classes of $G$. The symmetry is non-trivial if it permutes between inequivalent classes of fluxes, {\it i.e.} if $\psi$ is a non-trivial outer automorphism. 

\item {\bf Charges.} 

Recall that for a non-trivial irrep $\pi_\sigma$ of $G$, the operator associated with a charge-anticharge pair is  $\Pi_{\sigma}=\sum_{g,h\in G}\chi_\sigma(gh^{-1})|g\rangle \langle g|\otimes |h\rangle \langle h|$.
The result of braiding a flux $p$ with one charge of the pair is just the phase factor $\chi_\sigma(p)$ when $G$ is abelian. In any case the effect can be measured, see Section \ref{subsec:braiding}, giving  $ \left |\chi_\sigma(p) /d_\sigma \right|^2$.

To analyze the symmetry action on charges we study the topological, {\it i.e.} braiding, properties of the charge modified by the symmetry operators:
$$\Pi^{[q]}_{\sigma}=\sum_{g,h\in G}\chi_\sigma(h^{-1}g)v_q|g\rangle \langle g| v^{-1}_q\otimes v_q |h\rangle \langle h| v^{-1}_q.$$
If we braid one charge of the modified pair with the flux corresponding to $p\in G$ we obtain:
\begin{align}
  B^{[\sigma]}_p(\Pi^{[q]}_{\sigma})&=\sum_{g,h\in G}\chi_\sigma(h^{-1}g)L_pv_q|g\rangle \langle g| v^{-1}_qL^{\dagger}_p\otimes v_q |h\rangle \langle h| v^{-1}_q\notag \\
 &=\sum_{g,h\in G}\chi_\sigma(h^{-1}g)v_q  L_{\phi^{-1}_q(p)}|g\rangle \langle g| L^{\dagger}_{\phi^{-1}_q(p)} v^{-1}_q\otimes v_q |h\rangle \langle h| v^{-1}_q\notag \\
 &=\sum_{g,h\in G}\chi_\sigma( h^{-1} \phi_q(p) g)v_q  |g\rangle \langle g| v^{-1}_q\otimes v_q |h\rangle \langle h| v^{-1}_q. \notag
\end{align}
When $G$ is abelian the braiding operation is a phase factor:
$$  B^{[\sigma]}_p(\Pi^{[q]}_{\sigma})=\chi_\sigma(\phi_q(p))\Pi^{[q]}_{\sigma}.$$
This phase factor is equal to $\chi_\sigma(p)$ if and only if $\phi_q$ is an inner automorphism of $G$ because for abelian groups $G$ is isomorphic to the group of its irreps. 
Thus, if $\phi$ is non-trivial, $\chi_\sigma(\phi_q(p))$ has to be identified with the phase factor corresponding to the braiding of the flux $p$ with some other charge. We denote this charge as $\sigma^{[q]}$. It satisfies
$$  B^{[\sigma^{[q]}]}_p(\Pi_{\sigma^{[q]}})=\chi_\sigma(\phi_q(p))\Pi_{\sigma^{[q]}}.$$
Since the type of the charge is defined by its transformation under braiding with fluxes we can conclude that the symmetry permutes between charge types: $\Pi^{[q]}_{\sigma}\equiv \Pi_{\sigma^{[q]}}$.

In the case where $G$ is non-abelian, the effect of the braiding can be more complex than a phase factor. We can then measure the probability of zero total charge by projecting on the initial state $\Pi^{[q]}_{\sigma}$, obtaining: 
$$ \left |  \frac{\chi_\sigma(\phi_q(p))}{d_\sigma}\right|^2,$$
where $d_\sigma$ is the dimension of the irrep $\sigma$. If $\phi$ is an inner automorphism, {\it i.e.} if $\psi$ is trivial, the modified charge $\Pi^{[q]}_{\sigma}$ transforms equivalently as $\Pi_{\sigma}$ under braiding operations, so it has to be identified with the same topological excitation. In contrast, if $\phi$ is not an inner automorphism, the representation $\pi_\sigma \circ \phi_q$, with character $\chi_\sigma(\phi_q(\cdot))$, is irreducible (with the same dimension as $\sigma$) and can be inequivalent to $\pi_\sigma$ so we denote its character as $\chi_{\sigma^{[q]}}(\cdot)$. Then, the braiding of a flux $p$ over $\Pi^{[q]}_{\sigma}$ gives the result
$$ \left |  \frac{\chi_\sigma(\phi_q(p))}{d_\sigma}\right|^2 \equiv \left |  \frac{\chi_{\sigma^{[q]}}(p)}{d_\sigma}\right|^2, $$
which implies that $\Pi^{[q]}_{\sigma}$ has to be identified with the irrep $\sigma^{[q]}$ and therefore the global symmetry has permuted between charges.

\item {\bf Dyons.} \\
A dyon is characterized by a pair $([h],\alpha)$, see \cref{sec:introGinj}. We will focus on the dyon of the pair particle-antiparticle so the virtual operator associated is -see Eq.\eqref{dyonend}-:
  
$$h^{\otimes \ell}\otimes h  \sum_{n \in N_{h}} \chi_\alpha(w n) \sum^\kappa_{j=1} |nk_j\rangle\langle nk_j|,$$

where $k_j$ runs over the representatives of right cosets of $G$ by $N_h$ and the chain $ h^{\otimes \ell}$ corresponds to the flux part of the dyon that connects with the antiparticle. The global symmetry for an element $q\in Q$ acts as follows:
$$\phi_q(h)^{\otimes \ell}\otimes v_q \;h  \sum_{n \in N_{h}} \chi_\alpha(w n) \sum^\kappa_{j=1} |nk_j\rangle\langle nk_j|\; v^{-1}_q,$$
which does not change the phase factor of the self-braiding since the virtual symmetry operators cancel out. This particle has to be associated with the dyon-type $([\phi_q(h)],\alpha^{[q]})$. That is, when we braid with $g\in N_{\phi_q(h)}$, that satisfies $gv_qh v^{-1}_q g^{-1}=v_qh v^{-1}_q$, the chain remains invariant  and the operator of the charge part changes to

\begin{align}
\phi_q(h)^{\otimes \ell}\otimes g v_q h   \sum_{n \in N_{h}} \chi_\alpha(w n) \sum^\kappa_{j=1} |nk_j\rangle\langle nk_j|\;  v^{-1}_q g^{-1}  &=\notag \\
\phi_q(h)^{\otimes \ell}\otimes \phi_q(h)   v_q \phi^{-1}_q(g)  \sum_{n \in N_{h}} \chi_\alpha(w n) \sum^\kappa_{j=1} |nk_j\rangle\langle nk_j|\;  \phi^{-1}_q(g^{-1})  v^{-1}_q  &=\notag \\
\phi_q(h)^{\otimes \ell}\otimes  v_q  h  \sum_{n \in N_{h}} \chi_\alpha(w\phi_q(g) n) \sum^\kappa_{j=1} |nk_j\rangle\langle nk_j|\;   v^{-1}_q  &\notag,
\end{align}
which corresponds to the change of the internal state of the irrep $\pi_\alpha \circ \phi_q$ of the group $N_{\phi_q(h)}$ associated with the braiding with $g\in N_{\phi_q(h)}$.

\item {\bf Ground subspace.}\\
To study the effect of a symmetry on fluxes, we have considered strings living in the virtual d.o.f. A similar computation allows to analyze how the ground subspace is affected by the symmetry operators. The ground state basis, formed by pair conjugacy classes, is $|\Psi(g,h)\rangle$ for $g,h\in G$ -see Eq.\eqref{eq:gsubspace}- where $g,h$ represent the two non-contractible loops acting on the torus satisfying $gh=hg$. The symmetry acts according to Eq.\eqref{locsym}, transforming $|\Psi(g,h)\rangle$ as follows:
\begin{equation}\label{permflux}
U^{\otimes n}_q|\Psi(g,h)\rangle =|\Psi(\phi_q(g),\phi_q(h))\rangle.
\end{equation}
The state $|\Psi(\phi_q(g),\phi_q(h))\rangle $ is a ground state. Indeed, $\phi_q(g)\phi_q(h)=\phi_q(h)\phi_q(h)$ shows that $(\phi_q(g),\phi_q(h))$ represents a well-defined pair conjugacy class. Then the homomorphism $\psi$ can permute between different pairs of conjugacy classes. 
We notice that the permutation of ground states is well defined also in $G$-injective PEPS and not only in the isometric point. 
\begin{equation}
|\Psi_A(g,h)\rangle=
\begin{tikzpicture}
        \pic at (0,0,0.7) {3dpeps};
        \pic at (0,0,1.4) {3dpeps};   
         \pic at (0,0,2.1) {3dpeps};  
          \pic at (0,0,0) {3dpeps};
      \pic at (0.5,0,0) {3dpeps};
      \pic at (0.5,0,0.7) {3dpeps};
      \pic at (0.5,0,1.4) {3dpeps};
     \pic at (0.5,0,2.1) {3dpeps};
      \pic at (1,0,0) {3dpeps};
        \pic at (1,0,0.7) {3dpeps};
        \pic at (1,0,1.4) {3dpeps};
         \pic at (1,0,2.1) {3dpeps};     
	  \pic at (1.5,0,0) {3dpeps};
        \pic at (1.5,0,0.7) {3dpeps};
        \pic at (1.5,0,1.4) {3dpeps};
         \pic at (1.5,0,2.1) {3dpeps};  
         
         \filldraw [draw=black, fill=blue]  (0.75,0,0.7) circle (0.04);
          \filldraw [draw=black, fill=blue]  (0.75,0,1.4) circle (0.04);
          \filldraw [draw=black, fill=blue]  (0.75,0,2.1) circle (0.04);
          \filldraw [draw=black, fill=blue]  (0.75,0,0) circle (0.04);
         \draw[densely dotted, blue](0.75,0,-0.5)--(0.75,0,2.6);
         \node[anchor=south] at (0.78,0.1,0) {${g}$};

         \filldraw [draw=black, fill=green]  (0,0,1.05) circle (0.04);
	 \filldraw [draw=black, fill=green]  (0.5,0,1.05) circle (0.04);
	 \filldraw [draw=black, fill=green]  (1,0,1.05) circle (0.04);
	 \filldraw [draw=black, fill=green]  (1.5,0,1.05) circle (0.04);
	   \draw[densely dotted, green](-0.5,0,1.05) --(2,0,1.05);
	 \node[anchor=west] at (1.8,0,1.3) {${h}$};
 \end{tikzpicture}
 \stackrel{U_q^{\otimes n}}{\longrightarrow}
\begin{tikzpicture}
        \pic at (0,0,0.7) {3dpeps};
        \pic at (0,0,1.4) {3dpeps};   
         \pic at (0,0,2.1) {3dpeps};  
          \pic at (0,0,0) {3dpeps};
      \pic at (0.5,0,0) {3dpeps};
      \pic at (0.5,0,0.7) {3dpeps};
      \pic at (0.5,0,1.4) {3dpeps};
     \pic at (0.5,0,2.1) {3dpeps};
      \pic at (1,0,0) {3dpeps};
        \pic at (1,0,0.7) {3dpeps};
        \pic at (1,0,1.4) {3dpeps};
         \pic at (1,0,2.1) {3dpeps};     
	  \pic at (1.5,0,0) {3dpeps};
        \pic at (1.5,0,0.7) {3dpeps};
        \pic at (1.5,0,1.4) {3dpeps};
         \pic at (1.5,0,2.1) {3dpeps};  
         
         \filldraw [draw=black, fill=teal]  (0.75,0,0.7) circle (0.04);
          \filldraw [draw=black, fill=teal]  (0.75,0,1.4) circle (0.04);
          \filldraw [draw=black, fill=teal]  (0.75,0,2.1) circle (0.04);
          \filldraw [draw=black, fill=teal]  (0.75,0,0) circle (0.04);
         \draw[densely dotted, teal](0.75,0,-0.5)--(0.75,0,2.6);
         \node[anchor=south] at (0.78,0.1,0) {${\phi_k(g)}$};

         \filldraw [draw=black, fill=lime]  (0,0,1.05) circle (0.04);
	 \filldraw [draw=black, fill=lime]  (0.5,0,1.05) circle (0.04);
	 \filldraw [draw=black, fill=lime]  (1,0,1.05) circle (0.04);
	 \filldraw [draw=black, fill=lime]  (1.5,0,1.05) circle (0.04);
	   \draw[densely dotted, lime](-0.5,0,1.05) --(2,0,1.05);
	 \node[anchor=west] at (1.8,0,1.3) {${\phi_k(h)}$};
 \end{tikzpicture}
 =|\Psi_A(\phi_k(g),\phi_k(h))\rangle. \notag
\end{equation}
The action on the whole ground subspace extends from the basis by linearity. It is worth computing the action of the symmetry in the another basis, the MES basis. That basis in one-to-one correspondence with the anyon types of the model. For a pair ($[g], \alpha$) the associated MES is the following state
\begin{equation*}
|\Psi_A([g],\alpha)\rangle= \sum_{n\in N_g}\chi_\alpha(n)
\begin{tikzpicture}
        \pic at (0,0,0.7) {3dpeps};
        \pic at (0,0,1.4) {3dpeps};   
         \pic at (0,0,2.1) {3dpeps};  
          \pic at (0,0,0) {3dpeps};
      \pic at (0.5,0,0) {3dpeps};
      \pic at (0.5,0,0.7) {3dpeps};
      \pic at (0.5,0,1.4) {3dpeps};
     \pic at (0.5,0,2.1) {3dpeps};
      \pic at (1,0,0) {3dpeps};
        \pic at (1,0,0.7) {3dpeps};
        \pic at (1,0,1.4) {3dpeps};
         \pic at (1,0,2.1) {3dpeps};     
	  \pic at (1.5,0,0) {3dpeps};
        \pic at (1.5,0,0.7) {3dpeps};
        \pic at (1.5,0,1.4) {3dpeps};
         \pic at (1.5,0,2.1) {3dpeps};  
         
         \filldraw [draw=black, fill=blue]  (0.75,0,0.7) circle (0.04);
          \filldraw [draw=black, fill=blue]  (0.75,0,1.4) circle (0.04);
          \filldraw [draw=black, fill=blue]  (0.75,0,2.1) circle (0.04);
          \filldraw [draw=black, fill=blue]  (0.75,0,0) circle (0.04);
         \draw[densely dotted, blue](0.75,0,-0.5)--(0.75,0,2.6);
         \node[anchor=south] at (0.78,0,0) {${g}$};

         \filldraw [draw=black, fill=green]  (0,0,1.05) circle (0.04);
	 \filldraw [draw=black, fill=green]  (0.5,0,1.05) circle (0.04);
	 \filldraw [draw=black, fill=green]  (1,0,1.05) circle (0.04);
	 \filldraw [draw=black, fill=green]  (1.5,0,1.05) circle (0.04);
	   \draw[densely dotted, green](-0.5,0,1.05) --(2,0,1.05);
	 \node[anchor=west] at (1.5,0,1.05) {${n}$};
 \end{tikzpicture},
 \end{equation*}
 which transforms under the symmetry as:
 $$U^{\otimes n}_q|\Psi_A([g],\alpha)\rangle =|\Psi([\phi_q(g)],\alpha^{[q]})\rangle.$$
This action corresponds to the one of dyons, the more general anyon, which emphasizes in the duality between ground states and anyons.  
 
 \end{itemize}
\subsection{ Proof of \cref{prop:sf}}\label{sec:SFcharges}
%
The symmetry effect of the map $\phi$ is a permutation of the anyon types, that is, a permutation between different eigenstates with the same eigenvalue (energy). This is the regular effect of a (symmetry) operator that commutes with the Hamiltonian of the system: its action preserves the eigenspaces. 
But there is a more subtle behaviour of the symmetry on the anyons: the action can be projective on each quasiparticle individually. This phenomenon is known as Symmetry Fractionalization (SF) -see \cite{Barkeshli14}. Let us explain what it is about.

\

Consider an on-site global symmetry of a topological phase. When anyons are present in the system the action of the symmetry is localized around the region where the quasiparticles are placed\footnote{The actual size of this region will be given by the correlation length so this is zero for RG fixed points.} because the regions between them correspond to the the ground state sector (the vacuum transforms trivially under the symmetry). We talk about SF when the symmetry acts \emph{projectively}, as opposed to linear, over the individual anyons. This freedom is allowed because only the vacuum, in general a collection of excitations with zero total charge, has to transform linearly under the symmetry group. 
So only the global effect of the individual projective actions has to become linear when a collection of anyons with zero total charge is considered. It turns out that these projective actions are equivalent to the braiding with some anyon that characterizes the SF pattern \cite{Barkeshli14}. 

\

When considering a pair particle-antiparticle, one has to transform inversely to the other. For example in abelian theories where the braiding is just a phase factor, given an anyon $\sigma$, its anti-particle $\bar{\sigma}$ transforms as the inverse projective representation (picking up the conjugate phase factor) -see \cite{Chen17} for a review. In any case the SF pattern has to be consistent with the fusion rules of the theory.

\

In the following we will explain how these concepts materialize for charges in the $G$-isometric PEPS picture. As shown in Eq.\eqref{eq:quaihomo}, the operators  $v_q$ do not have to form a linear representation. It actually turns out that  if $\omega(k,q) \neq 1$ the relation 
$$ v_k v_q = {\omega(k,q)} \; v_{kq}$$
means that $v_q$ is a \emph{projective} representation of $Q$. 
In that case $\{ v_q \}$ is an homomorphism up to a matrix since $\omega(k,q)\equiv u_{\omega(k,q)}$, where $u_g\equiv g$ is the representation of $G$ acting on the virtual d.o.f. ($u_g=L_g$ for $G$-isometric PEPS). 
Since the representation $\{U_q : q \in Q \}$ is assumed to be linear, the projective nature of $v_q$ does not show up in the action over $\ket{\Psi_A}$. 
It does not appear either in the action on the whole ground subspace nor on the fluxes. This is because $\phi$ is linear over conjugacy classes of $G$- see Eq.\eqref{eq:quaihomo}.
But the situation changes in regions that contain charges. Let us first define the conjugation map:
\begin{align}\label{eq:defPhi}
  \Phi_q \colon \mathcal{M}_D &\to \mathcal{M}_D \notag\\
  X &\mapsto \Phi_q(X)=v_q X v^{-1}_q.
\end{align}
One can use  Eq.\eqref{locsym} to calculate how the on-site symmetry $U_q$ affects a charge sitting on a virtual bond: $C_{\sigma,h} \to \Phi_q(C_{\sigma,h})$. Diagramatically: 
\begin{equation*}
  \begin{tikzpicture}
   \pic at (0,0,0.7) {3dpeps}; 
        \pic at (0,0,2.1) {3dpeps};
        \pic at (0,0,1.4) {3dpeps}; 
      \pic at (0.5,0,0.7) {3dpeps};
      \pic at (0.5,0,1.4) {3dpeps};
     \pic at (0.5,0,2.1) {3dpeps};
     	 \filldraw[draw=black,fill=purple] (0,0.13,0.7) circle (0.04);
	 \filldraw[draw=black,fill=purple] (0,0.13,2.1) circle (0.04);
	 \filldraw[draw=black,fill=purple] (0,0.13,1.4) circle (0.04);
	 \filldraw[draw=black,fill=purple] (0.5,0.13,0.7) circle (0.04);
	 \filldraw[draw=black,fill=purple] (0.5,0.13,1.4) circle (0.04);
	 \filldraw[draw=black,fill=purple] (0.5,0.13,2.1) circle (0.04);
	    \node[anchor=south] at (-0.1,0.15,1)  {$U_q$};
	         \filldraw[draw=black,fill=orange] (0.15,0,1.27) rectangle (0.35,0,1.53); 
 \end{tikzpicture} 
 =
 \begin{tikzpicture}
     \pic at (0,0,0.7) {3dpeps}; 
        \pic at (0,0,2.1) {3dpeps};
        \pic at (0,0,1.4) {3dpeps}; 
      \pic at (0.5,0,0.7) {3dpeps};
      \pic at (0.5,0,1.4) {3dpeps};
     \pic at (0.5,0,2.1) {3dpeps};

         \filldraw[draw=black,fill=red]  (0.12,0,1.4) circle (0.04);
	  \filldraw[draw=black,fill=orange] (0.15,0,1.27) rectangle (0.35,0,1.53); 
	\filldraw[draw=black,fill=red]  (0.38,0,1.4) circle (0.04);
 \end{tikzpicture}.
  \end{equation*}
We will define
\begin{equation*}
      \Phi_q(C_{\sigma,h})
      \equiv 
        \begin{tikzpicture}
      \draw (0.1,0,0) -- (0.7,0,0);
      \filldraw[draw=black,fill=red] (0.2,0,0) circle (0.04);
      \filldraw[draw=black,fill=red] (0.6,0,0) circle (0.04);
      \node[anchor=south] at (0.1,0,0) {$v_q$};
      \node[anchor=north] at (0.75,0,-0.15) {$\myinv{v_q}$};
       \filldraw[draw=black,fill=orange] (0.4,0,-0.2) rectangle (0.4,0,0.2); 
    \end{tikzpicture}.
         \end{equation*}
\
If the symmetry is applied for two elements of $Q$, we see that 
\begin{equation}\label{eq:symcharge}
(\Phi_k\circ \Phi_q)(C_{\sigma,h})=(\tau_{\omega(k,q)}\circ \Phi_{kq})(C_{\sigma,h}),
\end{equation}
where $\tau_\omega$ denotes conjugation by $\omega$. This implies that the symmetry action over the charge sector can be \emph{projective}, \emph{i.e.} the symmetry fractionalizes. 

We remark a fundamental point for our work:
\begin{observation}[Relation  between braiding and SF]\label{ob:brsf}
In virtue of Eq.\eqref{braidcwf} we can say that the factor that relates the action of $\Phi_k\circ \Phi_q$ and $\Phi_{kq}$ over the charge, see Eq.\eqref{eq:symcharge}, is equal to the braiding with the flux $\omega(k,q)\in G$ on the corresponding charge. That is, the conjugation by $\omega(k,q)$, $\tau_{\omega(k,q)}$, that defines the braiding in \eqref{braidcwf}, appears in Eq.\eqref{eq:symcharge}.
\end{observation}

The linearity of the symmetry action on a pair of charges is clear in  $G$-isometric PEPS since $B^{[\sigma-\bar{\sigma}]}_\omega(\Pi_\sigma)=\Pi_\sigma$; the braiding of a flux around a composite charge-anticharge is trivial. This can be understood as the fact that a charge and an anti-charge transform inversely under the braiding of a flux:
\begin{align}\label{sginv}
B^{[\sigma]}_\omega(\Pi_{\sigma})& = \sum_{h,t\in G}\chi_\sigma(t^{-1}\omega h) |h\rangle \langle h| \otimes |t\rangle \langle t|\equiv V^{[\sigma]}_\omega, \notag \\
 B^{[\bar{\sigma}]}_\omega(\Pi_{\sigma}) &= \sum_{h,t\in G}\chi_\sigma(t^{-1}\omega^{-1}h) |h\rangle \langle h| \otimes |t\rangle \langle t|\equiv V^{[\bar{\sigma}]}_\omega.\notag
\end{align}
That is,
$$V^{[\sigma]}_{\omega(k,q)} =V^{[\bar{\sigma}]}_{\omega(k,q)^{-1}},$$
which is equivalent in terms of braiding to the expression: $B^{[\bar{\sigma}]}_\omega(\Pi_{\sigma}) = B^{[\sigma]}_{\omega^{-1}}(\Pi_{\sigma})$. Note that this effect is only a phase factor if $\omega\in G$ is abelian or $\sigma$ is a unidimensional irrep; {\it i.e.} abelian anyons. In that case $B^{[\sigma]}_{\omega(k,q)}(\Pi_\sigma)=\chi_\sigma(\omega(k,q))\Pi_\sigma$ and
$$\overline{\chi_\sigma(\omega(k,q))}=\chi_{\bar{\sigma}}(\omega(k,q)).$$

This effect can be generalised for any collection of charges which can fuse to the vacuum. Given two charges $\alpha$ and $\beta$, which correspond to two irreps, their tensor product is:
$$\pi_\alpha \otimes \pi_\beta \cong \bigoplus_\sigma \id_{N_{\alpha \beta}^\sigma} \otimes  \pi_\sigma \Rightarrow \chi_\alpha \times \chi_\beta =\sum_\sigma N_{\alpha \beta}^\sigma \chi_\sigma.$$
The previous equation characterizes the superposition of charges appearing in the fusion $\alpha \times \beta$.
The probability of $\alpha$ and $\beta$ to fuse in $\sigma$ is $ N_{\alpha \beta}^\sigma d_\sigma / d_\alpha d_\beta$ \cite{Preskill04}. 
For any triple $\alpha, \beta , \bar{\sigma}$ that fuse to the vacuum the braiding with any flux is trivial. Therefore  for the abelian case $\chi_\alpha(\omega(k,q))\times \chi_\beta(\omega(k,q))\times\chi_{\bar{\sigma}}(\omega(k,q))=1$ and
$$\chi_\alpha(\omega(k,q))\times \chi_\beta(\omega(k,q))=\chi_{{\sigma}}(\omega(k,q)),$$
for any pair $\alpha,\beta$ that can fuse to $\sigma$ ($N_{\alpha \beta}^\sigma\neq 0$). This is a compatibility condition between the SF effect and the fusion rules of the theory (see \cite{Barkeshli14}).

\section{Symmetry defects as domain walls}

Symmetry defects can be created by acting with the symmetry operators over a compact region. The boundary of the region acts as a Domain Wall (DW) which act over the anyons when they cross it. This DW corresponds to a loop of virtual symmetry operators acting on the virtual d.o.f.:

\begin{equation}
  \begin{tikzpicture}
   	\pic at (0,0,0) {3dpeps};
        \pic at (0,0,0.7) {3dpeps}; 
        \pic at (0,0,1.4) {3dpeps}; 
        \pic at (0,0,2.1) {3dpeps};
        \pic at (0,0,2.8) {3dpeps};
        \pic at (0.5,0,0) {3dpeps};
      \pic at (0.5,0,0.7) {3dpeps};
      \pic at (0.5,0,1.4) {3dpeps};
     \pic at (0.5,0,2.1) {3dpeps};
     \pic at (0.5,0,2.8) {3dpeps};
     \pic at (1,0,0) {3dpeps};
       \pic at (1,0,0.7) {3dpeps};
        \pic at (1,0,1.4) {3dpeps};
         \pic at (1,0,2.1) {3dpeps};
         \pic at (1,0,2.8) {3dpeps};
              \pic at (1.5,0,0) {3dpeps};
       \pic at (1.5,0,0.7) {3dpeps};
        \pic at (1.5,0,1.4) {3dpeps};
         \pic at (1.5,0,2.1) {3dpeps};
         \pic at (1.5,0,2.8) {3dpeps};
              \pic at (2,0,0) {3dpeps};
       \pic at (2,0,0.7) {3dpeps};
        \pic at (2,0,1.4) {3dpeps};
         \pic at (2,0,2.1) {3dpeps};
         \pic at (2,0,2.8) {3dpeps};
	 \filldraw [draw=black, fill=purple] (0.5,0.13,0.7) circle (0.04);
	 \filldraw [draw=black, fill=purple] (0.5,0.13,1.4) circle (0.04);
	 \filldraw [draw=black, fill=purple] (0.5,0.13,2.1) circle (0.04);
	 \filldraw [draw=black, fill=purple] (1,0.13,0.7) circle (0.04);
	  \filldraw [draw=black, fill=purple] (1,0.13,1.4) circle (0.04);
	   \filldraw [draw=black, fill=purple] (1,0.13,2.1) circle (0.04);
	   	 \filldraw [draw=black, fill=purple] (1.5,0.13,0.7) circle (0.04);
	  \filldraw [draw=black, fill=purple] (1.5,0.13,1.4) circle (0.04);
	   \filldraw [draw=black, fill=purple] (1.5,0.13,2.1) circle (0.04);
	   
 \end{tikzpicture} =
   \begin{tikzpicture}
   
    \begin{scope}[canvas is zx plane at y=0]
 \draw[epsilon] (0.35,1.75) rectangle (2.45,0.25);
 \end{scope}
  
   	\pic at (0,0,0) {3dpeps};
        \pic at (0,0,0.7) {3dpeps}; 
        \pic at (0,0,1.4) {3dpeps}; 
        \pic at (0,0,2.1) {3dpeps};
        \pic at (0,0,2.8) {3dpeps};
        \pic at (0.5,0,0) {3dpeps};
      \pic at (0.5,0,0.7) {3dpeps};
      \pic at (0.5,0,1.4) {3dpeps};
     \pic at (0.5,0,2.1) {3dpeps};
     \pic at (0.5,0,2.8) {3dpeps};
     \pic at (1,0,0) {3dpeps};
       \pic at (1,0,0.7) {3dpeps};
        \pic at (1,0,1.4) {3dpeps};
         \pic at (1,0,2.1) {3dpeps};
         \pic at (1,0,2.8) {3dpeps};
              \pic at (1.5,0,0) {3dpeps};
       \pic at (1.5,0,0.7) {3dpeps};
        \pic at (1.5,0,1.4) {3dpeps};
         \pic at (1.5,0,2.1) {3dpeps};
         \pic at (1.5,0,2.8) {3dpeps};
              \pic at (2,0,0) {3dpeps};
       \pic at (2,0,0.7) {3dpeps};
        \pic at (2,0,1.4) {3dpeps};
         \pic at (2,0,2.1) {3dpeps};
         \pic at (2,0,2.8) {3dpeps};
 \filldraw [draw=black, fill=red] (0.5,0,0.35) circle (0.04);   
 \filldraw [draw=black, fill=red] (1.5,0,0.35) circle (0.04); 
  \filldraw [draw=black, fill=red] (1,0,0.35) circle (0.04);
   \filldraw [draw=black, fill=red] (0.5,0,2.45) circle (0.04);   
   \filldraw [draw=black, fill=red] (1.5,0,2.45) circle (0.04);
      \filldraw [draw=black, fill=red] (1,0,2.45) circle (0.04);
       \filldraw [draw=black, fill=red] (0.25,0,0.7) circle (0.04); 
       \filldraw [draw=black, fill=red] (0.25,0,1.4) circle (0.04); 
       \filldraw [draw=black, fill=red] (0.25,0,2.1) circle (0.04); 
          \filldraw [draw=black, fill=red] (1.75,0,0.7) circle (0.04); 
       \filldraw [draw=black, fill=red] (1.75,0,1.4) circle (0.04); 
       \filldraw [draw=black, fill=red] (1.75,0,2.1) circle (0.04); 
 \end{tikzpicture}.
  \end{equation} 
In this subsection we show the following for $G$-isometric PEPS:

 \begin{tcolorbox}
\begin{proposition}[Domain wall permutation]\label{prop:DWperm}
An anyon is able to cross a domain wall, coming from a global on-site symmetry, unitarily. When the anyon crosses the domain wall, the type of the anyon changes according to $\phi$. This gives a method to detect the permutation pattern in small regions.
\end{proposition}
\end{tcolorbox}

\begin{proof}
We first show how to move a flux through a DW by local unitaries. To be self-contained, we first describe the procedure of Ref.\cite{Schuch10} to move a flux. We need to consider two bonds, one with the operator $L_g$ (the flux part) and the other empty. We can use a unitary operation on the adjacent physical sites to synchronize both bonds, that is, go from $\sum_x L_x \otimes \sum_y L_y$ to $\sum_x L_x \otimes  L_x$. 
The operation is sketched as follows:

\begin{equation}
   \begin{tikzpicture}
 \pic at (0,0,0) {pepsGisopen};
  \pic at (0,0,0.9) {pepsGisopen};
   \pic at (0.7,0,0) {pepsGisopen};
  \pic at (0.7,0,0.9) {pepsGisopen};
   \pic at (1.4,0,0) {pepsGisopen};
  \pic at (1.4,0,0.9) {pepsGisopen};
         \node at (-0.1,0.1,0.55)  {$g$};
   \filldraw [draw=black, fill=blue] (0,0,0.45) circle (0.04);
    \draw[densely dotted, blue](0,0,0.45) --(-0.35,0,0.45);
    \begin{scope}[canvas is zx plane at y=0.1]
    \draw[thin, brown] (0,0.35) ellipse (0.1cm and 0.31cm);
    \draw[thin, brown] (0.9,0.35) ellipse (0.1cm and 0.31cm);
\draw[thin, brown] (0,1.05) ellipse (0.1cm and 0.31cm);
    \draw[thin, brown] (0.9,1.05) ellipse (0.1cm and 0.31cm);
    \draw[thin, brown] (0.45,1.4) ellipse (0.37cm and 0.06cm);
 \end{scope}
   \end{tikzpicture} 
 \longrightarrow
  \begin{tikzpicture}
    \draw (0,0.1,0.5)--(0,0,0.5)--(0,0,-0.5)--(0,0.1,-0.5);
      \draw (0.4,0.1,0.5)--(0.4,0,0.5)--(0.4,0,-0.5)--(0.4,0.1,-0.5);
      \filldraw [draw=black, fill=blue] (0,0,0) circle (0.04);
      \filldraw (0,0,0.3) circle (0.03);
      \filldraw (0,0,-0.3) circle (0.03);
    \draw[densely dotted, blue](0,0,0) --(-0.35,0,0);
      \node at (0,0,0.65)  {$c$};
      \node at (0.05,0,-0.7)  {$b$};
        \node at (0.45,0,-0.7)  {$a$};
  \end{tikzpicture} 
  = \sum_{p \in G} |xg\myinv{x}  p\rangle_b \langle p|_c \otimes  \sum_{s \in G} |s\rangle_a \langle s |= L_{xg\myinv{x}}\otimes \id,
\end{equation}
where we do not have knowledge of $x\in G$ (the black dots), we have labelled each ket and bra and the synchronization is denoted by the ellipses in yellow. The transformation $|a\rangle|b\rangle \langle c| \to |b\myinv{c}a\rangle|b\rangle \langle c|$ is implemented by the operator $M_f(\cdot)=\sum_{b,c} L_{b\myinv{c}}\otimes |b\rangle \langle b|(\cdot) |c\rangle \langle c|$ which finally goes to $L_{xg\myinv{x}}\otimes L_{xg\myinv{x}}$:

\begin{equation}
M_f\left(
  \begin{tikzpicture}
    \draw (0,0.1,0.5)--(0,0,0.5)--(0,0,-0.5)--(0,0.1,-0.5);
      \draw (0.4,0.1,0.5)--(0.4,0,0.5)--(0.4,0,-0.5)--(0.4,0.1,-0.5);
      \filldraw [draw=black, fill=blue] (0,0,0) circle (0.04);
    \draw[densely dotted, blue](0,0,0) --(-0.35,0,0);
  \end{tikzpicture}
  \right )=
  \begin{tikzpicture}
    \draw (0,0.1,0.5)--(0,0,0.5)--(0,0,-0.5)--(0,0.1,-0.5);
      \draw (0.4,0.1,0.5)--(0.4,0,0.5)--(0.4,0,-0.5)--(0.4,0.1,-0.5);
      \filldraw [draw=black, fill=blue] (0,0,0) circle (0.04);
     \filldraw [draw=black, fill=blue] (0.4,0,0) circle (0.04);
    \draw[densely dotted, blue](0.4,0,0) --(-0.35,0,0);
  \end{tikzpicture} .
  \end{equation}
To move a flux through a DW we have to consider three bonds;  the flux outside, the boundary of the DW and the empty bond inside. We suppose that the bond where the flux is placed is synchronized already. We first synchronize as follows: 
\begin{equation*}
   \begin{tikzpicture}
 \pic at (0,0,0) {pepsGisopen};
   \draw(0.2,0,0) --(0.5,0,0);
    \draw(0.2,0,0.9) --(0.5,0,0.9);
       \draw(0.7,0,0.2) --(0.7,0,0.7);
  \pic at (0,0,0.9) {pepsGisopen};
   \pic at (0.7,0,0) {pepsGisopen};
   \draw(0.9,0,0) --(1.2,0,0);
    \draw(0.9,0,0.9) --(1.2,0,0.9);
       \draw(1.4,0,0.2) --(1.4,0,0.7);
  \pic at (0.7,0,0.9) {pepsGisopen};
   \pic at (1.4,0,0) {pepsGisopen};
  \pic at (1.4,0,0.9) {pepsGisopen};
   \filldraw [draw=black, fill=blue] (0,0,0.45) circle (0.04);
    \draw[densely dotted, blue](0,0,0.45) --(-0.35,0,0.45);
         \filldraw [draw=black, fill=red] (0.35,0,0) circle (0.04);   
 \filldraw [draw=black, fill=red] (0.35,0,0.9) circle (0.04); 
    \begin{scope}[canvas is zx plane at y=0.1]
\draw[thin, brown] (0,1.05) ellipse (0.1cm and 0.31cm);
    \draw[thin, brown] (0.9,1.05) ellipse (0.1cm and 0.31cm);
    \draw[thin, brown] (0.45,1.4) ellipse (0.37cm and 0.06cm);
 \end{scope}
   \end{tikzpicture} 
 \longrightarrow
  \begin{tikzpicture}
    \draw (0.35,0.1,-0.8) --(0.35,0,-0.8) --(1.15,0,-0.8)--(1.15,0.1,-0.8);
        \draw (0,0.1,0.7)--(0,0,0.7)--(0,0,-0.7)--(0,0.1,-0.7);
      \draw (1.5,0.1,0.7)--(1.5,0,0.7)--(1.5,0,-0.7)--(1.5,0.1,-0.7);
         \filldraw [draw=black, fill=red] (0.75,0,-0.8) circle (0.05); 
            \filldraw [draw=black, fill=gray] (0.9,0,-0.8) circle (0.04);  
          \filldraw [draw=black, fill=black] (0.6,0,-0.8) circle (0.04);     
      \filldraw [draw=black, fill=blue] (0,0,0) circle (0.05);
      \filldraw (0,0,0.3) circle (0.04);
      \filldraw (0,0,-0.3) circle (0.04);
         \filldraw[draw=black, fill=gray] (1.5,0,0.3) circle (0.04);
      \filldraw[draw=black, fill=gray] (1.5,0,-0.3) circle (0.04);
    \draw[densely dotted, blue](0,0,0) --(-0.35,0,0);
      \node at (-0.05,0,0.8)  {$c$};
      \node at (0.05,0,-0.8)  {$b$};
        \node at (1.55,0,-0.8)  {$a$};
        \node at (0.25,0,-0.8) {$f$};
        \node at (1.2,0,-0.8) {$d$};
  \end{tikzpicture} 
  = \sum_{l,s \in G} |\myinv{y} g y  l\rangle_b \langle l|_c \otimes  |\myinv{x} v_q y s  \rangle_d \langle s|_f \otimes \id_a,
\end{equation*}
where the black and grey dots represent unknown elements $y,x\in G$ respectively. The transformation we want to implement is 
$$  |a\rangle \to  |  (d \myinv{f}) (b \myinv{c}) (d \myinv{f})^{-1}a\rangle.$$
It can be implemented by 
$$\sum_{b,c,d,f} |d\rangle \langle d|(\cdot)  |f\rangle \langle f|\otimes   |b\rangle \langle b|(\cdot) |c\rangle \langle c|  \otimes L_{d \myinv{f}} L_{b \myinv{c}}L^\dagger_{d \myinv{f}}.$$
This operation results in
\begin{equation*}
  \begin{tikzpicture}
    \draw (0.35,0.1,-0.8) --(0.35,0,-0.8) --(1.15,0,-0.8)--(1.15,0.1,-0.8);
        \draw (0,0.1,0.7)--(0,0,0.7)--(0,0,-0.7)--(0,0.1,-0.7);
      \draw (1.5,0.1,0.7)--(1.5,0,0.7)--(1.5,0,-0.7)--(1.5,0.1,-0.7);
         \filldraw [draw=black, fill=red] (0.75,0,-0.8) circle (0.05); 
            \filldraw [draw=black, fill=gray] (0.9,0,-0.8) circle (0.04);  
          \filldraw [draw=black, fill=black] (0.6,0,-0.8) circle (0.04);
                \node at (-0.1,0.1,0.1)  {$g$};     
      \filldraw [draw=black, fill=blue] (0,0,0) circle (0.05);
      \filldraw (0,0,0.3) circle (0.04);
      \filldraw (0,0,-0.3) circle (0.04);
         \filldraw[draw=black, fill=gray] (1.5,0,0.3) circle (0.04);
      \filldraw[draw=black, fill=gray] (1.5,0,-0.3) circle (0.04);
      \node at (-0.05,0,0.8)  {$c$};
      \node at (0.05,0,-0.8)  {$b$};
        \node at (1.55,0,-0.8)  {$a$};
        \node at (0.25,0,-0.8) {$f$};
        \node at (1.2,0,-0.8) {$d$};
    \draw[densely dotted, blue](0.75,0,0) --(-0.35,0,0);
     \draw[densely dotted, purple](0.75,0,0) --(1.5,0,0);
           \filldraw [draw=red, fill=blue] (1.5,0,0) circle (0.05);
               \node at (2,0,0)  {$\phi_q(g)$};
  \end{tikzpicture} 
   = \sum_{l,s,t \in G} |\myinv{y} g y  l\rangle_b \langle l|_c \otimes  |\myinv{x} v_q y s  \rangle_d \langle s|_f \otimes  |\myinv{x} \overbrace{v_q g \myinv{v}_q}^{\phi_q(g)} x t \rangle_a \langle t|,
  \end{equation*}
which describes the permutation effect of moving a flux through a DW. \\

To move a charge, one just has to apply a swap operation between two bonds. This is because the virtual operator of a charge acts only in one bond and not as a string-like operator as the flux case. Let us consider two bonds and synchronize them in the following way:
\begin{equation*}
   \begin{tikzpicture}
 \pic at (0,0,0) {pepsGisopen};
  \pic at (0,0,0.9) {pepsGisopen};
   \pic at (0.7,0,0) {pepsGisopen};
  \pic at (0.7,0,0.9) {pepsGisopen};
         \node[anchor=east] at (-0.1,0.1,0.55)  {$\ket{g}\bra{g}$};
   \filldraw [draw=black, fill=orange] (0,0,0.4) circle (0.04);
    \begin{scope}[canvas is zx plane at y=0.1]
    \draw[thin, brown] (0,0.35) ellipse (0.1cm and 0.31cm);
    \draw[thin, brown] (0.9,0.35) ellipse (0.1cm and 0.31cm);
 \end{scope}
   \end{tikzpicture} 
 \longrightarrow
  \begin{tikzpicture}
    \draw (0,0.1,0.5)--(0,0,0.5)--(0,0,-0.5)--(0,0.1,-0.5);
      \draw (0.4,0.1,0.5)--(0.4,0,0.5)--(0.4,0,-0.5)--(0.4,0.1,-0.5);
      \filldraw [draw=black, fill=orange] (0,0,0) circle (0.04);
      \filldraw (0,0,0.3) circle (0.03);
      \filldraw[draw=black, fill=gray]  (0,0,-0.3) circle (0.03);
            \filldraw (0.4,0,0.3) circle (0.03);
      \filldraw[draw=black, fill=gray]  (0.4,0,-0.3) circle (0.03);
            \node at (0,0,0.65)  {$c$};
      \node at (0.05,0,-0.7)  {$b$};
        \node at (0.45,0,-0.7)  {$a$};
                 \node at (0.45,0,0.65)  {$d$};
  \end{tikzpicture} 
  =|x g \rangle_c \langle yg |_b \otimes  \sum_p |x y^{-1} p  \rangle_d \langle p|_a,
\end{equation*}
Then one has to permute the two bonds to move the charge from one site to the other. In the case when we have to cross the DW the synchronization results in:

 \begin{equation*}
   \begin{tikzpicture}
 \pic at (0,0,0) {pepsGisopen};
  \pic at (0,0,0.9) {pepsGisopen};
   \pic at (0.7,0,0) {pepsGisopen};
  \pic at (0.7,0,0.9) {pepsGisopen};
         \node[anchor=east] at (-0.1,0.1,0.55)  {$\ket{g}\bra{g}$};
   \filldraw [draw=black, fill=orange] (0,0,0.4) circle (0.04);
    \begin{scope}[canvas is zx plane at y=0.1]
    \draw[thin, brown] (0,0.35) ellipse (0.1cm and 0.31cm);
    \draw[thin, brown] (0.9,0.35) ellipse (0.1cm and 0.31cm);
 \end{scope}
              \filldraw [draw=black, fill=red] (0.35,0,0) circle (0.04);   
 \filldraw [draw=black, fill=red] (0.35,0,0.9) circle (0.04); 
   \end{tikzpicture} 
 \longrightarrow
  \begin{tikzpicture}
    \draw (0,0.1,0.5)--(0,0,0.5)--(0,0,-0.5)--(0,0.1,-0.5);
      \draw (0.4,0.1,0.5)--(0.4,0,0.5)--(0.4,0,-0.5)--(0.4,0.1,-0.5);
            \filldraw [draw=black, fill=orange] (0,0,0) circle (0.04);
      \filldraw (0,0,0.38) circle (0.03);
      \filldraw[draw=black, fill=gray]  (0,0,-0.4) circle (0.03);
       \filldraw[draw=black, fill=red]  (0,0,0.2) circle (0.03);
        \filldraw[draw=black, fill=red]  (0,0,-0.2) circle (0.03);
            \filldraw (0.4,0,0.3) circle (0.03);
      \filldraw[draw=black, fill=gray]  (0.4,0,-0.3) circle (0.03);
            \node at (0,0,0.65)  {$c$};
      \node at (0.05,0,-0.7)  {$b$};
        \node at (0.45,0,-0.7)  {$a$};
                 \node at (0.45,0,0.65)  {$d$};
  \end{tikzpicture} .
\end{equation*}
Therefore after permuting the bonds and recovering the DW, we recall that the synchronization is reversible, the charges have been permuted inside the DW by $\phi_q$. 
Analogously, a dyon is also permuted by $\phi$ when it crossed a DW.
\end{proof}

\cref{prop:DWperm} gives us a method to obtain the function $\phi_q$ in small regions, {\it i.e.} determine the anyonic permutation pattern, by measuring the type of the anyon after crossing the domain wall.

\

Moreover, together with the closed loops of operators $v_q$, formed by acting with $U_q$, string-like symmetry defects can be created. These strings, of length $\ell$, in the virtual d.o.f. are created by acting physically, for example, with $\mathcal{O}^{\otimes \ell}_q$, where 
\begin{equation*} \mathcal{O}_q=
   \begin{tikzpicture}
      \draw (-0.8,0.5)--(-0.8,0)--(0.8,0)--(0.8,0.5)--(-0.8,0.5);
      \draw[canvas is zy plane at x=0] (-1,0.25)--(-1,0)--(1,0)--(1,0.25);
      \draw [canvas is zy plane at x=0] (-1,0.25)--(-1,0.5)--(1,0.5)--(1,0.25);
      \filldraw [draw=black, fill=red] (0,0.25,1) circle (0.05);
      \pic at (0,0,0) {3dpepsdown};
      \pic at (0,0.5,0) {3dpeps};
      \node[anchor=north east] at (0,0.25,1) {$\sum_g u_gv_q u_g^{-1}$};
   \end{tikzpicture}
   \; \forall q\in Q.
\end{equation*}
The result of applying $\mathcal{O}^{\otimes \ell}_q$ on $\ket{\Psi_A}$ is the following:
\begin{equation*}
\begin{tikzpicture}
    \foreach \z in {0,0.7,1.4,2.1}{
   	   \foreach \x in {0,0.5,1,1.5,2,2.5}{
           	\pic at (\x,0,\z) {3dpeps}; 
             					} } 
     \foreach \x in {0.5,1,1.5,2}{
             \filldraw [draw=black, fill=magenta]  (\x,0.1,0.7) circle (0.04);
             \node[anchor=east] at (\x+0.08,0.12,0.7) {\tiny{$\mathcal{O}_q$}}; }
 \end{tikzpicture}
 =
\begin{tikzpicture}
    \foreach \z in {0,0.7,1.4,2.1}{
   	   \foreach \x in {0,0.5,1,1.5,2,2.5}{
           	\pic at (\x,0,\z) {3dpeps}; 
             					} } 
 \draw[red](0.5,0,1.05)--(2,0,1.05);
     \foreach \x in {0.5,1,1.5,2}{
             \filldraw [draw=black, fill=red]  (\x,0,1.05) circle (0.04);
             				   }
 \end{tikzpicture}
 \; .
\end{equation*}
We notice that this string cannot be moved freely using the $G$-injectivity of the tensors. The anyons are also permuted by $\phi_q$ when they cross these strings. This is because the exact same operation can be used here in order to cross the string with an anyon (we did not use the fact that the operators $v_q$ were part of a loop).

\section{Discussion}

In this chapter, based on \cref{theo:FTGinjective}, we have classified the different realizations of a global on-site symmetry coming from a finite group $Q$ on $G$-injective PEPS. We have linked this classification to the theory of group extensions and we have analyzed the action of the symmetry on the excitations and on the ground subspace. We have also studied the gauging procedure and domain walls properties in these phases.\\

We anticipate that for each class in our classification, we can construct a representative. This is the content of the next chapter. 
In that sense, if we restrict ourselves to the case where the symmetry, the topological order and the gauge theory are associated to groups, {\it i.e.} $Q$, $\mathcal{D}(G)$ and $\mathcal{D}(E)$ respectively, our classification is complete.

A more general picture of the thesis and the relationship between the chapters can be seen in the following diagram:

\begin{equation}\label{Diathesis}
\tikzstyle{mybox} = [draw=blue, fill=gray!20, very thick,
    rectangle, rounded corners, inner sep=5pt, inner ysep=10pt]
\tikzstyle{fancytitle} =[fill=blue, text=white, ellipse]
\begin{tikzpicture}[transform shape]
 
\node [mybox] (C2) at (-1,0){%
    \begin{minipage}[t!]{0.2\textwidth}
    $E$-isometric PEPS, where $E$ is a group extension of $G$ and $Q$
    \end{minipage}
    };

\node [mybox] (C3) at (5,2){%
    \begin{minipage}[t!]{0.2\textwidth}
    $G$-injective PEPS with a global symmetry given by $Q$.
    ($G \triangleleft E$ and $Q= E/G$)
    \end{minipage}
    };
      
    \node [mybox] (C4) at (5,-4.5){%
    \begin{minipage}[t!]{0.2\textwidth}
  Maps $(\phi,\omega)$ that characterize the action of the symmetry on anyons, {\it i.e.} permutation and SF.
    \end{minipage}
    };
    
\node [mybox] (C5) at (-1,-3){%
    \begin{minipage}[t!]{0.2\textwidth}
Local order parameter for the detection of the SF pattern, {\it i.e.} the class of $\omega$
    \end{minipage}
    };
    
\node [mybox] (C6) at (6,-1){%
    \begin{minipage}[t!]{0.2\textwidth}
Characterization of global symmetries in terms of local tensors
    \end{minipage}
    };

\path[->,draw=gray,line width=1mm, -{Triangle[]}]
    (C2) edge[line width=0.742mm] node[  fill=white, anchor=center, pos=0.5] {Chapter 4}
    						    node[ anchor=center, below right , pos=0.5, text width= 1.9cm ] {(via anyon  condensation)} (C3);

\path[->,draw=gray,line width=1mm, -{Triangle[]}]
    (C3) edge[line width=0.742mm] node[fill=white, anchor=center, above, pos=0.6] {Chapter 2 (via FT)}
    						     (C6); 

\path[->,draw=gray,line width=1mm, -{Triangle[]}]
    (C6) edge[line width=0.742mm] node[  fill=white,anchor=center, above, pos=0.5] {Chapter 3}
    						   node[ anchor=center, below right , pos=0.35, text width=3.5cm ] {(via interplay between topology and symmetry)} (C4); 

\path[->,draw=gray,line width=1mm, -{Triangle[]}]
    (C4) edge[line width=0.742mm] node[ anchor=center, above, pos=0.5] {Chapter 5}
    						   (C5); 						    
\path[<->,draw=gray,line width=1mm]
    (C2) edge[line width=0.742mm] node[ fill=white, anchor=center, above, pos=0.4] {Appendix A}
    node[ anchor=center, right , pos=0.5, text width=3cm ] {(relation between  }
    node[ anchor=center, right , pos=0.6, text width=2cm ] { $(\phi,\omega)$ and $E$ ) }
    						   (C4); 						    
						   
\end{tikzpicture}
\end{equation}

The previous diagram shows how the different chapters are used to study global on-site symmetries in $G$-injective PEPS to give a classification, construction and detection of those phases.

Some comments are in order. 
The classification of SET phases in terms of PEPS requires a Fundamental Theorem for the more general class of PEPS that describes topological order phases, the so-called MPO-injective PEPS \cite{Sahinoglu14}, already mentioned in \cref{dis:qd}. 
At the abstract level, in the language of modular tensor categories, SET phases have been classified in \cite{Barkeshli14} for on-site global symmetries. Since we restrict our study to groups we have two main limitations. First, we do not cover topological ordered phases outside quantum double models of finite groups. Second, there are symmetries of $\mathcal{D}(G)$ not described by our formalism, for example charge-flux permutation or symmetry fractionalization of fluxes. This is related to the fact that the gauged theories of our construction give quantum doubles of the extension, $\mathcal{D}(E)$. For example the toric code with a symmetry permuting the charge and the flux is mapped to the  double Ising model \cite{Barkeshli14}.

\newpage \cleardoublepage 

\chapter{Construction of symmetric phases via \emph{ungauging}} \label{chap:cond}
This chapter is devoted to the realization of the symmetric phases classified in \cref{theo:class}. 
We construct a representative of the phase corresponding to a $G$-isometric PEPS enriched with global on-site symmetry of the group $Q$, characterized by the maps $(\phi, \omega)$ introduced in the previous chapter. For a more detailed connection with the other chapters see diagram \eqref{Diathesis}. 

Let us denote as $E$ the extension group associated with the maps $(\phi, \omega)$, see Appendix \ref{ap:ext}. We will perform a local \emph{ungauging} to the $E$-isometric PEPS tensor. 
This procedure consists in explicitly breaking part of the local symmetry (virtual invariance) to induce a global symmetry. As opposed to gauging, where additional d.o.f. are introduced and the group extension is effectively reconstructed in \cref{sec:gauging}, here we go backwards in the group extension picture. This is done at the level of the local tensor. We start from
\begin{equation}\label{tensorE}
A_E=\frac{1}{|E|}\sum_{\epsilon \in E}L_\epsilon \otimes L_\epsilon \otimes L^{\dagger}_\epsilon \otimes L^{\dagger}_\epsilon\equiv
\begin{tikzpicture}[scale=1.2]
 \pic at (0,0,0) {3dpeps};
\end{tikzpicture}
\end{equation}
where the shape is rounded and restricting the sum to the elements of $G \triangleleft E$, we end up with the following tensor:
\begin{equation}\label{tensorGres}
A_G^{\rm res}=\frac{1}{|E|}\sum_{g\in G}L_g\otimes L_g\otimes L^{\dagger}_g\otimes L^{\dagger}_g \equiv
\begin{tikzpicture}[scale=1.2]
 \pic at (0,0,0) {3dpepsres};
\end{tikzpicture}
\end{equation}
which we denote as \emph{restricted} tensor and we represent it as a squared shape. Note that the local Hilbert space and the dimension of the virtual d.o.f. of $A_E$ and $A_G^{\rm res}$ coincide; it is $|E|$ because $L_g$ is the left regular representation of $E$. This will allow us to compare the action of the same operators acting on both tensors. The main result of this chapter is the following:

\begin{tcolorbox}
\begin{theorem}[Construction of the representatives] \label{theo:consrep}
The following statements hold:
\begin{enumerate}
\item The parent Hamiltonian of $|\Psi_{A_G^{\rm res}}\rangle$ is in the same phase as $\mathcal {D}(G)$ and it breaks part of the local symmetry of $|\Psi_{A_E}\rangle$. The broken local symmetry is degraded to a global symmetry of the group $Q\cong E/G$. 

\item The symmetry of $|\Psi_{A_G^{\rm res}}\rangle$ also corresponds to a symmetry that acts on the space of quasiparticle excitations and of the ground subspace. This action permutes between particle types and can act \emph{projectively} over charges. These effects are characterized by the maps $(\phi,\omega)$ which determine the extension group $E$ as explained in Appendix \ref{ap:ext}.
\end{enumerate}
\end{theorem}
\end{tcolorbox}

In Ref.\cite{Gu14}, the authors propose, at the level of modular tensor categories, that anyon condensation can be a mechanism to enrich topological phases with global symmetries. 
Refs. \cite{Bais09,ThesisMark, Bais03, Bais02, Bais06, Bais07} relate anyon condensation with an explicit symmetry breaking in gauge theories. Moreover, Ref.\cite{Bombin08} exhaustively study confinement and condensation for quantum doubles in terms of groups algebras, using lattice models.

The construction proposed in this chapter relates both approaches. Explicitly, we show how breaking partially a local symmetry, a global symmetry can emerge and the anyon condensation pattern can be identified.

If we consider the construction of SET phases without the connection with anyon condensation, the phases carried out in this chapter include the ones constructed in \cite{Hermele14,Tarantino16}, for $G$ abelian, as exactly solvable lattice models. 
Ref.\cite{Jiang15} also studies how the symmetry fractionalizes on $G$-injective PEPS when $G$ is an abelian group. There, the action of a global symmetry is assumed to be a gauge transformation. This is proven in the \cref{chap:FT}. Ref.\cite{Jiang15} also restricts itself to the case where the symmetry does not permute the anyons.
We also notice the general construction of SET phases of Refs.\cite{Heinrich16,Cheng17} as exactly solvable lattice models and Ref.\cite{Williamson17} as PEPS realizations.

\section{Proof of \cref{theo:consrep}}
To prove statement 1 of \cref{theo:consrep}, we first analyze the properties of the tensor $A_G^{\rm res}$, that is, we characterize its local and global symmetries. 

The state $|\Psi_{A_G^{\rm res}}\rangle$ is in the same topological phase as the $G$-isometric PEPS up to discarded local entangled degrees of freedom. This is derived from the fact that $G$ is isomorphic to the normal subgroup $\{ (g,e) |\; g\in G \}\subset E$ so we can choose $L_{(g,e)}^E= L_g^{G}\otimes\mathds{1}_{Q}$.
This is because acting on the group algebra basis $\mathbb{C}[E]=\mathbb{C}[G] \otimes \mathbb{C}[Q]$: 
\begin{equation}
\begin{split} 
L_{(g,e)}^E|n,k\rangle&=|g\phi_{e}(n) \omega(e,h),k\rangle=\\
 &=|gn,k\rangle= (L_g^{G}\otimes L_e^{Q})|g,k\rangle,\notag
\end{split}
\end{equation}
where we have used that $\omega(e,h)=e$ and $\phi_{e}=\id$ since these choices do not change the class of the extension. Thus,
$$A_G^{\rm res}\cong A_G\otimes \mathds{1}_{|Q|}^{\otimes 4};$$
where the identity operators form maximally entangled pair states between neighbour sites. Let us compare the local symmetries of both tensors by defining the operator $U_\epsilon=L_\epsilon \otimes L_\epsilon \otimes L^{\dagger}_\epsilon \otimes L^{\dagger}_\epsilon$ for all $\epsilon\in E$. The local symmetry, equivalent to the virtual invariance, of the tensors implies the corresponding invariance of the states:
$$U_\epsilon A_E=A_E \Rightarrow   (U^{\mathcal{R}}_\epsilon \otimes \mathds{1}^{{\rm rest}})|\Psi_{A_E}\rangle = |\Psi_{A_E}\rangle \; \forall \epsilon \in E,$$
where $\mathcal{R}$ is any region of the lattice; and analogously for the restricted tensor:
$$U_g A_G^{\rm res}=A_G^{\rm res} \Rightarrow   (U^{\mathcal{R}}_g\otimes \mathds{1}^{{\rm rest}})|\Psi_{A_G^{\rm res}}\rangle = |\Psi_{A_G^{\rm res}}\rangle \; \forall g \in G.$$
However the elements $\epsilon \in E \setminus G$ do not belong to the local symmetry of $A_G^{\rm res}$; the restricted tensor breaks the local symmetry of $A_E$ to the normal subgroup $G$ of $E$. The cosets of $E$ by $G$ form a group $Q\cong E/G$ whose elements we denote as $q,k$ and $z$. We can take representatives $\epsilon_q\in E$ of these cosets $\epsilon_q G=\{g\epsilon_q; g\in G\}$ in correspondence with all $q\in Q$.
We now consider the action on $A_G^{\rm res}$ of the operators that do break the symmetry, that is $U_{\epsilon_k}$ for $k\neq e $  in $Q$; 
\begin{equation}\label{symAres}
U_{\epsilon_k} A_G^{\rm res}= \sum_{\tilde{\epsilon}_k\in \epsilon_kG}U_{\tilde{\epsilon}_k}=A_G^{\rm res} U_{\epsilon'_k}\neq A_G^{\rm res} \quad \forall \epsilon'_k \in \epsilon_kG,
\end{equation}
where we have used that $G$ is normal in $E$.  The differences of the actions on the two tensors  are represented graphically as follows:
  \begin{equation*} 
    \begin{tikzpicture}[baseline=-1mm, scale=1.2]
      \pic at (0,0,0) {3dpepsres};
      \draw (0,0,0) -- (0,0.35,0);
      \filldraw[draw= black,fill=purple] (0,0.2,0) circle (0.05);
      \node[anchor=east] at (0,0.3,0) {$U_{\epsilon_k}$};
    \end{tikzpicture} =
    \begin{tikzpicture}[scale=1.2]
      \draw (-0.7,0,0) -- (0.7,0,0);
      \draw (0,0,-0.9) -- (0,0,0.9);
      \pic at (0,0,0) {3dpepsres};
      \filldraw [draw= black,fill=red] (-0.5,0,0) circle (0.05);
      \filldraw [draw= black,fill=red] (0.5,0,0) circle (0.05);
      \node[anchor=south] at (-0.5,0,0) {$\myinv{{\epsilon}_k}$};
      \node[anchor=north] at (0.5,0,0) {${\epsilon}_k$};
      \filldraw [draw= black,fill=red] (0,0,-0.6) circle (0.05);
      \filldraw [draw= black,fill=red] (0,0,0.6) circle (0.05);
       \node[anchor=south] at (0,0,-0.6) {${\epsilon}_k$};
      \node[anchor=north] at (0,0,0.6) {$\myinv{{\epsilon}_k}$};
    \end{tikzpicture}
   \;\; \longleftrightarrow \;\;
      \begin{tikzpicture}[baseline=-1mm, scale=1.2]
      \pic at (0,0,0) {3dpeps};
      \draw (0,0,0) -- (0,0.35,0);
      \filldraw[draw= black,fill=purple] (0,0.2,0) circle (0.05);
      \node[anchor=east] at (0,0.3,0) {$U_{\epsilon_k}$};
    \end{tikzpicture} =
    \begin{tikzpicture}[scale=1.2]
      \pic at (0,0,0) {3dpeps};
    \end{tikzpicture}.
  \end{equation*}
The local symmetry allows us to write:
$$U_{\epsilon_k} A^{\rm{res}}_G =U_{\epsilon'_k} A^{\rm{res}}_G \quad {\rm if} \;  \epsilon'_kG= \epsilon_kG,$$
and also that $ U_{\epsilon_k} U_{\epsilon'_k} A^{\rm{res}}_G= U_{\epsilon_{kk'}} A^{\rm{res}}_G$. From Eq.(\ref{symAres}) it follows that concatenating the tensors $U_{\epsilon_k} A^{\rm{res}}_G$ the virtual operators $L_{\epsilon'_k}$ cancel out in the contracted legs. Combining this fact with the previous equations, the most general form of symmetry over the state reads
\begin{equation}\label{eq:symstateAres}
\bigotimes_{x\in \Lambda}  U_{\epsilon(x)} |\Psi_{A^{\rm{res}}_G}\rangle = |\Psi_{A^{\rm{res}}_G}\rangle  \;  {\rm if} \; \exists ! k \in Q:\; \epsilon(x) \in \epsilon_kG \; \forall x\in \Lambda,
\end{equation}
where $\Lambda$ denotes the set of lattice sites. The operators of the global symmetry forms a representation of the group $Q$, using the projection from $E$ to $E/G: \epsilon'_k\mapsto k$, on the state $|\Psi_{A^{\rm{res}}_G}\rangle$. Therefore, the $E$ gauge symmetry of  $|\Psi_{A_E}\rangle$ has been reduced to a $G$ gauge symmetry plus a $Q\cong E/G$ global symmetry on $|\Psi_{A_G^{\rm res}}\rangle$. Let us note that the transformation from Eq.(\ref{tensorE}) to Eq.(\ref{tensorGres}) gives us naturally the symmetry operators as the ones discarded in the sum. This finishes the proof of statement 1.

We now analyze the effect of the global symmetry on the anyons, proposing first the appropriate anyonic operators for the restricted tensor, and on the ground subspace ($|\Psi_{A_G^{\rm res}}\rangle$ is only one state of the ground subspace basis). To prove statement 2 we separate the analysis of fluxes, charges, dyons and ground subspace:

\begin{itemize}

\item {\bf Fluxes.} A pair of fluxes is represented as a string of $L_g$ and $L^\dagger_g$  operators, where $g\in G$, placed on the virtual d.o.f. of $|\Psi_{A_G^{\rm res}}\rangle$.  The class of the flux is determined by the conjugacy class of $g$: $[g]$. The string can be deformed freely due to the $G$-invariance of $A_G^{\rm res}$ except on the endpoints; where the excitations are placed. We apply the operator $U_{\epsilon_k}$, $ \epsilon_k \in E \setminus G$ to each lattice site on the state with the flux $[g]$. This action only changes the string from $g$ to $\epsilon'_kg{\epsilon'}^{-1}_k\in G$. Therefore, the operators also correspond to a symmetry in the space of flux excitations.  In fact, this action that we denote as $\phi_k(g)= \epsilon_kg \epsilon_k^{-1}$ can map a conjugacy class $[g]$ into another $[\phi_k(g)]$ if $E$ is non-abelian. This operation corresponds to a representation of the group $Q$ which permutes the class of the fluxes. Let us prove this. 
Given a conjugacy class of $G$ the action $\phi_k$ only depends on the coset $\epsilon_kG\equiv k\in Q$. Let us take $\epsilon_k$ and $\epsilon'_k\in \epsilon_kG$, we can write $\epsilon'_k=g'\epsilon_k$ for some $g'\in G$ and then $\phi'_{k}(g)=\epsilon'_kg{\epsilon'}^{-1}_k= g'\phi_k(g) g'^{-1}\in [\phi_k(g)]$. Also if $\tilde{g}\in [g]$ it follows that $\phi'_{k}(\tilde{g})=\epsilon'_k\tilde{g}{\epsilon'}^{-1}_k= \epsilon'_kg_1g g^{-1}_1{\epsilon'}^{-1}_k= g_2\epsilon'_kg {\epsilon'}^{-1}_k g^{-1}_2= g_2 g'\phi_q(g) (g_2g')^{-1}$ for some $g_1,g_2\in G$. In the same way it can be shown that $\phi_q \circ \phi_k(g')$ belongs to the same conjugacy class as $\phi_{qk}(g)$ for any $g'\in [g]$. The previous construction depends on the extension group $E$. In fact we have shown that the map $\phi$ characterizing this extension is recovered in the action over the fluxes.

\item  {\bf Charges.}
We propose as the virtual operator of a pair charge-anticharge the following:
$$\Pi^{\rm res}_{\sigma,t}=\sum^{G}_{g,h}\chi_{\sigma}(th^{-1}g)\left(\sum_{q\in Q}|g\epsilon_q\rangle\langle g \epsilon_q| \right)\otimes \left(\sum_{z\in Q}|h\epsilon_z\rangle\langle h \epsilon_z|\right),$$
where $\chi_{\sigma}$ is the character of the irrep $\sigma$ of $G$.
This operator is constructed in order to have the correct braiding properties with the flux operator $L_g$ where $g\in G \subset E$. $\Pi^{\rm res}_{\sigma,t}$ is invariant under conjugation by $L^{\otimes2}_g$ for all $g\in G$ which ensures a zero total charge. 
To study the effect of the symmetry operators over the charges it would be enough to proceed analogously to  section \ref{sect:anyons}. That is, analyze how the charge, modified by the symmetry operators, behaves under braiding with fluxes. Instead of that, we work explicitily with the action of the operator $U_{\epsilon_k}$ over an individual charge:
\begin{equation}\label{eq:defQres}
C_{\sigma}^t=  \sum_{g\in G}\chi_{\sigma}(tg) \sum_{q\in Q}|g\epsilon_q\rangle\langle g \epsilon_q|.
\end{equation}
The action on the charge is 
$$L_{\epsilon'_k}C_{\sigma}^t L^{\dagger}_{\epsilon'_k}=\sum_{g\in G}\chi_{\sigma}\circ\phi^{-1}_k(\phi_k(t){g'}^{-1}g) \sum_{q\in Q}|g\epsilon_k\epsilon_q\rangle\langle g \epsilon_k \epsilon_q|,$$
where we have decomposed $\epsilon'_k=g'\epsilon_k$. Clifford's Theorem \cite{Clifford37} establishes that given an irrep $\pi_{\sigma}$ of $G$, the operator $\pi_{\sigma}\circ \phi^{-1}_k$ corresponds to $\pi_{\sigma'}$ where $\sigma'\equiv \sigma' (\sigma, \epsilon^{-1}_k)$ is another irrep of $G$ and only depends on $k\in Q$. Therefore the action of the symmetry operator includes a permutation of the particle type of the charge (according to $\phi_k$):
$$L_{\epsilon'_k}C_{\sigma}^t L^{\dagger}_{\epsilon'_k}=\sum_{g\in G}\chi_{\sigma'}(t'g) \sum_{q\in Q}|g\epsilon'_q\rangle\langle g  \epsilon'_q|,$$
where $t'=\phi_k(t){g'}^{-1}$ and $\epsilon'_q=\epsilon_k\epsilon_q$ is just a relabeling which changes the representative of each coset. The map $\omega$ that characterizes the extension is defined as $\omega(q,k)\equiv \epsilon_k \epsilon_q\epsilon^{-1}_{kq}$ and if it is trivial, i.e. $\omega(q,k)=e$ for all $q,k\in Q$, then
$$L_{\epsilon'_k}C_{\sigma}^t L^{\dagger}_{\epsilon'_k}= C_{\sigma' }^{t'}.$$
We can also show the following:
$$ \label{eq:SF}L_{\epsilon'_k}L_{\epsilon'_q} C_{\sigma}^t L^{\dagger}_{\epsilon'_q}  L^{\dagger}_{\epsilon'_k}=  L_{\epsilon'_{kq}} C_{\sigma}^{t g_{kq}g_q^{-1} g_k^{-1} \omega^{-1}(k,q)} L^{\dagger}_{\epsilon'_{kq}}, 
$$
when $\phi$ is trivial. The virtual action of the operators $U_{\epsilon_k}U_{\epsilon_q}$ and $U_{\epsilon_{kq}}$ over $C_{\sigma}^t$ is related by the braiding with the flux $\omega(k,q)$ up to gauge redundancies ($g_k g_q g_{kq}^{-1}$). That is, the symmetry acts projectively on individual charges and this model realizes the  symmetry fractionalization class corresponding to $\omega$.

 \item {\bf Dyons.} The action of the operator $U_{\epsilon_k}$ over a state with a dyon is the conjugation by $L_{\epsilon'_k}$ on the virtual d.o.f. The dyon is associated with an irrep $\alpha$ of the normalizer of a conjugacy class ok $G$. If $C_h$ is this conjugacy class we choose as its representative $h$ and the associated normalizer $N_h$. Again $k_j$ will denote the representatives of the right cosets of $G$ by $N_h$ ($k_1=e,k_2,\cdots,k_\kappa$ where $\kappa=|G|/|N_h|$). With this notation the charge part of the restricted dyon can be associated with the following operator (at the end plaquette of a string of $L_h$ corresponding to the flux part):  
$$
\sum_{n \in N_{h}} \chi_\alpha(w n)\sum_{q\in Q} \sum^\kappa_{j=1} \ket{n k_j \epsilon_q}\bra{n k_j \epsilon_q},
$$
where $w$ belongs to $N_h$. We decompose $\epsilon'_k=g\epsilon_k$ and then $L_{\epsilon'_k}\ket{n k_j \epsilon_q}=\ket{g\phi_k(nk_j)\epsilon_k\epsilon_{q}}$, where $\phi_k(n)=\epsilon_kn\epsilon^{-1}_k$  which goes from $N_{h}$ to $N_{\phi_k(h)}$. If $\epsilon'_k\in G$ the action is equivalent to a braiding with one of the fluxes of the model: the charge part will transform equivalently as Eq.(\ref{eq:braidyon}) and the flux part will be $\epsilon'_kh{\epsilon'}^{-1}_k\in [h]$ (which does not change the flux type). The action of the conjugation over the charge part of the dyon is:
\beq \label{eq:restdyonsym}
\sum_{m \in N_{\phi_k(h)}} \chi_\alpha\circ \phi^{-1}_k( \phi_k(w) m)\sum_{q\in Q} \sum^\kappa_{j=1} \ket{g m \phi_k(k_j) \epsilon_k\epsilon_q}\bra{g m \phi_k(k_j) \epsilon_k\epsilon_q}
\eeq
and over the flux part is $\phi_k(h)$. If $\epsilon_k\in E\setminus G$ it can be the case that $[h]\neq [\phi_k(h)] $ and then the flux part has been permuted to another class. Let us see that the action of Eq.(\ref{eq:restdyonsym}) also describes a permutation in the charge part. It is clear that $\phi_k(N_h)=N_{\phi_k(h)}$ and also that the representatives of the cosets can be given by $\phi_k(k_j)$. Since $\phi_k$ is an automorphism of $G$ the normalizers of $[\phi_k(h)]$ and $[h]$ are isomorphic. Then by Clifford Theorem we can conclude that $\chi_\alpha\circ \phi^{-1}_k$ is the character of another irrep of the group $N_{[\phi_k(h)]}\cong N_{[h]}$.

\item {\bf Composition of excitations.} We have analyzed above the action of the symmetry over individual pairs of different classes of anyons. Here we show the general setting where a superposition of anyons is placed on the lattice. Let us denote as $\ket{\Psi_A(a^{[x_1]}_1,\cdots,a^{[x_N]}_N)}$ the PEPS associated with the anyon $a_j$ placed on the plaquette/bond $x_j$ (depending whether the anyon is a flux/charge) for all $j=1,\cdots, N$ where they do not overlap $x_j\neq x_i$. We require that the global topological charge is zero, otherwise $\ket{\Psi_A(a^{[x_1]}_1,\cdots,a^{[x_N]}_N)}=0$. This can be satisfied easily if, for each anyon, its antiparticle is in the superposition. We can conclude from the previous points that
$$ U^{\otimes n}_{\epsilon_k} \ket{\Psi_A(a^{[x_1]}_1,\cdots,a^{[x_N]}_N)}= \ket{\Psi_A(\phi_k(a_1)^{[x_1]},\cdots,\phi_k(a_N)^{[x_N]})}.$$
Since the symmetry goes from fluxes to fluxes and charges to charges, the energy does not change by acting with the symmetry: 
\begin{align*}
H U^{\otimes n}_{\epsilon_k} \ket{\Psi_A(a^{[x_1]}_1,\cdots,a^{[x_N]}_N)}&= H \ket{\Psi_A(\phi_k(a_1)^{[x_1]},\cdots,\phi_k(a_N)^{[x_N]})} \\
 \mathcal{E}_N \ket{\Psi_A(\phi_k(a_1)^{[x_1]},\cdots,\phi_k(a_N)^{[x_N]})} &= U^{\otimes n} H\ket{\Psi_A(a^{[x_1]}_1,\cdots,a^{[x_N]}_N)}
\end{align*}
This can be seen as the commutation between the symmetry operators and the Hamiltonian.

 \item {\bf Ground states.} The parent Hamiltonian corresponding to $A_G^{\rm res}$ has a degenerate ground subspace on the torus. This subspace is spanned by placing two non-contractible loops, one in each direction, of virtual operators $L_g$ and $L_h$ on the torus. We denote these states as $|\Psi\left[A_G^{\rm res}|(g,h)\right]\rangle$ with $gh=hg$ and two of theses states $(g,h)$ and $(g',h')$ are equivalent if there exists a $p\in G$ such that $g=pg'p^{-1}$ and $h=ph'p^{-1}$. If we apply the symmetry operator $U_{\epsilon_k}$ at each lattice site we obtain:
$$ U^{\otimes \Lambda}_{\epsilon_k}|\Psi\left[A_G^{\rm res}|(g,h)\right]\rangle = |\Psi\left[A_G^{\rm res}|(\phi_k(g),\phi_k(h))\right]\rangle, $$ 
which also belongs to the ground subspace because $\phi_k(h)\phi_k(g)=\phi_k(g)\phi_k(h)$ and then this action is a representation of $Q$ permuting the different ground states.
\end{itemize}

This finishes the proof of \cref{theo:consrep} since we have shown that the symmetry of $|\Psi_{A_G^{\rm res}}\rangle$ acts on the anyons and on the ground subspace via the maps $(\phi,\omega)$ (that are determined by $E$).

\section{Condensation and Confinement}
In this section we compare the excitations of both models. In particular we place the excitations of the parent Hamiltonian of $|\Psi (A_E)\rangle$ on a background of $A_G^{\rm res}$ tensors to study their topological properties. To do so we will use the parent Hamiltonian of $A_G^{\rm res}$ which is defined as $H^{\rm res}=\sum_{i\in\Lambda}\mathfrak{h}^{\rm res}_i$ where the local projector is $\mathfrak{h}^{\rm res}_i=\id-\Pi_{\mathcal{S}^{\rm res}_{2\times2}}$ and the subspace $\mathcal{S}^{\rm res}_{2\times2}$ is defined as:
 \begin{equation}   
  \mathcal{S}^{\rm res}_{2\times2}= \left \{  
  \Gamma^{\rm res}_{2\times 2}(B)=
 \begin{tikzpicture}
            \begin{scope}[canvas is zx plane at y=0]
                   \filldraw[pink]  (-0.6,-0.5) rectangle (1.1,1.3); 
                   \draw  (-0.6,-0.5) rectangle (1.1,1.3);
                   \filldraw[white]  (-0.5,-0.4) rectangle (1,0.9); 
                   \draw  (-0.5,-0.4) rectangle (1,0.9); 
            \end{scope}
         \pic at (0,0,0) {3dpepsres};
        \pic at (0,0,0.5) {3dpepsres};
         \pic at (0.5,0,0) {3dpepsres};
        \pic at (0.5,0,0.5) {3dpepsres};
    \node at (1.1,0,0.25) {$B$};
  \end{tikzpicture}
  |B\in \left( \mathbb{C}^{|E|}\right)^{\otimes 8}\right \}.
\end{equation}  
The operators corresponding to single excitations can behave differently depending on the background; they can no longer be associated with topological quasiparticles (showing confinement or condensation) or they have to be associated with a superposition of quasiparticles (i.e spliting) or two of them represent the same particle type (identification).

These behaviours has been previously studied as lattice models in Ref \cite{Bombin08}. We will analyze these behaviours for each type of excitation.

\begin{tcolorbox}
\begin{theorem}\label{fluxrestmodel}
The following statements hold
\begin{enumerate}
\item The flux excitations of the state formed with $A_E$ are also excitations of the parent Hamiltonian of $A^{\rm{res}}_G$. 
\item The fluxes associated with elements of $E$ that do not belong to $G$ cannot be moved using the $G$-invariance of $A^{\rm{res}}_G$, they become confined in that model. The string of confined fluxes cannot be extended freely through the lattice of $A^{\rm{res}}_G$ and the energy penalty of the excitation depends on the length of the string. 
\item Some of the fluxes corresponding to $A_E$ that are not confined can be split in a superposition of fluxes of $G$: they are no longer simple anyons.
 \end{enumerate}
\end{theorem} 
\end{tcolorbox}

\begin{proof}

We calculate explicitly the scalar product between an arbitrary element of $\mathcal{S}^{\rm res}_{2\times2}$ and a $2\times 2$ subset of the lattice containing part of a flux of the state formed with $A_E$. If the result is zero, the state with a string is locally orthogonal to $\mathcal{S}_{2\times2}$ and then an excitation of $\mathfrak{h}^{\rm res}_i$.

First, let us place a string of $L_g$ operators on $|\Psi (A_G^{\rm res})\rangle$ where $g\in E \setminus G$ and check whether the middle part of the string is an excitation of $\mathfrak{h}^{\rm res}_i$. That is,

 \begin{equation}    
  \mathfrak{h}^{\rm res}
 \begin{tikzpicture}
            \begin{scope}[canvas is zx plane at y=0]
                   \filldraw[pink]  (-0.6,-0.5) rectangle (1.1,1.3); 
                   \draw  (-0.6,-0.5) rectangle (1.1,1.3);
                   \filldraw[white]  (-0.5,-0.4) rectangle (1,0.9); 
                   \draw  (-0.5,-0.4) rectangle (1,0.9); 
            \end{scope}
            \draw[red] (-0.3,0,0.25)--(0.8,0,0.25);
         \pic at (0,0,0) {3dpepsres};
        \pic at (0,0,0.5) {3dpepsres};
         \pic at (0.5,0,0) {3dpepsres};
        \pic at (0.5,0,0.5) {3dpepsres};
    \node at (1.1,0,0.25) {$B$};
    	\filldraw [draw= black,fill=red]  (0,0,0.25) circle (0.04);
	\filldraw [draw= black,fill=red]  (0.5,0,0.25) circle (0.04);
  \end{tikzpicture}
  =
   \begin{tikzpicture}
            \begin{scope}[canvas is zx plane at y=0]
                   \filldraw[pink]  (-0.6,-0.5) rectangle (1.1,1.3); 
                   \draw  (-0.6,-0.5) rectangle (1.1,1.3);
                   \filldraw[white]  (-0.5,-0.4) rectangle (1,0.9); 
                   \draw  (-0.5,-0.4) rectangle (1,0.9); 
            \end{scope}
            \draw[red](-0.3,0,0.25)--(0.8,0,0.25);
         \pic at (0,0,0) {3dpepsres};
        \pic at (0,0,0.5) {3dpepsres};
         \pic at (0.5,0,0) {3dpepsres};
        \pic at (0.5,0,0.5) {3dpepsres};
    \node at (1.1,0,0.25) {$B$};
    	\filldraw [draw= black,fill=red]  (0,0,0.25) circle (0.04);
	\filldraw [draw= black,fill=red]  (0.5,0,0.25) circle (0.04);
  \end{tikzpicture}
  ,\; \forall g\in E\setminus G.
\end{equation} 
We can express the scalar product pictorially as follows:
\begin{equation}\label{eq:scapro}
 \left \langle 
 \begin{tikzpicture}
            \begin{scope}[canvas is zx plane at y=0]
                   \filldraw[pink]  (-0.6,-0.5) rectangle (1.1,1.3); 
                   \draw  (-0.6,-0.5) rectangle (1.1,1.3);
                   \filldraw[white]  (-0.5,-0.4) rectangle (1,0.9); 
                   \draw  (-0.5,-0.4) rectangle (1,0.9); 
            \end{scope}
         \pic at (0,0,0) {3dpepsres};
        \pic at (0,0,0.5) {3dpepsres};
         \pic at (0.5,0,0) {3dpepsres};
        \pic at (0.5,0,0.5) {3dpepsres};
    \node at (1.1,0,0.25) {$X$};
  \end{tikzpicture}
 \Bigg |
   \begin{tikzpicture}
            \begin{scope}[canvas is zx plane at y=0]
                   \filldraw[pink]  (-0.6,-0.5) rectangle (1.1,1.3); 
                   \draw  (-0.6,-0.5) rectangle (1.1,1.3);
                   \filldraw[white]  (-0.5,-0.4) rectangle (1,0.9); 
                   \draw  (-0.5,-0.4) rectangle (1,0.9); 
            \end{scope}
            \draw[red](-0.3,0,0.25)--(0.8,0,0.25);
         \pic at (0,0,0) {3dpepsres};
        \pic at (0,0,0.5) {3dpepsres};
         \pic at (0.5,0,0) {3dpepsres};
        \pic at (0.5,0,0.5) {3dpepsres};
    \node at (1.1,0,0.25) {$Z$};
    	\filldraw [draw= black,fill=red]  (0,0,0.25) circle (0.04);
	\filldraw [draw= black,fill=red]  (0.5,0,0.25) circle (0.04);
  \end{tikzpicture}
 \right \rangle= 
 \begin{tikzpicture}
             \begin{scope}[canvas is zx plane at y=0]
                   \filldraw[pink]  (-1,-1) rectangle (1,1); 
                   \draw  (-1,-1) rectangle (1,1);
                   \filldraw[white]  (-0.9,-0.9) rectangle (0.9,0.9); 
                   \draw  (-0.9,-0.9) rectangle (0.9,0.9); 
                                      \draw (-0.9,-0.3)--(-0.6,-0.3);
                                      \draw (-0.9,0.3)--(-0.6,0.3);
                                        \draw (0.9,-0.3)--(0.6,-0.3);
                                      \draw (0.9,0.3)--(0.6,0.3);
                                      \draw (-0.3,-0.9)--(-0.3,-0.6);
                                      \draw (0.3,-0.9)--(0.3,-0.6);
                                          \draw (-0.3,0.9)--(-0.3,0.6);
                                      \draw (0.3,0.9)--(0.3,0.6);
                    \filldraw[pink]  (-0.7,-0.7) rectangle (0.7,0.7); 
                   \draw  (-0.7,-0.7) rectangle (0.7,0.7);
                   \filldraw[white]  (-0.6,-0.6) rectangle (0.6,0.6); 
                   \draw  (-0.6,-0.6) rectangle (0.6,0.6); 
            \end{scope}
                \node at (0,0,0.3) {$Z$};
                \node at (-1.2,0,-0.5) {$X$};
  \end{tikzpicture}
  \times
   \begin{tikzpicture}
   \draw  (-0.3,-0.3,-0.5) rectangle (0.3,0.3,-0.5); 
              \begin{scope}[canvas is yz plane at x=0.3]
             \draw[preaction={draw, line width=1pt, white}] (-0.3,-0.45) rectangle(0.3,0.45);
              \end{scope}
             \begin{scope}[canvas is yz plane at x=-0.3]
             \draw[preaction={draw, line width=1pt, white}]  (-0.3,-0.45) rectangle(0.3,0.45);
              \end{scope}
      \draw [preaction={draw, line width=1pt, white}] (-0.3,-0.3,0.5) rectangle (0.3,0.3,0.5);
       \filldraw [draw= black,fill=red]  (0.3,-0.3,0) circle (0.05);
       \filldraw [draw= black,fill=red]  (-0.3,-0.3,0) circle (0.05);
    \end{tikzpicture}
 \end{equation}
where we have separated the boundary term and the inside term in virtue of the tensor product form of $A^{\rm{res}}_G$ (see Eq.(\ref{tensorGres})). For the sake of simplicity, we have omitted the sum running over the elements of $G$ in all tensors forming the construction. The scalar product of Eq.(\ref{eq:scapro}) is proportional to 
$$\sum^{G}_{a,b,c,d} \mathcal{C}(X,Z,a,b,c,d) \chi_{L}(a^{-1}b) \chi_{L}(b^{-1}gc) \chi_{L}(c d^{-1}) \chi_{L}(a^{-1}gd),$$
where each element of the sum comes from the individual tensors,  $\mathcal{C}(X,Z,a,b,c,d)$ is the boundary term and the traces come from the loops of the last drawing containing the operator of the string $L_g$. The left regular representation obeys $\chi_{L}(g)=|E|\delta_{e,g}$ so the scalar product \cref{eq:scapro} is proportional to $\sum_{h\in G} \mathcal{C} \delta_{h,g}$ which is zero for $g\notin G$. The previous calculation is not valid for $g \in G$. In fact, for $g\in G$ the middle part of the string is not locally detectable so it is not an excitation. However, with a similar calculation, we can verify that the operator, on the virtual d.o.f., corresponding to an end of the string gives rise to an eigenvector of $H^{\rm res}$:
 \begin{equation}    
  \mathfrak{h}^{\rm res}
 \begin{tikzpicture}
            \begin{scope}[canvas is zx plane at y=0]
                   \filldraw[pink]  (-0.6,-0.5) rectangle (1.1,1.3); 
                   \draw  (-0.6,-0.5) rectangle (1.1,1.3);
                   \filldraw[white]  (-0.5,-0.4) rectangle (1,0.9); 
                   \draw  (-0.5,-0.4) rectangle (1,0.9); 
            \end{scope}
            \draw[red](-0.3,0,0.25)--(0,0,0.25);
         \pic at (0,0,0) {3dpepsres};
        \pic at (0,0,0.5) {3dpepsres};
         \pic at (0.5,0,0) {3dpepsres};
        \pic at (0.5,0,0.5) {3dpepsres};
    \node at (1.1,0,0.25) {$B$};
    	\filldraw [draw= black,fill=red]  (0,0,0.25) circle (0.04);
  \end{tikzpicture}
  =
   \begin{tikzpicture}
            \begin{scope}[canvas is zx plane at y=0]
                   \filldraw[pink]  (-0.6,-0.5) rectangle (1.1,1.3); 
                   \draw  (-0.6,-0.5) rectangle (1.1,1.3);
                   \filldraw[white]  (-0.5,-0.4) rectangle (1,0.9); 
                   \draw  (-0.5,-0.4) rectangle (1,0.9); 
            \end{scope}
            \draw[red](-0.3,0,0.25)--(0,0,0.25);
         \pic at (0,0,0) {3dpepsres};
        \pic at (0,0,0.5) {3dpepsres};
         \pic at (0.5,0,0) {3dpepsres};
        \pic at (0.5,0,0.5) {3dpepsres};
    \node at (1.1,0,0.25) {$B$};
    	\filldraw [draw= black,fill=red]  (0,0,0.25) circle (0.04);
  \end{tikzpicture}
  ,\; \forall g\in E.
\end{equation} 
This proves point 1.

A string of $L_g$ operators, where $g\in E \setminus G$, cannot be deformed using the $G$-invariance of the tensors. Instead, the element $g$ of the operator $L_g$ can be sent to $h'gh$ (and then to the operator $L_{h'gh}$), where $h,h'\in G$, applying the $G$-invariance of the tensors for every virtual edge. This transformation cannot change the coset $g$ belongs to because $h'gh\in gG$ ($G$ is normal in $E$). Thus, an element $g\in E \setminus G$ cannot be transformed into one belonging to the trivial coset $G$, so the operator cannot be moved freely from its position. This fact shows that these strings can be detected locally in contrast with a string formed with operators $L_g$, where $g \in G$, that can be deformed freely. That is, these fluxes become confined.

The fact that the state with these operators placed on the virtual d.o.f. is an eigenstate with eigenvalue one of the local hamiltonian $\mathfrak{h}^{\rm res}$ and that a string of this kind cannot be deformed (using the $G$ invariance) through the lattice is equivalent to a string tension. The string tension is manifested in the fact that there exists an energy dependence on the string's length of the excitation. That is:
\begin{equation*} 
H^{\rm res} \left(  |\Psi_\ell^*(A^{\rm{res}}_G,L_g)\rangle\right) \propto \ell |\Psi_\ell^*(A^{\rm{res}}_G,L_g)\rangle,\quad {\rm with}  \; g\in E \setminus G,
\end{equation*}
where $|\Psi_\ell^*(A^{\rm{res}}_G,L_g)\rangle$ is the state constructed with the tensor $A^{\rm{res}}_G$ and placing a string flux of length $\ell$ with element $g\in E \setminus G$ in the virtual d.o.f. When $g$ belongs to $G$, the scalar product is not zero so we cannot conclude that it is an excitation. In fact the operator in the virtual d.o.f. can be cancelled out using the $G$-invariance of the tensor and it then corresponds to a string not locally detectable. 

This shows that we have two very different flux string excitations in the state formed with the tensor $A^{\rm{res}}_G$. The open string created by placing operators $L_g$ with $g\in G$ in the virtual d.o.f. has the freedom of being deformed in the entire chain except at the ending plaquettes. Therefore the energy penalty of this excitation comes from the two end points, regardless of the length of the string. As opposed to this case we have the string constructed with elements $g\in E\setminus G$ which cannot be deformed freely; the operators $L_g$ are confined in its position of the lattice. Therefore the energy penalty of this chain depends on the length of the string. We say that all the flux-type particles of the parent model, conjugacy classes of $G$, which do not belong to $G$ become confined fluxes in the state formed with $A^{\rm{res}}_G$. This finishes the proof of the statement 2.

\begin{figure}[ht!]
\begin{center}
\includegraphics[scale=1.5]{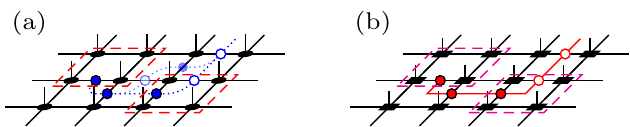}
\caption{A subset of the lattice, in a background of $A_E$ (a) and $A^{\rm res}_G$ (b)  tensors, having a flux excitation string from $E\setminus G$. The dots, blue or red depending on the background, represent the operator $L_{g}$ acting on the virtual d.o.f. of the tensors. Each colored square drawn represents the place where the corresponding local hamiltonian, $\mathfrak{h}$ on (a) and $\mathfrak{h}^{\rm res}$ on (b), acts.  In (a) only the end of the string is an excitation of the local hamiltonian $\mathfrak{h}$: adds $+1$ to the total energy of the state. While in (b) both plaquettes are excitations of the local hamiltonian $\mathfrak{h}^{\rm res}$: adds $+1$ to the total energy of the state.}
\label{fig:string}
\end{center}
\end{figure}

The fluxes that are not confined, where $g\in G$, are called \emph{deconfined} and they can split into a superposition of fluxes of the restricted model. This is because a conjugacy class of $E$ belonging to $G$ can be decomposed into multiple conjugacy classes of $G$.  We denote by $[g]^E$ the conjugacy class of $g$ in $E$. Different internal states of the same type of flux on the parent Hamiltonian of $A_E$, $h,p\in [g]^E \subset G$, should be considered now different types of flux of the restricted model if there is no  element $x\in G$ such that $h=xpx^{-1}$ and then, $h\in [{g_1}]^{G}$, $p\in [{g_2}]^{G}$ with $[{g_1}]^{G} \neq [{g_2}]^{G}$. This can be seen at the level of the creation operators; 
$$\sum_{g'\in [g]^E  }L_{g'}\otimes L_{g'}= \sum_{g_i}\left( \sum_{g'_i \in [g_i]^G } L_{g'_i}\otimes L_{g'_i} \right), $$
where the $g_i$'s runs over the representatives of each $G$-conjugacy class inside the $E$-conjugacy class of $g$. This proves point 3 which finishes the proof of \cref{fluxrestmodel}.
\end{proof}

\begin{tcolorbox}
\begin{theorem}\label{chargesrestmodel}
The following statements hold
\begin{enumerate}
 \item The charges of the state formed with $A_E$ are also quasiparticle excitations of the parent Hamiltonian of $A^{\rm{res}}_G$. 
 \item There is always a charge, or a superposition of charges, of $A_E$ that is invariant under braiding with any flux of $A^{\rm{res}}_G$. This charge has been condensed. 
 \item Also some of the charges of $A_E$ can split in a superposition of charges of $A^{\rm{res}}_G$.
 \end{enumerate}
\end{theorem} 
\end{tcolorbox}
\begin{proof}
Let us proof \cref{chargesrestmodel}  by studying the properties of the charge excitations of $A_E$ when they are placed on a background formed by $A^{\rm{res}}_G$ tensors. The single charge virtual operator is $\sum_{\epsilon \in E}\chi_\nu(p\epsilon)|\epsilon\rangle\langle \epsilon|$, where $\chi_\nu$ is the character of the irrep $\nu$ of $E$ \cite{Schuch10}. Unlike flux excitations, charge particles are not associated with a string (in this PEPS representation) and they are point defects in the lattice constructed with $G$-isometric PEPS. Therefore, although they are always created in pairs, we can safely just consider how one particle of the pair affects the background. These operators also correspond to excitations when they are placed in the state constructed with $A^{\rm{res}}_G$. To prove it we compute the scalar product between a state with this operator in the virtual d.o.f. and one arbitrary element of $\mathcal{S}^{\rm res}_{2\times2}$ as we did in the case of flux excitations. We obtain a result proportional to 
$$ \sum^{G}_{a,b,c,d} \mathcal{C}(X,Z,a,b,c,d)  \chi_{L}(a^{-1}b)\chi_{L}(bc^{-1}) {\rm tr} \left[L_{c  d^{-1}}  \sum_{\epsilon \in E}\chi_\nu(p\epsilon)|\epsilon\rangle\langle \epsilon| \right]\chi_{L}(ad^{-1}),$$
which reduces to
$${\rm tr}\left[\sum_{\epsilon \in E}\chi_\nu(p\epsilon)|\epsilon\rangle\langle \epsilon| \right]= \sum_{\epsilon \in E}\chi_\nu(p\epsilon)=\sum_{\epsilon \in E}\chi_\nu(\epsilon)= \sum_{\epsilon \in E}\chi_\nu(\epsilon)\chi_1(\epsilon)=|E|\delta_{\nu,1}.$$
Therefore states that contain plaquettes formed by restricted tensors with one charge operator placed on the virtual d.o.f. are orthogonal to the ground state of the parent Hamiltonian associated with $A^{\rm{res}}_G$. As for the fluxes above, this implies that they are eigenstates with eigenvalue $1$ of $\mathfrak{h}^{\rm res}$ for each neighbouring plaquette and then, as there is no string associated, of the parent Hamiltonian. This proves statement 1.

The confinement of some fluxes is intimately related to the condensation of charge particles. An anyon condensation is a situation where a topological excitation cannot be distinguished from the vacuum with topological interactions (i.e. using braiding operations). Since some fluxes are confined in the $A_G^{\rm res}$ background, there are less fluxes 'available' (the remaining deconfined) to braid with the charge particles to topologically distinguish amongst them.
Let us consider an elementary charge of $A_E$ and try to identify the class of this charge in the restricted model. To perform this experiment we have to create a charge-pair excitation belonging to the class of the irrep $\nu$ of $E$, braid one charge of the pair with a flux characterized by the element $g\in G$ and then fuse the charge modified by the braiding with the other charge of the pair. The probability the charge pair fuses back to the vacuum is given by \cite{Preskill04, Schuch10} 
\begin{equation}
{\rm Prob(vacuum)}=\left | \frac{\chi_\nu(g)}{|\nu|} \right |^2.\notag
\end{equation}
This interferometric process can be conceived as a method to identify the irrep associated with the charge particle. To completely identify the irrep $\nu$ of $E$ with an $A_G^{\rm res}$ background, we would need to braid charges with confined fluxes also. We forbid this operation because we are restricting to topological interactions; we do not allow  processes whose energy cost depend on lengths. If two irreps of $E$ are the same when restricted to elements of $G$; they have to be considered equivalent in the $A_G^{\rm res}$ background. Then if one irrep, say $\mu$, is the identity for all elements of $G$, we obtain that the probability to fuse it with the vacuum after any available braiding is one. Therefore this charge is not modified by the braiding of any of the deconfined fluxes, and we now have to identify it with the trivial topological charge (the vacuum); we will call this phenomenon 'charge condensation'. 

We now claim that there is always a charge excitation of $A_E$ that is condensed in the $A_G^{\rm res}$ background. This charge excitation does not need to be elementary, i.e. associated with an irrep of $E$. It can be a composite of elementary charges (a reducible representation of $E$). We know that given any normal subgroup $G$ of $E$ there is always a representation $\rho$ of $E$ such that the kernel of the corresponding character is exactly $G$ \cite{Isaac76}. Therefore, if we perform an interferometric experiment with $\rho$ we obtain
\begin{equation}
{\rm Prob(vacuum)}=\left | \frac{\chi_\rho(g)}{|\rho|} \right |^2=1, \notag
\end{equation}
for all $g\in G$, that is, the charge associated with $\rho$ has condensed and then statement 2 is proven.

We conclude by proving point 3, which concludes the proof of \cref{chargesrestmodel}. We show that the charges of the model $A_E$ can split into a superposition of charges of the restricted model. As mentioned before, we only allow to braid them with the fluxes corresponding to the elements of $G$. Then, the irreducibility of $\nu$ in $E$ is broken to a superposition of irreps of $G$:
$$\chi^E_{\nu}(g)\cong\sum_{\sigma} m_{\sigma} \chi^{G}_{\sigma}(g),$$
where $m_{\sigma}$ is the multiplicity of the irrep $\sigma$ of $G$ that appears in the decomposition of  the irrep $\nu$ of $E$.

\end{proof}

\begin{tcolorbox}
\begin{corollary}
Similarly to fluxes and charges, dyons of the state formed with $A_E$ are also quasiparticle excitations of the parent Hamiltonian of $A^{\rm{res}}_G$.  They can also be confined or split in simple dyons of $A^{\rm{res}}_G$. 
\end{corollary}
\end{tcolorbox}
\begin{proof}
A dyon of $E$ with the flux part corresponding to a conjugacy class not in $G$, is an excitation using \cref{fluxrestmodel}. It is clear that the chain cannot be moved using the $G$-invariance so the dyon is confined.

Let us analyze the more involved case where the parent dyon is unconfined and its flux and charge parts split. Let $[h]^E=\{h_i, i=1,\cdots, \kappa \}$ be a conjugacy class of $E$ which is in $G$. This conjugacy class can be decomposed in conjugacy classes of $G$: $[h]^E= \cup_j [h_j]^{G}$ where $j$ only runs over the indices corresponding to the elements $h_i$ with disjoint conjugacy classes of $G$. Take now a representative element $h_j$ of $[h]^E$ and denote its normalizer in $E$ as $N^E_{h_j}=\{n\in E | nh_j=h_jn\}$. Trivially $N^{G}_{h_j}=\{k\in G | kh_j=h_jk\}$ is a subgroup of $ N^E_{h_j}$. It is also normal: $(nkn^{-1})h_j=h_j(nkn^{-1})$ $\forall n\in N^E_{h_j}$ and $\forall k\in N^{G}_{h_j}$. Therefore $N^{G}_{gh_jg^{-1}}$ is normal in $N^E_{gh_jg^{-1}}$. By Clifford's Theorem \cite{Clifford37} the irreps of $N_{[h]^E}$ will decompose into a direct sum with equal multiplicity of irreps of $N_{[h]^{G}}$, all of them related by conjugation.

This describes the splitting of the charge part of an unconfined parent dyon into dyons of $A^{\rm{res}}_G$ that is, the unconfined dyon is an excitation. We note that this also describes qualitatively the action of the symmetry over the charge part of a dyon of the restricted model. It is clear that any two of these conjugacy classes of $G$ can be related by conjugation with an element of $E$ and vice versa. This fact is what is causing the splitting of the flux part of an unconfined parent dyon and the action of the symmetry over the flux part of a dyon of the restricted model. 
\end{proof}


We finish this section with a illustrative example for the groups $G= \mathbb{Z}_n ,Q= \mathbb{Z}_2$.\\

Let us consider the dihedral group $G=D_{n}$, with $n$ odd, $\{ r,s | r^n=s^2=e, sr^ks=r^{-k}\}$. There are $1+1+(n-1)/2$ conjugacy classes: $[e], [s]=\{sr^{k}\}_{k=0}^{n-1}, [r^k]=\{r^k,r^{n-k}\}\;k=1,\ldots, n-1$ which correspond to fluxes in the quantum double of $D_{n}$: $\mathcal {D}(D_{n})$. There are two one-dimensional representation of $D_{n}$; the trivial and the irrep $Z$ given by $Z(r^k)=1$ and $Z(sr^{k})=-1 \; \forall k=0,\ldots,n-1$. There are $n-1$ two-dimensional irreps $\Pi_\nu$, labelled by $\nu=1,\cdots, n-1$ given by
\beq
\Pi_\nu(r^k)=\left(\begin{array}{cc}q^{ k\nu}& 0 \\ 0 & q^{- k\nu} \end{array}\right), \; \Pi_\nu(s)=\left(\begin{array}{cc}0&1 \\ 1 & 0 \end{array}\right),\notag
\eeq
where $q=e^{\frac{2\pi i}{n}}, \; k=0,\cdots, n-1$.
The normal subgroup $\{ r^k \}_{k=0}^n$ of $D_{n}$ is isomorphic to the abelian group $\mathbb{Z}_n$. The fluxes of $\mathcal {D}(\mathbb{Z}_n)$ are the $n$ conjugacy classes given by single elements: $\{r^k\}$ for $k=0,\ldots,n-1$. The charges correspond to the $n$ one-dimensional irreps, given by $\pi_\sigma(r^k)= q^{k\sigma}$ where $\sigma=0,\ldots,n-1$. The restriction of fluxes from $D_{n}$ to $\mathbb{Z}_n $ can be associated with a confinement of the flux $[s]$ and a condensation of the charge $Z$ (because $Z(r^k)=1$) of $\mathcal {D}(D_{n})$. The splitting of charges of $\mathcal {D}(D_{n})$ is given by $\Pi_\nu \cong \pi_\sigma \oplus \pi_{-\sigma}$ for $\nu=\sigma$. The associated global symmetry comes from the quotient group $\mathbb{Z}_2\cong D_{n} / \mathbb{Z}_n$. The non-trivial action of the symmetry on the anyons of $\mathcal {D}(\mathbb{Z}_n)$ is given by:
\beq
r^k\mapsto sr^ks=r^{-k}, \; \pi_\sigma(r^k)\mapsto  \pi_\sigma(sr^ks)=  \pi_{-\sigma}(r^k). \notag
\eeq
This example corresponds to the non-trivial extension group of $\mathbb{Z}_n$ by $\mathbb{Z}_2$. One also can consider the restriction from $\mathcal {D}(\mathbb{Z}_n\times \mathbb{Z}_2 )$ (the trivial extension) to $\mathcal {D}(\mathbb{Z}_n)$ where the corresponding condensation and confinement can be identified, see \cite{Bais03}, but the symmetry action is trivial on the anyons.

\subsection{Symmetry reduction induced by anyon condensation}\label{sec:symredcon}

In this subsection we make the connection between some of the properties previously analyzed. First the flux confinement and the emergence of a non-trivial on-site symmetry studied before (see Eq.(\ref{symAres}) and Eq.(\ref{eq:symstateAres})) are inextricably linked features. 
They both come from a reduction of the local symmetry of the tensor.
This is achieved by removing complete conjugacy classes of the local symmetry group of the tensor from $E$ to $G$. Because conjugacy classes are related to the fluxes of the underlying topological model, we are performing an effective confinement of fluxes which naturally gives rise to a charge condensation. Now two facts allow to realize a global symmetry. First, that the relation between the physical and virtual levels is an isometry. Second, $G$ is a normal subgroup of $E$, so that the action by conjugation of elements in $E\setminus G$ does not leave the relevant subspace. Therefore the operators associated with the confined fluxes are the ones used to construct the local symmetry operator of the emergent global symmetry. It is important to note that the restriction does not change the representation; this allows to compare the same operators in both models. Moreover the fact that the confined fluxes effectively represent the operators of the symmetry and that the action is given by conjugation means that the effect is equivalent to a braiding with these confined fluxes. Then this explains the SF effect in the charge sector.

When we apply the operator $U_{\epsilon_k}$ with $\epsilon_k \in E \setminus G$ on a connected subset of the lattice ($\mathcal{M}$) in a background of $A^{\rm{res}}_G$, we create a closed loop of $L_{\epsilon'_k}$ virtual operators around the affected region: 

\begin{equation*}
  \begin{tikzpicture}
     \foreach \z in {0,0.7,1.4,2.1,2.8}{
      \foreach \x in {0,0.5,1,1.5,2}{
           \pic at (\x,0,\z) {3dpepsres}; 
             } } 

	 \filldraw [draw=black, fill=purple] (0.5,0.13,0.7) circle (0.04);
	 \filldraw [draw=black, fill=purple] (0.5,0.13,1.4) circle (0.04);
	 \filldraw [draw=black, fill=purple] (0.5,0.13,2.1) circle (0.04);
	 \filldraw [draw=black, fill=purple] (1,0.13,0.7) circle (0.04);
	  \filldraw [draw=black, fill=purple] (1,0.13,1.4) circle (0.04);
	   \filldraw [draw=black, fill=purple] (1,0.13,2.1) circle (0.04);
	   	 \filldraw [draw=black, fill=purple] (1.5,0.13,0.7) circle (0.04);
	  \filldraw [draw=black, fill=purple] (1.5,0.13,1.4) circle (0.04);
	   \filldraw [draw=black, fill=purple] (1.5,0.13,2.1) circle (0.04);
	   
 \end{tikzpicture} =
   \begin{tikzpicture}
   
    \begin{scope}[canvas is zx plane at y=0]
 \draw[epsilon] (0.35,1.75) rectangle (2.45,0.25);
 \end{scope}
    \foreach \z in {0,0.7,1.4,2.1,2.8}{
      \foreach \x in {0,0.5,1,1.5,2}{
           \pic at (\x,0,\z) {3dpepsres}; 
             } } 
 \filldraw [draw=black, fill=red] (0.5,0,0.35) circle (0.04);   
 \filldraw [draw=black, fill=red] (1.5,0,0.35) circle (0.04); 
  \filldraw [draw=black, fill=red] (1,0,0.35) circle (0.04);
   \filldraw [draw=black, fill=red] (0.5,0,2.45) circle (0.04);   
   \filldraw [draw=black, fill=red] (1.5,0,2.45) circle (0.04);
      \filldraw [draw=black, fill=red] (1,0,2.45) circle (0.04);
       \filldraw [draw=black, fill=red] (0.25,0,0.7) circle (0.04); 
       \filldraw [draw=black, fill=red] (0.25,0,1.4) circle (0.04); 
       \filldraw [draw=black, fill=red] (0.25,0,2.1) circle (0.04); 
          \filldraw [draw=black, fill=red] (1.75,0,0.7) circle (0.04); 
       \filldraw [draw=black, fill=red] (1.75,0,1.4) circle (0.04); 
       \filldraw [draw=black, fill=red] (1.75,0,2.1) circle (0.04); 
 \end{tikzpicture}
  \end{equation*} 
 We will show the following:
 \begin{tcolorbox}
 \begin{proposition}
 The energy of the restricted state modified by  $U^{\otimes \mathcal{M}}_{\epsilon}$, over a compact region $ \mathcal{M}$, depends on the length of the boundary of the region $ \mathcal{M}$, when $\epsilon \in E\setminus G$. The virtual representation of this action is a confined loop defect which only depends on the quotient group $Q\cong E/G$. 
 \end{proposition}
 \end{tcolorbox}
\begin{proof}
The state with this closed string is an excitation of the parent Hamiltonian of $A^{\rm{res}}_G$ because it is formed with the same virtual operators as the confined fluxes. Then the virtual representation of the unitary $U_{\epsilon_k}$ over a connected region is a confined closed string. Therefore the energy of the state with the confined closed string depends on the length of the loop:
 \begin{equation}
H^{\rm res} \left(U_{\epsilon_k}^{\otimes \mathcal{M}}|\Psi_{A^{\rm{res}}_G}\rangle\right) \propto |\partial \mathcal{M}| \left(U_{\epsilon_k} ^{\otimes \mathcal{M}}|\Psi_{A^{\rm{res}}_G}\rangle \right),\notag
\end{equation}
where $k\neq e$. We notice that we can obtain the same defect acting on the complementary region of $\mathcal{M}$. The fact that these defects cannot be deformed freely through the lattice imply that the state does not remain invariant under the action of $U_{\epsilon_k}$  over $\mathcal{M}$. Unlike the unmodified model situation \cite{Kitaev03}, the action of closed loops of operators on the ground state may increase energy. 

In order to leave the restricted state invariant, i.e., remove the confined closed string, we have to act over the complementary of $\mathcal{M}$. That is, we need to act on the whole lattice to preserve the state, ending up with a global on-site symmetry of the group $Q\cong E/G$. 
\end{proof}

\section{Application to MPS: the 1D SPT  classification}\label{sec:phasesMPS}

We now wish to show how the notion of group extension discussed in the previous sections is also useful to analyse 1D systems with symmetries. Namely, it allows for a transparent derivation of the classification of 1D phases in the MPS formalism. We start with a short summary of this classification \cite{Chen11,Schuch11, Pollmann10,Fidkowski11}. Our main focus will be on symmetry protected topological (SPT) phases. 

\subsection{Overview of SPT classification in 1D} \label{secclass}

Formally, two systems are said to be in the same quantum phase if they can be connected by a smooth path of Hamiltonians, which is local, bounded-strength and uniformly gapped\footnote{One requires that the gap of the Hamiltonian along the path is uniformly lower bounded by a constant in the size of the system. This requirement ensures that the gap is preserved in the thermodynamic limit.}. Along this path the physical properties of the system will change smoothly. If at some point the gap closes, it may result in a change of the global properties and usually a phase transition will occur. When a symmetry is imposed on Hamiltonians and the paths connecting them, phase diagrams become richer; two systems are then said to be in the same phase if the previous path exists and if, moreover, there exists a representation of the symmetry which commutes with the Hamiltonian along the entire path. An example of a system with non-trivial SPT phase is the celebrated spin-1 Haldane chain \cite{Haldane83}.

The classification of quantum gapped phases can be restricted to the task of classifying {MPS}. This is justified, as it was already mentioned in \cref{chapter:Intro}, because it has been proven that the family of  {MPS} approximate efficiently ground states of gapped quantum Hamiltonians \cite{Hastings07A,Hastings07B,Arad13}\footnote{This reduction, mathematically speaking, is still open}. 
And for any MPS an associated parent Hamiltonian can be constructed. The classification can be further restricted to the so-called {\it isometric form} of an  {MPS}: those  {MPS} which are renormalization fixed points.  The reason is that, as shown in \cite{Schuch11}, a gapped path of Hamiltonians can be constracted connecting {\it any} MPS with its corresponding isometric form. The final step is to identify the obstructions to design gapped paths of Hamiltonians between the different isometric forms. 

The main conclusions of \cite{Chen11,Schuch11, Pollmann10} are: (i) without symmetries, all systems with the same ground state degeneracy are in the same phase, where the representative states are the product state for the unique ground state case and the GHZ state for the degenerate case. (ii) When on-site linear symmetries are imposed to the systems, the different phases are classified, in the unique ground state case, by the second cohomology group $H^2(G,U(1))$ of  the symmetry group $G$ over $U(1)$. This classification is best understood if one considers the virtual d.o.f. of the  {MPS}: the unitary representation $U_g$ realising the physical global symmetry translates into an action $V_g \otimes V^{\dagger}_g$ on these virtual d.o.f., where $V_g$ is a projective representation of $G$ (see \cref{Ap:projrep}). When a symmetry is imposed, the possible phases that can be obtained are labeled by the equivalence classes of the representation $V_g$; $H^2(G,U(1))$ precisely identifies them. 

In the case where the ground state is not unique (non-injective case), the tensor $A$ of a system is supported on a "block-diagonal" space
\begin{equation}
\mathcal{H}=\bigoplus^{\mathcal{A}}_{\alpha=1} \mathcal{H}_{\alpha},\notag
\end{equation}
which is the known block structure of the matrices forming the tensor of the {MPS}. The phases are determined by a representation of the symmetry group in terms of permutations between the blocks of the  {MPS}, and the $\mathcal{A}$-fold cartesian product of $H^2(G,U(1))$ with itself.

An on-site global symmetry of an  {MPS} under a linear unitary representation of the group $G$, which we will call $U_g$, is given by the following action on the virtual d.o.f. of the tensor:

\begin{equation}\label{eq:symblocks}
U_g A = A \left ( P_g\left [ \bigoplus^{\mathcal{A}}_{\alpha=1}V^{\alpha}_g\otimes \bar{V}^{\alpha}_g\right] \right).
\end{equation}

The operators $\{ V^{\alpha }_g: g \in G \}$ form a (projective) unitary representation acting on $\mathcal{H}_{\alpha}$, as in the case of a unique ground state. The operators $\{ P_g: g \in G \}$ form a representation of $G$ that acts as a permutation between the subspaces $\mathcal{H}_{\alpha}$. This representation is in general reducible: $\{P_g \mathcal{H}_{\alpha}: g \in G\} \neq \mathcal{H}$. As a result, the subspaces $\mathcal{H}_\alpha$ can be lumped into larger subspaces of $\mathcal{H}$, $\mathcal{H}_{\mathfrak{a}}$, that are irreducible under the action of $G$: $\{P_g \mathcal{H}_{\mathfrak{a}}: g \in G\}= \mathcal{H}_{\mathfrak{a}}$. From the splitting $\mathcal{H}=\bigoplus_{\mathfrak{a}} \mathcal{H}_{\mathfrak{a}}$, a decomposition of the operators $P_g$ into irreducible representations $P_g^{\mathfrak{a}}$ can be derived, which we can use to re-express Eq.(\ref{eq:symblocks}): 
\begin{equation}\label{eq:symblocks_2}
U_gA=A\left (\bigoplus_{ \mathfrak{a} }P^{\mathfrak{a}}_g\left [ \bigoplus_{\alpha\in \mathfrak{a} }V^{ \alpha }_g\otimes \bar{V}^{\alpha }_g\right]\right).
\end{equation}
Interestingly, the decomposition in terms of $\mathcal{H}_{\mathfrak{a}}$ is generally unstable under perturbations, even those that preserve the symmetry \cite{Schuch11}. 

We now study each summand $P^{\mathfrak{a}}_g\left ( \bigoplus_{\alpha\in \mathfrak{a} }V^{\alpha }_g\otimes \bar{V}^{\alpha }_g\right)$ of Eq.(\ref{eq:symblocks_2}) and explain how it relates to the concept of \emph{induced} representation \cite{Simon96}. For a fixed summand index $\mathfrak{a}$, we pick a reference block $\alpha_0\in \mathfrak{a}$ and we define the subgroup:
$$ H:=\{h\in G : \; P^{\mathfrak{a}}_h (\mathcal{H}_{\alpha_0})=\mathcal{H}_{\alpha_0}  \}\subset G.$$
We can split $G$ in disjoint cosets $k_{\beta}H$ labelled by the blocks $\beta \in \mathfrak{a}$ for a (non-unique) choice of $k_{\beta}\in G$ chosen such that $P^{\mathfrak{a}}_{k_{\beta}} (\mathcal{H}_{\alpha_0})=\mathcal{H}_{\beta}$ (let us notice that this is possible because $P_g$ is irreducible in this subset $\mathfrak{a}$). Then we can use that for every $g$ and $\alpha$, there exist unique $h\in H$ and $\beta$ such that
\begin{equation}\label{eq:galphah}
g k_{\alpha}=k_{\beta}h.
\end{equation}
Ref. \cite{Schuch11} shows that the action on each summand is unitarily equivalent to
\begin{equation}
P^{\mathfrak{a}}_g\left ( \bigoplus_{\alpha \in \mathfrak{a} }V^{\alpha_0 }_h\otimes \bar{V}^{\alpha_0 }_h\right),\notag
\end{equation}
where $h$ is determined by $g,\alpha$ in Eq.(\ref{eq:galphah}). Also Ref. \cite{Schuch11} shows that two systems are in the same phase if the permutation representation $P_g$ are the same and if for each irreducible subset $\mathfrak{a}$, the projective representation $V^{\alpha_0 }_h$ has the same cohomology class. Since the permutation is effectively encoded in $H$, a phase is characterized by the choice of $H$ together with one of its cohomology classes.

\subsection{Non-trivial virtual representation from restriction in  {MPS}}

Let us start with a $G$-isometric  {MPS} with tensor $ A_G=|G|^{-1}\left (\sum_{ g \in G}  L_g\otimes L^{\dagger}_g \right )$, where $L_g$ denotes again the left regular representation of $G$ \cite{Schuch10}. This {MPS} has the local symmetry (see Fig. \ref{Aresym}(a)):
\begin{equation}
(L_g\otimes L^{\dagger}_g)_p A_G \; =A_G (L_g\otimes L^{\dagger}_g)_v =A_G \quad \forall g \in G. \notag
\end{equation}
Its parent Hamiltonian has a degenerate ground subspace whose dimension is equal to the number of inequivalent irreps of $G$. In analogy to the 2D case, we wish to study the restricted tensor 
\begin{equation} \label{redtens}
 A^{{\rm res}}_N=\frac{1}{|N|}\left (\sum_{ g \in N}  L_g\otimes L^{\dagger }_g \right ), 
\end{equation} 
where $N$ is a normal subgroup of $G$. Unlike the 2D case, there is no topological content associated with $N$.

\begin{figure}[ht!]
\begin{center}
\includegraphics[scale=0.3]{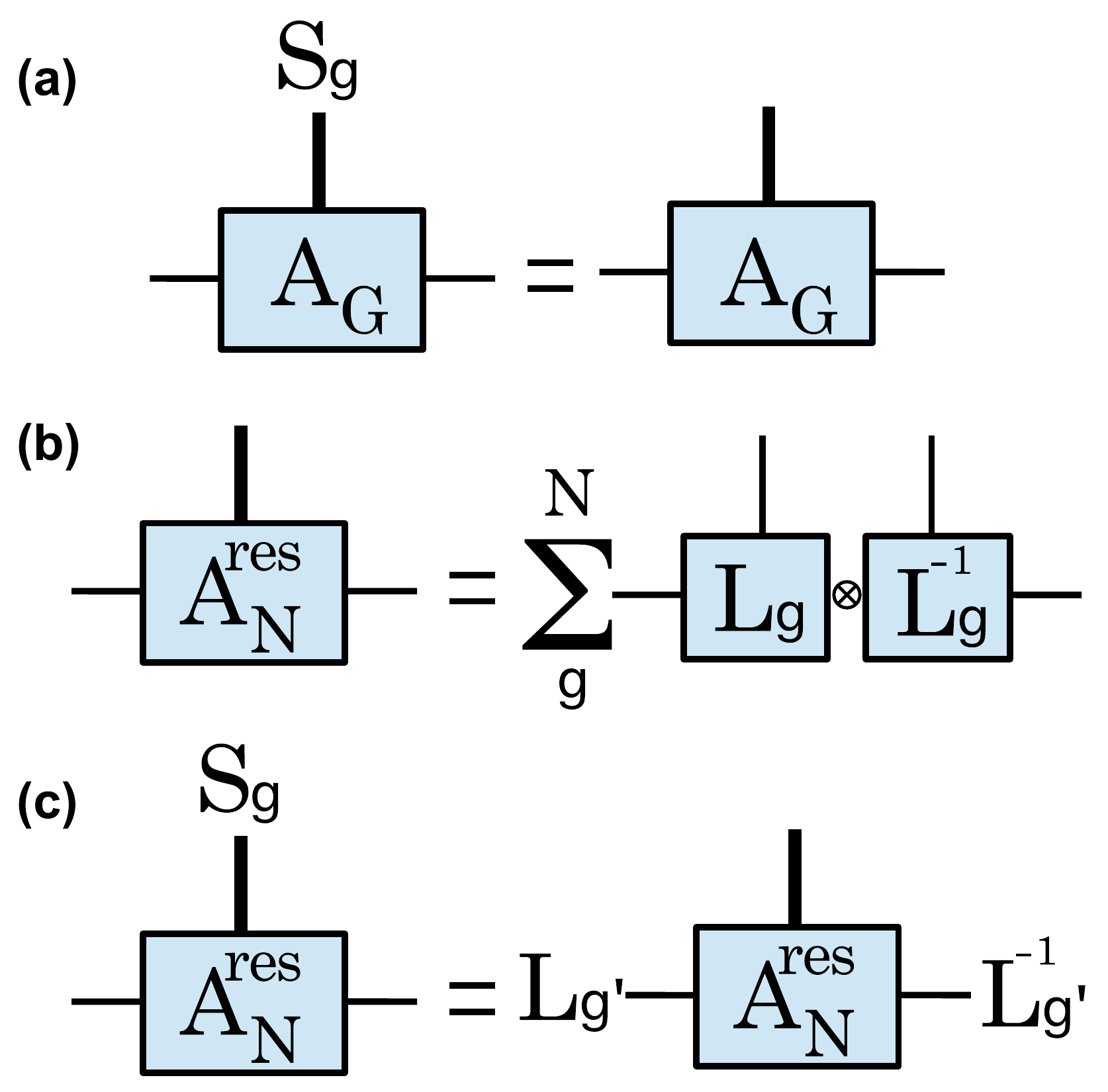}
\caption{(a) The local symmetry of the tensor $A_G$ is represented. (b) The tensor $A^{{\rm res}}_N$ of equation (\ref{redtens}) is depicted. The physical Hilbert space is decomposed into a tensor product form. (c) The action of the operators $S_g=L_g\otimes L_{g^{-1}}$ is translated into the virtual d.o.f. with a freedom on the choice of the element $g'$ within coset $[g]$.} 
\label{Aresym}
\end{center}
\end{figure}

The operators $\{S_g=L_g\otimes L_{g^{-1}}: g\in G\}$ represent a global on-site symmetry of both {MPS} constructed with $A_G$ and $A^{{\rm res}}_N$, as shown in Fig. \ref{Aresym}. As in 2D, these operators no longer represent a local symmetry of the tensor $A^{{\rm res}}_N$ for all the elements $g\in G$. In fact, $S_g$ is a local symmetry of $A^{{\rm res}}_N$ if and only if $g \in N$. 

Given $g,g' \in G$, as already seen in the more general 2D case, whenever $g N=g' N$, the actions of $S_g$ and $S_{g'}$ on $|{\mathcal{M}(A^{{\rm res}}_N)} \rangle$ are identical; $S_g$ is actually a representation of $G_{{\rm sym}}\cong G/N$. Part of the local symmetry of the state $|\mathcal{M}(A_G)\rangle$ is degraded to a global symmetry in the state $|\mathcal{M}(A^{\rm{res}}_N)\rangle$. The local $G$-symmetry is reduced to a local $N$-symmetry plus a global on-site $G_{{\rm sym}}$ symmetry. Applying $S_g$ on $A^{\rm{res}}_N$ translates into a \emph{non}-trivial representation of $G_{{\rm sym}}$ on the virtual d.o.f., in contrast to applying $S_g$ on $A_G$. 

We now study the effect of $S_g$ on $A^{{\rm res}}_N$. To do so, we will exploit the block diagonal structure of $A^{{\rm res}}_N$ and we will express $\{ L_g: g \in G \}$ as its direct sum decomposition in terms of irreps of $G$. It will turn out that the block structure of the virtual matrices of $A^{{\rm res}}_N$ is related to the irreps of $N$. For our purposes, we will be led to study how, given a proper normal subgroup $N \subset G$, and an irrep of $G$, $\Pi_\nu$, the irreps of $N$ contained in $\Pi_\nu$ give a particular structure to the matrices $\Pi_\nu$. This issue has been analysed by Clifford in \cite{Clifford37}. Using his results, we will: 1) obtain all the possible phases in 1D with symmetries and degenerate ground space, 2) show how the notion of induced representation appears naturally in the 1D phase classification, 3) exhibit an explicit method to construct the state and operators of each phase 4) associate the restriction $G \to N$ to an appealing physical mechanism in 2D (see \cref{sec:symredcon}). 

To begin with, let us write $A_G$ a bit more explicitly: 
\begin{equation}\label{eq:generictensorA}
A_G= \frac{1}{|G|} \sum_{g\in G} \sum_{\alpha,\beta,i,j}[L_g]_{\alpha i} \overline{[L_g]}_{\beta j }|\alpha ) \langle i| \otimes|j \rangle (\beta|,
\end{equation}
where  $|i\rangle,|j\rangle,|\alpha)$ and $|\beta)$ are basis elements of the vector space that supports $L$. The basis of the space describing the virtual d.o.f. of the tensor $A_G$ is represented with round brackets in Eq.(\ref{eq:generictensorA}) whereas regular brackets are used for the physical Hilbert space.
Let us denote 
\beq 
L_g \cong \bigoplus_\nu  \id_{m_\nu} \otimes \Pi_\nu(g), \notag
\eeq
the decomposition into irreps of the left regular representation, which acts on $\mathbb{C}^{|G|} \cong \mathcal{H} = \bigoplus_\nu \mathcal{K}_\nu \otimes \mathcal{H}_\nu$. That is, $\Pi_\nu$ are the irreps of $G$, and $m_\nu$ their multiplicities; $\Pi_\nu$ acts on $\mathcal{H}_\nu$ and $\mathcal{K}_\nu$ is the multiplicity space associated with $\Pi_\nu$. Some Clebsch-Gordan matrix allows to write 
$$ A_G \cong \frac{1}{|G|}  \sum_{\substack{g \in G\\ \alpha,\beta,i,j}}   \left[ C \big(\bigoplus_{\nu}  \id_{m_{\nu}} \otimes \Pi_{\nu}(g)\big) C^{\dagger}  \right]_{\alpha i}   
\overline{\left[   C \big(\bigoplus_{\nu'}  \id_{m_{\nu'}} \otimes \Pi_{\nu'}(g)  \big) C^{\dagger} \right]}_{\beta j } |\alpha )\langle i| \otimes|j \rangle (\beta|. $$

If we express $A_G$ using orthonormal bases $\{ |k^{(\nu)},l^{(\nu)}\rangle\}$ for each subspace $\mathcal{K}_{\nu}\otimes \mathcal{H}_{\nu}$, and using the orthogonality relations of irreps, we obtain:
$$ A_G \cong \sum_{\nu}  \frac{1}{d_{\nu}} \sum_{ \substack{ k^{(\nu)}_1,k^{(\nu)}_3 \\ l^{(\nu)}_1,l^{(\nu)}_2 }} |k^{(\nu)}_1,l^{(\nu)}_1)\langle k^{(\nu)}_1,l^{(\nu)}_2 |  
\otimes   |k^{(\nu)}_3,l^{(\nu)}_2 \rangle (k^{(\nu)}_3,l^{(\nu)}_1 |. $$
This tensor exhibits an obvious block diagonal form, in line with \cite{Schuch10},
\begin{equation*}
A_G\cong \bigoplus_{\nu}  \frac{1}{d_{\nu}}  A_G[\nu],
\end{equation*}
where each block $A_G[\nu]$ admits a very simple diagrammatic representation:
\beq \label{eq:tendiagram}
A_G[\nu]= 
\parbox[c]{0.22\textwidth}{ \includegraphics[scale=0.25]{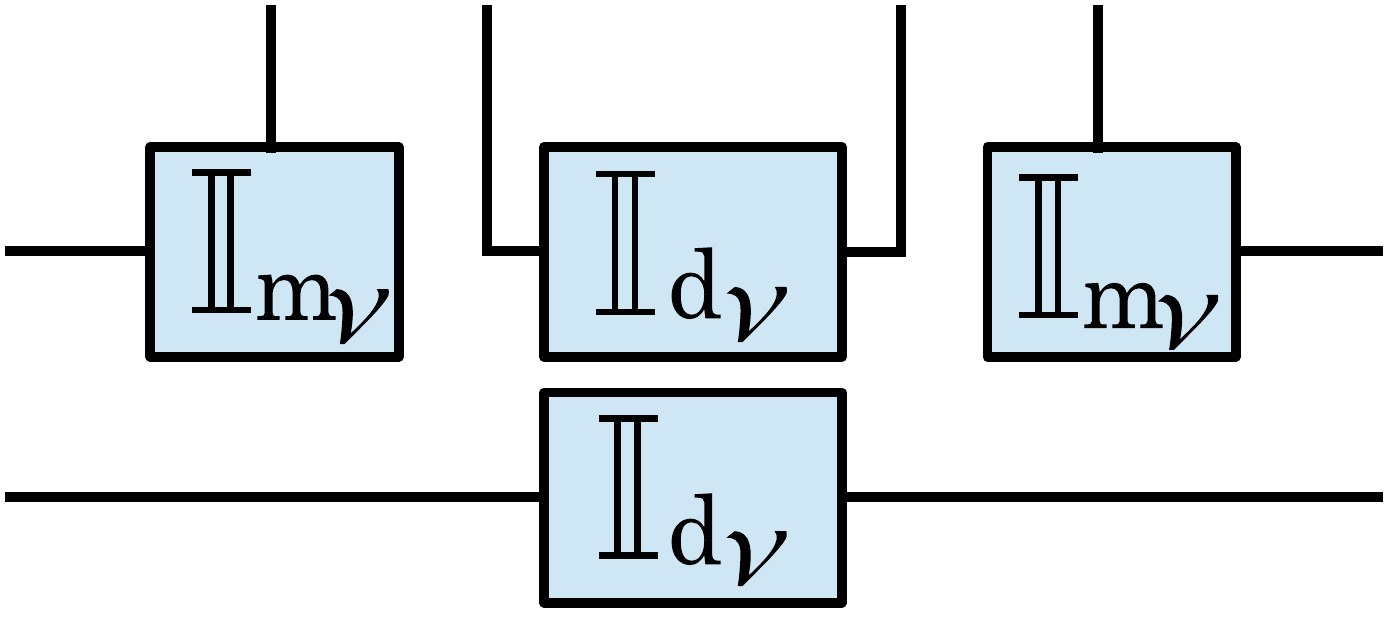}}.
\eeq
Note that the left regular representation is such that $m_\nu=d_\nu$.
Eq.({\ref{eq:tendiagram}) shows that if the irreps of two different groups, $G_1$ and $G_2$, have the same dimensions (and multiplicities), the corresponding tensors are identical modulo a local (Clebsch-Gordan) transformation. For example the left regular representations of $\mathbb{Z}_4$ and $\mathbb{Z}_2\times \mathbb{Z}_2$ both decompose into four one-dimensional irreps. Another example of this situation is that of the quaternionic group $Q_8$ and the order 4 dihedral group $D_4$: in each case, the left regular representation is made of two inequivalent one-dimensional irrep and two equivalent copies of a two-dimensional irrep. \\

Given an irrep $\gamma$ of $G$ and an element $h\in G$, we consider the following operator:
$$S_{\gamma}(h)=\left(\bigoplus_{\nu \neq \gamma}  \id_{m_{\nu}}\otimes  \id_{d_{\nu}} \oplus  \id_{m_{\gamma}}\otimes \Pi^G_{\gamma}(h)\right)
 \otimes  \left(\bigoplus_{\nu \neq \gamma}  \id_{m_{\nu}}\otimes  \id_{d_{\nu}} \oplus  \id_{m_{\gamma}}\otimes \bar{\Pi}^G_{\gamma}(h)\right).$$
Modulo an appropriate change of basis, the action of $S_{\gamma}(h)$ reads:
$$\tilde{A}_G=S_{\gamma}(h) \left( C\otimes C)A_G( C^{\dagger} \otimes C^{\dagger} \right ).$$
It is clear that $\tilde{A}_G[\nu]={A}_G[\nu]$ for $\nu\neq \gamma$, but the block $\gamma$ is modified as
$$\tilde{A}_G[\gamma]  = \frac{1}{d_{\gamma}}\sum_{k^{(\gamma)},k'^{(\gamma)}} 
 \sum_{ \substack{ m^{(\gamma)},n^{(\gamma)} \\ l^{(\gamma)},l'^{(\gamma)}}}
 [\Pi^G_{\gamma}(h)]_{n^{(\gamma)},l'^{(\gamma)}} \overline {[\Pi^G_{\gamma}(h)]}_{m^{(\gamma)},l'^{(\gamma)}  }  
 |k^{(\gamma)},l^{(\gamma)})\langle k^{(\gamma)},n^{(\gamma)} | \otimes|k'^{(\gamma)},m^{(\gamma)} \rangle (k'^{(\gamma)},l^{(\gamma)} |.$$

Since $\sum_{l'} \left([\Pi^G_{\gamma}(h)]_{n,l'} \overline {[\Pi^G_{\gamma}(h)]}_{m,l' }\right)=\delta_{m,n}$, the tensor remains invariant. That $S_{\gamma}(h)$ is a symmetry can be straightforwardly seen diagrammatically:
\begin{equation}\label{eq:locsympa}
\tilde{A}_G[\gamma]= \parbox[c]{0.25\textwidth}{ \includegraphics[scale=0.25]{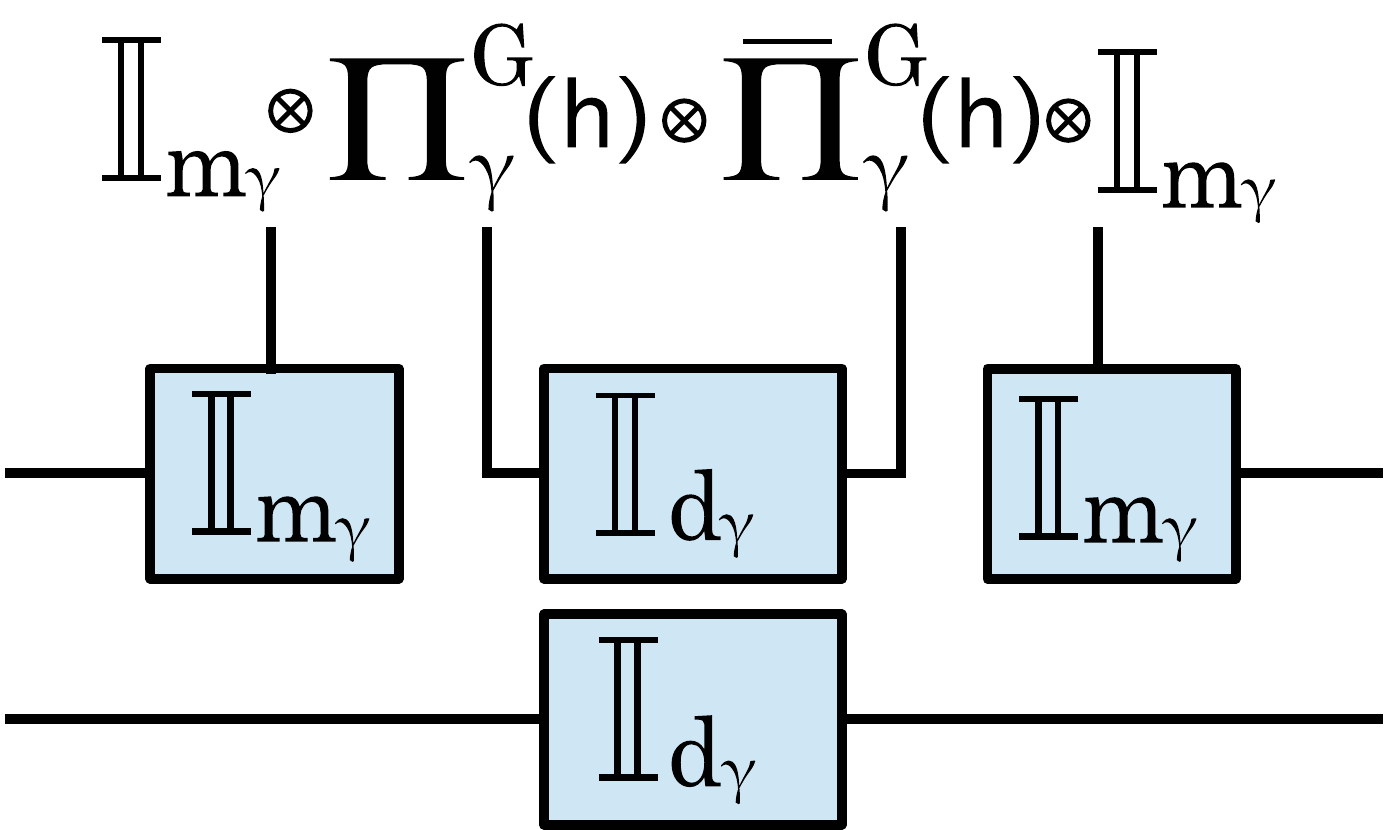}}
= \parbox[c]{0.25\textwidth}{ \includegraphics[scale=0.25]{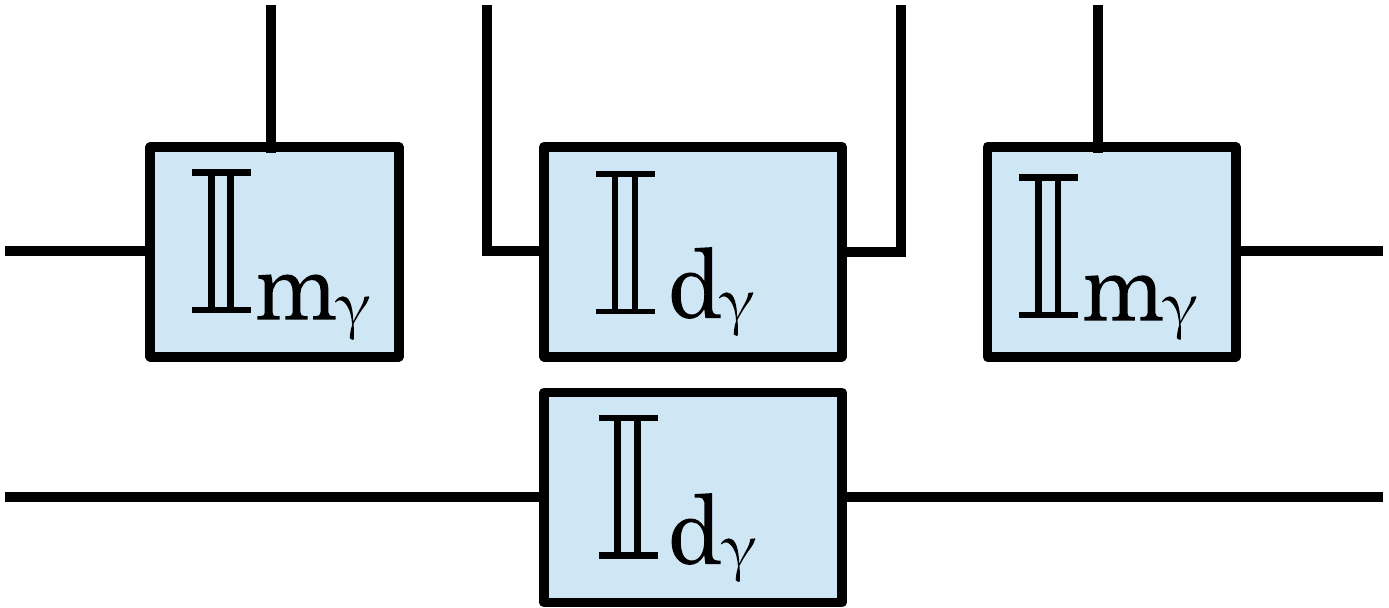}}= A_G[\gamma].
\end{equation}
The operators $S_{\gamma}(h)$ for all $ h\in G$, represent a local symmetry of the state constructed with the tensor $A_G$. 
Similarly the diagrams show that $S_{\gamma}(h)$ for each $ h\in G$ have a trivial action on the virtual d.o.f. of $A_G$. Summarizing, we have the following statement:

\begin{tcolorbox}
\begin{proposition}
For $G$-isometric  {MPS},
$$[S_{\gamma}(h)]_pA_G= A_G=[\Pi^G_{\gamma}(h)]_v A_G [\Pi^G_{\gamma}(h^{-1})]_v,$$
 $$ \left (S_{\gamma}(h)^{\otimes \mathcal{R}} \right) |\mathcal{M}(A_G) \rangle = |\mathcal{M}(A_G) \rangle,$$
where the subscripts $p$ and $v$ stand for the physical Hilbert space and virtual d.o.f. respectively, and where $\mathcal{R}$ is any region of the lattice.
\end{proposition}
\end{tcolorbox}

We now turn to the restricted  {MPS} (\ref{redtens}) and study its symmetries. For that, we begin with the decomposition of $ \Pi_\nu(n)$, $ n \in N$, into irreps of $N$: 
\begin{equation}\label{eq:Irdec}
\Pi_\nu(n)\cong \bigoplus_{\sigma \subset \nu}  \id_{\widehat{\mathcal{K}}_\sigma} \otimes \pi_\sigma(n),
\end{equation} 
where $\{\pi_\sigma(n): n \in N\}$ denote the irreps of $N$ and $\Pi_\nu(n)$ acts on  $\mathcal{H}_\nu= \bigoplus_{\sigma \subset \nu} \widehat{\mathcal{K}}_\sigma \otimes \widehat{\mathcal{H}}_\sigma$. As before $\widehat{\mathcal{K}}_\sigma$ is a multiplicity space, i.e. the number of copies of $\pi_\sigma(n)$ contained in $\Pi_\nu$ is equal to $\textrm{dim} \; \widehat{\mathcal{K}}_\sigma$.  Clifford's Theorem \cite{Clifford37} states that if $\sigma, \sigma' \subset \nu$, then (i) $\textrm{dim} \; \widehat{\mathcal{H}}_\sigma=\textrm{dim} \; \widehat{\mathcal{H}}_{\sigma'}=d_{\sigma(\nu)}$, (ii) the multiplicity spaces $\widehat{\mathcal{K}}_\sigma$ and $\widehat{\mathcal{K}}_{\sigma'}$  are isomorphic: $\textrm{dim} \widehat{\mathcal{K}}_\sigma= \textrm{dim} \widehat{\mathcal{K}}_{\sigma'}=\ell_{\nu}$ (iii) there exists $g \in G$ (that depends on $\sigma,\sigma'$) such that $\pi_{\sigma'}(n)=\pi_{\sigma}(g n g^{-1}), \forall n \in N$, i.e. $\pi_\sigma$ and $\pi_{\sigma'}$ are related by conjugation. 

Let $\{ \ket{l^{(\sigma)},q^{(\sigma)}} \}$ denote an orthonormal basis of $\widehat{\mathcal{K}}_\sigma \otimes \widehat{\mathcal{H}}_\sigma$. Using again irrep orthogonality relations, one can readily show that 
$$ A^{\rm res}_N =  \frac{1}{|N|}\sum_{n \in N} L (n)\otimes L^{\dagger}(n) \cong   \sum_{\nu,\mu}  \;  \id_{\mathcal{K}_\nu}\otimes A^{\rm res}_G(\nu,\mu) \otimes  \id_{\mathcal{K}_\mu}$$
where
\begin{equation} \label{eq:blockdesc}
A^{\rm res}_N(\nu,\mu)= \sum_{\sigma(\nu) \sim \rho(\mu)}  \frac{1}{d_{\sigma(\nu)}}  
\sum_{\substack{ q^{(\sigma)},q^{(\rho)}\\l^{(\sigma)},l^{(\rho)} }} | l^{(\sigma)},q^{(\sigma)} ) 
\langle  l^{(\sigma)},q^{(\rho)}|  \otimes    |l^{(\rho)},q^{(\rho)}\rangle  ( l^{(\rho)},q^{(\sigma)}|,
 \end{equation}
and where $\sigma \sim \rho$ indicates that both representations are equivalent. We observe that if $\nu$ and $\mu$ do not contain any common irrep of $N$, $A^{\rm res}_N(\nu,\mu)=0$. We also point out that, in virtue of Clifford's Theorem, if $\nu$ and $\mu$ have one common irrep of $N$, then \emph{any} irrep of $N$ contained in $\nu$ is also contained in $\mu$ and vice versa. They are so-called associate.

We can represent Eq.(\ref{eq:blockdesc}) diagrammatically:
\beq \label{eq:diares}
A^{\rm res}_N(\nu,\mu)=  \sum_{\sigma(\nu) \sim \rho(\mu)}  \frac{1}{d_{\sigma(\nu)}}
\parbox[c]{0.2\textwidth}{ \includegraphics[scale=0.25]{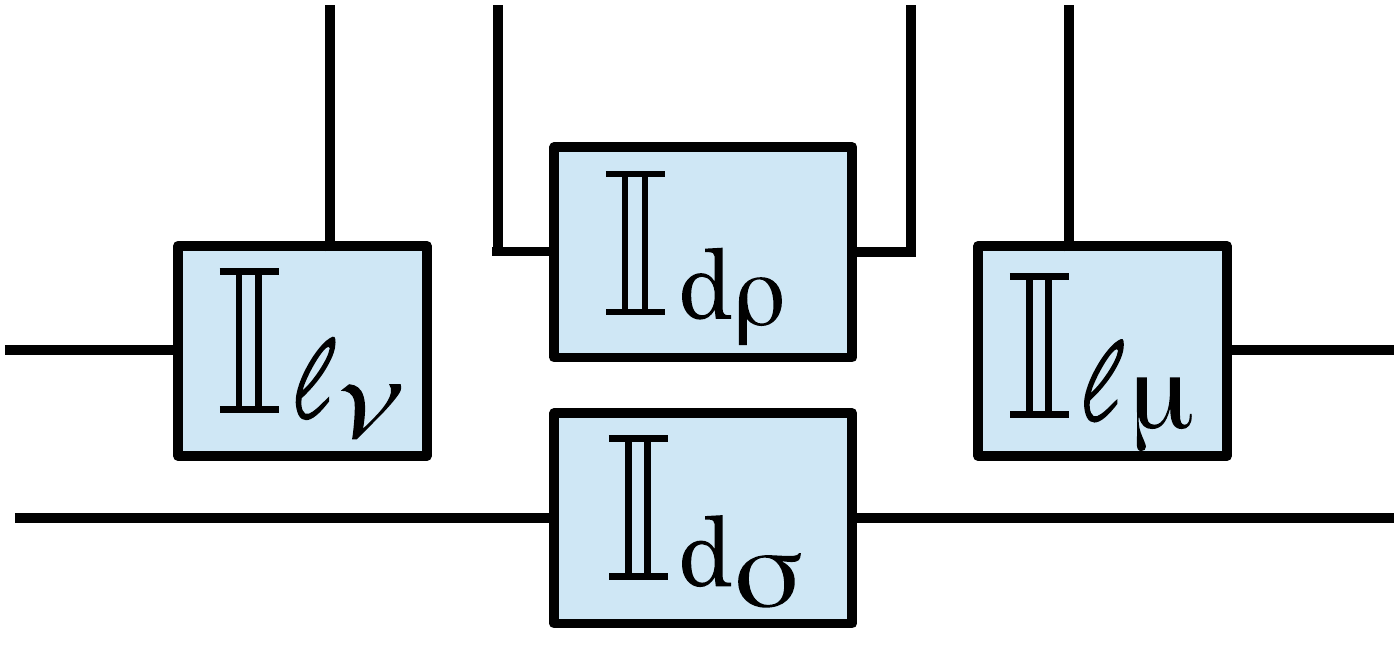}}
\eeq

In analogy with our discussion of $A_G$, we will construct symmetry operators in terms of irreps of $G$. For the sake of clarity, it is desirable that these irreps reflect the decomposition theory of the restricted tensor (\ref{eq:diares}). Clifford's theorem shows us exactly how to do that. Let $\{ \pi_{\sigma_{i}} : i=1 \ldots r_\nu\}$ represent the set of inequivalent irreps of $N$ that appear in the decomposition theory of $\{ \Pi_\nu(n): n \in N \}$. We can express the matrix $\Pi_\nu(g)$ in terms of submatrices $T_{ij}(g): \widehat{\mathcal{K}}_{\sigma_j} \otimes \widehat{\mathcal{H}}_{\sigma_j} \mapsto \widehat{\mathcal{K}}_{\sigma_i} \otimes \widehat{\mathcal{H}}_{\sigma_i}$ as:
\beq\label{eq:pi_nu_cliff}
\Pi_{\nu}(g)=\left( \begin{array}{ccc}T_{11}(g)&\cdots&T_{1r}(g)\\ \vdots &\ddots& \vdots\\ T_{r1}(g) &\cdots & T_{rr}(g)\end{array}\right).
\eeq
It is clear that $T_{ij}(n)=\delta_{i,j}   \id_{\widehat{\mathcal{K}}_{\sigma_i}} \otimes \pi_{\sigma_i}(n) \forall n \in N$. For fixed $i$ and $g$, it can be shown that there is one and only one value of $j$ for which $T_{ij}(g) \neq 0$. That is, $\Pi_\nu(g)$ has a permutation form, for example:
$$\Pi_{\nu}(g)=\left( \begin{array}{cccc} 0&*&0&0\\  0&0&*&0 \\ 0&0 &0& *\\ *&0 &0 & 0\end{array}\right),$$ interchanging the subspaces $\{ \widehat{\mathcal{K}}_{\sigma_i} \otimes \widehat{\mathcal{H}}_{\sigma_i} :i =1 \ldots r_\nu\}$, and acting non-trivially on them.

To elucidate the content of the operators $T_{ij}$, it is convenient to introduce the subgroup $G' \subset G$ that leaves $\widehat{\mathcal{K}}_{\sigma_1} \otimes \widehat{\mathcal{H}}_{\sigma_1}$ invariant: 
\begin{equation*}
G' \equiv \{g'\in G ; T_{11}(g') \neq 0\},
\end{equation*}
and we choose a set of elements $\{ \check{g}_2,\cdots, \check{g}_{r_\nu} \}$, such that $  T_{i1}(\check{g}_i) \neq 0$. These elements index the cosets of $G'$ in $G$:
$$G=G'+\check{g}_2 G'+\cdots+\check{g}_{r_\nu} G'. $$
We want to describe each matrix $T_{ij}(g)$ in terms of the representation $T_{11}(g')$ of $G'$. We notice that if one takes an element $g\in G$ and some element $\check{g}_i$, there exists a unique $\check{g}_j$ and $g'\in G'$ such that $g \check{g}_j=\check{g}_i g'$. Building on this observation, we see that
\begin{equation}\label{eq:indurep} T_{ij}(g) \cong  \tilde{T}_{ij}(g) \equiv
\Big\{
	\begin{array}{ll}
		T_{11}(\check{g}_i ^{-1}g \check{g}_j)  & \mbox{if } \check{g}^{-1}_i g \check{g}_j\in G' \\
		0 & \mbox{otherwise } 
	\end{array}.
\end{equation}
Hence $\Pi_\nu$ can be expressed as an induced representation of an irrep of $G'$. Moreover one can prove that $T_{11}(g' )$ admits a tensor product decomposition:
\beq \label{eq:VC}
 T_{11}(g')=V_\nu(g')\otimes C_\nu(g'),
\eeq
 where $V_\nu(g'): \widehat{\mathcal{K}}_{\sigma_1} \to \widehat{\mathcal{K}}_{\sigma_1}$ and $C_\nu(g'): \widehat{\mathcal{K}}_{\sigma_1} \to \widehat{\mathcal{K}}_{\sigma_1}$ are irreducible representations of $G'$. Since $T_{11}(n)= \id_{\widehat{\mathcal{K}}_{\sigma_1}} \otimes \pi_{\sigma_1}(n) \forall n \in  N$, we find by identification that $V_\nu(n)= \id_{\widehat{\mathcal{K}}_{\sigma_1}}$ for all $n\in N$. 

It can be proven that $V_\nu$ is projective:
\begin{equation*}
V_\nu(g) V_\nu(h')=\omega(g',h') V_\nu(g'h'),
\end{equation*}
where $\omega$ satisfies the cocycle condition:
$$ \omega(k'g',h') \omega(k',g')=\omega(k',g'h') \omega(g',h').$$
Since $\omega(g'n,h'm)=\omega(g',h')$ for all $n,m\in N$, $V_\nu(g')$ is actually a representation of $G'/N$.
Finally we can write:
\begin{equation}\label{eq:simdec}
 \tilde{T}_{ij}(g)\equiv
\Big\{
	\begin{array}{ll}
		V_\nu(\check{g}_i^{-1}g \check{g}_j)\otimes C_\nu(\check{g}_i ^{-1}g \check{g}_j)  & \mbox{if } \check{g}^{-1}_i g \check{g}_j\in G' \\
		0 & \mbox{otherwise } 
	\end{array}.
\end{equation}
See Fig. \ref{Aresymblock2}(c) for a diagrammatic representation. 
\begin{figure}[ht!]
\begin{center}
\includegraphics[scale=0.7]{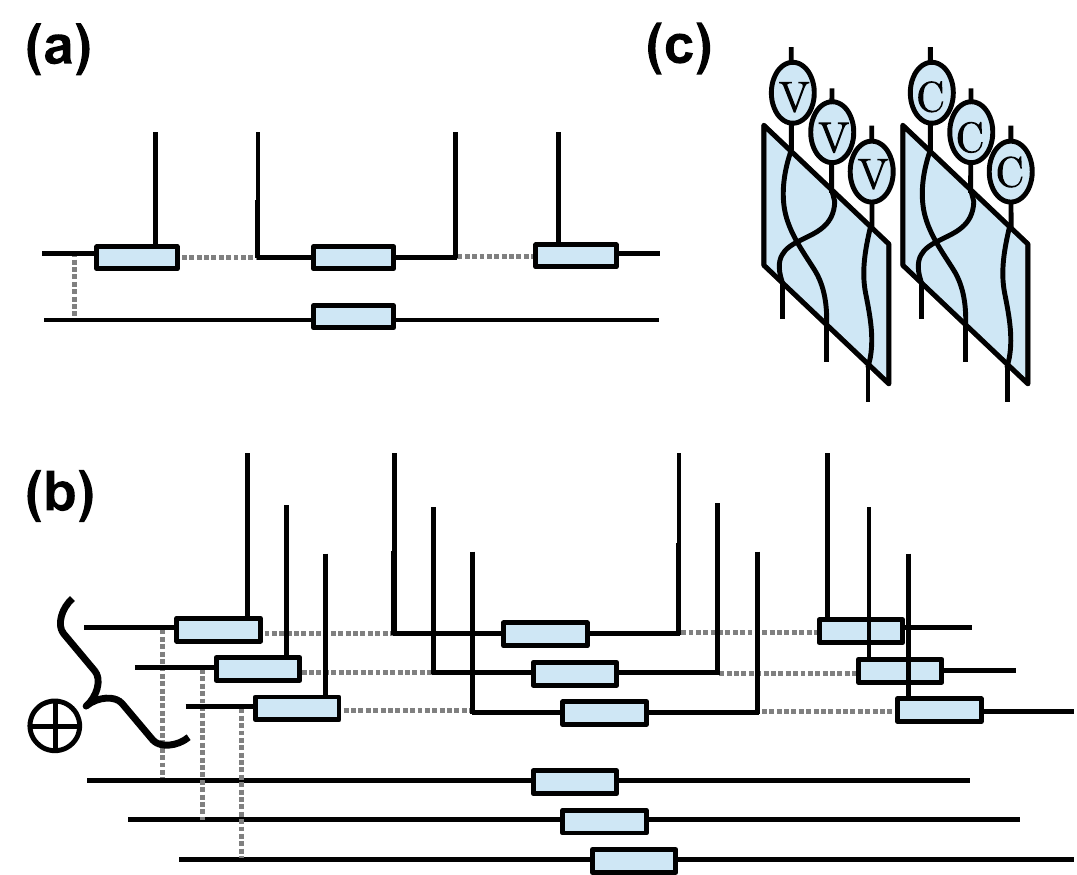}
\caption{ (a) Diagrammatic representation of a term appearing in the block $A^{\rm res}_N(\nu,\mu)$; see Eq.(\ref{eq:diares}). The grey dashed line is meant to indicate that the four boxes represent the operator $ \id_{\ell_\nu} \otimes  \id_{d_\rho} \otimes  \id_{d_\sigma} \otimes  \id_{\ell_\mu}$ corresponding to \emph{fixed} values of $\rho, \sigma$ in the sum. (b) Diagrammatic representation of the \emph{whole sum} Eq.(\ref{eq:diares}): each set of four tensors in a same plane parallel to the sheet relate to a same term of the sum, i.e. a same value of the pair $(\rho,\sigma)$. The $\oplus$ symbol is meant to mark the summation over all possible values of $(\rho,\sigma)$. (c) Decomposition of an irrep of $G$, described by the Eq.(\ref{eq:simdec}). We omit to represent the dependence of $V$ and $C$ on $g,\nu, i, j$, see Eq.(\ref{eq:simdec}).} 
\label{Aresymblock2}
\end{center}
\end{figure}
The matrix $\tilde{\Pi}_\nu(g)$ defined as (\ref{eq:pi_nu_cliff}), with $T_{ij}$ replaced with $\tilde{T}_{ij}$ maps $\widehat{\mathcal{K}}_{\sigma_i} \otimes \widehat{\mathcal{H}}_{\sigma_i}$ to $\widehat{\mathcal{K}}_{\sigma_j} \otimes \widehat{\mathcal{H}}_{\sigma_j}$ bijectively according to $gg_j=g_i g'$. If we further define $\psi_{ij}(g) \equiv \check{g}_i^{-1} g \check{g}_j$, we are in a position to specify the action of $\tilde{\Pi}_\nu(g)$,
$$\bra{l^{(\sigma_j)},q^{(\sigma_j)}} \tilde{\Pi}_\nu(g) \ket{l^{(\sigma_i)},q^{(\sigma_i)}}=
\bra{l^{(\sigma_j)}} V_\nu(\psi_{ij}(g))\ket{l^{(\sigma_i)}}  
\times \bra{q^{(\sigma_j)}} C_\nu(\psi_{ij}(g)) \ket{q^{(\sigma_i)}},$$
\begin{figure}[ht!]
\begin{center}
\includegraphics[scale=0.6]{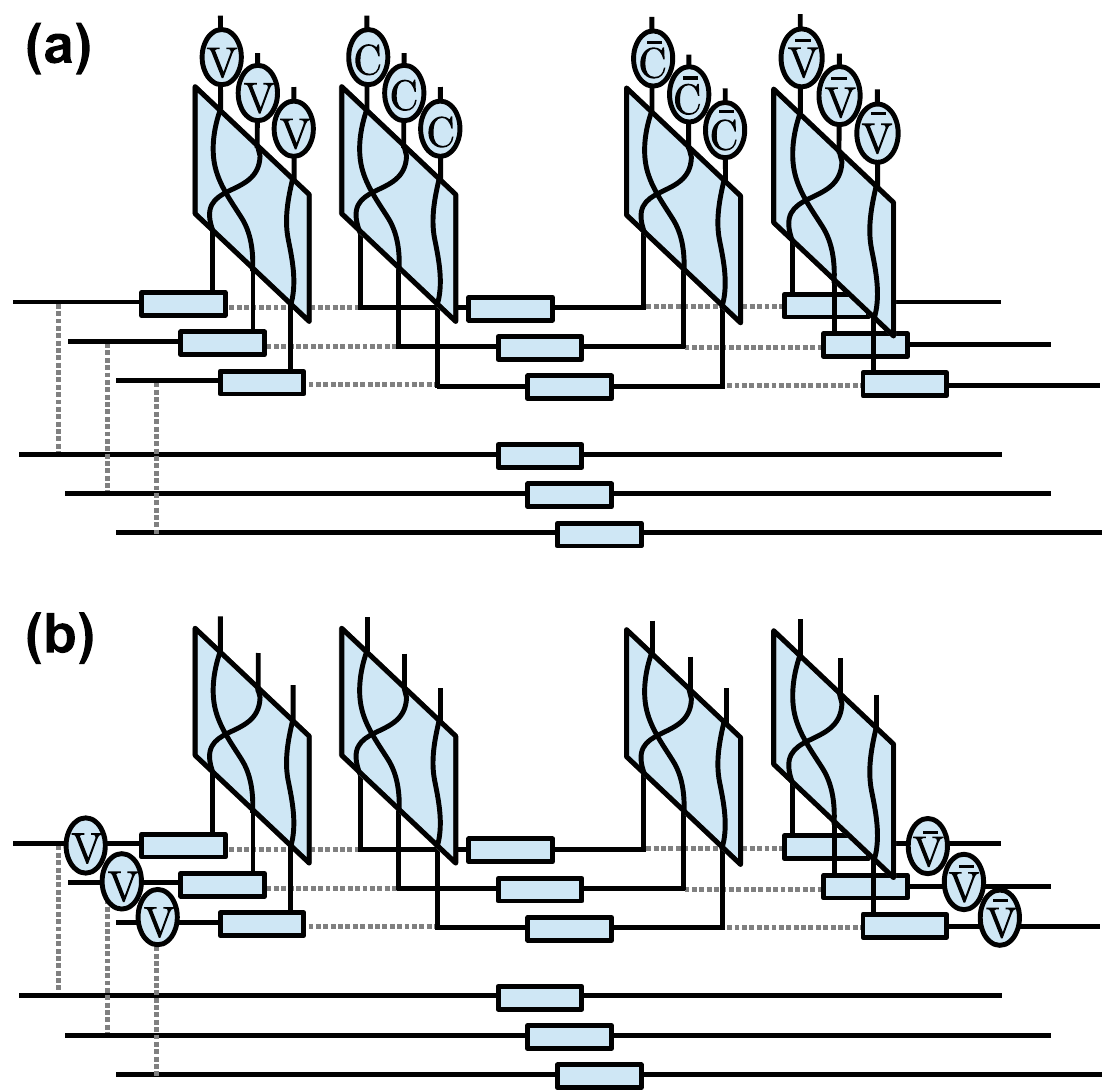}
\caption{ (a) Action of the irrep operator on the tensor according to Eq.(\ref{eq:symresfor1}). (b) Final result of the action on the restricted tensor given by Eq.(\ref{eq:actionblock}).} 
\label{Aresymblock}
\end{center}
\end{figure}

and to apply the operator corresponding to two associate irreps $\nu$ and $\mu$ over $A^{\rm res}_N(\nu,\mu)$:

\begin{equation}\label{eq:symresfor1}
\begin{split}
\Big(&\tilde{\Pi}_\nu(g)\otimes \tilde{\Pi}^{\dagger}_\mu(g)\Big)_p A^{\rm res}_N(\nu,\mu)=
\sum_{\sigma_i(\nu) \sim \rho_i(\mu)}  \frac{1}{d_{\sigma_i}}  
\sum_{\substack{ q^{(\rho_j)}\\l^{(\sigma_j)},l^{(\rho_j)} }}  
\sum_{\substack{ q^{(\sigma_i)}\\l^{(\sigma_i)},l^{(\rho_i)} }} 
 \bra{q^{(\rho_j)}}C^{\dagger}_\mu(\psi_{ij}(g))  \id_{d_{\rho_i}} C_\nu(\psi_{ij}(g))\ket{q^{(\rho_j)}}\times \\
&\hspace{3cm}|l^{(\sigma_i)},q^{(\sigma_i)} ) \langle  l^{(\sigma_j)},q^{(\rho_j)}| \bra{l^{(\sigma_i)}} V_\nu(\psi_{ij}(g))\ket{l^{(\sigma_j)}}\otimes \bra{l^{(\rho_j)}}V^{\dagger}_\mu(\psi_{ij}(g))\ket{l^{(\rho_i)}} 
 |l^{(\rho_j)},q^{(\rho_j)}\rangle  ( l^{(\rho_i)},q^{(\sigma_i)}|,
\end{split}
\end{equation}
where $ \bra{q^{(\rho_j)}}C^{\dagger}_\mu(\psi_{ij}(g))  \id_{d_{\rho_i}} C_\nu(\psi_{ij}(g))\ket{q^{(\rho_j)}}=1 $ since $C_\nu=C_\mu$ for associate irreps $\nu$ and $\mu$ \cite{Clifford37}. Then
\begin{equation}
\begin{split}
\Big(&\tilde{\Pi}_\nu(g)\otimes \tilde{\Pi}^{\dagger}_\mu(g)\Big)_p A^{\rm res}_N(\nu,\mu)=\\
&\sum_{\sigma_i(\nu) \sim \rho_i(\mu)}  \frac{1}{d_{\sigma_i}}  
\sum_{\substack{ q^{(\rho_j)}\\l^{(\sigma_i)},l^{(\rho_i)} \\l^{(\sigma_j)},l^{(\rho_j)} }}  
 |l^{(\sigma_i)},q^{(\sigma_i)} ) \langle  l^{(\sigma_j)},q^{(\rho_j)}|  \bra{l^{(\sigma_i)}} V_\nu(\psi_{ij}(g))\ket{l^{(\sigma_j)}} 
  \bra{l^{(\rho_j)}}V^{\dagger}_\mu(\psi_{ij}(g))\ket{l^{(\rho_i)}}  |l^{(\rho_j)},q^{(\rho_j)}\rangle  ( l^{(\rho_i)},q^{(\sigma_i)}|= \\
  &\sum_{\sigma_j(\nu) \sim \rho_j(\mu)}  \frac{1}{d_{\sigma_j}}  
\sum_{\substack{ q^{(\rho_j)}\\l^{(\sigma_i)},l^{(\rho_i)} \\l^{(\sigma_j)},l^{(\rho_j)} }}   
(l^{(\sigma_i)}| V_\nu(\psi_{ij}(g))|l^{(\sigma_j)})
 |l^{(\sigma_i)},q^{(\sigma_i)} ) \langle  l^{(\sigma_j)},q^{(\rho_j)}| \otimes 
   |l^{(\rho_j)},q^{(\rho_j)}\rangle  ( l^{(\rho_i)},q^{(\sigma_i)}| (l^{(\rho_j)}|V^{\dagger}_\mu(\psi_{ij}(g)) | l^{(\rho_i)}) .\notag
\end{split}
\end{equation}

Therefore 
\begin{equation}\label{eq:actionblock}
\Big(\tilde{\Pi}_\nu(g)\otimes \tilde{\Pi}^{\dagger}_\mu(g)\Big)_p A^{\rm res}_N(\nu,\mu)=  
A^{\rm res}_N(\nu,\mu)\left( P_\nu(g)\bigoplus V_\nu(\psi_{ij}(g))\otimes P_\mu(g) \bigoplus\bar{V}_\mu(\psi_{ij}(g))\right)_v,
\end{equation}
where $P_\nu(g)$ represents the permutation part of the irrep $\nu$ given by the induced representation of $G'$ (see Eq.(\ref{eq:indurep})) mapping $\widehat{\mathcal{K}}_{\sigma_i} \otimes \widehat{\mathcal{H}}_{\sigma_i}$ to $\widehat{\mathcal{K}}_{\sigma_j} \otimes \widehat{\mathcal{H}}_{\sigma_j}$ and where $V_\nu(\psi_{ij}(g))$ is the projective representation of $G'/N$, appearing in the decomposition of  Eq.(\ref{eq:VC}), acting on $\widehat{\mathcal{K}}_{\sigma_j}$. The diagrams for this action are shown in Fig. \ref{Aresymblock}.

We now consider the operators
$$S_{\nu,\mu}(g)=\left(\bigoplus_{\gamma \neq \nu,\mu}  \id_{m_{\gamma}}\otimes  \id_{d_{\gamma}} \oplus  \id_{m_{\nu}}\otimes \tilde{\Pi}_\nu(g) \oplus  \id_{m_{\mu}}\otimes \tilde{\Pi}_\mu(g) \right)  \otimes \left(\bigoplus_{\gamma \neq \nu,\mu}  \id_{m_{\gamma}}\otimes  \id_{d_{\gamma}} \oplus  \id_{m_{\nu}}\otimes \tilde{\Pi}^{\dagger}_\nu(g) \oplus  \id_{m_{\mu}}\otimes \tilde{\Pi}^{\dagger}_\mu(g)
\right) $$
and
\begin{equation}\label{eq:sinop}
S_{\nu}(g)=\left(\bigoplus_{\gamma \neq \nu,\mu}  \id_{m_{\gamma}}\otimes  \id_{d_{\gamma}} \oplus  \id_{m_{\nu}}\otimes \tilde{\Pi}_\nu(g) 
\right) \otimes \left(\bigoplus_{\gamma \neq \nu,\mu}  \id_{m_{\gamma}}\otimes  \id_{d_{\gamma}} \oplus  \id_{m_{\nu}}\otimes \tilde{\Pi}^{\dagger}_\nu(g) \right).
\end{equation}
In virtue of Eq.(\ref{eq:actionblock}), these operators correspond to an on-site global symmetry of the state constructed with $A^{\rm res}_G$:
 
$$S^{\otimes L}_{\nu,\mu}(g)|\mathcal{M}(A^{\rm res}_G) \rangle=|\mathcal{M}(A^{\rm res}_G) \rangle.$$
Summarizing:

\begin{tcolorbox}
\begin{theorem}
The action of the symmetry operator of the group $G$ on each block is a conjugation by a projective representation ($V(g')$ on the virtual d.o.f.) of the group $ G'/N$ which can be extended to the group $G/N$ as an induced representation carrying intrinsically the pattern of permutation-action between blocks. The group $G'$ is defined as the elements of $G$ leaving one (chosen) block invariant under the permutation action of the operator for these elements. The normal subgroup $N$ of $G$ corresponds to the local invariance of the tensor and it encodes the splitting of irreducible blocks under the permutation action of the symmetry.
\end{theorem}
\end{tcolorbox}

This structure fits nicely with the one explained in \cref{secclass}. Let us explained the role of the groups involved.
We consider $G$ as the group associated with the physical symmetry of the MPS. This would correspond with a representation, potentially projective, in the virtual d.o.f. In our construction the $N$-injectivity of the {MPS}, formed with the tensor $A_N^{\rm res}$, reveals an effective $G/N$ representation of the symmetry in the virtual d.o.f. In the case of degenerate ground space, the role of the subgroup $H\subset G$ is played by the quotient $G'/N \subset G/N$.

\subsubsection{An example: classification for $G_{sym}= \mathbb{Z}_2 $}

As an illustration of the analysis of the previous subsection, we consider the case where $G_{\rm sym}=\mathbb{Z}_2$. According to \cite{Schuch11}, the phases will be given solely by permutations between the blocks forming the tensor of the  {MPS} in consideration, since the second cohomology group $H^2(\mathbb{Z}_2,U(1))$ is trivial. There are two phases: one where $\mathbb{Z}_2$ is represented trivially at the virtual level, and another where $\mathbb{Z}_2$ is faithfully represented, by an appropriate permutation, at the virtual level. This permutation will be the product of disjoint transpositions acting on blocks with the same size. Thus the number of disjoint transpositions characterizes the phase, i.e., the way in which the symmetry acts. 

In order to obtain a nontrivial permutation we have to take the operator associated with the non-trivial semidirect product. The number of disjoint transpositions of this permutation will be given by the number of irreps of the extension in the decomposition of the operator that realizes the symmetry. 

The degeneracy of the ground state manifold is given by the block structure of the matrices and will depend on the group $N$; see Eq.(\ref{eq:blockdesc}). For our purposes, we study the case of the abelian group $N=\mathbb{Z}_n$, with $n$ odd. It is known that all possible extensions of these two groups are semidirect products \cite{Rotman95}. Then, we have the different extensions $G=\mathbb{Z}_n \rtimes_{\rho} \mathbb{Z}_2$  (see Appendix \ref{ap:ext}) resulting in the direct product $\mathbb{Z}_n \times \mathbb{Z}_2$ and the dihedral group $D_{n}$ as the non-trivial semidirect product. The latter is built choosing the inverse automorphism $\rho_1(g)=-g, \; \forall g\in \mathbb{Z}_n$. We will write $g^{-1}$, $n-g$ or $-g$ indistinctly.

The left regular representation of $D_{n}$ decomposes as (we only consider elements of the normal subgroup):
\begin{equation}
L_{(k,0)}^{D_{2n}}  \cong \left(1\oplus 1 \bigoplus_m \Pi_m(k,0)\oplus \Pi_m(k,0) \right),\notag
\end{equation}
where $\Pi_m$, $m=1,\cdots, (n-1)/2$, are the bidimensional irreps of $D_{2n}$ \cite{Simon96} :
\beq
\Pi_m(k,0)=\left(\begin{array}{cc}q^{ km}& 0 \\ 0 & q^{- km} \end{array}\right); \; q=e^{\frac{2\pi i}{n}}, \; k=0,\cdots, n-1. \notag
\eeq
This form for the representation is analogous to Eq.(\ref{eq:Irdec}).

One easily checks that
 $$\Pi_m(k,0)\oplus \Pi_m(k,0)=P( \Pi_m(k,0)\otimes  \id_2) P,$$
 where $ P=1\oplus \sigma_x\oplus1$  and
$$\Pi_m(k,0)\otimes  \id_2= \left(\begin{array}{cccc}q^{km}& 0& 0& 0 \\ 0 & q^{km}& 0& 0\\ 0 & 0&q^{-km}&  0 \\ 0 & 0& 0&q^{-km} \end{array}\right).$$ 
A bit more explicitly, 
$$ \Pi_m(k,0)\otimes  \id_2=q^{km} (|m,0\rangle \langle m,0|+|m,1\rangle \langle m,1|)+
 q^{-km} (|-m,0\rangle \langle -m,0|+|-m,1\rangle \langle -m,1|).$$
So, up to local unitary equivalence, 
$$L_{(k,0)} \cong  \sum_{u=0}^{1}\bigg [ |0,u\rangle \langle 0,u| + \sum_{m=1}^{(n-1)/2} \big( q^{km} |m,u\rangle \langle m,u|+q^{-km} |-m,u\rangle \langle -m,u|\big)  \bigg]$$
After decomposing $L^{\dagger}(k,0)=L(-k,0)$ similarly, we get
$$L_{(-k,0)} \cong  \sum_{u=0}^{1}\bigg [ |0,u\rangle \langle 0,u| + \sum_{m=1}^{(n-1)/2} \big( q^{-km} |m,u\rangle \langle m,u|+ q^{km} |-m,u\rangle \langle -m,u|\big)  \bigg]$$
Putting all this together, we find

\begin{equation}
\begin{split}
 A^{\rm res}_{\mathbb{Z}_n} & \cong \sum_{k=0}^{n-1}  \Bigg\{\sum_{u=0,1}\bigg [  |0,u ) \langle 0,u| +\sum_{m=1}^{(n-1)/2}  \big( q^{km} |m,u) \langle m,u|+q^{-km} |-m,u) \langle -m,u|\big)  \bigg] \Bigg\} \otimes\\
& \Bigg\{ \sum_{u=0,1}\bigg [ |0,u\rangle ( 0,u| +\sum_{m=1}^{(n-1)/2}  \big( q^{-km} |m,u\rangle ( m,u|+q^{km} |-m,u\rangle ( -m,u| \big)  \bigg] \Bigg\}\frac{1}{n}.\notag
\end{split}
\end{equation}

Using $\frac{1}{n}\sum_k e^{2\pi k(i-j)/n}=\delta_{i,j}$ we get the matrix 
\begin{equation}
A^{ {\rm res} \{(m,b)(m',b')\} }_N=\delta_{m,m'}|m,b)(m',b'|, \notag
\end{equation}
where $\{(m,b)(m',b')\}$ label the physical indices. For each value of the physical index, the virtual matrices, of size $2n \times 2n$, of the tensor has a diagonal structure with $n$ two-dimensional blocks related to the $n$ irreps of $\mathbb{Z}_n$. The first one is denoted by the label $m=0$ and the others are grouped in pairs, those which labels $\pm m$, related to the $(n-1)/2$ bidimensional irreps of $D_{n}$. 

We denote such a pair by $(m)$. The tensor given by Eq.(\ref{eq:blockdesc}) is also diagonal in terms of the irreps of $D_{n}$, i.e. the associate irreps of $D_{n}$ are a single irrep. If we fix one block of the pair $(m)$, say $+m$, 
we obtain four different matrices
\begin{equation}
A^{{\rm res}\{(+m,b)(+m,b')\}}_N = |+m,b)(+m,b'|\equiv B_{b,b'}^{+m},
\end{equation}
where $[B_{b,b'}^{+m}]_{\alpha,\alpha'}=\delta_{b,\alpha}\delta_{b',\alpha'}$. These matrices span the whole space of $2\times 2$ matrices, so each block is injective \cite{PerezGarcia07}. Now we are going to act on our tensor with different symmetry operators at the physical level, and we will recover the permutations of the blocks of the matrices. The exchange will be between the blocks $\pm m$ belonging to the pair $(m)$. The symmetry operators are the irreps of the non-trivial extension, evaluated at the elements belonging to the non-trivial coset of $G$ by $N$, $\{ (g,1)|g\in N \}$, acting on each pair $(m_0)$ .

The  $(n-1)/2$ two-dimensional irreps of $D_{n}$ in the coset $\{ (k,1)|k\in \mathbb{Z}_n \}$  take the form
$$\Pi_m(k,1)=\Pi_m(k,0)\Pi_m(0,1)=\left(\begin{array}{cc}0&q^{ km} \\ q^{- km} & 0 \end{array}\right),$$
where $\Pi_m(0,1)$ is nothing but the Pauli matrix $\sigma_x$ which does not depend on $m$. For simplicity, we will deal only with the case of one block in detail reaching a single transposition. The operator associated with one block $(m_0)$, in analogy with Eq.(\ref{eq:sinop}), and is given by
$$S_{m_0}(l,l')=[( \Pi_{m_0}(l,1)\otimes \id_2)  \otimes  \id_{\rm rest})] \otimes[(\bar{\Pi}_{m_0}(l',1) \otimes  \id_2 ) \otimes  \id_{\rm rest}].$$
The left-hand side of the previous operator in the selected basis can be expressed as
$$\Pi_{m_0}(l,1)\otimes \id_2= \sum_{u=0}^{1} \big[ q^{ lm_0} |m_0,u\rangle \langle -m_0,u|+ q^{-lm_0} |-m_0,u\rangle \langle m_0,u| \big] ,$$
acting on the pair $(m_0)$ and the identity operator in the rest of the blocks. Therefore, if we act with this operator on the tensor, we notice that all blocks with $m\neq \pm m_0$ are not affected, but the pair with $(m)=(m_0)$ changes as  
\begin{equation*}
\begin{split}
\Pi_{m_0}(k,0)\Pi_{m_0}(l,1)=\Pi_{m_0}(k+l,1) \quad {\rm and}\\
\Pi_{m_0}(l',1)\Pi_{m_0}(-k,0)=\Pi_{m_0}(k+l',1)
\end{split}
\end{equation*}
on each side of the tensor product respectively. Let us analyze the action of the operator $S_{m_0}(l,l')$ looking at the virtual matrices of the modified tensor $\tilde{A}^{ {\rm res} }_N=S_{m_0}(l,l') {A}^{ {\rm res} }_N$:

$$\tilde{A}^{ {\rm res} \{(+m,b)(+m,b')\} }_N
 =\bigg\{ \begin{array}{lcc}   |+m,b)(+m, b'| & {\rm if} &+m\notin (m_0)
\\
q^{(l-l') m_0}|-m,b)(-m,b'|  & {\rm if} &+m\in (m_0) \end{array},$$

where only the non-zero elements are written. That is:
\begin{equation}
 \tilde{B}_{b,b'}^{+m}
 =\bigg\{ \begin{array}{lcc}   B_{b,b'}^{+m} & {\rm if} &+m\notin (m_0)
\\
 B_{b,b'}^{-m} & {\rm if} &+m\in (m_0) \end{array},\notag
\end{equation}
As a consequence, we obtain that the two blocks of the matrices, associated with the pair $(m_0)$ are exchanged:
\begin{equation}
\begin{tikzpicture}
                        \matrix (m) [matrix of math nodes,left delimiter=(,right delimiter=)] {
                                B^{+m_0}       & 0 \\
                                0       & B^{-m_0}     \\
                };
                        \draw[<->] ($(m-1-1.center) + (0.1,-0.1)$) to ($(m-2-2.center) + (-0.45,0.1)$);
\end{tikzpicture} \notag
\end{equation}
This permutation between the blocks $+m_0$ and $-m_0$ is nothing but the single transposition that we were looking for. The action of the operator does not depend on the element of $\mathbb{Z}_n$ in the set $\{ (l,1)|l\in\mathbb{Z}_n \}$, the non-trivial coset $[1]\neq[e]\cong \mathbb{Z}_n$, so the result is uniquely determined by the quotient group $\mathbb{Z}_2$ (in general $G_{sym} \cong G/N$). The action of the operator using elements of the subgroup $\{ (g,0)|g\in\mathbb{Z}_n \}$ is trivial because it does not permute the blocks of the virtual matrices. In this example, the multiplicity of the irreps of $N$ in each irrep of $G$ is one. Thus, the projective representation $V$ in Eq.(\ref{eq:VC}) does not play any role here. Instead, the only non-trivial action for this case is a permutation carried out by the induced representation of Eq.(\ref{eq:indurep}). Eq.(\ref{eq:actionblock}) translates as:
\beq
\left(S_{m_0}(g)\right )_p{A}^{ {\rm res} }_N={A}^{ {\rm res} }_N  \left(  (\sigma_x\otimes  \id_2)_{m_0}\otimes(\sigma_x\otimes  \id_2)_{m_0}   \right )_v,\notag
\eeq
where $g$ represents the non-trivial element of the group $\mathbb{Z}_2$. We can interpret this result as a symmetry breaking phase since the blocks exchanged correspond to linearly independent states of the ground subspace. When we take the operator from the trivial extension, we find that the symmetry is in a non-equivalent phase, characterized by a non-symmetry breaking pattern in the ground subspace.

In order to recover the other permutations, related to disjoint transpositions, we just act with the operator created by adding the different irreps associated  to the two interchangeable blocks. The operator associated with the pairs $(m_0),\cdots, (m_i)$ is given by
$$( \Pi_{m_0}(l_0,1)\oplus \cdots \oplus  \Pi_{m_i}(l_i,1))\otimes \id_2)  \otimes  \id^{\rm rest}] \otimes 
 [( \bar{\Pi}_{m_0}(l_0,1)\oplus \cdots \oplus  \bar{\Pi}_{m_i}(l_i,1))\otimes \id_2)  \otimes  \id^{\rm rest}] .$$
This operator carries out the transposition between the blocks $\pm m_0, \dots, \pm m_i$ in the virtual matrices. Again, the action is independent from the element $l_0,\cdots, l_i$, so it is uniquely determined by the element of the quotient group just as in the single transposition case. Therefore, we have recovered all the possible phases with symmetry group $\mathbb{Z}_2$ and degenerate ground state for  {MPS}. \\
It is straightforward to use Eq.(\ref{eq:locsympa}) to show that the parent tensor is left invariant by the symmetry operators.

\section{Discussion}

In this chapter we have studied two classes of {PEPS} related by anyon condensation (parent and restricted model). We have seen that the local invariance of the first, under the action of a group G, is broken in the second model. However, some residual symmetry persists: the restricted {PEPS} is left invariant by a smaller local symmetry ($G \vartriangleleft E$) plus a global symmetry $(Q \equiv E/G)$. This symmetry change, from local to global, is closely related to flux confinement and charge condensation. To get a microscopic understanding of these phenomena, we have analyzed how a background defined by the restricted model is affected by the insertion of (virtual) excitations of the parent model. Besides, we have seen that the residual global symmetry is represented by permutations of particle types within each anyonic sector. Also, the fractionalization of the symmetry on the charge sector has been identified. Similarly, when the model is placed on a non-trivial manifold, this residual (global) symmetry leaves the ground subspace invariant and does not act trivially on it. 
On top of that, Wilson loops (corresponding to unconfined excitations) also leave the ground subspace invariant and act non-trivially on it. This coincidence leads us to believe that the two types of operators might be related.

To summarize, we have constructed a representative of each phase, classified in the previous chapter as shown in diagram \eqref{Diathesis}, via an explicit (local) symmetry breaking or ungauging. 
The representatives are equivalent to $G$-isometric PEPS, {\it i.e.} renormalization fixed point so $(\phi, \omega)$ can be obtained straightforwardly. In the next chapter we deal with the local detection of $\omega$, see diagram \eqref{Diathesis} for the connection between chapters.

Next, we have researched {MPS} analogues of our findings. By combining the symmetry reduction discussed above with classical results in the theory of group representations \cite{Clifford37}, we have been able to re-derive all possible representations of an on-site global symmetry at the virtual level and hence the SPT classification of MPS.

The approach for charge condensation studied here could also be applied to charge confinement using the different tensor network realizations of the same quantum phase described in \cite{Shukla16}. In that case, the restricted tensor would act as a {\it flux condensator} for the parent model. In \cite{Haegeman15,Marien17,Duivenvoorden17} anyon condensation has been numerically studied in the framework of {PEPS} for different topological orders without symmetries. The authors of these works performed a local parametrized perturbation on the tensor and successfully identified the condensed and confined anyons pattern. In contrast we have analytically studied  pairs of phases, corresponding to the extreme points of an anyon condensation process modelled with {PEPS}. We have focused on discrete gauge theories (quantum doubles) which have allowed us to analyze the behavior of local/global symmetries in both phases through the  condensation. 

\newpage \cleardoublepage 

\chapter{Detection of symmetry fractionalization}\label{detecSF}

In the previous chapter we have described how to construct a representative of the $\mathcal{D}(G)$ phase enriched with an on-site symmetry $Q$. These representatives are given in a renormalization fixed point form, {\it i.e.}, zero correlation length and commuting parent Hamiltonian. This allows us to extract the characteristic maps $(\phi, \omega)$ of the quantum phase straightforwardly, see Diagram \eqref{Diathesis}. But in general, this is not an easy task. Besides that, the map $\phi$ can be extracted easily for any model just by checking how the anyons or the ground states permute under the symmetry. This is feasible since we are relying just on the distinguishability of the ground states or that of the anyons (this is taken for granted since we know the topological order). 
Then, we are still lacking a method to obtain $\omega$. 

Here we propose a universal (model independent) local order parameter $\langle  X_{\rm local}(\phi,\omega) \rangle$ that detects the class of $\omega$, given a $\phi$. 
We argue that it should work for models beyond renormalization fixed points as long as they are in the same topological phase as the $\mathcal{D}(G)$ and the correlation length remains finite. The numerical implementation of the proposed order parameter, left for future work, would be desirable. We will analytically show that, at least in the $G$-isometric point, $\langle  X_{\rm local}(\phi,\omega) \rangle$ correctly detects $\omega$, as stated in the following theorem.

\begin{tcolorbox}
\begin{theorem}[Symmetry fractionalization detection]\label{theo:sfdetect}
Given the groups $G$ and $Q$ and $H^2_\phi(Q,G)$, where the latter classifies the different patterns of symmetry fractionalization on the charge sector, there exists a local order parameter, constructed in \cref{sec:opproposal},
$$\langle \Psi_A  | X_{\rm local}(\phi,\omega)  | \Psi_A\rangle,$$
that completely detects the class of $\omega \in H^2_\phi(Q,G)$ for the following pairs $( G, Q )$.  $( G, Q )= ( \mathbb{Z}_2, \mathbb{Z}_2 ), ( \mathbb{Z}_2, \mathbb{Z}_2 \times \mathbb{Z}_2 ), \{ (\mathbb{Z}_p, \mathbb{Z}_p): p \; \text{prime} \}, ( \mathbb{Z}_4, \mathbb{Z}_2 ), ( Q_8, \mathbb{Z}_2 )$ and the state $| \Psi_A\rangle$ is a $G$-isometric PEPS with a global on-site symmetry of $Q$ described by $(\phi,\omega)$.
\end{theorem}
\end{tcolorbox}

The proof of \cref{theo:sfdetect} is divided in two parts. 
First, in \cref{sec:opproposal} we construct the general form of our local order parameter motivating its definition. Second, in \cref{sec:examples} we show how it works when applied to the considered pairs $(G,Q)$.
We notice that the classification of systems is defined by incorporating the relabelling of the elements of $Q$ in $H^2_\phi(Q,G)$ commented in \cref{obs:relabelling}. We show how, in some pairs $(G,Q)$, this relabelling can be incorporated in the order parameter. Let us first contextualize our results.\\

The novelty of our approach is that the order parameters are calculated locally, they rely on operations on few neighbouring particles on the bulk, that is, in a genuinely two-dimensional geometry. Previous proposals to detect the SF pattern \cite{Cincio15,Zaletel16, Zaletel17, Wang15, Qi15, Saadarmand16, Huang14} are based on reductions to effective 1D systems (dimensional compactification). The studied {2D} system is given the geometry of a long cylinder, and one-dimensional techniques, used in the context of {SPT} phases, are employed to detect the {SF} class of the {2D} system. Refs. \cite{Cincio15,Zaletel16, Zaletel17, Wang15, Qi15, Saadarmand16} focus on lattice and time reversal symmetries, on-site symmetries are dealt with in \cite{Huang14}. However, these compactification techniques miss resolution power: we show how this happens and also how we obtain a strictly finer phase distinction. \\

The power of our order parameters is illustrated with several combinations of topological content and symmetry. As we already mention in the statement of \cref{theo:sfdetect} we will consider the pairs $( G, Q ) = ( \mathbb{Z}_2, \mathbb{Z}_2 ), ( \mathbb{Z}_2, \mathbb{Z}_2 \times \mathbb{Z}_2 ), \{ (\mathbb{Z}_p, \mathbb{Z}_p): p \; \text{prime} \}, ( \mathbb{Z}_4, \mathbb{Z}_2 ), ( Q_8, \mathbb{Z}_2 )$ and exhibit order parameters that allow for full resolution between the various {SET} phases. These instances have been chosen because they simply illustrate a variety of scenarios, including why {SET} detection based on dimensional compactification may miss {SF} patterns. For all these examples, a \emph{strictly finer} phase resolution will be demonstrated. In the simplest case of the toric code, $( \mathbb{Z}_2, \mathbb{Z}_2 )$, we will find two distinct {SF} classes for which the {SPT} order parameter assumes the same value. The case $( \mathbb{Z}_2, \mathbb{Z}_2 \times \mathbb{Z}_2 )$ is interesting since {SPT} order parameter exhibits \emph{some coarse} distinction. The case $(\mathbb{Z}_4, \mathbb{Z}_2)$ involves permutation of anyon types; the case $( Q_8, \mathbb{Z}_2 )$ is particular in the sense that the topological content is non-abelian.  
 
The main idea of our proposal is to use the equivalence between symmetry fractionalization and braiding with an anyon, a relation explained in \cref{sec:SFcharges}, to detect the former using the latter. Our order parameters are given as an expectation value of a local operator which could facilitate their experimental realizations. They are similar in spirit to those proposed in \cite{Pollmann12, Haegeman12} for {1D} {SPT} phases.\\

Our work is based on the formalism of tensor network states. This has the advantage that any state that is potentially approximated by a PEPS is suitable for the use of our order parameters. On top of that, our order parameters are independent of the explicit realization of the quantum phase; they only depend on how the symmetry acts via the virtual indices (which only depends on the SF pattern). To obtain our order parameters we have found the gauge invariant quantities that determine the SF pattern. These quantities are given in terms of the virtual symmetry operators.

\section{Identifying the SF pattern}\label{sec:opproposal}

Let us consider a model which hosts a topological phase with an on-site global symmetry. We are interested in obtaining the SF pattern of the anyons present in the model, {\it i.e.} in resolving between the different possible {SET} phases. Our main assumption is that the low energy sector of the model can be represented as a PEPS. In particular as a $G$-injective PEPS which has a global symmetry and then, the tensors $A$ transform under the symmetry as in Eq.\eqref{locsym}: $U_q A= A (v_q\otimes w_q\otimes \myinv{v_q}\otimes \myinv{w_q} )$.\\

If we knew the matrices $v_q$, we could identify $\omega(k,q)=v_kv_qv^{-1}_{kq}$ just computing it. But in general, this is impossible. It is our goal to construct order parameters that identify the projective representation of the charges without relying on the knowledge of $v_q$. Given a $G$-isometric {PEPS} with a global on-site symmetry group $Q$, and a permutation action over the anyons $\phi$, we will look for a gauge-invariant quantity which distinguishes between inequivalent $2$-cocycles.
An important observation is that $\omega$ itself is not gauge-invariant, as Eq.\eqref{cocyclefreedom} shows. Its physical detection would be a meaningless endeavour.  But products of virtual symmetry operators,  
\beq
\lambda(v_{q_1},\dots, v_{q_n}) \equiv v_{q_1}v_{q_2}\cdots v_{q_n} \in G
\eeq
can be gauge invariant if the elements $\{ q_i \in Q \}$ are appropriately chosen.\\

Of course, we already know an important example of a gauge-invariant quantity: the result of braiding a flux around a probe charge. In \cref{chap:classsymGPEPS}, \cref{ob:brsf}, we have also mentioned a relation between braiding and symmetry fractionalization. These two facts suggest to construct an operator whose mean value will be analogous to the quantity of Eq.\eqref{braid-result}: the effect of braiding some flux, related to $\omega$ through $\lambda$, around a probe charge. \\

The order parameter we propose is: 
\begin{equation*}\label{orderparameter}
 \Lambda\equiv \langle \;  \mathcal{O}^\dagger_{ \sigma} (x',y)  \; P_{\pi} \; \left( U^{[s_1]}_{q_1}\otimes  \cdots \otimes U^{[s_m]}_{q_m}\right ) \;\mathcal{O}_{ \sigma} (x,y) \; \rangle,
\end{equation*}
where $\mathcal{O}_{ \sigma} (x,y)$ is the operator creating a pair of charges (see \cref{sec:introGinj}), $P_{\pi}$ is an operator representing some permutation $\pi$ of $m$ sites, the upper index $s_j$ denotes the site where $U_{q_j}$ acts, and the indices $x, y$ denote the sites where the charges are created. 
The permutation rearranges the virtual symmetry operators to obtain a gauge-invariant quantity $\lambda$ acting on the charges analogously as the element $g$ in Eq.\eqref{braid-result}:
$$ \langle  \;  \mathcal{O}^\dagger_{ \sigma} (x,y) \; \mathcal{B}^{[\sigma]}_\lambda \; \mathcal{O}_{ \sigma} (x,y) \; \rangle, $$ 
that is, $\lambda$ plays the role of the flux.
 An interesting identity relates $\Lambda$ with the phase factors resulting from braiding discussed above:
$$\hat{\Lambda} \equiv \Lambda/  \langle \; \mathcal{O}^\dagger_{ \sigma} (x',y) \; P_{\pi} \; \mathcal{O}_{ \sigma} (x,y)\; \rangle = \chi_\sigma(\lambda)/d_\sigma.$$
We remark that the proposed order parameter is \emph{local} in the sense that it can be written as an operator, $X_{\rm local}$, acting on a finite region of the lattice $\Lambda=\langle X_{\rm local} \rangle=\tr[\rho X_{\rm local} ]$. \\

The classification of $(\omega, \phi)$ is equivalent to the classification of inequivalent group extensions, $E$, of $Q\cong E/G$ by $G \triangleleft E$. In fact the maps $(\omega, \phi)$ can also be obtained when considering $v_q$ as an element of $E$. The assignment is done by choosing $v_q$ as some element of $E$ for each $q\in Q$ which preserves the quotient map $E\to E/G\cong Q$. This analogy allows us to see $\lambda$ as an invariant quantity of the extension group. Therefore, given ($G,Q, \phi$), the strategy is: $(i)$ we search for all the possible group extensions which characterize all possible symmetry actions, $(ii)$ we look for a gauge-invariant quantity (in the form of $\lambda=\Pi_i v_{q_i}$) that distinguishes between the different extensions, $(iii)$ we design the order parameter $\Lambda$, that is, we identify the appropriate permutation of sites, and apply the suitable symmetry operators, so that the result is equivalent to the braiding of a probe charge with the 'flux' $\lambda \in G$.

\section{Proof of \cref{theo:sfdetect}}\label{sec:examples}
We prove \cref{theo:sfdetect} case by case. To do that we first show which are the gauge invariant quantities and then we explicitly construct the associated order parameter. 

We notice that the SET phases analyzed in this section are introduced in \cite{Hermele14,Tarantino16} as exactly solvable lattice models or in their $G$-isometric {PEPS} representation in \cref{chap:cond} (see also subsequent work \cite{Williamson17}). The necessary background on group extensions can be found in Appendix \ref{ap:ext}.

\subsection{Toric Code with $\mathbb{Z}_2$ symmetry}\label{TCZ2}

 The Toric Code, $ G = \mathbb{Z}_2 = \{ +1, -1; \times \}$, has three non-trivial anyons: the charge $\sigma$, the flux $m$ and the combination of both. We consider this model with an internal symmetry  $Q = \mathbb{Z}_2 = \{ e, q; q^2 = e \}$ which does not permute the flux and the charge. There are two possible {SF} patterns for the charge as predicted by the identity $H^2(\mathbb{Z}_2,\mathbb{Z}_2)=\mathbb{Z}_2$. The non-trivial {SF} class is characterized by the following projective action on the charge: 
 $$
(\Phi_q \circ \Phi_q)(C_\sigma)=  -1\times C_\sigma \equiv B^{[\sigma]}_m (C_\sigma).
$$ 
That is, the charge picks up a sign, equivalent to the braiding with the flux, when the symmetry acts twice on it. Notice that because the charge and anticharge are the same particle and they are created together under a local operation, the two minus signs globally cancel out. This is consistent with the fact that globally the symmetry acts linearly. The only non-trivial cocycle here is $\omega(q,q)=-1$. One possible gauge-invariant quantity is $\lambda=\omega(e,e)\omega(q,q)=v^2_{q}$ which gives $u_{-1}$ in the non-trivial case and is the identity element $u_{+1}=\id$ for the trivial {SF} class. $\lambda$ is gauge-invariant because $v'^2_{q} = v^2_{q} l^2_q$ and $l_q \in \mathbb{Z}_2$, where of course $l^2_q=e$.  The associated group extensions are $\mathbb{Z}_4$ for the non-trivial case and $\mathbb{Z}_2\times \mathbb{Z}_2$ for the trivial {SF} class. In this example, the gauge-invariant quantity detects whether there are elements of the extension group with order greater than two.
  
We now describe a protocol, and the corresponding diagrams, that distinguishes between these two {SET} phases:
\begin{equation}\label{transprotocol}
   \begin{tikzpicture}[scale=1]
              \pic at (0,0,0.7) {3dpepsshort};
        \pic at (0,1.5,0.7) {3dpepsdownshort};
        \pic at (0.5,0,0) {3dpeps};
        \pic at (0.5,1.5,0) {3dpepsdown};
         \pic at (0.5,0,1.4) {3dpeps};
        \pic at (0.5,1.5,1.4) {3dpepsdown};
                \pic at (1,0,0) {3dpeps};
        \pic at (1,1.5,0) {3dpepsdown};
         \pic at (1,0,1.4) {3dpeps};
           \pic at (1.5,0,0.7) {3dpeps};
        \pic at (1.5,1.5,0.7) {3dpepsdown};
	           \pic at (2,0,0.7) {3dpepsshort};
        \pic at (2,1.5,0.7) {3dpepsdownshort};
          \begin{scope}[canvas is xy plane at z=0]
          \draw[preaction={draw, line width=1.2pt, white},blue, ] (0.5,1.5,0.7) to  (0.5,1.05,0.7);
          \draw[preaction={draw, line width=1.2pt, white},blue, ] (1,1.5,0.7) to  (1,1.05,0.7);
          \draw[preaction={draw, line width=1.2pt, white},blue, ] (1,1.1,0.7) to [out=-90, in=90] (0.5,0,0.7);
          \draw[preaction={draw, line width=1.2pt, white},blue, ] (0.5,1.1,0.7) to [out=-90, in=90] (1,0,0.7);
        \end{scope}   
        \draw[ ] (1,1.5,1.4)--(1,1.31,1.4);
    \begin{scope}[canvas is zx plane at y=1.5]
      \draw[preaction={draw, line width=1.2pt, white}, ] (0.9,1)--(1.9,1);
      \draw[preaction={draw, line width=1.2pt, white}, ] (1.4,0.6)--(1.4,1.4);
      \filldraw (1.4,1) circle (0.07);
    \end{scope} 
         \pic at (1,0,0.7) {3dpepsp};
        \pic at (1,1.5,0.7) {3dpepsdownp};
        \pic at (0.5,0,0.7) {3dpepsp};
        \pic at (0.5,1.5,0.7) {3dpepsdownp};
             \filldraw[draw=black,fill=orange, ] (1.25,0,0.5) rectangle (1.25,0,0.9 );
            \filldraw[draw=black,fill=orange, ] (0.25,1.5,0.5) rectangle (0.25,1.5,0.9); 
              \filldraw[draw=black,fill=orange, ] (1.75,1.5,0.5) rectangle (1.75,1.5,0.9 );  
            \filldraw[draw=black,fill=orange, ] (1.75,0,0.5) rectangle (1.75,0,0.9  );

              \filldraw[draw=black,fill=purple, ] (0.52,0.95,0.7) circle (0.07);
              \filldraw[draw=black,fill=purple, ] (0.98,0.95,0.7) circle (0.07);
               \node at (4,1.6,0) {{\sf (i)  Create the excited state with}};
               \node at (4.1,1.3,0) {{\sf  2 charges: $\mathcal{O}_{\sigma} (x,y) |{\rm TC} \rangle$}};
               \node at (4,0.9,0) {{\sf (ii) Apply the on-site symmetry}};
                 \node at (4,0.6,0) {{\sf operators: $U^{\otimes 2}_q$}};
                 \node at (4,0.2,0) {{\sf (iii) Permute sites: $P_{(12)}$}};
                   \node at (4,-0.2,0) {{\sf (iv) Project onto $\langle {\rm TC}|  \mathcal{O}^\dagger_{\sigma} (x',y)$}};
        \end{tikzpicture},
        \end{equation}
where $P_{(12)}$ denotes the permutation of the two sites where the symmetry acts. 
We are assuming the contractions of the physical indices, between the bra and the ket layers. For the sake of clarity we will not draw them.
The order parameter associated to this protocol is the following:
 $$
 \Lambda = \langle {\rm TC}| \mathcal{O}^\dagger_{\sigma} (x',y) P_{(12)} U^{\otimes 2}_q  \mathcal{O}_{\sigma} (x,y)|{\rm TC}\rangle.
 $$
Notice that the charges appearing in the ket and the bra are not placed on the same links: $x'\neq x$. Using Eq.(\ref{Gisoconca}), we obtain the following:
\begin{equation*}
\hat{\Lambda}=
\frac{1}{|G|}
\sum_{b\in G}\;
   \begin{tikzpicture}[scale=1]
                \node[anchor= south west] at (0.5,0.75,-0.2) {$\myinv{b}$};
        \filldraw (0.5,0.75,-0.2) circle (0.05);
                      \node[anchor=south west] at (1,0.75,-0.2) {$\myinv{b}$};
        \filldraw (1,0.75,-0.2) circle (0.05);
\draw[ ] (0.5,1,0.7)--(0.5,1,-0.2)--(0.5,0,-0.2)--(0.5,0,0.7);
\draw[ ] (1,1,0.7)--(1,1,-0.2)--(1,0,-0.2)--(1,0,0.7);
\draw[preaction={draw, line width=1.2pt, white}, ] (0,0,0.7) rectangle (1.5,1,0.7);
          \draw[preaction={draw, line width=1.2pt, white}, , blue] (1,1,0.7) to [out=-90, in=90] (0.5,0,0.7);
          \draw[preaction={draw, line width=1.2pt, white}, , blue] (0.5,1,0.7) to [out=-90, in=90] (1,0,0.7);  
              \filldraw[draw=black,fill=purple, ] (0.55,0.75,0.7) circle (0.07);
              \filldraw[draw=black,fill=purple, ] (0.95,0.75,0.7) circle (0.07);
              \draw[preaction={draw, line width=1.2pt, white}, ] (0.5,1,0.7)--(0.5,1,1.6)--(0.5,0,1.6)--(0.5,0,0.7);
              \draw[preaction={draw, line width=1.2pt, white}, ] (1,1,0.7)--(1,1,1.6)--(1,0,1.6)--(1,0,0.7);
         \pic at (1,0,0.7) {3dpepsp};
        \pic at (1,1,0.7) {3dpepsdownp};
        \pic at (0.5,0,0.7) {3dpepsp};
        \pic at (0.5,1,0.7) {3dpepsdownp};
                             \filldraw[draw=black,fill=orange, ] (1.25,0,0.5) rectangle (1.25,0,0.9 );
            \filldraw[draw=black,fill=orange, ] (0.25,1,0.5) rectangle (0.25,1,0.9);
                        \node[anchor= north east] at (0.5,0.25,1.6) {${b}$};
        \filldraw (0.5,0.25,1.6) circle (0.05);
                                \node[anchor= north east] at (1,0.25,1.6) {${b}$};
        \filldraw (1,0.25,1.6) circle (0.05); 
                                  \node[anchor=west] at (1.5,0.5,0.7) {$\myinv{b}$};
        \filldraw (1.5,0.5,0.7) circle (0.05);
 \node[anchor=east] at (0,0.5,0.7) {${b}$};
        \filldraw (0,0.5,0.7) circle (0.05);
\end{tikzpicture}
\times
\begin{tikzpicture}[scale=1]    
 \draw[preaction={draw, line width=1.2pt, white}, ] (1.8,0,0.7) rectangle (2.6,1,0.7);
              \filldraw[draw=black,fill=orange, ] (2.2,1,0.5) rectangle (2.2,1,0.9 );  
            \filldraw[draw=black,fill=orange, ] (2.2,0,0.5) rectangle (2.2,0,0.9  );
             \node[anchor=east] at (1.8,0.5,0.7) {${b}$};
        \filldraw (1.8,0.5,0.7) circle (0.05);
         \node[anchor=west] at (2.6,0.5,0.7) {$\myinv{b}$};
        \filldraw (2.6,0.5,0.7) circle (0.05);
   \end{tikzpicture},
  \end{equation*}
where the sum over $b\in \mathbb{Z}_2$ comes from the concatenation of the non-permuted sites. 
We notice that the position of the two charges placed on the rightmost edges, one on top of the other, in \eqref{transprotocol} can be changed without altering the value of $\hat{\Lambda}$.  
Using Eq.(\ref{Ginje}) and Eq.\eqref{locsym} we get:
\begin{align}\label{loopcalc}
\hat{\Lambda} = \frac{1}{|G|^3}&\sum_{a,b,c\in G} 
  \begin{tikzpicture}[baseline=+5mm]
        \draw[ , blue] (1.5,1) to [out=-90, in=90] (0.6,0);
        \draw[ ,blue] (1.4,1) to [out=-90, in=90] (0.5,0);
       \draw[preaction={draw, line width=1.2pt, white}, ,blue] (0.6,1) to [out=-90, in=90] (1.5,0);
        \draw[preaction={draw, line width=1.2pt, white}, ,blue] (0.5,1) to [out=-90, in=90] (1.4,0);
        \draw[ ] (0.6,0)--(1.4,0);
        \draw[ ] (1.5,0)--(1.9,0);
        \draw[ ] (1.9,1)--(1.9,0);
        \draw[ ] (0.6,1)--(1.4,1);
        \draw[ ] (1.5,1)--(1.9,1);
        \draw[ ] (0.5,1)--(0,1);
        \draw[ ] (0.5,0)--(0,0);
        \draw[ ] (0,1)--(0,0);
        \node[anchor=south] at (0.35,1) {$\myinv{v_{q}}$};
        \filldraw[draw=black,fill=red, ] (0.4,1) circle (0.05);
         \node[anchor=east] at (0.7,0.7) {$\myinv{a}$};
        \filldraw (0.58,0.75) circle (0.05);
        \node[anchor=west] at (1.4,0.7) {$c$};
        \filldraw (1.43,0.75) circle (0.05);
        \node[anchor=east] at (0.7,0.2) {$\myinv{c}$};
        \filldraw (0.58,0.25) circle (0.05);
        \node[anchor=west] at (1.4,0.2) {$a$};
        \filldraw (1.43,0.25) circle (0.05);
        \node[anchor=south] at (0.8,1) {$v_{q}$};
        \filldraw[draw=black,fill=red, ] (0.8,1) circle (0.05);
         \node[anchor=south] at (1.3,1) {$\myinv{v_{q}}$};
        \filldraw[draw=black,fill=red, ] (1.3,1) circle (0.05);
        \node[anchor=south, ] at (1.9,1) {$v_{q} $};
        \filldraw[draw=black,fill=red] (1.7,1) circle (0.05);
        \node[anchor=east] at (0,0.5) {$b$};
        \filldraw (0,0.5) circle (0.05);
        \node[anchor=west, ] at (1.9,0.5) {$\myinv{b}$};
        \filldraw (1.9,0.5) circle (0.05);
        \filldraw[draw=black,fill=orange, ] (1.65,-0.1) rectangle (1.85,0.1);
        \filldraw[draw=black,fill=orange, ] (0.05,0.9) rectangle (0.25,1.1);
         \end{tikzpicture}  \times  \;
  \begin{tikzpicture}
       \draw[ ] (0,0) rectangle (1,1);
        \node[anchor=west] at (0,0.5) {$b$};
        \filldraw (0,0.5) circle (0.05);
       \node[anchor=east] at (1.1,0.5) {$\myinv{b}$};
        \filldraw (1,0.5) circle (0.05);
        \filldraw[draw=black,fill=orange, ] (0.4,0.9) rectangle (0.6,1.1);
      \filldraw[draw=black,fill=orange, ] (0.4,-0.1) rectangle (0.6,0.1);
     \end{tikzpicture} \notag \\
 &  \times \;
    \begin{tikzpicture}
          \draw[preaction={draw, line width=1.2pt, white}, ,blue] (0,0,-1) to [out=90, in=-90] (1,1,-1);
          \draw[preaction={draw, line width=1.2pt, white}, , blue] (1,0,-1) to [out=90, in=-90] (0,1,-1);
           \draw[ ] (0,0,0)--(0,1,0);
       \draw[preaction={draw, line width=1.2pt, white}, ] (1,0,0)--(1,1,0);
         \node[anchor=west] at (0.05,0.8,-1) {$\myinv{a}$};
        \filldraw (0.05,0.8,-1) circle (0.05);
        \node[anchor=east] at (0.2,0.3,-1) {$\myinv{c}$};
        \filldraw (0.05,0.2,-1) circle (0.05);
        \node[anchor=east] at (0,0.5,0) {$b$};
        \filldraw (0,0.5,0) circle (0.05);
        \node[anchor=east] at (1,0.5,0) {$b$};
        \filldraw (1,0.5,0) circle (0.05);
   \begin{scope}[canvas is zx plane at y=0]
     \draw[ ] (0,0)--(-1,0);
     \draw[ ] (0,1)--(-1,1);
   \end{scope} 
  \begin{scope}[canvas is zx plane at y=1]
      \draw[ ] (0,0)--(-1,0);
      \draw[ ] (0,1)--(-1,1);
   \end{scope} 
            \node[anchor=south] at (0,1,-0.5) {$v_{q}$};
        \filldraw[draw=black,fill=red, ] (0,1,-0.5) circle (0.06);
      \node[anchor=south] at (1,1,-0.5) {$v_{q}$};
        \filldraw[draw=black,fill=red, ] (1,1,-0.5) circle (0.06);
      \end{tikzpicture}
      \; \times  
        \begin{tikzpicture}
        \draw[ ] (0,0,0)--(0,1,0);
         \draw[ ] (1,0,0)--(1,1,0);
          \draw[preaction={draw, line width=1.2pt, white}, ,blue] (0,0,1) to [out=90, in=-90] (1,1,1);
          \draw[preaction={draw, line width=1.2pt, white}, ,blue] (1,0,1) to [out=90, in=-90] (0,1,1);
         \node[anchor=east] at (0.05,0.8,1) {$a$};
        \filldraw (0.05,0.8,1) circle (0.05);
        \node[anchor=east] at (0.05,0.2,1) {$c$};
        \filldraw (0.05,0.2,1) circle (0.05);
        \node[anchor=west] at (0,0.5,0) {$\myinv{b}$};
        \filldraw (0,0.5,0) circle (0.05);
        \node[anchor=east] at (1.1,0.5,0) {$\myinv{b}$};
        \filldraw (1,0.5,0) circle (0.05);
     \begin{scope}[canvas is zx plane at y=0]
     \draw[ ] (0,0)--(1,0);
     \draw[ ] (0,1)--(1,1);
      \end{scope} 
     \begin{scope}[canvas is zx plane at y=1]
      \draw[ ] (0,0)--(1,0);
      \draw[ ] (0,1)--(1,1);
     \end{scope} 
     \node[anchor=south] at (0,1,0.5) {$\myinv{v_{q}}$};
        \filldraw[draw=black,fill=red, ] (0,1,0.5) circle (0.06);
      \node[anchor=south] at (1,1,0.5) {$\myinv{v_{q}}$};
        \filldraw[draw=black,fill=red, ] (1,1,0.5) circle (0.06);
   \end{tikzpicture}.
 \end{align}
 The first factor in the sum is $\tr[ v^{-1}_q c^{-1} b  C^{[u]}_\sigma v^{-1}_q a^{-1}  c v_qb^{-1} C^{[d]}_\sigma a v_q ]$ and simplifies to  $\tr[  C^{[u]}_\sigma {b}^{-1} (a {c}^{-1})^{-1}  C^{[d]}_\sigma (a {c}^{-1}) b ]$. Diagrammatically,
  \begin{equation*}
 \begin{tikzpicture}
        \draw[ , blue] (1.5,1) to [out=-90, in=90] (0.6,0);
        \draw[ , blue] (1.4,1) to [out=-90, in=90] (0.5,0);
       \draw[ , preaction={draw, line width=1pt, white}, blue] (0.6,1) to [out=-90, in=90] (1.5,0);
        \draw[preaction={draw, line width=1pt, white}, , blue] (0.5,1) to [out=-90, in=90] (1.4,0);
        \draw[ ] (0.6,0)--(1.4,0);
        \draw[ ] (1.5,0)--(1.9,0);
        \draw[ ] (1.9,1)--(1.9,0);
        \draw[ ] (0.6,1)--(1.4,1);
        \draw[ ] (1.5,1)--(1.9,1);
        \draw[ ] (0.5,1)--(0,1);
        \draw[ ] (0.5,0)--(0,0);
        \draw[ ] (0,1)--(0,0);
        \node[anchor=south] at (0.35,1) {$\myinv{v_{q}}$};
        \filldraw[ ,draw=black,fill=red] (0.4,1) circle (0.06);
         \node[anchor=east] at (0.7,0.7) {$\myinv{a}$};
        \filldraw[ ] (0.58,0.75) circle (0.05);
        \node[anchor=west] at (1.4,0.7) {$c$};
        \filldraw[ ] (1.43,0.75) circle (0.05);
        \node[anchor=east] at (0.7,0.2) {$\myinv{c}$};
        \filldraw[ ] (0.58,0.25) circle (0.05);
        \node[anchor=west] at (1.4,0.2) {$a$};
        \filldraw[ ] (1.43,0.25) circle (0.05);
        \node[anchor=south] at (0.8,1) {$v_{q}$};
        \filldraw[ , draw=black,fill=red] (0.8,1) circle (0.06);
         \node[anchor=south] at (1.3,1) {$\myinv{v_{q}}$};
        \filldraw[ , draw=black,fill=red] (1.3,1) circle (0.06);
        \node[anchor=south] at (1.9,1) {$v_{q} $};
        \filldraw[ ,draw=black,fill=red] (1.7,1) circle (0.06);
        \node[anchor=east] at (0,0.5) {$b$};
        \filldraw[ ] (0,0.5) circle (0.05);
        \node[anchor=west] at (1.9,0.5) {$\myinv{b}$};
        \filldraw[ ] (1.9,0.5) circle (0.05);
        \filldraw[ ,draw=black,fill=orange] (1.65,-0.1) rectangle (1.85,0.1);
        \filldraw[ ,draw=black,fill=orange] (0.05,0.9) rectangle (0.25,1.1);
         \end{tikzpicture} =
             \begin{tikzpicture}[baseline=-2mm]
 \draw[ ] (-0.4,0.3) rectangle (0.4,-0.3); 
    \filldraw[ , draw=black,fill=orange] (-0.1,-0.4) rectangle (0.1,-0.2);
    \filldraw[ , draw=black,fill=orange] (-0.1,0.4) rectangle (0.1,0.2);
    \filldraw[draw=blue,fill=white, ]  (-0.25,-0.3) circle (0.06);
   \node[anchor=north ] at (-0.5,-0.3)  {$\myinv{(a\myinv{c})}$}; 
     \filldraw[blue, ]  (0.25,-0.3) circle (0.06);
     \node[anchor=north] at (0.5,-0.3)  {${a\myinv{c}}$}; 
         \node[anchor=west] at (-0.4,0)  {${b}$};
          \filldraw (-0.4,0)  circle (0.05);
 \node[anchor=west] at (0.4,0)  {$\myinv{b}$};
        \filldraw (0.4,0) circle (0.05);
 \end{tikzpicture}.
 \end{equation*}
We see that all symmetry operators cancel out. Together with the second factor, $\tr[  \bar{C}^{[u]}_\sigma  b^{-1} \bar{C}^{[d]}_\sigma b ]$, these two loops are equal to the result of the braiding detection, Eq.(\ref{braid-result}), with $g\equiv a {c}^{-1}$ . The dependence on the {SF} class lies in the remaining two loops, both equal to $\tr[a v^{-1}_q b^{-1} c v^{-1}_q b^{-1}]= \tr[a c v^{-2}_q]= |G|\delta_{a ,v^2_qc^{-1}}$. These factors can only be non-zero if $ac^{-1}= v^2_q \equiv \lambda$. Notice that the two loops implying the previous identity are not part of the plane formed by the charges and the permutation in Eq.(\ref{transprotocol}): it is an intrinsic 2D effect of $G$-injective PEPS. In summary, we can write
\begin{align*}
\hat{\Lambda}=& \;\frac{1}{2} \sum_{b\in \mathbb{Z}_2} \tr[ C^{[u]}_\sigma  v^{-2}_q b^{-1} C^{[d]}_\sigma b v^2_q  ]\times \tr[  \bar{C}^{[u]}_\sigma  b^{-1} \bar{C}^{[d]}_\sigma b ]  \\
=&\; \frac{1}{2} \sum_{b\in \mathbb{Z}_2} \;
 \begin{tikzpicture}[baseline=-1mm]
 \draw[ ] (-0.4,0.3) rectangle (0.4,-0.3); 
    \filldraw[draw=black,fill=orange, ] (-0.1,-0.4) rectangle (0.1,-0.2);
    \filldraw[draw=black,fill=orange, ] (-0.1,0.4) rectangle (0.1,0.2);
    \filldraw[draw=blue,fill=white, ]  (-0.25,-0.3) circle (0.06);
     \filldraw[blue]  (0.25,-0.3) circle (0.06);
        \node[anchor=north ] at (-0.4,-0.3)  {$\myinv{\lambda}$}; 
     \node[anchor=north] at (0.4,-0.3)  {${\lambda}$}; 
         \node[anchor=west] at (-0.4,0)  {${b}$};
          \filldraw (-0.4,0)  circle (0.05);
 \node[anchor=west] at (0.4,0)  {$\myinv{b}$};
        \filldraw (0.4,0) circle (0.05);
 \end{tikzpicture}
\times
   \begin{tikzpicture}
    \draw[ ] (-0.4,0.3) rectangle (0.4,-0.3); 
    \filldraw[draw=black,fill=orange, ] (-0.1,-0.4) rectangle (0.1,-0.2);
    \filldraw[draw=black,fill=orange, ] (-0.1,0.4) rectangle (0.1,0.2);
      \node[anchor=west] at (-0.4,0)  {${b}$};
          \filldraw (-0.4,0)  circle (0.05);
 \node[anchor=west] at (0.4,0)  {$\myinv{b}$};
        \filldraw (0.4,0) circle (0.05);
 \end{tikzpicture}
=\chi_{\sigma}(\lambda). 
\end{align*}
The value of $\hat{\Lambda}$ is equal to $+ 1$ or $-1$, depending on the {SF} class; trivial or non-trivial respectively.      
\\

Let us compare our approach with a discrimination based on mapping a 2D {SET} system to a cylinder, and using {SPT} classification tools. Concretely, for each anyon type one constructs a 1D system, and places an anyon and its antiparticle on each edge of the cylinder. Then one could wonder if the symmetric projective action on the bonds in those chains, characterized by SPT phases in 1D, is equivalent to the SF pattern of the corresponding anyons. The algebraic object that classifies SPT phases in 1D is $H^2(Q,U(1))$, whereas in SF for 2D systems is $H^2(Q,G)$: the topological order constraints the values of the projective actions from $U(1)$ to $G$ when $G$ is abelian. Therefore for each anyon of the TC with a $\mathbb{Z}_2$ symmetry, the corresponding SPT phase is in $H^2(\mathbb{Z}_2,U(1))$, which is trivial. Therefore \emph{no} signature of the SF pattern of the charge for a unitary simmetry given by $\mathbb{Z}_2$ can be found with compactification methods.\\
 
In general for symmetries coming from a cyclic group, the 1D {SPT} phase is always trivial: $H^2(\mathbb{Z}_n,U(1))=1$, so any non-trivial SF pattern of a cyclic group have no 1D SPT analogue (this situation includes examples \ref{secq}, \ref{ex:perm} and \ref{ex:nonab} below). In the next example we show that even when there are non-trivial {SPT} phases after compactification, they do not fully resolve between all the SF patterns of the anyons. \\

Another reason for the compactification to fail, apart from the mismatch between cohomological identities, is that the mapping is highly non-local. This is because in the 1D picture a site corresponds to the blocked sites of a non-contractible loop in the torus; in the direction that we have closed to form the cylinder. Then one would expect that a non-local mapping could not capture completely a local effect, the SF pattern.

\subsection{ Toric Code with $\mathbb{Z}_2 \times \mathbb{Z}_2$ symmetry }

The symmetry $\mathbb{Z}_2 \times \mathbb{Z}_2=\{e,x,y,z\} \subset SO(3)$, considered as $\pi$ rotations over each axis, acting on the Toric Code gives four inequivalent patterns of {SF} on the charge (when the {SF} of the flux is trivial and there is no permutation of anyons). There are three non-trivial {SF} classes, associated with the group extensions $\mathbb{Z}_4\times \mathbb{Z}_2, D_8,Q_8$ and one trivial class, associated with $\mathbb{Z}_2\times\mathbb{Z}_2\times\mathbb{Z}_2\equiv \mathbb{Z}^3_2$. These four classes come from $H^2( Q, G ) = H^2(\mathbb{Z}_2 \times \mathbb{Z}_2 ,\mathbb{Z}_2)=\mathbb{Z}^3_2$ when one has incoporated the redundancy of cocycles by relabeling the elements of $\mathbb{Z}_2 \times \mathbb{Z}_2$ (see end of this section). 
The gauge-invariant quantity that identifies any of the four {SF} classes is $v^2_q$ for each $q=x,y,z$. We are led to choose gauge-invariant quantities to be the triple
\begin{equation}\label{eq:z2z2z2} 
\lambda = \{v^2_x,v^2_z,v^2_y\}.
\end{equation} 
The concrete values for each SF pattern are shown in Table \ref{tablev}. Then, the scheme presented in Section \ref{TCZ2} can be used in order to calculate each element of the triple (the three non-trivial elements of the symmetry group) to distinguish between all SF patterns.

 \begin{table}[htb]
\centering
\begin{tabular}{c|c|c|c|c|}
 & $\mathbb{Z}^3_2$ & $\mathbb{Z}_4\times \mathbb{Z}_2$ & $D_8$ & $Q_8$ \\ \hline 
$\{v^2_x,v^2_z,v^2_y\}$ & $\{1,1,1\}$ & $\{1,-1,-1\}$ & $\{1,1,-1\}$ & $\{-1,-1,-1\}$ \\  \hline
$v_xv_zv^{-1}_x v^{-1}_z$ & $+1$ & $+1$ & $-1$ & $-1$ \\ \hline
\end{tabular}
\caption{Comparison between the values of the gauge-invariant quantities.}
\label{tablev}
 \end{table}
 
Let us compare this order parameter with the analogous quantity used in the detection of {1D} system invariants, SPT phases, under the symmetry  $\mathbb{Z}_2 \times \mathbb{Z}_2 \subset SO(3)$. The discrete magnitude relevant in that problem is $H^2(\mathbb{Z}_2 \times \mathbb{Z}_2,U(1))=\mathbb{Z}_2$, predicting two non-equivalent phases where the non-trivial one (the trivial corresponds to a product state) is the Haldane phase \cite{Pollmann12}. 

The order parameter for the effective 1D system obtained when putting the system around a (long) cylinder is \cite{Haegeman12, Pollmann12}:
$$
v_xv_zv^{-1}_x v^{-1}_z=\omega(x,z)v_{xz}v^{-1}_{zx}\omega^{-1}(z,x)=\omega(x,z)\omega^{-1}(z,x).
$$
A comparison between Table \ref{tablev} and Fig. \ref{fig:cocycleconstruction} allows to contrast the {SPT} approach with ours: the former doesn't resolve between the 4 {SF} phases, whereas the latter does. The {SPT} approach is only able to discriminate between the sets corresponding to $\{ \mathbb{Z}^3_2, \mathbb{Z}_4\times \mathbb{Z}_2\}$ and $\{D_8,Q_8\}$. This has important implications since using compactification methods cannot even ensure us that we are in the trivial phase, no SF on the anyons.
\begin{figure}[ht!]
\begin{center}
\includegraphics[scale=1]{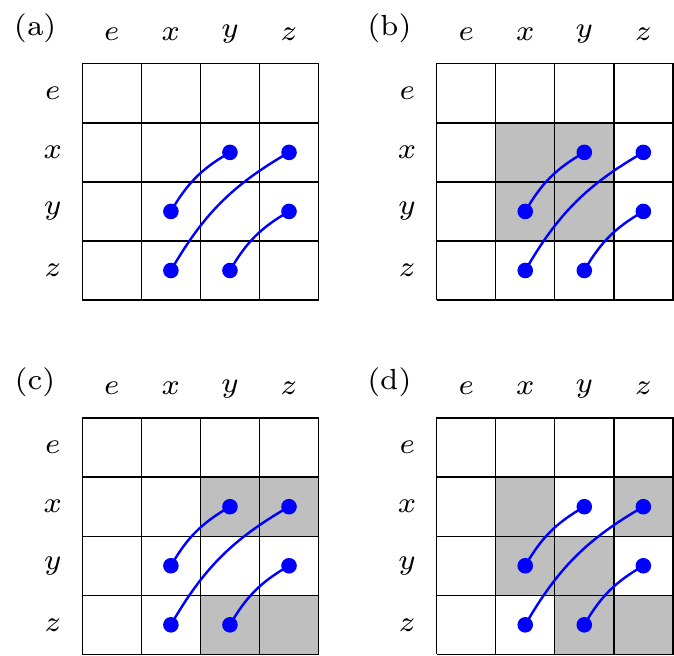}
\caption{(Color online) Matrix representation, $(M)_{q,k}=\omega(q,k)$, of the cocycles related to the four $\mathbb{Z}_2\times \mathbb{Z}_2$ {SF} patterns in the TC. The grey shaded areas correspond to the cocycle value $-1$ and the white to $+1$ (the matrices are shown for a specific gauge choice). (a) corresponds to $\mathbb{Z}^3_2$, (b) to $\mathbb{Z}_4 \times \mathbb{Z}_2$ (c) to $D_8$ and (d) to $Q_8$.  The order parameter (\ref{eq:z2z2z2}), the set made of the lower 3 diagonal elements, is distinct for each phase. In turn, the {SPT}-induced order parameter, the set of product of upper and lower diagonal elements identifies (a) with (b) and (c) with (d), blue line connecting $\omega(k,q)$ and $\omega^{-1} (q,k)$ (the concrete values are given in Table \ref{tablev}). }
\label{fig:cocycleconstruction}
\end{center}
\end{figure}

Now we show how the group $H^2(\mathbb{Z}_2 \times \mathbb{Z}_2 ,\mathbb{Z}_2)=\mathbb{Z}_2 \times\mathbb{Z}_2 \times\mathbb{Z}_2 \equiv \mathbb{Z}^3_2$ is reduced to four classes, related to $\{ \mathbb{Z}^3_2, \mathbb{Z}_4\times \mathbb{Z}_2, D_8,Q_8\}$, after taking the relabelling in the group $\mathbb{Z}_2 \times \mathbb{Z}_2=\{e,x,y,xy\}$ into account. 

We can always choose that $\omega(e,q)=\omega(q,e)=1$ for any $q\in \mathbb{Z}_2 \times \mathbb{Z}_2$   so we only write the elements $\omega(q,k)$ where $q$ and $k$ are both different from $e$. Thus we represent the different 2-cocycles as a $3\times 3$ table where white indicates $+1$ and grey $-1$. The trivial 2-cocycles, 2-coboundaries, can be constructed as $\omega(q,k)=g_q g_k g^{-1}_{qk}$ where $g: \mathbb{Z}_2 \times \mathbb{Z}_2 \to \mathbb{Z}_2$.  There are two distinct 2-coboundaries shown in Fig. \ref{fig:cocyclesQ8Z2}. These 2-coboundaries are the gauge freedom of the 2-cocycles in the second cohomology group and correspond to the trivial extension, this is the direct product of $\mathbb{Z}_2 \times \mathbb{Z}_2$ and $\mathbb{Z}_2$. We note that the quantities ${\lambda}=\{\omega(q,q)\}$ and $\lambda'=\omega(q,k)\omega(k,q)^{-1}$ for $q,k\in \mathbb{Z}_2\times \mathbb{Z}_2$ are invariant under this gauge freedom. Given an automorphism $\alpha$ of $\mathbb{Z}_2\times \mathbb{Z}_2$, we group together 2-cocycles $\omega, \omega'$ related by $\omega'(q,k)=\omega(\alpha(q),\alpha(k))$. It is also clear that this composition does not change the value of $\lambda$, considered it as a unordered triple, and $\lambda'$; these quantities are invariant under all the allowed freedom. The possible 2-cocycles are shown in Fig. \ref{fig:cocyclesQ8Z2}.

\begin{figure}[ht!]
\begin{center}
\includegraphics[scale=0.8]{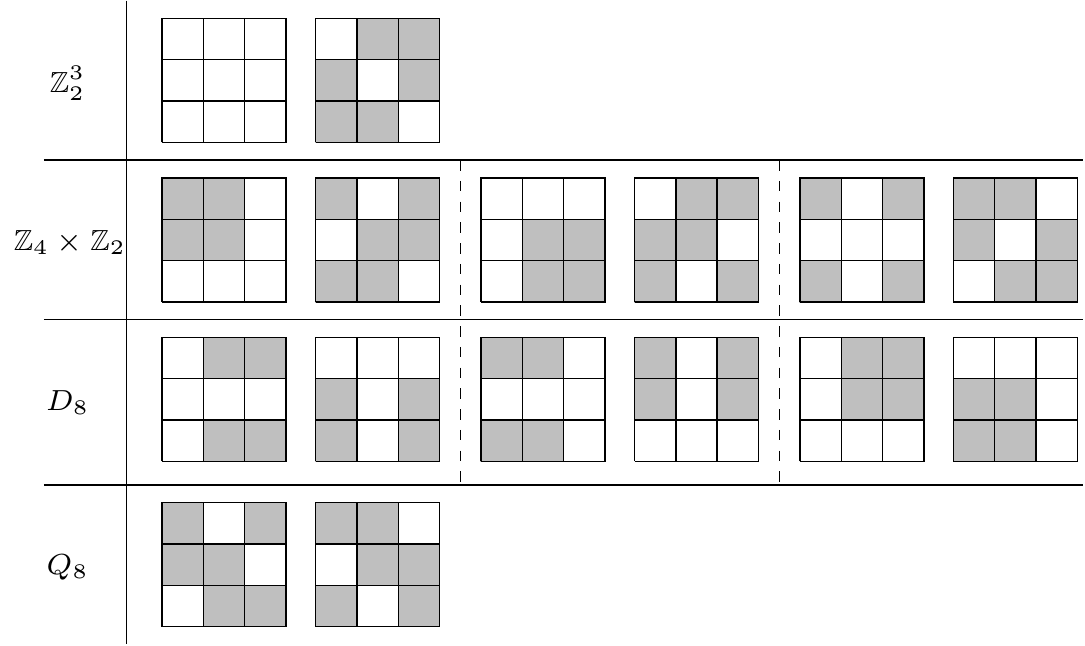}
\caption{The different cocycles of $\mathbb{Z}_2\times \mathbb{Z}_2$ by $\mathbb{Z}_2$ related to the extensions. The dashed lines separate the inequivalent cocycles coming from the classification of the second cohomology group; these are related by an automorphism of $\mathbb{Z}_2\times \mathbb{Z}_2$ giving the same group.}
\label{fig:cocyclesQ8Z2}
\end{center}
\end{figure}        
.

\subsection{$\mathcal{D}$($\mathbb{Z}_p$) with $\mathbb{Z}_p$ symmetry, $p$ prime}\label{secq}

We are going to work with the following presentation of the cyclic group with $p$ elements: $\mathbb{Z}_p = \{ \langle q \rangle: q^p = e \}$. Since $p$ is prime, the homomorphism from $Q=\mathbb{Z}_p$ to ${\rm Aut}(G) = {\rm Aut}(\mathbb{Z}_p)\cong \mathbb{Z}_{p-1}$ is trivial ($\phi_q=\id$\footnote{
Since $\phi$ is an homomorphism from $\mathbb{Z}_p \to \mathbb{Z}_{p-1}$ it satisfies $\id=\phi^{p-1}_q=\phi_{q^{p-1}}$ for any $q\in \mathbb{Z}_p$ but since $p-1$ and $p$ are coprimes $\{ q^{p-1}\}_{q\in \mathbb{Z}_p}= \mathbb{Z}_p$ so $\phi_q=\id , \; \forall q\in \mathbb{Z}_p$. 
}) 
so there is no permutation of anyons. This implies that the possible phases are distinguished only by the inequivalent {SF} patterns on the charges.
The cohomological classification gives $H^2(\mathbb{Z}_p,\mathbb{Z}_p)\cong \mathbb{Z}_p$ but there are only two group extensions $\mathbb{Z}_p \times \mathbb{Z}_p$ and $\mathbb{Z}_{p^2}$. There are $p-1$ non-trivial inequivalent cocycles corresponding to $\mathbb{Z}_{p^2}$ that can be related by relabelling the symmetry operators. That is, $\omega \not \equiv \omega'$ but $\omega' \equiv \omega \circ (\rho\times \rho)$ where $\rho \in {\rm Aut}(Q)$. The non-trivial {SF} class is characterized by a projective action on all non-trivial charges of the model.
The quantity we want to measure is 
$$\lambda=v^p_q=\omega(q,e)\omega(q,q)\omega(q,q^2)\cdots \omega(q,q^{p-1}), $$ 
which is gauge-invariant since  $v'^p_q=v^p_q l^p_q=v^p_q$ whenever $l_q\in \mathbb{Z}_p$. The values are $v^p_q=e$ if it is the trivial cocycle (group extension $\mathbb{Z}_p \times \mathbb{Z}_p$) or $v^p_q =\alpha$
where $\alpha\in [1,\dots,p-1]$ represents the $p-1$ inequivalent non-trivial cocycles. $\lambda$ measures if the group extension has elements of order greater than $p$.
To construct the order parameter we apply $U^{\otimes p}_q$ and the permutation $(12\cdots p)$ to $p$ consecutive sites on the state with a pair of charges and then we project onto the state with the pair of charges placed as in the previous example:$$\Lambda \equiv  \langle \;  \mathcal{O}^\dagger_{\sigma}(x',y) \;P_{(1\cdots p)}\; U^{\otimes p}_q \; \mathcal{O}_{\sigma} (x,y) \; \rangle .$$ For instance, in the case $p=3$ the order parameter would correspond to
   \begin{equation*} 
   \Lambda=\begin{tikzpicture}[scale=1]
        \pic at (1,0,0) {3dpeps};
        \pic at (1,1,0) {3dpepsdown}; 
        \pic at (0.5,0,0) {3dpeps};
        \pic at (0.5,1,0) {3dpepsdown};
         \pic at (1.5,0,0) {3dpeps};
        \pic at (1.5,1,0) {3dpepsdown};
                   \pic at (0,0,0.7) {3dpepsshort};
        \pic at (0,1,0.7) {3dpepsdownshort};
            \pic at (2,0,0.7) {3dpeps};
       \pic at (2,1,0.7) {3dpepsdown};
          \begin{scope}[canvas is xy plane at z=0]
          \draw[preaction={draw, line width=1.2pt, white}, ,blue] (1.5,1,0.7) to [out=-90, in=90] (0.5,0,0.7);
          \draw[preaction={draw, line width=1.2pt, white}, ,blue] (0.5,1,0.7) to [out=-90, in=90] (1,0,0.7);
          \draw[preaction={draw, line width=1.2pt, white}, ,blue] (1,1,0.7) to [out=-90, in=90] (1.5,0,0.7);
        \end{scope}
         \filldraw[draw=black,fill=purple, ] (1.47,0.85,0.7) circle (0.06);
         \filldraw[draw=black,fill=purple, ] (1.03,0.85,0.7) circle (0.06);
         \filldraw[draw=black,fill=purple, ] (0.53,0.85,0.7) circle (0.06);
            \pic at (1.5,0,0.7) {3dpepsp};
        \pic at (1.5,1,0.7) {3dpepsdownp};
                \pic at (1,0,0.7) {3dpepsp};
        \pic at (1,1,0.7) {3dpepsdownp}; 
                \pic at (0.5,0,0.7) {3dpepsp};
        \pic at (0.5,1,0.7) {3dpepsdownp};
         \pic at (0.5,0,1.4) {3dpeps};
        \pic at (0.5,1,1.4) {3dpepsdown};
                \pic at (1,0,1.4) {3dpeps};
                \pic at (1.5,0,1.4) {3dpeps};
        \draw[ ] (1,1,1.4)--(1,0.81,1.4);
    \begin{scope}[canvas is zx plane at y=1]
      \draw[preaction={draw, line width=1.2pt, white}, ] (0.9,1)--(1.9,1);
      \draw[preaction={draw, line width=1.2pt, white}, ] (1.4,0.6)--(1.4,1.4);
      \filldraw (1.4,1) circle (0.07);
    \end{scope} 
            \draw[ ] (1.5,1,1.4)--(1.5,0.81,1.4);
    \begin{scope}[canvas is zx plane at y=1]
      \draw[preaction={draw, line width=1.2pt, white}, ] (0.9,1.5)--(1.9,1.5);
      \draw[preaction={draw, line width=1.2pt, white}, ] (1.4,1.1)--(1.4,1.9);
      \filldraw (1.4,1.5) circle (0.07);
    \end{scope} 
             \filldraw[draw=black,fill=orange, ] (1.75,0,0.5) rectangle (1.75,0,0.9);
            \filldraw[draw=black,fill=orange, ] (0.25,1,0.5) rectangle (0.25,1,0.9); 
          \pic at (2.5,0,0.7) {3dpepsshort};
        \pic at (2.5,1,0.7) {3dpepsdownshort};
             \filldraw[draw=black,fill=orange, ] (2.25,0,0.5) rectangle (2.25,0,0.9);
            \filldraw[draw=black,fill=orange, ] (2.25,1,0.5) rectangle (2.25,1,0.9); 
\end{tikzpicture}.
\end{equation*}
Using Eq.(\ref{Gisoconca}):
\begin{equation*}
\hat{\Lambda}=
\frac{1}{|G|}
\sum_{b\in G}\;
   \begin{tikzpicture}[scale=1]
                \node[anchor= south west] at (0.5,0.75,-0.2) {$\myinv{b}$};
        \filldraw (0.5,0.75,-0.2) circle (0.05);
                      \node[anchor=south west] at (1,0.75,-0.2) {$\myinv{b}$};
        \filldraw (1,0.75,-0.2) circle (0.05);
                              \node[anchor=south west] at (1.5,0.75,-0.2) {$\myinv{b}$};
        \filldraw (1.5,0.75,-0.2) circle (0.05);
\draw[ ] (0.5,1,0.7)--(0.5,1,-0.2)--(0.5,0,-0.2)--(0.5,0,0.7);
\draw[ ] (1,1,0.7)--(1,1,-0.2)--(1,0,-0.2)--(1,0,0.7);
\draw[ ] (1.5,1,0.7)--(1.5,1,-0.2)--(1.5,0,-0.2)--(1.5,0,0.7);
\draw[preaction={draw, line width=1.2pt, white}, ] (0,0,0.7) rectangle (2,1,0.7);
          \draw[preaction={draw, line width=1.2pt, white}, ,blue] (0.5,1,0.7) to [out=-90, in=90] (1,0,0.7);
          \draw[preaction={draw, line width=1.2pt, white}, ,blue] (1,1,0.7) to [out=-90, in=90] (1.5,0,0.7);
          \draw[preaction={draw, line width=1.2pt, white}, ,blue] (1.5,1,0.7) to [out=-90, in=90] (0.5,0,0.7);  
              \filldraw[draw=black,fill=purple, ] (0.57,0.75,0.7) circle (0.06);
              \filldraw[draw=black,fill=purple, ] (1.05,0.75,0.7) circle (0.06);
               \filldraw[draw=black,fill=purple, ] (1.4,0.75,0.7) circle (0.06);
               
              \draw[preaction={draw, line width=1.2pt, white}, ] (0.5,1,0.7)--(0.5,1,1.6)--(0.5,0,1.6)--(0.5,0,0.7);
              \draw[preaction={draw, line width=1.2pt, white}, ] (1,1,0.7)--(1,1,1.6)--(1,0,1.6)--(1,0,0.7);
                            \draw[preaction={draw, line width=1pt, white}, ] (1.5,1,0.7)--(1.5,1,1.6)--(1.5,0,1.6)--(1.5,0,0.7);
                       \pic at (1.5,0,0.7) {3dpepsp};
        \pic at (1.5,1,0.7) {3dpepsdownp};
         \pic at (1,0,0.7) {3dpepsp};
        \pic at (1,1,0.7) {3dpepsdownp};
        \pic at (0.5,0,0.7) {3dpepsp};
        \pic at (0.5,1,0.7) {3dpepsdownp};
                             \filldraw[draw=black,fill=orange, ] (1.75,0,0.5) rectangle (1.75,0,0.9 );
            \filldraw[draw=black,fill=orange, ] (0.25,1,0.5) rectangle (0.25,1,0.9);
                        \node[anchor= north east] at (0.5,0.25,1.6) {${b}$};
        \filldraw (0.5,0.25,1.6) circle (0.05);
                                \node[anchor= north east] at (1,0.25,1.6) {${b}$};
        \filldraw (1,0.25,1.6) circle (0.05); 
                                \node[anchor= north east] at (1.5,0.25,1.6) {${b}$};
        \filldraw (1.5,0.25,1.6) circle (0.05);         
                                  \node[anchor=west] at (2,0.5,0.7) {$\myinv{b}$};
        \filldraw (2,0.5,0.7) circle (0.05);
 \node[anchor=east] at (0,0.5,0.7) {${b}$};
        \filldraw (0,0.5,0.7) circle (0.05);
\end{tikzpicture}
\times
\begin{tikzpicture}[scale=1]    
 \draw[preaction={draw, line width=1.2pt, white}, ] (1.8,0,0.7) rectangle (2.6,1,0.7);
              \filldraw[draw=black,fill=orange, ] (2.2,1,0.5) rectangle (2.2,1,0.9 );  
            \filldraw[draw=black,fill=orange, ] (2.2,0,0.5) rectangle (2.2,0,0.9  );
             \node[anchor=east] at (1.8,0.5,0.7) {${b}$};
        \filldraw (1.8,0.5,0.7) circle (0.05);
         \node[anchor=west] at (2.6,0.5,0.7) {$\myinv{b}$};
        \filldraw (2.6,0.5,0.7) circle (0.05);
   \end{tikzpicture}.
  \end{equation*}
Using now Eq.\eqref{Ginje} and Eq.\eqref{locsym}, we obtain a sum over $G=\mathbb{Z}_3$ for each of the three permuted sites (we denote these elements $s_1,s_2$ and $s_3$). The final expression to compute is the following
\begin{align}
\hat{\Lambda} &= \frac{1}{|G|^4}\sum_{s_1,s_2,s_3,b\in G} 
  \begin{tikzpicture}[baseline=+5mm]
        \draw[ , blue] (0.6,0)..controls (1,0.5) and (2,0.5)..(2.4,1);
                \draw[ , blue] (0.5,0)..controls (0.9,0.5) and (1.9,0.5)..(2.3,1);
                       \draw[preaction={draw, line width=1.2pt, white}, ,blue] (0.6,1) to [out=-90, in=90] (1.5,0);
        \draw[preaction={draw, line width=1.2pt, white}, ,blue] (0.5,1) to [out=-90, in=90] (1.4,0);
                       \draw[preaction={draw, line width=1.2pt, white}, ,blue] (0.6+0.9,1) to [out=-90, in=90] (1.5+0.9,0);
        \draw[preaction={draw, line width=1.2pt, white}, ,blue] (0.5+0.9,1) to [out=-90, in=90] (1.4+0.9,0);
                \draw[ ](0.6,0)--(1.4,0);
        \draw[ ](1.5,0)--(2.3,0);
        \draw[ ](2.4,0)--(2.8,0)--(2.8,1)--(2.4,1);
        \draw[ ](0.6,1)--(1.4,1);
        \draw[ ](1.5,1)--(2.3,1);
        \draw[ ](0.5,1)--(0,1);
        \draw[ ](0.5,0)--(0,0);
        \draw[ ](0,1)--(0,0);
        \node[anchor=south] at (0.35,1) {$\myinv{v_{q}}$};
        \filldraw[draw=black,fill=red, ] (0.4,1) circle (0.06);
         \node[anchor=east] at (0.7,0.7) {$\myinv{s}_1$};
        \filldraw (0.58,0.75) circle (0.05);
        \node[anchor=east] at (1.5,0.7) {$\myinv{s}_2$};
        \filldraw (1.47,0.75) circle (0.05);
          \node[anchor=west] at (2.27,0.85) {${s_3}$};
        \filldraw (2.27,0.85) circle (0.05);
         \node[anchor=east] at (0.8,0.3) {$\myinv{s}_3$};
        \filldraw (0.76,0.25) circle (0.05);
        \node[anchor=west] at (2.33,0.3) {${s_2}$};
        \filldraw (2.33,0.25) circle (0.05);
        \node[anchor=west] at (1.4,0.2) {$s_1$};
        \filldraw (1.43,0.25) circle (0.05);
        \node[anchor=south] at (0.8,1) {$v_{q}$};
        \filldraw[draw=black,fill=red, ] (0.8,1) circle (0.06);
         \node[anchor=south] at (1.3,1) {$\myinv{v_{q}}$};
        \filldraw[draw=black,fill=red, ] (1.3,1) circle (0.06);
        \node[anchor=south] at (1.7,1) {$v_{q} $};
        \filldraw[draw=black,fill=red, ] (1.7,1) circle (0.06);
                 \node[anchor=south] at (1.3+0.9,1) {$\myinv{v_{q}}$};
        \filldraw[draw=black,fill=red, ] (1.3+0.8,1) circle (0.06);
        \node[anchor=south] at (2.6,1) {$v_{q} $};
        \filldraw[draw=black,fill=red, ] (2.6,1) circle (0.06);
        \node[anchor=east] at (0,0.5) {$b$};
        \filldraw (0,0.5) circle (0.05);
        \node[anchor=west] at (2.8,0.5) {$\myinv{b}$};
        \filldraw (2.8,0.5) circle (0.05);
        \filldraw[draw=black,fill=orange, ] (2.5,-0.1) rectangle (2.7,0.1);
        \filldraw[draw=black,fill=orange, ] (0.05,0.9) rectangle (0.25,1.1);
         \end{tikzpicture}  \times 
 \notag \\
  &  \begin{tikzpicture}
       \draw[ ] (0,0) rectangle (1,1);
        \node[anchor=west] at (0,0.5) {$b$};
        \filldraw (0,0.5) circle (0.05);
       \node[anchor=east] at (1.1,0.5) {$\myinv{b}$};
        \filldraw (1,0.5) circle (0.05);
        \filldraw[draw=black,fill=orange, ] (0.4,0.9) rectangle (0.6,1.1);
      \filldraw[draw=black,fill=orange, ] (0.4,-0.1) rectangle (0.6,0.1);
     \end{tikzpicture}
  \times 
    \begin{tikzpicture}
          \draw[preaction={draw, line width=1.2pt, white}, ,blue] (0,0,-1) to [out=90, in=-90] (1.2,1,-1);
          \draw[preaction={draw, line width=1.2pt, white}, ,blue] (0.6,0,-1) to [out=90, in=-90] (0,1,-1);
 \draw[preaction={draw, line width=1.2pt, white}, ,blue] (1.2,0,-1) to [out=90, in=-90] (0.6,1,-1);
        \foreach \x in {0,0.6,1.2}{
        \draw[preaction={draw, line width=1.2pt, white}, ] (\x,1,-1)--(\x,1,0)--(\x,0,0)--(\x,0,-1);
        \node[anchor= north east] at (\x,0.4,0) {$b$};
        \filldraw (\x,0.35,0) circle (0.05);
      \node[anchor=south] at (\x,1,-0.5) {$v_{q}$};
        \filldraw[draw=black,fill=red, ] (\x,1,-0.5) circle (0.06);
        }
         \filldraw (0.15,0.3,-1) circle (0.05);
          \node[anchor=south east] at (0.2,0.25,-1) {$\myinv{s}_3$};
          \filldraw (0.5,0.3,-1) circle (0.05);
          \node[anchor=west] at (0.45,0.3,-1) {$\myinv{s}_1$};
          \filldraw (1.1,0.3,-1) circle (0.05);
          \node[anchor=west] at (1.05,0.3,-1) {$\myinv{s}_2$};
      \end{tikzpicture}
 \times  
        \begin{tikzpicture}
                        \draw[preaction={draw, line width=1.2pt, white}, , blue] (0,0,1) to [out=90, in=-90] (1.2,1,1);
          \draw[preaction={draw, line width=1.2pt, white}, ,blue] (0.6,0,1) to [out=90, in=-90] (0,1,1);
 \draw[preaction={draw, line width=1.2pt, white}, ,blue] (1.2,0,1) to [out=90, in=-90] (0.6,1,1);
        \foreach \x in {0,0.6,1.2}{
        \draw[preaction={draw, line width=1.2pt, white}, ] (\x,1,1)--(\x,1,0)--(\x,0,0)--(\x,0,1);
        \node[anchor=  east] at (\x+0.15,0.6,0) {$\myinv{b}$};
        \filldraw (\x,0.5,0) circle (0.05);
      \node[anchor=south] at (\x,1,0.5) {$\myinv{v}_{q}$};
        \filldraw[draw=black,fill=red, ] (\x,1,0.5) circle (0.06);
        }
         \filldraw (0.15,0.3,1) circle (0.05);
          \node[anchor=south east] at (0.2,0.25,1) {$s_3$};
          \filldraw (0.5,0.3,1) circle (0.05);
          \node[anchor=west] at (0.45,0.3,1) {$s_1$};
          \filldraw (1.1,0.3,1) circle (0.05);
          \node[anchor=west] at (1.05,0.3,1) {$s_2$};
 \end{tikzpicture}.\notag
 \end{align}
 Each term of this sum contains four diagrams. The first is made of two loops 
$$
\tr[   C^{[u]}_\sigma  \myinv{s}_1  s_3 \myinv{b} C^{[d]}_\sigma s_2 \myinv{s}_3 b ] \times \tr[ \myinv{s}_2 s_1].
$$
The last factor of this expression is equal to $|G|\delta_{s_2,s_1}$. The last two diagrams are both equal to $\tr[s_1s_2s_3v^{-3}_q b^{-3}]$ which reduces the sum to its terms that satisfy $v^3_q=s^2_1s_3$. Putting it all together, we get
$$
\hat{\Lambda} =\frac{1}{|G|}  \tr[ C^{[u]}_\sigma  v^3_q b^{-1} C^{[d]}_\sigma b v^3_q  ]\times \tr[  \bar{C}^{[u]}_\sigma  b^{-1} \bar{C}^{[d]}_\sigma b ]=\chi_{\sigma}( v^3_q).
$$
An analogous calculation can be carried out for arbitrary $p$ prime: we would also obtain a sum of four diagrams. The first contains $p-1$ loops:
$$\tr[   C^{[u]}_\sigma  \myinv{s}_1  s_p \myinv{b} C^{[d]}_\sigma s_2 \myinv{s}_p b ]  \tr[  s_1\myinv{s}_2] \tr[  s_2\myinv{s}_3]  \cdots \tr[  s_{p-2}\myinv{s}_{p-1}], $$
where the last $p-2$ factors reduces the sum to its elements that satisfy $s_1=s_2=\cdots =s_{p-2}=s_{p-1}$. The last two diagrams are both equal to $ \tr[s_1\cdots s_p v^{-p}_q b^{-p}] = \tr[s^{p-1}_1 s_p v^{-p}_q]$. So the only terms that survive in the sum are those satisfying $s_1s^{-1}_p=  v^p_q$. Finally
$$\hat{\Lambda }=  \chi_\sigma(v^p_q),$$ 
 where $\chi_\sigma$ denotes one of the $p-1$ non-trivial irreps of $\mathbb{Z}_p$. Therefore the order parameter is only equal to one if the {SF} pattern is trivial.

\subsection{ $\mathcal{D}$($\mathbb{Z}_4$) with $\mathbb{Z}_2$ symmetry }\label{ex:perm}
       
We denote the topological group as $\mathbb{Z}_4=\{+1,-1,i,-i;\times \}$ and the symmetry group as $\mathbb{Z}_2=\{e,q;q^2=e \}$. There are two cases here to be analyzed depending on whether there is a non-trivial permutation action over the anyons or not. These permutations come from the possible homomorphism from $\mathbb{Z}_2$ to ${\rm Aut}(\mathbb{Z}_4)$. In each case there are two inequivalent {SF} classes.\\
 
 {\bf Non-trivial permutation}. In this case there is a non-trivial action of the symmetry operators over the topological group: $\phi_q(g)=v_q g v^{-1}_q= g^{-1}$;  $\forall g \in \mathbb{Z}_4$. That is, it permutes the fluxes $i$ and $-i$. There are two inequivalent cocycles since $H_\phi^2(\mathbb{Z}_2,\mathbb{Z}_4)= \mathbb{Z}_2$. 
If we denote the irreps of $\mathbb{Z}_4$ as $\sigma=0,1,2,3$, the non-trivial {SF} class (group extension $Q_8$) is characterized by a projective action on charges $1$ and $3$. Applying this action twice is equivalent to braiding with the flux $-1$. The symmetry permutes between charges $1$ and $3$, and multiplies the wave function of the system by minus one phase factor.

The class can be distinguished with the quantity $\lambda=\omega(e,e)\omega(q,q)=v^2_{q}\in \mathbb{Z}_4$, which is represented by the identity element for the trivial {SF} class (group extension $D_8$) and $-1$ for the non-trivial one. $\lambda$ is not affected by gauge transformations because $v'^2_{q}= v^2_{q}l_q \phi_q(l_q)= v^2_{q}l_q l^{-1}_q=v^2_{q}$. This quantity measures if there are elements of the extension, that do not belong to $\mathbb{Z}_4$, with order greater than two.
We use the same configuration of operators as for the TC example, see Eq.(\ref{transprotocol}),  to construct the order parameter but creating the charges, irreps $\sigma=1,3$ of $\mathbb{Z}_4$, of the $\mathcal{D}$($\mathbb{Z}_4$). The final expression is the one of Eq.\eqref{loopcalc} particularising for the above relations of $G=\mathbb{Z}_4$. We get:
\begin{align*}
        \begin{tikzpicture}
        \draw[ , blue] (1.5,1) to [out=-90, in=90] (0.6,0);
        \draw[ , blue] (1.4,1) to [out=-90, in=90] (0.5,0);
       \draw[ , preaction={draw, line width=1pt, white}, blue] (0.6,1) to [out=-90, in=90] (1.5,0);
        \draw[preaction={draw, line width=1pt, white}, , blue] (0.5,1) to [out=-90, in=90] (1.4,0);
        \draw[ ] (0.6,0)--(1.4,0);
        \draw[ ] (1.5,0)--(1.9,0);
        \draw[ ] (1.9,1)--(1.9,0);
        \draw[ ] (0.6,1)--(1.4,1);
        \draw[ ] (1.5,1)--(1.9,1);
        \draw[ ] (0.5,1)--(0,1);
        \draw[ ] (0.5,0)--(0,0);
        \draw[ ] (0,1)--(0,0);
        \node[anchor=south] at (0.35,1) {$\myinv{v_{q}}$};
        \filldraw[ ,draw=black,fill=red] (0.4,1) circle (0.06);
         \node[anchor=east] at (0.7,0.7) {$\myinv{a}$};
        \filldraw[ ] (0.58,0.75) circle (0.05);
        \node[anchor=west] at (1.4,0.7) {$c$};
        \filldraw[ ] (1.43,0.75) circle (0.05);
        \node[anchor=east] at (0.7,0.2) {$\myinv{c}$};
        \filldraw[ ] (0.58,0.25) circle (0.05);
        \node[anchor=west] at (1.4,0.2) {$a$};
        \filldraw[ ] (1.43,0.25) circle (0.05);
        \node[anchor=south] at (0.8,1) {$v_{q}$};
        \filldraw[ , draw=black,fill=red] (0.8,1) circle (0.06);
         \node[anchor=south] at (1.3,1) {$\myinv{v_{q}}$};
        \filldraw[ , draw=black,fill=red] (1.3,1) circle (0.06);
        \node[anchor=south] at (1.9,1) {$v_{q} $};
        \filldraw[ ,draw=black,fill=red] (1.7,1) circle (0.06);
        \node[anchor=west] at (0,0.5) {$b$};
        \filldraw[ ] (0,0.5) circle (0.05);
        \node[anchor=east] at (1.9,0.5) {$\myinv{b}$};
        \filldraw[ ] (1.9,0.5) circle (0.05);
        \filldraw[ ,draw=black,fill=orange] (1.65,-0.1) rectangle (1.85,0.1);
        \filldraw[ ,draw=black,fill=orange] (0.05,0.9) rectangle (0.25,1.1);
         \end{tikzpicture}
     & = \tr[ v^{-1}_q c^{-1} b  C^{[u]}_\sigma v^{-1}_q a^{-1}  c v_q b^{-1} C^{[d]}_\sigma a v_q ]\\
     &= \tr[C^{[u]}_\sigma a c^{-1} b^{-1} C^{[d]}_\sigma a c^{-1} b ],
       \end{align*}
and the last two diagrams, see Eq. \eqref{loopcalc}, both equal to $\tr[ v_q a^{-1}b  v_q c^{-1} b]= \tr[ v^2_q a^{-1} c]$, are $|G|\delta_{ac^{-1}  ,v^{2}_q}$.
Next, we observe that $ \tr[C^{[u]}_\sigma a c^{-1} b^{-1} C^{[d]}_\sigma a c^{-1} b ]=\tr[ C^{[u]}_\sigma  v^2_q b^{-1} C^{[d]}_\sigma b v^2_q  ]$, which implies that together with the factor $\tr[  C^{[u]}_\sigma  b^{-1} C^{[d]}_\sigma b  ]$
$$
\hat{\Lambda} =\frac{1}{|G|}  \tr[ C^{[u]}_\sigma  v^2_q b^{-1} C^{[d]}_\sigma b v^2_q  ]\times \tr[  \bar{C}^{[u]}_\sigma  b^{-1} \bar{C}^{[d]}_\sigma b ]=\chi_{\sigma}( v^2_q).
$$

This calculation is valid for any $\sigma$ but to distinguish between phases, we have to choose the two irreps, $\sigma=1,3$, that satisfy $\chi_{\sigma}( -1)\neq1$ (the ones that fractionalize the symmetry).\\

 {\bf Trivial permutation}. 
 There are two fractionalization classes since $H^2(\mathbb{Z}_2,\mathbb{Z}_4)= \mathbb{Z}_2$. A gauge-invariant quantity is $\lambda=v^4_q=\omega^2(e,q) \omega^2(q,q)$ since $v'^4_q=v^4_ql^4_q$ and $l_q\in \mathbb{Z}_4$. The value of $\lambda$ is $-1\in \mathbb{Z}_4$ for the non-trivial class (group extension $\mathbb{Z}_8$) and the identity for the trivial one (group extension $\mathbb{Z}_4\times \mathbb{Z}_2$). $\lambda$ measures whether there are elements of order greater than $4$ in the group extension. The order parameter is
   \begin{equation} 
   \Lambda=\begin{tikzpicture}[scale=1]
         \pic at (1,0,0) {3dpeps};
        \pic at (1,1,0) {3dpepsdown};
        \pic at (0.5,0,0) {3dpeps};
        \pic at (0.5,1,0) {3dpepsdown};
        \pic at (1.5,0,0) {3dpeps};
        \pic at (1.5,1,0) {3dpepsdown};
          \pic at (2,0,0) {3dpeps};
        \pic at (2,1,0) {3dpepsdown};
                 \filldraw[draw=black,fill=orange, ] (0.75,1,-0.2) rectangle (0.75,1,0.2);  
          \filldraw[draw=black,fill=orange, ] (0.75,0,-0.2) rectangle (0.75,0,0.2 );
                  \begin{scope}[canvas is xy plane at z=0]
          \draw[preaction={draw, line width=1.2pt, white}, ,blue] (2,1,0.7) to [out=-90, in=90] (0.5,0,0.7);
          \draw[preaction={draw, line width=1.2pt, white}, ,blue] (0.5,1,0.7) to [out=-90, in=90] (1,0,0.7);
          \draw[preaction={draw, line width=1.2pt, white}, ,blue] (1,1,0.7) to [out=-90, in=90] (1.5,0,0.7);
          \draw[preaction={draw, line width=1.2pt, white}, ,blue] (1.5,1,0.7) to [out=-90, in=90] (2,0,0.7);
        \end{scope}
        \filldraw[draw=black,fill=purple, ] (1.97,0.85,0.7) circle (0.07);
                 \filldraw[draw=black,fill=purple, ] (1.53,0.85,0.7) circle (0.07);
         \filldraw[draw=black,fill=purple, ] (1.03,0.85,0.7) circle (0.07);
         \filldraw[draw=black,fill=purple, ] (0.53,0.85,0.7) circle (0.07);
                 \pic at (0,0,0.7) {3dpeps};
        \pic at (0,1,0.7) {3dpepsdown};
        \pic at (0.5,0,0.7) {3dpepsp};
        \pic at (0.5,1,0.7) {3dpepsdownp};
        \pic at (1,0,0.7) {3dpepsp};
        \pic at (1,1,0.7) {3dpepsdownp};
           \pic at (1.5,0,0.7) {3dpepsp};
        \pic at (1.5,1,0.7) {3dpepsdownp};
                \pic at (2,0,0.7) {3dpepsp};
        \pic at (2,1,0.7) {3dpepsdownp};
                   \pic at (2.5,0,0.7) {3dpeps};
        \pic at (2.5,1,0.7) {3dpepsdown};
         \pic at (1,0,1.4) {3dpeps};
         \pic at (0.5,0,1.4) {3dpeps};
        \pic at (0.5,1,1.4) {3dpepsdown};
         \pic at (1.5,0,1.4) {3dpeps};
         \pic at (2,0,1.4) {3dpeps};      
        \draw[ ] (1,1,1.4)--(1,0.81,1.4);
    \begin{scope}[canvas is zx plane at y=1]
      \draw[preaction={draw, line width=1.2pt, white}, ] (0.9,1)--(1.9,1);
      \draw[preaction={draw, line width=1.2pt, white}, ] (1.4,0.6)--(1.4,1.4);
      \filldraw (1.4,1) circle (0.07);
    \end{scope} 
            \draw[ ] (1.5,1,1.4)--(1.5,0.81,1.4);
    \begin{scope}[canvas is zx plane at y=1]
      \draw[preaction={draw, line width=1pt, white}, ] (0.9,1.5)--(1.9,1.5);
      \draw[preaction={draw, line width=1pt, white}, ] (1.4,1.1)--(1.4,1.9);
      \filldraw (1.4,1.5) circle (0.07);
    \end{scope} 
             \draw[ ] (2,1,1.4)--(2,0.81,1.4);
    \begin{scope}[canvas is zx plane at y=1]
      \draw[preaction={draw, line width=1.2pt, white}, ] (0.9,2)--(1.9,2);
      \draw[preaction={draw, line width=1.2pt, white}, ] (1.4,1.6)--(1.4,2.4);
      \filldraw (1.4,2) circle (0.07);
    \end{scope} 
      \filldraw[draw=black,fill=orange, ] (2.25,0,0.5) rectangle (2.25,0,0.9);
       \filldraw[draw=black,fill=orange, ] (0.25,1,0.5) rectangle (0.25,1,0.9); 
 \end{tikzpicture} \notag.\end{equation}
The calculation done in Section \ref{secq} (for $p=4$) is also valid here, since $\phi_q=\id$, therefore  
$$\hat{\Lambda} =  \chi_{\sigma}( v^4_q),$$
where  $\sigma$ denotes  the two irreps of $\mathbb{Z}_4$ that satisfy $\chi_{\sigma}( -1)\neq1$.

\subsection{ $\mathcal{D}(Q_8)$ with $\mathbb{Z}_2$ symmetry}\label{ex:nonab}

Here we study a case of non-abelian topological order: $\mathcal{D}$($Q_8$) with $Q=\mathbb{Z}_2=\{e,q;q^2=e \}$ symmetry and trivial anyon permutation action which can host two inequivalent {SF} classes. The non-trivial class is associated with a projective action of $\mathbb{Z}_2$ over the two-dimensional charge (irrep). The two related group extensions are $\mathbb{Z}_2 \times Q_8 \equiv \{ (s,g)| s=0,1 ; g\in Q_8 \}$ and $(\mathbb{Z}_4 \times Q_8)/\mathbb{Z}_2 \equiv \{ (t,g)| t=0,1,2,3 ; g\in Q_8 \} / \{(0,e),(2,-1)  \}$.
Let us consider the posible values of $v_q$ and $v^2_q=\omega(e,e)\omega(q,q)\in Q_8 $ in both group extensions:

\begin{table}[htb]
\centering
\begin{tabular}{|c|c|}
\hline
\multicolumn{2}{|c|}{$\mathbb{Z}_2 \times Q_8$} \\
 \hline\hline
$v_q$ & $v^2_q\in Q_8$ \\ \hline 
$(1,+1)$ & $+1$ \\ \hline
$(1,-1)$ & $+1$ \\ \hline
$(1,+i)$ & $-1$ \\ \hline
$(1,-i)$ & $-1$ \\ \hline
$(1,+j)$ & $-1$ \\ \hline
$(1,-j)$ & $-1$ \\ \hline
$(1,+k)$ & $-1$ \\ \hline
$(1,-k)$ & $-1$ \\ \hline
\end{tabular}
\qquad
\begin{tabular}{|c|c|}
\hline
\multicolumn{2}{|c|}{$(\mathbb{Z}_4 \times Q_8)/\mathbb{Z}_2 $} \\
 \hline\hline
$v_q$ & $v^2_q\in Q_8$ \\ \hline 
$(1,+1)\equiv (3,-1)$ & $-1$ \\ \hline
$(1,-1)\equiv (3,+1)$ & $-1$ \\ \hline
$(1,+i)\equiv (3,-i)$ & $+1$ \\ \hline
$(1,-i)\equiv (3,+i)$ & $+1$ \\ \hline
$(1,+j)\equiv (3,-j)$ & $+1$ \\ \hline
$(1,-j)\equiv (3,+j)$ & $+1$ \\ \hline
$(1,+k)\equiv (3,-k)$ & $+1$ \\ \hline
$(1,-k)\equiv (3,+k)$ & $+1$ \\ \hline
\end{tabular}
\end{table}
The fact that the topological group is non-abelian implies that the natural candidate for $\lambda$ of a $\mathbb{Z}_2$ symmetry, $v^2_q$, is not gauge-invariant. Besides that, there exists a difference in $v^2_q$ between extensions in the number of times that it can be $-1$ (or $+1$). This distinction can be captured by the magnitude:
\begin{align*}
\lambda&= \sum_{g\in Q_8} (v_q\times g)^2 =  \sum_{g\in Q_8}(1,g)^2 \\
 &=\bigg\{ \begin{array}{lcc}   
 6u_{+1}+2u_{-1}& {\rm in} & (\mathbb{Z}_4 \times Q_8)/\mathbb{Z}_2 
\\
  6u_{-1}+2u_{+1} &  {\rm in} &\mathbb{Z}_2 \times Q_8 \end{array},
\end{align*}
which does not depend on the gauge and where $u_{+1},u_{-1}$ correspond to the representation of $Q_8$ of the elements $+1,-1$.  Let us notice that $\lambda$ here does not belong to $Q_8$, but it belongs to the algebra generated by the representation of $G$ acting on the virtual d.o.f. This suggests that $\lambda$ can be interpreted as a superposition of fluxes which depends on the extension, {\it i.e.} SF pattern. To obtain $\lambda$ at the virtual level, we use again the same configuration of operators as for the TC example, see Eq.(\ref{transprotocol}),  to construct the order parameter. But we now create the non-abelian charge of  $\mathcal{D}$($Q_8$) which corresponds to the two-dimensional irrep of $Q_8$.
Then, the final expression is the one of Eq.\eqref{loopcalc} particularising it for this case. There are two factors equal to $\tr[c v^{-1}_q b^{-1} a v^{-1}_q b^{-1}]$ which reduce the sum over $c$ in Eq.\eqref{loopcalc} to the terms where $c^{-1}= v^{-1}_qb^{-1}av^{-1}_q b^{-1}$. This implies that 
 $$ac^{-1}b =(av^{-1}_qb^{-1})(av^{-1}_q b^{-1})b= [v^{-1}_q a^{-1}  c v_q b^{-1}]^{-1}.$$
The factor $av^{-1}_q b^{-1}$ can be relabelled as $a\phi_q(b^{-1})v^{-1}_q\equiv g^{-1} v^{-1}_q$ with $g\in Q_8$  replacing the sum over $a$ and $b$ with the sum over $g$ (and a factor $|Q_8|$). The remaining two factors are $\tr[  \bar{C}^{[u]}_{\sigma,h}  b^{-1} \bar{C}^{[d]}_{\sigma,h} b ]$ and $\tr[ v^{-1}_q c^{-1} b  C^{[u]}_\sigma v^{-1}_q a^{-1}  c v_q b^{-1} C^{[d]}_{\sigma,h} a v_q ]= \tr[ C^{[u]}_{\sigma,h} v^{-1}_q a^{-1}  c v_q b^{-1} C^{[d]}_{\sigma,h} ac^{-1}b ]$ where in the last relation we can identify the factor $ac^{-1}b$ previously discussed. We finally obtain:

\begin{align}
\hat{\Lambda} & \propto  \sum_{b,g,h\in Q_8} \tr[ C^{[u]}_{\sigma, h}  b^{-1}(v_q g)^{2} C^{[d]}_{\sigma , h}  (v_q g )^{-2} b ] \tr[  \bar{C}^{[u]}_{\sigma , h}  b^{-1} \bar{C}^{[d]}_{\sigma ,h} b ] \notag \\
& =  \bigg\{ \begin{array}{ccc}  
\left(6\chi_\sigma(+1)+2\chi_\sigma(-1)\right)/8& {\rm for} & E = (\mathbb{Z}_4 \times Q_8)/ \mathbb{Z}_2\\ 
\left( 6 \chi_\sigma(-1) +2\chi_\sigma(+1) \right) /8 &  {\rm for} & E = \mathbb{Z}_2 \times Q_8 
		\end{array}\notag \\
 & = \bigg\{ \begin{array}{lcc} 
  			+1 & {\rm for } &  E = (\mathbb{Z}_4 \times Q_8)/\mathbb{Z}_2  \\
  			-1 &  {\rm for } & E = \mathbb{Z}_2 \times Q_8,
   		\end{array}\notag
\end{align}
where in the last step we have used that $\sigma$ is the two-dimensional irrep of $Q_8$. The fact that $\lambda$ belongs to the group algebra instead of just to the group comes from $Q_8$ being non-abelian. \\

This non-abelian case finishes the proof of \cref{theo:sfdetect} since we have shown that there exists an order parameter, for every case, that detects all the symmetry fractionalization patterns.

\section{Discussion}

In this chapter, we have presented order parameters to distinguish the fractionalization class of an internal symmetry in $G$-injective {PEPS}. These are calculated in the bulk of the {2D} system and only depend on the virtual symmetry operators --and not on the explicit representation of the symmetry--. Our technique works for some very interesting cases like non-abelian topological order or {SF} classes involving permutations of anyons. 
In all the examples provided, we have also shown that the {SF} classes are better resolved with our order parameters than it would be possible through dimensional compactification.\\

The calculations have been done using square lattices but they can be generalized. 
We notice that \cref{theo:sfdetect} can be potentially useful for a generalisation of our order parameter in general string nets models \cite{Levin05} enriched with symmetries. To do so, it would be interesting to study how our findings could be adapted to the case where the symmetry is not internal and when it is not represented virtually as a tensor product, but as a matrix product operator (MPO) \cite{Williamson17}. In fact, the MPO-like form would allow to carry out more SF patterns: in our approach we only consider SF patterns equivalent to braiding with fluxes.

The construction of order parameters for non-unitary symmetries and lattice symmetries is left for future work; some gauge-invariant quantities of these symmetries are calculated in \cref{symTRS}.

Our order parameter is proven to work for the RG fixed points of $G$-injective {PEPS} where the correlation length is zero and Eq.(\ref{Gisoconca}) can be used to compute all the expressions analytically.  At those points we can ensure that the state is in the $\mathcal{D}$($G$) phase in the thermodynamic limit. 

One would expect that, after blocking $G$-injective {PEPS} tensors in the phase of the $\mathcal{D}$($G$), the blocked tensor will satisfy the following:

\begin{equation*}
  \begin{tikzpicture}
  \node[anchor=south] at (0,0.2,0) {$\bar{A}$};
  \node[anchor=north] at (0,0,1.4) {$A$};
     \foreach \z in {0,0.7,1.4}{
      \foreach \x in {0,0.5,1}{
           \pic at (\x,0,\z) {3dpeps}; 
           \pic at (\x,0.2,\z) {3dpepsdown}; 
             } } 
   \end{tikzpicture}
   \approx
     \begin{tikzpicture}
       \pic at (0,0,0) {3dGisopepsproj};
                      \node at (-0.25,0.2,0.3) {$\myinv{g}$};
              \node at (0.15,-0.1,0.3) {$\myinv{g}$};
        \node at (0.37,0.35,0.3) {${g}$};
                \node at (0.5,-0.03,0.3) {${g}$};
\end{tikzpicture}
+{\rm error \; terms} \;.
\end{equation*}
This would allow us to obtain a leading term of the form of Eq.\eqref{Gisoconca}, when we concatenate tensors. Then, since the maps $(\omega, \phi)$ are well defined in arbitrary $G$-injective {PEPS} and also when doing blocking, the order parameter for the blocked tensor should work, at least approximately.

These considerations are important for the numerical implementation of our order parameters which is left for future work. 
We notice that the locality of our approach allows to use 2D techniques as the Corner Transfer Method (CTM) -see \cite{Nishino96} and \cite{Fishman18} for a recent review.\\

We also point out that the locality of our order parameters does not contradict the fact that the topological order cannot be detected locally. We are measuring a local symmetric effect.
But remarkably, the SF pattern is identified thanks to its duality with braiding, so it would be interesting to clarify if the SF of some quasiparticles implies a non-trivial braiding of the quantum phases that we are considering. Moreover the SF detection in our calculations involves only one type of anyon: the other anyon to carry out the braiding is simulated by the symmetry operators.\\


\newpage \cleardoublepage 

\chapter{Mathematical open problems in PEPS}\label{Openproblems}

Tensor network states play a prominent role in the rigorous study of central results of quantum many-body problems -see \cite{reviewPEPS} for a complete review. In particular, PEPS capture the relevant physics in the low-energy sector of locally interacting systems. Then, the study of these systems is translated into the formal PEPS framework where different mathematical techniques have been developed. Fruitful results in this field have derived from the interplay with different areas in mathematics. 
The aim of this chapter is to present some mathematical problems related to PEPS, providing the necessary background and motivation for them. We will separate the questions in three main topics, each of them presented in a separate section. The first one deals with questions related to the correspondence between PEPS and ground states. The second section deals with the use of PEPS to prove rigorous results in condensed-matter problems. The last section collects some open questions about PEPS that appear in different fields.

\section{Are PEPS and GS of local gapped Hamiltonians the same set?}

As commented in the \cref{chapter:Intro}, one of the key features of PEPS is that they are conjectured to {\it correspond}  to the set of ground states of gapped and locally interacting Hamiltonians (modified Area Law Conjecture). This is motivated by the fact that this is the situation for the one dimensional case with MPS. This correspondence can be divided in two statements:

\begin{enumerate}
\item[1] Ground states of gapped locally interacting Hamiltonians can be well approximated by PEPS with a \emph{small} bond dimension (i.e. GS $\subset$ PEPS).
\item[2] PEPS are exact ground states of (gapped) locally interacting Hamiltonians (i.e. PEPS $\subset$ GS).
\end{enumerate}

Some comments are in order: 

As shown in \cite{Ge16}, there are examples of states in 2D that fulfill an area law and, however, are not ground states of local Hamiltonians (nor well approximated by PEPS). In this sense, the set of area-law states is too big to capture the desired set of ground states and it is precisely the family of PEPS the one that seems to capture better such set.

The gap in statement 2 cannot be always guaranteed, as there are examples of PEPS that cannot be ground states of any gapped Hamiltonian \cite{Verstraete06}. This will be commented in detail in subsection \ref{sec:spectral-gap} below.

\

In the following we will pose the main open questions concerning points 1 and 2, together with the state of the art for both of them. 
\subsection{Are all PEPS the GS of a local gapped Hamiltonian?}

We recall from \cref{chapter:Intro} that every PEPS is the GS of a locally interacting Hamiltonian, called {\it parent Hamiltonian}. This parent Hamiltonian is constructed using the map $\Gamma_\mathcal{R}$ defined in a region $\mathcal{R}$ with the $A$ tensors of the PEPS. The basic open question here is the following 

\begin{question}
Which are the minimal requirements on $A$ and the minimal size of $\mathcal{R}$ under which one can guarantee that the given PEPS is the unique ground state of $H$ and in addition $H$ is gapped? 
\end{question}

This question turns out to be very difficult, specially beyond 1D systems. Let us now go slowly  through the known results and divide this question into more specific ones.

For that, we use the key concepts of {\it normal} and {\it injective} tensors, see \cref{eq:injec}, \cref{eq:PEPSinj} and \cref{def:nomalTN}, which endow $A$ with some special properties. 

\begin{definition}\label{def:injindex}
The injectivity index of a normal tensor $A$, $i(A)$, is the smallest $n\in\N$ so that $\Gamma_\mathcal{R}$ is an injective map for the square region $\mathcal{R}$ of size $n\times n$.
\end{definition}

It is known \cite{Fannes92,reviewPEPS} that, given a normal tensor $A$ with injectivity index $i(A)$, by taking $\mathcal{R}$ as the square region of size $i(A)+1$, the parent Hamiltonian associated to $\mathcal{R}$ with the above construction has the PEPS $\ket{\Psi_A}$ as the unique ground state with zero energy $H\ket{\Psi_A}=0$. 

Therefore, the bounds on the injectivity index correspond to the bounds on the interaction length of the parent Hamiltonian.  To comment on such bounds we will start with the case of 1D. There, in order to briefly illustrate about the techniques used so far, we will make a small detour and talk about a classic inequality of Wielandt in the context of stochastic matrices.

\subsubsection{Wielandt inequalities}

In 1950 \cite{Wielandt}, Wielandt proved that the index of primitivity of a primitive stochastic matrix $A\in \mathcal{M}_{D\times D}$ must be less or  equal than $ D^2-2D+2$, and that this bound is optimal. 

Let us recall that an stochastic matrix $A= (A_{i,j})_{i,j}\in \mathcal{M}_{D\times D}$ is a matrix with $A_{i,j}\ge 0$ for all $i,j$ and $\sum_{i} A_{i,j}=1$ for all $j$. This implies that if $p=(p_i)_i$ is a probability distribution ($p_i\ge 0$ and $\sum_i p_i=1$), the same holds for $Ap=(\sum_j A_{i,j}p_j)_i$. In this sense, $A$ models a noisy memoryless communication channel acting on an alphabet of size $D$ -- the basic object in Shannon's information theory.  

A stochastic matrix $A$ is called primitive if there exists $n\in \N$ such that $(A^n)_{i,j}>0$ for all $i,j$. The minimum of such $n$ is called the index of primitivity of $A$. 

The range of applications of Wielandt's inequality is wide: Markov chains \cite{Seneta}, graph theory and number theory \cite{Alfonsin}, or numerical analysis \cite{Varga} to name a few. 

In quantum information theory, the object that models a memoryless noisy channel is a trace-preserving completely positive linear map (also called quantum channel) $T:\mathcal{M}_{D\times D} \rightarrow \mathcal{M}_{D\times D}$ \cite{NielsenChuang10}. The quantum channel  $T$, by means of its Kraus decomposition, is nothing but a map of the form $T(X)=\sum_{i=1}^d A_i X A_i^{\dagger}$, where the {\it Kraus operators} $A_i$ are $D\times D$ matrices fulfilling $\sum_i A_i^{\dagger} A_i=\id$ (this is precisely the trace preserving condition) and $A^{\dagger}$ denotes the adjoint matrix of $A$.

Note that quantum channels include stochastic matrices as particular cases. Given a stochastic matrix $A=(a_{i,j})$, the quantum channel  $T_A$ with Kraus operators $\sqrt{a_{i,j}} \ket{i}\bra{j}$ has the following property: given a probability vector $p$, if we consider the diagonal matrix $\rho={\rm diag} (p)=\sum_i p_i \ket{i}\bra{i}$ then, $T_A(\rho)$ is exactly ${\rm diag} (Ap)$.  That is, the quantum channel $T_A$ restricted to the diagonal matrices is exactly the stochastic matrix $A$.

The following definition is the natural quantum (non-commutative) analogue of the notion of primitivity for a stochastic matrix \cite{Sanz10}.

\begin{definition}
A quantum channel is called primitive if there exists an $n\in \N$ so that $T^n(\rho)$ is full rank for all positive-semidefinite input $\rho$.   The minimum of such $n$ is called the primitivity index $p(T)$. 
\end{definition}

Note that given a stochastic matrix $A$, the associated quantum channel $T_A$ is primitive if and only if $A$ is primitive. Moreover, the corresponding primitivity indices coincide. There is an equivalent notion of primitive quantum channel, related to the classical Perron-Frobenius-like characterization of primitivity for the stochastic case \cite{Sanz10}:

\begin{proposition}
A quantum channel $T$ is primitive if and only if $T$ has a unique non-degenerate eigenvalue $\lambda$ with $|\lambda|=1$ and the corresponding eigenvector (which is necessarily semidefinite positive) is full rank. 
\end{proposition}

A natural question arises:

\begin{question}
Which is the optimal upper bound for the primitivity index of a primitive quantum channel $T$ acting on $\mathcal{M}_{D\times D}$?
\end{question}

In \cite{Sanz10} it is shown that $p(T)\le (D^2-d+1)D^2$. This result has been recently improved \cite{Rahaman} to $p(T)\le 2(D-1)^2$. The order $O(D^2)$ is optimal just by invoking the optimality of the classical Wielandt inequality. However, the exact optimal bound is still unkown.

As shown in \cite{Sanz10}, this type of bounds gives universal thresholds for the behaviour in time of the {\it zero-error classical capacity of a quantum channel} ,denoted by $C_0(T)$, defined as the optimal rate (measured in number of bits per use of the channel) at which a quantum channel can transmit classical information without errors \cite{Leung12}. Concretely, the following dichotomy result can be shown \cite{Sanz10}: \\
\begin{proposition}
Let $T$ be a quantum channel with a full-rank fixed point then, either $C_0(T^n)\ge1$ for all $n\in \N$ or $C_0(T^{p(T)})=0$,
\end{proposition}

\subsubsection{Index of injectivity of a MPS}

Let us now connect the previous discussion with the injectivity index of an MPS, as defined in Definition \ref{def:injindex}.

We recall that a translationally invariant MPS is given by a rank-3 tensor $A$, which is nothing but a set of matrices $A_i\in \mathcal{M}_{D\times D}$, $i=1,\ldots, d$, and hence it naturally defines a completely positive linear map $\mathcal{E}_{{A}}(X)= \sum_i A_i XA_i^\dagger$. Such map is usually called the transfer operator associated to the MPS. Using the transformation $A_i\mapsto YA_iY^{-1}$ that leaves invariant the MPS $\ket{\psi_A}$, one can assume w.l.o.g that the transfer operator $\mathcal{E}_{A}$ is trace-preserving and hence a quantum channel.

It is easy to see \cite{Sanz10} that the MPS is injective if and only if its associated transfer operator is primitive. In the normal case, the injectivity index of an MPS is an upper bound to the index of primitivity of its associated transfer operator, i.e. $i(A)\ge p(\mathcal{E}_{A})$. This finally brings us to the following key question: 

\begin{question} \label{question:injectivity-length-MPS}
Which is the optimal upper bound for the injectivity index of a normal MPS in terms of its bond dimension?
\end{question}

In \cite{Sanz10} it is shown that if $A$ is normal then $i({A})\le (D^2-d+1)D^2$. This result has been recently improved \cite{Michalek} to $p(T)\le 2D^2(6+\log_2(D))$. Up to a   logarithmic factor, the order $O(D^2\log(D))$ is optimal just by invoking the optimality of the classical Wielandt inequality. As before, the exact optimal bound is still unkown.\\

\subsubsection{Index of injectivity of a PEPS}

Motivated by the connection between the injectivity index and the interaction length of the parent Hamiltonian, one may ask the analogue of Question \ref{question:injectivity-length-MPS} in 2D.

\begin{question} \label{question:injectivity-PEPS1}
Which is the optimal upper bound for the injectivity index of a normal  PEPS in terms of its bond dimension?
\end{question}

As opposed to the 1D case, the only known result, proven recently in \cite{Michalek-2D} is the existence of a function of the bond dimension $f(D)$ that bounds $i(A)$ for every PEPS with bond dimension $D$. Unfortunately, in principle such function could be uncomputable.

Indeed, checking normality becomes undecidable if one generalizes the notion of normal PEPS as those tensors ${A}$ with the following properties: 
\begin{enumerate}
\item There exists an orthogonal projector $P:\mathbb{C}^D\to \mathbb{C}^D$ so that the tensor ${B}=(\id_d \otimes P^{\otimes 4}){A}$ is normal and
\item the PEPS associated to ${A}$ and ${B}$ coincide for every system size, i.e. $\ket{\Psi_{{A}}}= \ket{\Psi_{B}}$.
\end{enumerate}
This can be proven easily with the techniques in \cite{Scarpa18}. Therefore a weaker version of  Question \ref{question:injectivity-PEPS1} should be considered:

\begin{question} \label{question:injectivity-PEPS2}
Given an explicit (computable and if possible polynomial) function $f(D)$, which is an upper bound for the injectivity index of a normal PEPS?
\end{question}

\subsubsection{Spectral gap in PEPS} \label{sec:spectral-gap}

Let us finish this subsection tackling the problem of the spectral gap of the parent Hamiltonian. It is proven in \cite{Fannes92} (see \cite{Kastoryano19} for an alternative proof) that the parent Hamiltonian of a normal MPS is always gapped.  Unfortunately, this is not the case for 2D in PEPS, as it is shown in \cite{Verstraete06} by constructing an explicit counterexample.

In fact, for general PEPS the existence of gap in the parent Hamiltonian is undecidable, as shown in \cite{Scarpa18}, which highlights the complexity of the problem.  Moreover, the spectral gap of even the simplest non-trivial PEPS --the AKLT model \cite{AKLT88} as the paradigmatic example-- is still open.

However, some light has been shed on checking whether a Hamiltonian is gapped or not translating the question into a problem on the boundary. For instance, in \cite{Cirac11}, motivated by the holographic correspondence uncovered by Li and Haldane in \cite{LiHaldane}, an exact bulk-boundary correspondence was found, constructing for every PEPS a (family of) 1D mixed states, named as boundary states. In that work it is conjectured (see also \cite{Kastoryano19}), based on numerical evidence that

\begin{conj} \label{gap2Dboundary1dlocal}The gap of the parent Hamiltonian of a PEPS corresponds exactly to the possibility of writing the boundary states as Gibbs states of 1D short range Hamiltonians 
$$\rho=e^{-\beta H}, \quad \text{with } H=\sum_{i,j} h_{i,j},\quad \|h_{i,j}\|\le Je^{-\alpha |i-j|},$$
where $h_{i,j}$ acts non-trivially only on spins $i$ and $j$. 
\end{conj}

The boundary states are simply the semidefinite operators defined on $(\C^D)^{\otimes |\partial \mathcal{R}|}$ obtained in the boundary of a region when tracing out the bulk as shown in the figure

\begin{equation*}
\rho_\mathcal{R}=
 \begin{tikzpicture}
             \pic at (-0.3,0) {tensor};
             \pic at (-0.3,-0.2) {tensord};
            \pic at (0,0) {tensor};
             \pic at (0,-0.2) {tensord};
                 \draw (-0.3,0)--(0.15,0);
                 \draw (-0.3,-0.2)--(0.15,-0.2);
                 \node at (0.35,0) {${\cdots}$};
                 \node at (0.35,-0.2) {${\cdots}$};
    \pic at (0.7,0) {tensor};
    \pic at (0.7,-0.2) {tensord};
                     \draw (0.55,0)--(1.15,0);
                     \draw (0.55,-0.2)--(1.15,-0.2);
                         \node at (0.85,0.1) {$\alpha$};
                         \node at (0.85,-0.35) {$\beta$};
    \pic at (1,0) {tensor};
    \pic at (1,-0.2) {tensord};
                \draw[preaction={draw, line width=1pt, white}][line width=0.5pt] (1,0.19) to [out=45, in=-45] (1,-0.39);
    \node at (1.35,0) {${\cdots}$};
                 \node at (1.35,-0.2) {${\cdots}$};
                  \pic at (1.7,0) {tensor};
             \pic at (1.7,-0.2) {tensord};
                               \pic at (2,0) {tensor};
             \pic at (2,-0.2) {tensord};
             \draw (2,0)--(1.55,0);
                 \draw (2,-0.2)--(1.55,-0.2);
            \draw[preaction={draw, line width=1pt, white}][line width=0.5pt] (1.7,0.19) to [out=45, in=-45] (1.7,-0.39);
                        \draw[preaction={draw, line width=1pt, white}][line width=0.5pt] (2,0.19) to [out=45, in=-45] (2,-0.39);
    \end{tikzpicture}
    = \sum_{\alpha,\beta} (\ket{\Psi^{[\mathcal{R}]}_A})_\alpha  (\sigma_{\mathcal{R}^c})_{\alpha,\beta}(\bra{\Psi^{[\mathcal{R}]}_A})_\beta
\end{equation*}
As in any holographic correspondence, one is interested in creating a dictionary that maps bulk properties to boundary properties. The reason that such dictionary is expected in PEPS comes from the way in which expectation values are computed (see Eq.(\ref{ExpValue}) in Section \ref{sec:TNS} ): the boundary states are exactly the operators that mediate at the virtual level the correlations present at the physical level. Then,
 
\begin{question}
Is Conjecture \ref{gap2Dboundary1dlocal} true?
\end{question}

An important step in this direction was given in \cite{Kastoryano19}, proving one of the implications for the case of a faster than exponential decay in $\|h_{i,j}\|$.

\subsection{Can any GS of a local gapped Hamiltonian be represented as a PEPS?}

One of the main features of PEPS, and the one that makes them a relevant ansatz in the classical simulation of quantum systems, is the conjectured fact that PEPS approximate well ground states of locally interacting gapped Hamiltonians. To formalize this, we consider a gapped, translationally invariant Hamiltonian on an $L\times L$ torus given by a finite range interaction $\mathfrak{h}$, $H=\sum_\tau \tau(\mathfrak{h})$. We will assume a unique ground state denoted by  $\ket{\Psi_{\rm GS}}$.

There are two types of relevant approximations, global and local, depending on whether one is interested in approximating an extensive or an intensive quantity in the ground state. 

In the global approximation problem, the aim is to find a function $f(L)$ such that one can guarantee the existence of a (non-necessarily translational invariant) PEPS $\ket{\Psi_{\rm PEPS}}$ with bond dimension $D\le f(L)$ so that in the Hilbert norm
$$\|\ket{\Psi_{\rm PEPS}}-\ket{\Psi_{\rm GS}}\|_2\le \frac{1}{{\rm poly}(L)} \; .$$

For the local approximation problem, the goal is to find a function, if it exists, $g(\epsilon)$, so that one can guarantee the existence of a translational invariant PEPS $\ket{\Psi_A}$ given by a tensor $A$, with bond dimension $D\le g(\epsilon)$, so that in trace-class norm, 
$$\lim_{L\rightarrow \infty} \|\rho^{[L]}_{\mathcal{R}, GS}-\rho^{[L]}_{\mathcal{R}, A}\|_1\le \epsilon\; ,$$
where $\rho^{[L]}_{\mathcal{R}, GS}$ is the reduced density matrix of the region $\mathcal{R}$ associated to $\ket{\Psi_{\rm GS}}$ in the torus of size $L\times L$ ($\rho^{[L]}_{\mathcal{R}, A}$ is defined analogously). Note that being both $\ket{\Psi_A}$ and $\ket{\Psi_{\rm GS}}$ translationally invariant, the exact position of region $\mathcal{R}$ in the torus is irrelevant.  This type of approximation guarantees that in the thermodynamic limit, compactly supported observables can be well approximated by translational invariant PEPS (with finite bond dimension).

Both the global and local approximation problems have a positive satisfactory solution in 1D, with the current best bounds being  
\begin{align}\label{eq:bounds-approx}
f(L))&= e^{O(\log^{3/4}L)} \\
g(\epsilon)&= e^{O(\log^{3/4}\frac{1}{\epsilon})} \nonumber \; .
\end{align}
 proven in \cite{Arad13} and \cite{Huang15} respectively .

Both results come from refined versions of the so-called detectability lemma \cite{Aharonov,Anshu}. For simplicity, we will state it in 1D for nearest-neighbor interactions but a similar result holds in any dimension for finite range interactions. 

\begin{lemma}[Detectability Lemma in 1D]\label{lemma:DL}
Let $P$ be an orthogonal  projector on $\C^d\otimes \C^d$ and $Q=\id-P$ its orthogonal complement. Denote by $P_i$ the projector $P$ acting on sites $i,i+1$ of a chain of $L$ spins with periodic boundary conditions. Let $H=\sum_{i=1}^L P_i$ be a frustration free Hamiltonian and let 
$DL(H)$ be the operator
$$DL(H)=\left(\bigotimes_{i\text{ even}} Q_i\right)\left(\bigotimes_{i\text{ odd}} Q_i\right)$$ 

Then $$\left\| \ket{\Psi_{\rm GS}}\bra{\Psi_{\rm GS}} - DL(H)^\ell\right\|_\infty\le \left(\frac{1}{\sqrt{\frac{\Delta}{4}+1}} \right)^\ell= e^{-\alpha\ell},$$
where $\Delta$ is the spectral gap of $H$ (and $\alpha=\frac{1}{2}\log(\frac{\Delta}{4}+1)$).
\end{lemma}

To get an intuition of its application, let us briefly show how to use Lemma \ref{lemma:DL} to show approximation in operator norm of the ground state projector of $H$ by a MPO. Each $Q_i$ in $DL(H)$ is a two-body operator so both operators can be represented graphically as:

\begin{equation*}
Q_i= 
\begin{tikzpicture}
\pic at (0,0) {O2};
\end{tikzpicture}
\Rightarrow DL(H)=
\begin{tikzpicture}
 \node at (-0.6,0.15) {${\cdots}$};
  \pic at (0,0) {O2};
  \pic at (0.6,0.3) {O2};  
   \pic at (1.2,0) {O2};
     \pic at (1.8,0.3) {O2};  
        \pic at (2.4,0) {O2};
     \pic at (3,0.3) {O2}; 
      \node at (3.6,0.15) {${\cdots}$};
      \end{tikzpicture}.
 \end{equation*}
Then, by doing a SVD decomposition in each $Q_i$; 
 \begin{equation*}
\begin{tikzpicture}
\pic at (0,0) {O2};
\end{tikzpicture}
=
\begin{tikzpicture}
\draw (-0.3,0)--(0.3,0);     
  \draw (-0.3,0.2)--(-0.3,-0.2);
    \draw (0.3,0.2)--(0.3,-0.2);
     \filldraw[draw=black, fill=red] (0,0) circle (0.06);
      \filldraw[draw=black, fill=blue] (-0.35,-0.05) rectangle (-0.25,0.05);
        \filldraw[draw=black, fill=blue] (0.25,-0.05) rectangle (0.35,0.05);
             \node at (-0.47,-0.15)  {$U$};
                \node[anchor=north] at  (0,0)  {$\Sigma$};
                \node at (0.5,-0.13)  {$V^\dagger$};
 \end{tikzpicture}
 \equiv
 \begin{tikzpicture}
\pic at (0,0) {O2SVD};
 \end{tikzpicture},
 \end{equation*}
 it is easy to see that $DL(H)^{\ell}$ is an MPO with bond dimension $D\le d^{2\ell}$:
 \begin{equation*}
 \begin{tikzpicture}
  \node at (-0.3,0.075) {$\big($};
  \pic at (0,0) {O2p};
  \pic at (0.3,0.15) {O2p};  
   \pic at (0.6,0) {O2p};
     \pic at (0.9,0.15) {O2p};  
      \node at (1.2,0.075) {$\big)$};
            \node at (0.5,0.3) {${\cdot}$};
            \node at (0.5,0.35) {${\cdot}$};
              \node at (0.5,0.4) {${\cdot}$};
  \node at (-0.3,0.625) {$\big($};
  \pic at (0,0.55) {O2p};
  \pic at (0.3,0.7) {O2p};  
   \pic at (0.6,0.55) {O2p};
     \pic at (0.9,0.7) {O2p};  
      \node at (1.2,0.625) {$\big)$};
      \end{tikzpicture}
      \equiv
  \begin{tikzpicture}
  \node at (-0.45,0.085) {$\big($};
  \pic at (0,0) {O2SVD};
          \pic at (0.8,0) {O2SVD}; 
      \pic at (0.4,0.17) {O2SVD};  
     \pic at (1.2,0.17) {O2SVD}; 
      \node at (1.45,0.085) {$\big)$};
            \node at (0.6,0.3) {${\cdot}$};
            \node at (0.6,0.35) {${\cdot}$};
              \node at (0.6,0.4) {${\cdot}$};
       \node at (-0.45,0.635) {$\big($};       
  \pic at (0,0.55) {O2SVD};
          \pic at (0.8,0.55) {O2SVD}; 
      \pic at (0.4,0.72) {O2SVD};  
     \pic at (1.2,0.72) {O2SVD}; 
      \node at (1.45,0.635) {$\big)$};
      \end{tikzpicture}.
 \end{equation*}

Now, fixing $\epsilon$ and solving $\epsilon= e^{-\alpha\ell}$ (see Lemma \ref{lemma:DL}), we get $\ell= \frac{1}{\alpha}\log\frac{1}{\epsilon}$ and, by Lemma \ref{lemma:DL}, the operator $DL(H)^\ell$ approximates within $\epsilon$ the ground state projector on operator norm and has bond dimension 	
$$D\le d^{2\ell}=d^{\frac{2}{\alpha}\log\frac{1}{\epsilon}}= \left(\frac{1}{\epsilon}\right)^{\frac{\log d^2}{\alpha}}={\rm poly}\left(\frac{1}{\epsilon}\right).$$

In order to have the required approximation in trace class norm, and to get it beyond frustration free systems, more sophisticated versions of Lemma \ref{lemma:DL} are required \cite{Arad13}, leading to the bounds of Eq.(\ref{eq:bounds-approx}).

However in 2D the analogue problems are quite open. First of all, there is no known solution of the local approximation problem. Second, the best known function associated to the global approximation problem is superpolynomial $f(L)=e^{O(\log^2 L)}$ and, moreover, it can only be guaranteed to work under extra spectral assumptions on the Hamiltonian. Specifically, under the following assumption about the absence of concentration of eigenvalues close to the ground state energy: for each $M>0$, the number of eigenstates with energy lower than $E_0+M$ grows at most polynomially with the system size $L$. 

Three questions arise here which can be seen as variants of the Area Law Conjecture.

\begin{question}
Does there exist a global approximation result in 2D only under the spectral gap assumption?
\end{question}

\begin{question}
Can the function $f(L)$ be taken polynomial in $L$?
\end{question}

\begin{question}
Does there exist a local approximation result in 2D? Is this possible assuming only the spectral gap assumption? Can $g(\epsilon)$ be taken polynomial in $\frac{1}{\epsilon}$?
\end{question}

\section{PEPS as a framework to give formal proofs in cond-mat problems}

The results and questions stated in the previous section point to the informal statement that $PEPS = GS$. This opens the possibility to analyze relevant questions for GS, that are really hard to solve in the case of arbitrary systems, using the framework of PEPS where rigorous mathematical proofs can be found.

An illustrative example is the study of 1D GS that are invariant under symmetries. In particular the question is the following: {\it in how many different ways a group can act as a symmetry in a quantum many-body system?} The inequivalent ways of that action classify the so-called Symmetry Protected Topological (SPT) phases and they are defined formally as follows:

\begin{definition}
Consider two gapped locally interacting Hamiltonians $H_0=\sum_\tau \tau(\mathfrak{h}_0)$ and $H_1=\sum_\tau \tau(\mathfrak{h}_1)$ on a ring $\Lambda$,  supported on local Hilbert spaces $\mathcal{H}_0=\C^{d_0}$ and $\mathcal{H}_1=\C^{d_1}$ respectively and such that they commute with unitary representations $U_0:G\rightarrow  \mathcal{U}(d_0)$, $U_1:G\rightarrow  \mathcal{U}(d_1)$ of a group $G$ (meaning that $[H_i,U_i(g)^{\otimes |\Lambda|}]=0$ for all $g\in G$) respectively. We say that $H_0$ and $H_1$ are in the same SPT phase if there exist another local ancillary Hilbert space $\mathcal{H}_a=\C^{d_a}$ and a locally interacting Hamiltonian $H_\lambda=\sum_\tau \tau(h_\lambda)$ with local Hilbert space $\mathcal{H}=\mathcal{H}_0\oplus  \mathcal{H}_1\oplus  \mathcal{H}_a$ so that
\begin{enumerate}
\item $[0,1]\ni\lambda\mapsto h_\lambda$ is smooth (real analytic) (where $\mathcal{H}_0$ and $\mathcal{H}_1$ are embedded in the corresponding sector of $\mathcal{H}$.)
\item There exists a representation $U_a:G\mapsto \mathcal{U}( d_a)$ so that $H_\lambda$ commutes with $(U_0\oplus U_1\oplus U_a)^{\otimes |\Lambda|}$  for all $\lambda$.
\item The spectral gap of $H_\lambda$ is bounded from below by a constant $c>0$ which is independent of $\lambda$ and the system size $|\Lambda|$.
\end{enumerate}
\end{definition}

It is clear that this definition gives rise to an equivalence relation, the different equivalent classes being the different SPT phases. Then, one can rephrase the question by: \emph{how many SPT phases are there for a given group $G$?}

In order to solve this question in unique GS of local gapped Hamiltonians, one can restrict to the case of injective MPS (and their parent Hamiltonians), see \cref{secclass} for the general case,  that are invariant under the action of a symmetry, i.e. MPS so that 
\begin{equation}\label{eq:MPS-sym}
\ket{\psi_{A}}=U(g)^{\otimes L}\ket{\psi_{A}}.
\end{equation}
It is proven in \cite{Schuch11} that Eq.\eqref{eq:MPS-sym} holds for all $L$ if and only if there exists a projective representation $V_g$ of $G$ acting on the virtual space $\mathcal{M}_{D\times D}$ so that

 \begin{equation} \label{fig:local-sym-MPS}
    \begin{tikzpicture}[baseline=-1mm]
      \pic at (0,0) {tensor};
      \draw (-0.25,0) -- (0.25,0);
      \draw (0,0) -- (0,0.25);
      \filldraw[draw=black, fill=purple] (0,0.15,0) circle (0.04);
      \node[anchor=east] at (0,0.25,0) {$U(g)$};
    \end{tikzpicture} =
    \begin{tikzpicture}[baseline=-1mm]
         \pic at (0,0) {tensor};
      \draw (-0.3,0) -- (0.3,0);
      \filldraw [draw=black, fill=red] (-0.2,0,0) circle (0.04);
      \filldraw [draw=black, fill=red] (0.2,0,0) circle (0.04);
      \node[anchor=north] at (-0.25,0.05) {$\myinv{V_g}$};
      \node[anchor=north] at (0.25,0) {$V_g$};
          \end{tikzpicture},
  \; \forall g\in G.
  \end{equation} 
 
From there, one can prove \cite{Pollmann10, Schuch11,Chen11} that 1D SPT phases in MPS are exactly given by the different non-equivalent projective representations of $G$, which is exactly the second cohomology group $H^2(G,U(1))$.

\subsection{Fundamental theorem in PEPS}

It is clear from the above argument how crucial it is to have a local (single tensor) characterization as in Eq. (\ref{fig:local-sym-MPS}) of the existence of a global symmetry \eqref{eq:MPS-sym}. In fact such local characterization is just a particular case, by fixing $g$ and defining ${B}= U(g) A$, of the following more general question for PEPS:
\begin{question}
What is the relation between two tensors $A$ and ${B}$  that define the same PEPS, i.e. $\ket{\Psi_A}=\ket{\Psi_{B}}$, on a torus $L\times L$ for all possible sizes $L$?
\end{question}

The Fundamental Theorem of MPS \cite{Cirac17A} shows that this happens in 1D if and only if there exists an invertible matrix $Y$ so that 
 \begin{equation} 
    \begin{tikzpicture}[baseline=-1mm]
      \pic at (0,0) {tensor};
      \draw (-0.25,0) -- (0.25,0);
      \draw (0,0) -- (0,0.25);
      \node[anchor=north] at (0,0) {${B}$};
    \end{tikzpicture} =
    \begin{tikzpicture}[baseline=-1mm]
         \pic at (0,0) {tensor};
         \node[anchor=north] at (0,0) {${A}$};
      \draw (-0.3,0) -- (0.3,0);
      \filldraw [draw=black, fill=red] (-0.2,0,0) circle (0.04);
      \filldraw [draw=black, fill=red] (0.2,0,0) circle (0.04);
      \node[anchor=north] at (-0.3,0.05) {$\myinv{Y}$};
      \node[anchor=north] at (0.3,0) {$Y$};
          \end{tikzpicture}.
  \end{equation} 
Note that in \cref{chap:FT} we have proven a Fundamental Theorem for normal PEPS in two and larger dimensions, in particular MPS, describing the same state without the condition of equality for all system sizes. Unfortunately, as opposed to the 1D case, in 2D  the restriction to normal tensors excludes all non-trivial SPT phases. This is why extending the Fundamental Theorem in 2D beyond normal tensors becomes a crucial question to solve (see \cite{Andras18B} for one such extension to the case of so-called quasi-injective PEPS).

On the opposite direction, it is shown in \cite{Scarpa18} that it is undecidable to know whether two general local tensors give rise to the same state for all system sizes in 2D. Therefore, if there is a local characterization of such fact must be an uncomputable (and hence useless) one.  

The big question then is to fill the gap in between these two extremes points: the true but rather incomplete normal case and the undecidable general case:

\begin{question}
Give a Fundamental Theorem in 2D (and higher dimensions) for the largest possible family of PEPS.
\end{question}

The relation between the tensors $A$ and ${B}$ has been investigated so far from the equality of their defining PEPS, nevertheless other conditions can be considered.  One of those could be the approximability of two PEPS in the thermodynamic limit:
\begin{question}
Given $A$ and ${B}$ such that there exists an $\epsilon >0$ and a system size $L_0$ such that for all $L \ge L_0$ 
$$\| \ket{\Psi_{A}} -\ket{\Psi_{{B}}}\|_2\le \epsilon,$$
is there a local relation between both tensors?
\end{question}
In contrast to previous questions, here there are no known results; one first step would be answering the question for normal PEPS.

\section{Miscellanea}

There are many other relevant questions about PEPS that were not formulated in the previous sections due to the need of introducing too specialized prerequisites. In this section we will list a selection of those, with the hope that researchers in the corresponding fields could be attracted to such problems:

\paragraph{Machine Learning.} MPS (and other Tensor Networks such as MERA) have been successfully used numerically in the context of Supervised Machine Learning (ML)  \cite{Stoudenmire}. They lack however an in-depth theoretical analysis. A concrete (relevant) question is the following: 

\begin{question}
Can one write the Rademacher complexity or the Vapnik-Chervonenkis (VC)-dimension for such ML algorithms as a function of the bond dimension?
\end{question}
 
 \paragraph{Computational Complexity.} Part of the difficulty of dealing with PEPS is that, as we saw before, they can encode hard (even undecidable) problems. For some type of problems concerning PEPS, the exact complexity class is known \cite{Schuch07}. In \cite{Scarpa18} it is shown that zero-testing in 2D PEPS is a central question to understand their fundamental limitations and the NP-hardness of that problem is proven (see also \cite{Gharibian}).
 
\begin{question}
Which is the exact complexity class for 2D PEPS zero-testing?
\end{question} 
 
 \paragraph{Topological complexity} The complexity of a state (in particular a PEPS) can also be measured in an operational way by the depth of the quantum circuit required to construct it from a different (usually simpler) state. Indeed, fast (meaning low-depth) convertibility in both directions between different states is the quantum-information-like definition for two states to belong to the same quantum phase (see \cite{Coser} for an in-depth discussion on that). One would expect however that one can always reduce complexity fast. Making this statement rigorous for topologically ordered phases boils down to find low-depth circuits of (noisy) gates that implement dynamically the notion of anyon condensation.  The formal question becomes (see \cite{Coser} for the necessary notions and definitions):
 
 \begin{question}
Is there a low-depth noisy quantum circuit that maps the quantum double phase associated to a finite group $G$ to the one associated to a normal subgroup $H$?. 
\end{question} 

\paragraph{Quantum Cellular Automata.} Quantum Cellular Automata (QCA) are unitary evolutions on a lattice that have a finite propagation cone \cite{QCA-Werner}. By means of the Lieb-Robinson bounds they can be seen as discrete analogues of time-evolutions of locally interacting systems. In \cite{Cirac17B} (see also \cite{Sahinoglu18}) it is shown that 1D translational invariant QCAs correspond exactly with the set of Matrix Product Unitaries MPU (MPOS that are unitary for all system size). This opens the possibility to combine techniques from MPS and QCAs in order to classify the different QCAs up to continuous deformations, as illustrated in \cite{Cirac17B} and \cite{QCA-onsite}.  The question is:

\begin{question}
Which is the exact relation between PEPS and QCA in 2D and higher dimensions? 
\end{question} 

See \cite{Haah} for recent work in this direction.

\newpage \cleardoublepage

\chapter{Conclusions and outlook}

In quantum mechanics, the number of parameters needed to describe a composite system grows exponentially with the number of sites of the system. 
Naturally, strongly-correlated systems deal with large composite systems.
%
The complexity of this scaling is captured by the entanglement pattern. Interestingly, when a gapped Hamiltonian is made of local interactions, the pattern of entanglement in the low-energy sector is restricted. This pattern is captured by tensor network states which arise as a suitable framework to study these systems numerically and analytically.

\

This thesis has the main goal of investigating symmetries in 2D topological tensor networks. For that purpose, we have studied what the allowed freedom in the tensors generating the same tensor network state is. These results are the so-called 'fundamental theorems' proven in \cref{chap:FT} and they give the necessary knowledge to properly study symmetries (actions that leave the states invariant). We focus our study on the family of PEPS describing quantum double models of $G$, the so-called $G$-injective PEPS, and on global on-site symmetries. 

\

The classification of symmetric $G$-injective PEPS is addressed in \cref{chap:classsymGPEPS}. We have obtained a finite number of phases for a $G$-injective PEPS with a global symmetry coming from a finite group $Q$. The classification of phases in that setting is closely related to the theory of group extensions. Remarkably, we provide the maps that appear in group extension theory, called $\phi$ and $\omega$ in this thesis, with physical meaning and we characterize their actions on the $G$-injective PEPS models: both on the ground subspace and on the anyons. 

\

A theoretical classification might in principle produce some non-realizable classes. We solve this issue in our setting by constructing a representative of each class of our classification. We do this in \cref{chap:cond} using the theory of group extensions as well. Moreover, we connect our construction with an interesting physical phenomenon, the so-called anyon condensation, which describes topological phase transitions.

The representatives of each of the phases we construct have a particular form, they are renormalization fixed points, that is,  their parent Hamiltonians are commuting and their ground states have zero correlation length. These properties allow for straightforward detection of the phase which the representative belongs to. But in general, we would like to identify the phase outside renormalization fixed points. To that end, we have proposed a family of gauge invariant quantities and their corresponding order parameters in \cref{detecSF}. These order parameters are formulated as expectation values of a local operator. We have shown that we can distinguish between all studied phases. Our approach does not rely on dimensional reduction, used in previous works, which fails to identify all phases.

The connection between the chapters of this thesis is summarized in the following diagram:

\begin{equation}\label{diathesiscon}
\tikzstyle{mybox} = [draw=blue, fill=gray!20, very thick,
    rectangle, rounded corners, inner sep=5pt, inner ysep=10pt]
\tikzstyle{fancytitle} =[fill=blue, text=white, ellipse]
\begin{tikzpicture}[transform shape]
 
\node [mybox] (C2) at (-1,0){%
    \begin{minipage}[t!]{0.2\textwidth}
    $E$-isometric PEPS, where $E$ is a group extension of $G$ and $Q$
    \end{minipage}
    };

\node [mybox] (C3) at (5,2){%
    \begin{minipage}[t!]{0.2\textwidth}
    $G$-injective PEPS with a global symmetry given by $Q$.
    ($G \triangleleft E$ and $Q= E/G$)
    \end{minipage}
    };
      
    \node [mybox] (C4) at (5,-4.5){%
    \begin{minipage}[t!]{0.2\textwidth}
  Maps $(\phi,\omega)$ that characterize the action of the symmetry on anyons, {\it i.e.} permutation and SF.
    \end{minipage}
    };
    
\node [mybox] (C5) at (-1,-3){%
    \begin{minipage}[t!]{0.2\textwidth}
Local order parameter for the detection of the SF pattern,{\it i.e.} the class of $\omega$
    \end{minipage}
    };
    
\node [mybox] (C6) at (6,-1){%
    \begin{minipage}[t!]{0.2\textwidth}
Characterization of global symmetries in terms of local tensors
    \end{minipage}
    };

\path[->,draw=gray,line width=1mm, -{Triangle[]}]
    (C2) edge[line width=0.742mm] node[  fill=white, anchor=center, pos=0.5] {Chapter 4}
    						    node[ anchor=center, below right , pos=0.5, text width= 1.9cm ] {(via anyon  condensation)} (C3);

\path[->,draw=gray,line width=1mm, -{Triangle[]}]
    (C3) edge[line width=0.742mm] node[fill=white, anchor=center, above, pos=0.6] {Chapter 2 (via FT)}
    						     (C6); 

\path[->,draw=gray,line width=1mm, -{Triangle[]}]
    (C6) edge[line width=0.742mm] node[  fill=white,anchor=center, above, pos=0.5] {Chapter 3}
    						   node[ anchor=center, below right , pos=0.35, text width=3.5cm ] {(via interplay between topology and symmetry)} (C4); 

\path[->,draw=gray,line width=1mm, -{Triangle[]}]
    (C4) edge[line width=0.742mm] node[ anchor=center, above, pos=0.5] {Chapter 5}
    						   (C5); 						    
\path[<->,draw=gray,line width=1mm]
    (C2) edge[line width=0.742mm] node[ fill=white, anchor=center, above, pos=0.4] {Appendix A}
    node[ anchor=center, right , pos=0.5, text width=3cm ] {(relation between  }
    node[ anchor=center, right , pos=0.6, text width=2cm ] { $(\phi,\omega)$ and $E$ ) }
    						   (C4); 						    
						   
\end{tikzpicture}
\end{equation}

\

This thesis contributes to the understanding of symmetry-enriched topological phases focusing on their descriptions in terms of tensor network states. The PEPS formalism allows us to locally encode the main properties of the models (like topological order, symmetries and their interplay) in the tensors. We have used that encoding to classify, construct and detect, see diagram \eqref{diathesiscon}, some class of symmetry-enriched topological phases in PEPS.

\

From our work, some open questions arise. The classification of the SF effect lacks a study of its effect on the entanglement spectrum in the anyonic sector. Based on the intuition on how a global symmetry modifies the entanglement spectrum in 1D SPT phases, see \cite{Pollmann10}, we would expect some analogy in the 2D case. This is because in the 1D case the gauge transformation of the MPS tensors under the symmetry allows for deduction of a degeneracy in the entanglement spectrum. It is also interesting to study the analog of string-orders for MPS, see \cite{PerezGarcia08B, Pollmann12}, in PEPS hosting a non-trivial SF pattern.

SET phases with PEPS can be potentially applied in quantum computation. Tensor network states have been used to perform measurement-based quantum computation, see \cite{Gross07,Stephen17,Raussendorf17,Raussendorf19}, using the gauge transformation of the symmetry in SPT phases. Topological order is used for topological quantum computing, see \cite{Nayak08,Freedman03}. Then, the construction of SET in PEPS suggests that some interplay between measurement-based and topological quantum computation can be achieved in this framework. There are already some works in this direction like Ref.\cite{Teo16} and Ref.\cite{Delaney18} that study how SET phases can be used for topological quantum computation in the framework of category theory.

An interesting point is the numerical and experimental implementation of our order parameter. 
For the experimental implementation, it would be interesting to generalize Ref. \cite{Elben19}, which proposes a method to implement in spin Hamiltonian the 1D SPT order parameters using randomized measurements.

\

Our work allows for some generalizations which are left for future work. First of all, one could consider continuous symmetry groups. It is not clear whether our approach for constructing the representative of each phase can be applied in this situation since the extension group $E$ would also be continuous and then the $E$-injective PEPS, see \cref{chap:cond}, is not well-defined. One can also consider more general classes of PEPS with topological order, like the so-called MPO-injective PEPS \cite{Sahinoglu14}, and obtain a fundamental theorem to characterize the corresponding topological phases under symmetries.

Finally, this work has been focused on bosonic phases of matter but fermionic systems also allow for a description in terms of PEPS \cite{Corboz10}. In particular, PEPS with fermionic topological order have been proposed in \cite{Wille17,Bultinck17} but fermionic SET phases in this formalism have not been constructed yet.

\begin{appendices}  
\addappheadtotoc

\chapter{Group extensions}\label{ap:ext} 

An extension of a group $Q$ is a group $E$, together with a surjective homomorphism $\pi:E\rightarrow Q$. The kernel of  $\pi$ is a normal subgroup $G$ of $E$. We say that the group $E$ with the homomorphism $\pi$ is an extension of $Q$ by $G$ \cite{Adem04} and it is encoded in the following short exact sequence:
$$1\rightarrow  G  \stackrel{i}{\rightarrow} E\stackrel{\pi}{\rightarrow} Q \rightarrow 1,$$
where $i$ is the inclusion map and $Q$ is isomorphic to the quotient group $E/G$.\\

In the case where $G$ is an abelian group, given a group extension $E$ of $Q$ by $G$ and the homomorphism $\pi$, two maps can be defined: (i) a homomorphism $\phi:Q  \rightarrow  {\rm Aut}(G)$ and (ii) a $2$-cocycle $\omega:Q\times Q\rightarrow G$ which satisfies 
\begin{equation} \label{concocycle}
\omega(k,q)\omega(kq,z)=\phi_k(\omega(q,z))\omega(k,qz).
\end{equation}
These maps are defined as follows. Given $k$, we pick a pre-image $\epsilon_k\in E$ such that $\pi(\epsilon_k)=k$, and we construct $\phi_k:g\mapsto \epsilon_k g \epsilon^{-1}_k$ and $\omega(k,q)=\epsilon_k \epsilon_q \epsilon^{-1}_{kq}$. There is some arbitrariness in the choice of the pre-image: $\epsilon_k$ and $g\epsilon_k $ are mapped to $k$ under $\pi$ for any $g \in G$.  This does not modify the map $\phi_k$ but this arbitrariness modifies the cocycle as follows
\begin{equation} \label{relcocycle}
\omega(k,q)\to g_k \phi_k(g_q)g^{-1}_{kq}\omega(k,q),
\end{equation}
where $g_k,g_q,g_{kq}\in G$ are the posible choices. The second cohomology group $H_\phi^2(Q,G)$ is defined as the quotient of the $2$-cocycles satisfying Eq.(\ref{concocycle}) by the $2$-coboundaries of the form $g_k \phi_k(g_q)g^{-1}_{kq}$, that is, we identify cocycles that are related by Eq.(\ref{relcocycle}). The elements of $E$ are decomposed uniquely as $g\epsilon_q$ for some $g\in G$ and $q\in Q$. The multiplication rule of $E$ can be written as
$$ (g\epsilon_q)\cdot (h\epsilon_k)= g\phi_q(h)\omega(g,h) \epsilon_{qk}.$$

As a set, $E$ can be expressed as the cartesian product $G\times Q$ with the rule for multiplication:
\begin{equation}\label{eq:carprod}
(g,q)(b,h)=(g\phi_q(h)\omega(q,k),qk).
\end{equation}
The product $(g,e)(h,e)=(gh,e)$ generates the normal subgroup $G$ and the product $(e,q)(e,k)=(\omega(q,k),qk)$ generates the group $Q$ after quotienting by $G$. The inverse of an element is $(g,q)^{-1}=(\phi_{q^{-1}}[(g\omega(q,q^{-1}))^{-1}],q^{-1})$.\\
If there exists a homomorphism $\phi:Q\rightarrow E$ such that $\pi \circ \phi =  \id_Q$, we say that the group extension {\it splits} and it is associated with the semidirect product $E=G \rtimes_{\rho} Q$. If such a homomorphism $\phi$ exists, the cocycle $\omega$ is trivial, i.e. $H^2(Q,G)=1$.
Two extensions are said to be equivalent if there is a isomorphism $\sigma:E\rightarrow E'$ such that the following diagram commutes:
\begin{equation} \label{eq:extdia}
\begin{array}{lllllllll}
1 & \longrightarrow  &G &\stackrel{i}{\longrightarrow} & E& \stackrel{\pi}{\longrightarrow}
& Q& \longrightarrow  & 1 \\
&  & \downarrow \,{\id}&  & \downarrow \,\sigma &  & \downarrow \,{\id} &  &  \\
1 & \longrightarrow  & G & \stackrel{i'}{\longrightarrow} &E' & \stackrel{\pi'}{\longrightarrow}
& Q & \longrightarrow  & 1.
\end{array}
\end{equation}
If $E$ and $E'$ come from a commutative diagram as Eq. (\ref{eq:extdia}) then, they are isomorphic as groups. However it is possible that the diagram (\ref{eq:extdia}) does not commute even though $E\equiv E'$ constructed with the same groups $Q$ and $G$. Hence equivalence of extensions is a more subtle notion than group equivalence. An important result is that if two extensions are equivalent then the action $Q\rightarrow {\rm Aut}(G)$ is the same for both extensions, and the cocycles describing the two extensions are in the same class in $H^2(Q,G)$.\\

To deal with the non-abelian case, two maps $\omega$ and $\phi$ can also be constructed as in the abelian case. But the map $\phi:Q\rightarrow {\rm Aut}(G)$ need not be a group homomorphism now. In fact it satisfies
$$\phi_k \circ \phi_q={\rm Inn}(\omega(k,q)) \circ \phi_{kq},$$
where Inn$(g)$ denotes the inner automorphism $h\mapsto g h g^{-1}: g,h \in G$. The map $\omega(k,q)$ is defined as in the abelian case and also satisfies the cocycle condition (\ref{concocycle}). However, the group homomorphism $\phi$ now maps elements of $Q$ to $ {\rm Out}(G) \equiv {\rm Aut}(G)/{\rm Inn}(G)$. The extension group equivalence is again defined as the commutation of the diagram (\ref{eq:extdia})  and is classified by $\phi$ and the cocycle $\omega$. It can be shown that the group $H_\phi^2(Q,Z(G))$ acts freely and transitively over the set of extensions, where $Z(G)$ is the center of $G$ \cite{Adem04}. In particular, this implies that $|H_\phi^2(Q,Z(G))|$ is equal to the number of the inequivalent cocycles. The elements of $H_\phi^2(Q,Z(G))$ are constructed as $c(q,k)=\omega'(q,k)\omega^{-1}(q,k)$, \textit{i.e.} the difference between cocycles, so that the non-trivial element maps one class into another. That is, the difference between  cocycles of non-abelian groups is given by a second cohomology group that classifies the general cocycles of the abelian groups. \\
\chapter{Projective representations}\label{Ap:projrep}

A projective representation of a group $Q$ is a homomorphism from $Q$ to $GL(n)$ up to a phase factor:
$$
V_qV_k=\rho(q,k) V_{qk};\; \forall q,k\in Q,
$$
where $\rho(q,k)\in U(1)$. Associativity of group multiplication implies the so called cocycle condition: $ \rho(q,k) \rho(qk,p)=  \rho(k,p)  \rho(q,kp)$. A change of basis of the vector space where $\{V_q\}$ act does not affect $\{\rho(q,k)\}$ but a phase redefinition $V_q\to V'_q \equiv \nu_q V_q$ induces the modification 
\beq\label{eq:equiv-rho}
\rho(q,k) \to \rho'(q,k) \equiv \nu^{-1}_q \nu^{-1}_k \nu_{qk}\rho(q,k).
\eeq
Eq.\eqref{eq:equiv-rho} is taken to be the equivalent relation to classify cocycles resulting in the group $H^2(Q,U(1))$.  An example of projective representation of $\mathbb{Z}_2\times \mathbb{Z}_2=\{x^2=y^2=z^2=e, xy=z \}$ is given by the Pauli matrices.
\

Given $Q$ there exists a group $E$ such that a projective representation $V$ of $Q$ can be expressed as a linear representation $U$ of $E$. $E$ is a so-called representation group of $Q$ \cite{Simon96}. The relation of the two groups is the extension $1\rightarrow G \rightarrow E \rightarrow Q \rightarrow1$ where $G$ is a group that satisfies $U_g\propto \id;\;\forall g \in G$. For the previous example $E=Q_8=\{ \pm 1,\pm i,\pm j,\pm k\}$ with the representation $U_i=i\sigma_z$, $U_j=-i\sigma_y$, $U_k=-i\sigma_x$.

\

Considering the definition of $\omega$ in the previous section
$$
\epsilon_k \epsilon_q =\omega(k,q) \epsilon_{kq},
$$
one can take a faithful representation $W$ of $E$ to realize this equation. The representation $V_q\equiv W_{\epsilon_q}$ can be seen as a \emph{projective} representation of $Q$:
$$
V_k V_q =W_{\omega(k,q)} V_{kq},
$$
where the projective factors are matrices: the representation $W$ restricted to the elements of $G\subset E$.
For example the case $G=Q=\mathbb{Z}_2$ realizing $\omega(-1,-1)=-1$ is $E=\mathbb{Z}_4$ where one can take the following representation:
$$
V_{-1}=
  \left( {\begin{array}{cccc}
   0 & 1 & 0&0 \\
 0 & 0 & 1&0 \\
  0 & 0 & 0&1 \\
   1 & 0 & 0&0 
  \end{array} } \right), \;
  W_{-1}=
  \left( {\begin{array}{cccc}
   0 & 0 & 1&0 \\
 0 & 0 & 0&1 \\
  1 & 0 & 0&0 \\
   0 & 1 & 0&0 
  \end{array} } \right).
  $$

\chapter{Creation of Dyons}\label{ap:dyonic}

A composite dyon-antidyon excitation is created by acting with certain combination of operators, the so-called ribbon operators \cite{Kitaev03, Bombin08}, over the ground state of $\mathcal{D}(G)$. This ground state can be constructed using $G$-isometric PEPS tensors \cite{Schuch10}. Here we have obtained the virtual representation of the ribbon operator corresponding to the composite dyon-antidyon excitation. In order to do so we apply that operator over the physical indices of a tensor network. Analyzing the virtual indices of the boundary we obtain the desired equivalence between physical and virtual operator.

\begin{figure}[ht!]
\begin{center}
\includegraphics{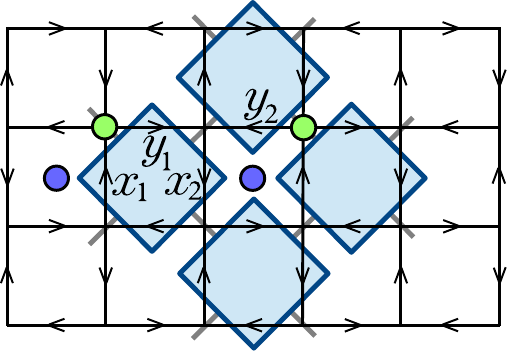}
\caption{The operator of Eq.\eqref{eq:ribbop} acts on four adjacent edges, denoted as $x_1,y_1,x_2,y_2$, and involves three vertices and three plaquettes. This operator depends on the orientation of each edge and we take the one represented by the arrows. The green and blue points identify the vertices and plaquettes excited respectively.}
\label{orientation}
\end{center}
\end{figure}

The ribbon operator we choose acts in four edges of the square lattice with the orientation illustrated in \cref{orientation}. The ribbon operator of a dyon can be written as follows \cite{Bombin08}:
\begin{equation}\label{eq:ribbop}
\mathcal{O}_{\alpha,h}\equiv \frac{d_{\alpha}}{|N_h|} \sum_{n\in N_h} \bar{\chi}_\alpha(n)\sum_{i,j=1}^{\kappa}  F^{\myinv{h}_i,k_in\myinv{k}_j},
\end{equation}
where the $k_i$'s are the representatives of the left cosets of $G$ by $N_h$ and the operator $ F^{h,g}$ acts over the four chosen edges as follows (see Fig. \ref{orientation} for clarification):
$$F^{h,g}\ket{x_1,y_1,x_2,y_2}= \delta_{g,x_1\myinv{x}_2}\ket{x_1\myinv{h},y_1,\myinv{y}_1h y_1x_2,y_2}. $$

The ground state of the quantum double of $G$ can be constructed with the following tensor \cite{Schuch10}:
$$K= \sum_{l,r,s,p\in G}\ket{p\myinv{l},l\myinv{r},r\myinv{s},s\myinv{p}} |p,l)(r,s|.$$
This tensor has the following virtual invariance:
\begin{equation} \label{eq:Rregular}
K\big[R_g\otimes R_g\otimes R^{\dagger}_g\otimes R^{\dagger}_g\big]_v =  
\sum_{l,r,s,p\in G}\ket{p\myinv{l},l\myinv{r},r\myinv{s},s\myinv{p}} \big[R_g\otimes R_g |p,l)(r,s| R^{\dagger}_g \otimes R^{\dagger}_g\big]=K \; \forall g \in G,
\end{equation}
which endows the state with topological properties. We now express the edges involved in the action of the operator of Eq.(\ref{eq:ribbop}) in its tensor network representation:
\begin{equation}\label{eq:physten}
 \mathcal{P}(K)\equiv \sum_{l,r,s,p,t\in G}\ket{p\myinv{l},l\myinv{r},r\myinv{s},s\myinv{p},t\myinv{r}} |p,l)(s,t|,
\end{equation}
for a driagrammatic representation see Fig. \ref{fig:indices}.
\begin{figure}[ht!]
\begin{center}
\includegraphics[scale=0.8]{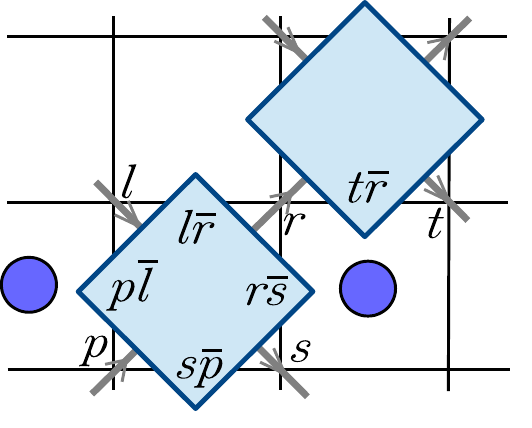}
\caption{Tensor network representation of the physical system involved in the creation of a dyon. The blue dots are depicted for comparison with Fig. \ref{orientation} and the virtual index $r$ is depicted for clarification.}
\label{fig:indices}
\end{center}
\end{figure}
The creation operator of the dyon acts on the tensor network representation as follows:
\begin{equation*}
\frac{|N_h|}{ d_{\alpha}}\mathcal{O}_{\alpha}  \mathcal{P}(K)=
 \sum_{\substack{ n\in N_h   \\  i,j=1,\cdots,\kappa  \\  l,r,s,p,t\in G }} \delta_{k_in\myinv{k}_j,l\myinv{t}} \;\bar{\chi}_\alpha(n) \; \ket{p\myinv{l}h_i,l\myinv{r},r\myinv{l}\; \myinv{h}_il\myinv{s},s\myinv{p},t\myinv{r}}  |p,l)(s,t|.
 \end{equation*}
We can now relabel the indices ($s'=s\myinv{l} h_i l$ and $p'=p\myinv{l} h_i l$) to obtain the action on the virtual d.o.f.:
\begin{equation}\label{eq:creadyon}
\sum_{\substack{  i,j=1,\cdots,\kappa  \\  l,r,s,p,t\in G }} \bar{\chi}_\alpha(\myinv{k}_i l \myinv{t} k_j)  \ket{p\myinv{l},l\myinv{r},r\myinv{s},s\myinv{p},t\myinv{r}}  |p\myinv{l}\;\myinv{h}_il,l)(s\myinv{l}\;\myinv{h}_il,t| 
=\sum_i\mathcal{F}_i \circ\mathcal{C}_i[ \mathcal{P}(K)],
\end{equation}
 where the operator $\sum_i\mathcal{F}_i \circ\mathcal{C}_i$ acts purely on the virtual d.o.f of $\mathcal{P}(K)$ and its components are defined as follows:
\begin{align}\label{eq:opedyon}
\mathcal{F}_i\big[| p,l)(s,t| \big]\equiv & \sum_{g\in G} R^{\dagger}_{\myinv{g}\;\myinv{h}_i g}\otimes |g)(g| \big[| p,l)(s,t|\big]R_{\myinv{g}\;\myinv{h}_i g}\otimes  \id,\notag \\
\mathcal{C}_i \big[| p,l)(s,t| \big]\equiv
&\sum_{n,m\in N_h} \bar{\chi}_\alpha(n\myinv{m})\;  \id\otimes|k_in)(k_i n|\big[| p,l)(s,t|\big]  \id\otimes \sum_j^\kappa |k_jm)(k_j m|, \notag
\end{align}

and where $\mathcal{F}_i$ and $\mathcal{C}_i$ can be regarded as the flux and charge part of the dyon respectively.
We can represent diagramatically the virtual operator corresponding to Eq.(\ref{eq:creadyon}) as follows:
$$\sum_{\substack{  i=1,\cdots,\kappa  \\ g\in G \\n,m\in N_h }} \bar{\chi}_\alpha(n\myinv{m})\;
\parbox[c]{0.4\textwidth}{ \includegraphics[scale=0.9]{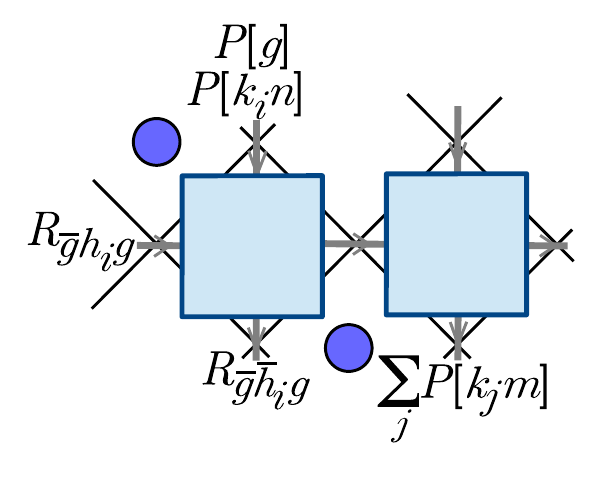}},$$
where $P[a]=\ket{a}\bra{a}$. Let us simplify this expression; $P_g P_{k_i n}= P_{k_i n}\delta_{g, k_i n}$  and also $\myinv{g} h_i g=h$, then the virtual operator is
$$\sum_{n,m\in N_h } \bar{\chi}_\alpha(n\myinv{m}) \sum_i P_{k_i n}\otimes R_h\otimes R_{\myinv{h}}\otimes \sum_j P_{k_j m},$$
which can be represented in the following forms using the $G$-invariance of the tensors:
\begin{equation*}    
 \begin{tikzpicture}
   \draw (-1,0)--(1,0);  
  \draw (-0.5,-0.5)--(-0.5,0.5);
    \draw (0.5,-0.5)--(0.5,0.5);   
 	\filldraw[draw=black,fill=orange] (-0.6,0.2) rectangle (-0.4,0.4); 
	     \node[anchor=west] at (-0.5,0.4) {$P_{k_in}$};
	 	\filldraw[draw=black,fill=orange] (0.6,-0.2) rectangle (0.4,-0.4);  
		 \node[anchor=west] at (0.5,-0.3) {$P_{k_j m}$};     
	\filldraw[draw=black,fill=blue]  (-0.5,-0.1) circle (0.05);
	 \node[anchor=west] at (-0.5,-0.2) {$\myinv{h}$};
	\filldraw[draw=black,fill=blue]  (-0.7,0) circle (0.05);
	 \node[anchor=south] at (-0.7,0) {${h}$};
\end{tikzpicture}
=
 \begin{tikzpicture}
   \draw (-1,0)--(1,0);  
  \draw (-0.5,-0.5)--(-0.5,0.5);
    \draw (0.5,-0.5)--(0.5,0.5);   
 	\filldraw[draw=black,fill=orange] (-0.6,0.2) rectangle (-0.4,0.4); 
	     \node[anchor=west] at (-0.5,0.4) {$P_{k_in}$};
	 	\filldraw[draw=black,fill=orange] (0.6,-0.2) rectangle (0.4,-0.4);  
		 \node[anchor=west] at (0.5,-0.3) {$P_{k_j m}$};     
	\filldraw[draw=black,fill=blue]  (-0.5,0.1) circle (0.05);
	 \node[anchor=east] at (-0.5,0.2) {$\myinv{h}$};
	\filldraw[draw=black,fill=blue]  (0,0) circle (0.05);
	 \node[anchor=north] at (0,0) {${h}$};
\end{tikzpicture}
=
 \begin{tikzpicture}
   \draw (-1,0)--(1,0);  
  \draw (-0.5,-0.5)--(-0.5,0.5);
    \draw (0.5,-0.5)--(0.5,0.5);   
 	\filldraw[draw=black,fill=orange] (-0.6,0.2) rectangle (-0.4,0.4); 
	     \node[anchor=west] at (-0.5,0.4) {$P_{k_in}$};
	 	\filldraw[draw=black,fill=orange] (0.6,-0.2) rectangle (0.4,-0.4);  
		 \node[anchor=west] at (0.5,-0.3) {$P_{k_j m}$};     
	\filldraw[draw=black,fill=blue]  (-0.5,0.1) circle (0.05);
	 \node[anchor=east] at (-0.5,0.2) {$\myinv{h}$};
	 \filldraw[draw=black,fill=blue]  (0.5,0.1) circle (0.05);
	 \node[anchor=east] at (0.5,0.2) {$\myinv{h}$};
	\filldraw[draw=black,fill=blue]  (0.7,0) circle (0.05);
	 \node[anchor=south west] at (0.7,0) {${h}$};
	 	\filldraw[draw=black,fill=blue]  (0.5,-0.1) circle (0.05);
	 \node[anchor= east] at (0.5,-0.1) {${h}$}; 
\end{tikzpicture}.
 \end{equation*} 
An analogous construction can be given for longer ribbon operators which create pairs dyon-antidyon separated for longer distancies. We can consider only part of the pair: an isolated dyon. This operator would correspond with a string of $R_g$ operators ended with the operator \cref{dyonend}.

We point out that the tensor $K$ has the invariance described in \cref{eq:Rregular} due to the chosen clockwise direction of the edges contained in $K$ (see \cref{orientation}) \cite{Schuch10}. A counterclockwise direction would give rise to a tensor with the virtual invariance represented by $L_g$ instead of $R_g$, which would be unitary equivalent to the $G$-isometric PEPS. This relation connects the tensor $K$, obtained in \cite{Schuch10} and used in this section, with the convention used through the main text.

\chapter{Time reversal symmetry and reflexion symmetry}\label{symTRS}

Let us consider a $G$-injective {PEPS} invariant under time reversal symmetry $ \mathcal{T} |\psi\rangle=  |\psi\rangle$ which is realised by an anti-unitary global operator $\mathcal{T}=U_{\mathcal{T}}^{\otimes N}K$, where $U_{\mathcal{T}}$ is unitary and $K$ is the complex conjugation operator and we will denote its action as $K v= v^*$ where $^*$ means complex conjugation. The local transformation of the tensors is
\begin{equation}\label{eq:TRS}
U_{\mathcal{T}}A^*=A(\myinv{V}_{\mathcal{T}}\otimes V_{\mathcal{T}}\otimes \myinv{V}_{\mathcal{T}}\otimes V_{\mathcal{T}}).
\end{equation}
The condition $ \mathcal{T}^2=\id$ implies $A= U_{\mathcal{T}}(U_{\mathcal{T}}A^*)^*= A ({\myinv{V}_{\mathcal{T}}}^* \myinv{V}_{\mathcal{T}}\otimes V_{\mathcal{T}}V^*_{\mathcal{T}}\otimes {\myinv{V}_{\mathcal{T}}}^* \myinv{V}_{\mathcal{T}}\otimes V_{\mathcal{T}}V^*_{\mathcal{T}})$
\cite{Molnarinprep}. Therefore, we conclude that $$V_{\mathcal{T}}V^*_{\mathcal{T}}\equiv\omega_{\mathcal{T}}\in G.$$ We point out that $A^*$ is a $G$-injective tensor with respect to the conjugated representation of $G$ acting on $A$ that we will denote as $g^*\in G^*$. Using Eq.(\ref{eq:TRS}) and the corresponding $G$-injectivity of both $A^*$ and $A$,  the following holds $$V_{\mathcal{T}}g^*\myinv{V}_{\mathcal{T}}\equiv\phi_{\mathcal{T}}(g^*)\in G\;\Rightarrow \phi_{\mathcal{T}}\circ \phi^*_{\mathcal{T}}=\tau_{\omega_{\mathcal{T}}},$$ where $\phi^*_{\mathcal{T}}(g)\equiv V^*_{\mathcal{T}}g{\myinv{V}_{\mathcal{T}}}^*\in G^*$. We notice the difference with an  internal $\mathbb{Z}_2=\{1,k\}$ symmetry where we would have obtained $V^2_k\in G$ and $V_k g \myinv{V}_k\in G$. If the representation of $G$, in some basis, is real, then $U_{\mathcal{T}}A=A(\myinv{V}_{\mathcal{T}}\otimes V_{\mathcal{T}}\otimes \myinv{V}_{\mathcal{T}}\otimes V_{\mathcal{T}})$. This is the case for $G$-isometric {PEPS} with the left regular representation in the group algebra basis: $L_g=\sum_{h\in G}|gh\rangle\langle h|\Rightarrow  L^*_g=L_g$ which implies that $ A=A^*$.
From Eq.(\ref{eq:TRS}), it is clear that the operator $V_{\mathcal{T}}$ is defined up to an element of $G$: $V_{\mathcal{T}}\sim gV_{\mathcal{T}}$ so%
$$ \omega'_{\mathcal{T}} = g \phi_{\mathcal{T}} (g^*) \omega_{\mathcal{T}}. $$%
We can define recursively the coefficient of the $m$-power of $\omega'_{\mathcal{T}}$: $${\omega'_{\mathcal{T}}}^m = h_m \omega_{\mathcal{T}}^m,$$ where $ h_m=h_1\tau_{\omega_{\mathcal{T}}} (h_{m-1})$  and $h_1= g\phi_{\mathcal{T}}(g^*)$. Given a finite group $G$ we are looking for an $m<|G|$ such that $h_m=e\Rightarrow {\omega'_{\mathcal{T}}}^m={\omega_{\mathcal{T}}}^m$. For the toric code, $G=\mathbb{Z}_2=\{1,g\}$ the two phases are distinguished by $ \omega_{\mathcal{T}}=\{1,g\}$, which is gauge-invariant ($m=1$), and it follows that in both cases $\phi_{\mathcal{T}}(g^*)=g$. The phase with $ \omega_{\mathcal{T}}=g$ corresponds to a non-trivial symmetry fractionalization of the charge and it is equivalent to the braiding with the $m$ particle, resulting in a $-1$ sign. It is left for future work to understand how these phases can be distinguished for any group $G$.

We consider now the case of a $G$-injective PEPS invariant under reflection with respect to a horizontal line. This symmetry is realized at the level of tensors by
$$ U_\sigma A \pi=A(\myinv{V}\otimes W\otimes V\otimes \myinv{W}),$$
where $U_\sigma$ is a transposition of the blocked sites of the tensors, $\pi$ is the horizontal flip operator in the virtual level (interchange plus transposition) and $V$ is the virtual operator acting on the horizontal part. Notice that we are assuming a translational invariance under blocked tensors. If we apply again another horizontal reflection, it follows that
\begin{align*}
A= & \;U_\sigma A(\myinv{V}\otimes W\otimes V\otimes \myinv{W}) \pi\\
 =& \; [U_\sigma A\pi](\myinv{V}\otimes (\myinv{W})^T \otimes V\otimes W^T )\\
 =&  \; A(V^{-2}\otimes W(\myinv{W})^T \otimes V^2\otimes \myinv{W} W^T ),
 \end{align*}
which implies that $V^2=W(\myinv{W})^T\in G$ .
\newpage \cleardoublepage 

\end{appendices}

\bibliographystyle{acm}
\bibliography{Bibliografia}

\end{document}